\documentclass[12pt,a4paper]{article}
\usepackage[latin1]{inputenc}
\usepackage[english]{babel}
\usepackage{amsmath,amsthm,times,fullpage}
\usepackage[toc,page]{appendix}
\usepackage{wrapfig,fancybox,makecell}
\usepackage{mwe,caption,mathtools}
\usepackage{amssymb,amsfonts,etoolbox}
\usepackage{relsize,proof,dsfont,stmaryrd}
\usepackage{algorithmic,shadow,empheq}
\usepackage{yfonts,tikz,pgf,shadow,pifont,datetime,xcolor}
\usepackage{graphicx,color}
\usepackage{xcolor,longtable,multirow}
\usepackage{wasysym,marvosym,setspace}
\usepackage{enumerate,mathrsfs,adjustbox,comment}
\usepackage[format=plain,
            font=it]{caption}
\usetikzlibrary{shadows}
\usetikzlibrary{scopes,shapes,snakes,arrows}

\numberwithin{equation}{section}

\renewcommand\le{\leqslant}
\renewcommand\ge{\geqslant}
\def\ve{\varepsilon}
\def\tr#1{\lfloor #1\rfloor}
\def\ltr#1{\bigl\lfloor #1\bigr\rfloor}
\def\cl#1{\lceil #1\rceil}
\def\lcl#1{\bigl\lceil #1\bigr\rceil}

\def\lpa#1{\bigl({#1}\bigr)}
\def\Lpa#1{\Bigl({#1}\Bigr)}
\def\llpa#1{\biggl({#1}\biggr)}

\def\eqtext#1{\quad\text{#1}\quad}

\newcounter{fg}
\newtheorem{thm}{Theorem}[section]
\newtheorem{cor}[thm]{Corollary}
\newtheorem{lemma}[thm]{Lemma}
\newtheorem{prop}[thm]{Proposition}
\newtheorem*{problem}{Problem}
\theoremstyle{definition}
\newtheorem{definition}[thm]{Definition}   
\newtheorem{rem}[thm]{Remark}   

\newtheorem{example}[thm]{Example}
\AtEndEnvironment{example}{\null\hfill\qedsymbol}
\newtheorem{remark}[thm]{Remark}
\let\oldrem\remark
\renewcommand{\remark}{\oldrem\normalfont}
\AtEndEnvironment{remark}{\null\hfill\qedsymbol}

\newenvironment{romenumerate}[1][0pt]{%
\addtolength{\leftmargini}{#1}\begin{enumerate}
 }{\end{enumerate}}

\newcounter{thmenumerate}
\newenvironment{thmenumerate}
{\setcounter{thmenumerate}{0}%
 \def\item{\par% \ifnum\thethmenumerate=0\else\par\fi %\noindent\fi
 \refstepcounter{thmenumerate}\textup{(\roman{thmenumerate})\enspace}}
}
{}

\newenvironment{alphenumerate}[1][5pt]{%
\addtolength{\leftmargini}{#1}\begin{enumerate}
\addtolength{\itemsep}{\smallskipamount}

\setlength\parindent{10pt} %works
 }{\end{enumerate}}

\newcommand\A{\textup{\textrm{A}}}

\newcommand\Ieven{\mathbf{1}_{n \text{ is even}}}
\newcommand\Iodd{\mathbf{1}_{n \text{ is odd}}}

\newcommand\dd{\,\mathrm{d}}
\newcounter{CC}
\newcommand{\CC}{\stepcounter{CC}\CCx} %new constant C_i
\newcommand{\CCx}{C_{\arabic{CC}}}     %repeats the last C_i
\newcounter{cc}
\renewcommand\rho{\varrho}
\newcommand\xpar[1]{(#1)}
\newcommand\bigpar[1]{\bigl(#1\bigr)}
\newcommand\Bigpar[1]{\Bigl(#1\Bigr)}
\newcommand\biggpar[1]{\biggl(#1\biggr)}
\newcommand\lrpar[1]{\left(#1\right)}

\newcommand\lrsqpar[1]{\left[#1\right]}

\newcommand\bigabs[1]{\bigl\lvert#1\bigr\rvert}

\newcommand{\refT}[1]{Theorem~\ref{#1}}

\newcommand{\refC}[1]{Corollary~\ref{#1}}

\newcommand{\refL}[1]{Lemma~\ref{#1}}

\newcommand{\refP}[1]{Proposition~\ref{#1}}
\newcommand{\refR}[1]{Remark~\ref{#1}}

\newcommand\refE[1]{Example \ref{#1}}
\newcommand\refEs[1]{Examples \ref{#1}}
\newcommand{\refS}[1]{Section~\ref{#1}}

\newcommand\refApp[1]{Appendix~\ref{#1}}
\newcommand\refTab[1]{Table~\ref{#1}}
\newcommand\ga{\alpha}
\newcommand\gb{\beta}
\newcommand\gd{\delta}
\newcommand\gD{\Delta}
\newcommand\gf{\varphi}
\newcommand\gam{\gamma}

\newcommand\gl{\lambda}
\newcommand\gL{\Lambda}

\newcommand\gs{\sigma}

\newcommand\gth{\theta}

\newcommand\eps{\varepsilon}

\newcommand\para[1]{\ensuremath{\{#1\}}}

\newcommand\oi{\ensuremath{[0,1]}}
\newcommand\ceil[1]{\lceil#1\rceil}
\newcommand\floor[1]{\lfloor#1\rfloor}
\newcommand\lrfloor[1]{\left\lfloor#1\right\rfloor}
\newcommand\frax[1]{\{#1\}}
\newcommand\qw{^{-1}}

\newcommand\qq{^{1/2}}

\newcommand\Holder{H\"older}
\newcommand\ii{\mathrm{i}}
\renewcommand\ii{i}

\newcommand\pfitemx[1]{\par#1:}
\newcommand\pfitemref[1]{\pfitemx{\ref{#1}}}

\newcommand\ntoo{\ensuremath{n\to\infty}}
\newcommand\bbR{\mathbb R}
\newcommand\bbC{\mathbb C}
\newcommand\bbZ{\mathbb Z}

\newcommand\BV{\mathrm{BV}}
\newcommand\HH{\mathrm{H}}
\newcommand\gab{{\ga,\gb}}
\newcommand\XLambda{\widehat\Lambda}
\newcommand\hP{\widehat P}
\newcommand\rchi{\rho+\chi}
\newcommand\richi{\rho-1+\chi}
\newcommand\Res{\operatorname{Res}}
\newcommand\FV{_{\mathrm{fin}}}
\newcommand\sqrthalf{\frac{\sqrt2}{2}}
\newcommand\xf{\widetilde{f}}
\newcommand\url{\texttt} %%%SJ
\newcommand\Tstrut{\rule{0pt}{2.6ex}}         % = `top' strut
\newcommand\Bstrut{\rule[-1.0ex]{0pt}{0pt}}   % = `bottom' strut

\usepackage[pdftex]{hyperref}  %%%SJ
\hypersetup{plainpages=True, pdfstartview=FitV,
pageanchor=false, colorlinks=true,linkcolor=blue!50,citecolor=blue!50}

\definecolor{dg}{rgb}{0,0.3,0}
\definecolor{lightblue}{rgb}{0.798,0.898,0.996}
\definecolor{gold}{rgb}{1,0.84,0.}

\title{Identities and periodic oscillations of divide-and-conquer
recurrences splitting at half\footnote{%
The work of the first author was partially supported by
National Science and Technology Council under the Grant
MOST-108-2118-M-001-005-MY3, and part of it was carried out while he
was visiting Department of Mathematics, Uppsala University; he thanks
the Department for hospitality and support. Part of the work of the
second author was carried out during visits to the Isaac Newton
Institute for Mathematical Sciences (EPSCR Grant Number EP/K032208/1)
and was partially supported by a grant from the Simons Foundation,
and a grant from the Knut and Alice Wallenberg Foundation; he thanks
these for hospitality and support.}}

\author{Hsien-Kuei Hwang \\
    Institute of Statistical Science, \\
    Academia Sinica\\
    Taipei 115\\
    Taiwan
\and Svante Janson\\ 
    Department of Mathematics\\ 
    Uppsala University\\ 
    Uppsala\\ 
    Sweden
\and Tsung-Hsi Tsai \\
    Institute of Statistical Science \\
    Academia Sinica\\
    Taipei 115\\
    Taiwan}
\date{\today}
\graphicspath{{./dac-ab-figs/}}

\begin{document}
    
\maketitle

\begin{abstract}
	
We study divide-and-conquer recurrences of the form 
\begin{equation*}
    f(n) 
    = \alpha f\lpa{\ltr{\tfrac n2}}
    + \beta f\lpa{\lcl{\tfrac n2}} 
    + g(n) \qquad(n\ge2),
\end{equation*}
with $g(n)$ and $f(1)$ given, where $\alpha,\beta\ge0$ with 
$\alpha+\beta>0$; such recurrences appear often in analysis of 
computer algorithms, numeration systems, combinatorial sequences,  
and related areas. We show that the solution satisfies always the 
simple \emph{identity}
\begin{equation*}
    f(n) 
    = n^{\log_2(\alpha+\beta)} P(\log_2n) - Q(n) 
\end{equation*}
under an optimum (iff) condition on $g(n)$. This form is not only an
identity but also an asymptotic expansion because $Q(n)$ is of a
smaller order. Explicit forms for the \emph{continuity} of the
periodic function $P$ are provided, together with a few other
smoothness properties. We show how our results can be easily applied
to many dozens of concrete examples collected from the literature,
and how they can be extended in various directions. Our method of
proof is surprisingly simple and elementary, but leads to the
strongest types of results for all examples to which our theory
applies.

\end{abstract} 
\tableofcontents

\section{Introduction}

This paper is a sequel to \cite{Hwang2017}, where we studied the 
case $(\alpha,\beta)=(1,1)$ of the following recurrence
\begin{equation}\label{a1}
    f(n)
    = \alpha f\lpa{\ltr{\tfrac{n}{2}}} 
    + \beta f\lpa{\lcl{\tfrac{n}{2}}} 
    + g(n) \qquad (n\ge 2),  
\end{equation}
with $f(1)$ and $\{g(n)\}_{n\ge2}$ given; we focus here on the
general case of two given constants $\alpha,\beta > 0$. (The case
when $\alpha\le0$ or $\beta\le0$ is briefly discussed in Section
\ref{S:nonpositive}.) As in \cite{Hwang2017}, our aim in this paper will be
\begin{itemize}
	
\item to establish \emph{optimum iff-conditions} for the identity
\begin{equation}\label{a2}
    f(n)
    = n^{\rho}P(\log_{2}n)-Q(n) \qquad(n\ge1),  
\end{equation}
which is also an asymptotic expansion, where 
\begin{align}\label{rho}
    \rho 
    := \log_2(\alpha+\beta),   
\end{align}
$P$ is a bounded, \emph{continuous}, \emph{periodic} function and
$Q(n)$ is of a smaller order $o(n^{\rho})$; and 

\item to explore the usefulness of such a result by examining other
associated properties and applying to many concrete examples.

\end{itemize}

In addition, we also examine further \emph{smoothness properties} of
the periodic function $P$, and introduce and explore a new notion to
describe the \emph{equivalence of different recurrences}.

\paragraph{An elementary interpolation approach.}
The crucial step of our approach is to identify a (generally
nonlinear) interpolation function $\varphi(x)$ such that the sequence
$f(n)$ as defined by \eqref{a1} for positive integers $n$ can be
extended to a \emph{continuous} function $f(x)$ defined for all real
$x\ge1$ in the way that $f(x)$ equals the original sequence $f(n)$
when $x=n$, a positive integer, and there is a version of the
recurrence \eqref{a1} valid for all $x$; see
Section~\ref{S:recurrence} for details.

Such an interpolation-based analysis for \eqref{a1} will then be 
extended (in \refS{S:qary}) to the more general $q$-ary recurrence 
($q\ge2$) of the form 
\begin{equation}\label{E:qary}
	f(n)
	= \sum_{0\le j<q}\alpha_{j}f\lpa{\ltr{\tfrac{n+j}q}}
	+g(n) \qquad (n\ge q),
\end{equation}
for some given constants $\alpha_0,\dots,\alpha_{q-1}$, and a result
of the form \eqref{a2} will also be derived under some conditions.
Typical situations where \eqref{E:qary} arises is the application of
divide-and-conquer into $q$ parts whose sizes are as evenly as
possible. The special case when $\alpha_j=1$ for $0\le j<q$ was
already discussed in \cite{Hwang2017}. Another special case is
$\ga_j=0$ for $1\le j\le q-2$, which yields (with $\ga=\ga_0$ and
$\gb=\ga_{q-1}$) the recursion $f(n) = \alpha f(\tr{\frac nq})+\beta
f(\cl{\frac nq})+g(n)$ considered in \cite[Theorem 4.1]{Cormen2022}
and \cite{Kuszmaul2021} although they allow also non-integer $q>1$.

\paragraph{Recurrences with or without floors and ceilings.}
While the divide-and-conquer paradigm with evenly divided parts is
widely used in computer algorithms, our formulation of the
divide-and-conquer recurrence \eqref{E:qary}, as well as the very
precise identity \eqref{a2}, is surprisingly rare in the computer
algorithm literature; instead one finds predominantly a recurrence of 
the form
\begin{align}\label{E:srr}
	f(n) 
	= (\alpha+\beta) f\lpa{\tfrac n2}+g(n),
\end{align}
or more generally
\begin{align}\label{E:srr2}
	f(n) 
	= \sum_{0\le j<d}\alpha_j f\lpa{\tfrac n{q_j}}
	+ g(n),
\end{align}
where $\alpha, \alpha_j>0$, $d=1,2,\dots$ and $q_j>1$. According to  
the first edition of Cormen et al.'s widely used textbook on 
Algorithms \cite[p.~54]{Cormen1990}: ``\emph{When we state and solve
recurrences, we often omit floors, ceilings, and boundary conditions.
We forge ahead without these details and later determine whether or
not they matter. \dots we shall address some of these details to show
the fine points of recurrence solution methods.}" 

Such a simplifying approach also appears in most publications on
Algorithms. One of our aims in this paper is to show that
\emph{retaining floors and ceilings is not much more complicated than
omitting them, and with various advantages that are mostly unnoticed
in the literature.} This suggests that one main reason of omitting
floors and ceilings in handling a divide-and-conquer recurrence lies
more in methodological deficiencies than simply technical
conventions; the approach proposed in this paper will then complete
to some extent the required methodological developments.

More precisely, when $\alpha,\beta, g(n)\ge0$, typical approaches
adopted in the computer algorithms community to solving the
recurrence \eqref{a1} include

\begin{itemize}
\item dropping floor and ceiling in \eqref{a1} by assuming $n$ to be 
a power of $2$, resulting in the closed-form expression 
\begin{align}\label{E:2^m}
	f(2^m)
	= \sum_{1\le j\le m}(\alpha+\beta)^{m-j}g(2^j)
	+ (\alpha+\beta)^m f(1),
\end{align}
and 
\item lower- and upper-bounding $f$ by keeping only floor and only 
ceiling function in \eqref{a1}, leading to 
\begin{align}\label{E:fmn}
	f_-(n) := (\alpha+\beta)f_-\lpa{\tr{\tfrac n2}} + g(n),
\end{align}
and
\begin{align}\label{E:fpn}
	f_+(n) := (\alpha+\beta)f_+\lpa{\cl{\tfrac n2}} + g(n),
\end{align}
respectively.
\end{itemize}
While simple and effective in estimating the growth order of $f(n)$
for large $n$, both approaches suffer from subtle oversights and 
intrinsic limitations, with different shortcomings.

\paragraph{Monotonicity.}
First, in either of the approaches \eqref{E:2^m} and
\eqref{E:fmn}--\eqref{E:fpn}, the next crucial property used in
estimating the asymptotic growth of $f(n)$ is \emph{monotonicity},
which in the first approach is of the form $f(2^m)\le f(2^m+\ell)\le
f(2^{m+1})$ for $0\le \ell \le 2^m$, and $f_-(n)\le f(n)\le f_+(n)$
in the second approach. However, these inequalities may not hold in
general due to the periodic nature of the recurrences \eqref{a1},
\eqref{E:fmn} and \eqref{E:fpn}. For example, take $\alpha=\beta=1$
and $g(n) = 1+\mathbf{1}_{n \text{ odd}}$ with $f(1)=f_+(1) = 0$.
Then $8=f(7)>f(8)=7$, and $f(15)=17>f_+(15)=16$. Thus the use of the
monotonicity is more subtle than it is generally taken to be.

On the other hand, when $g(n) = \mathbf{1}_{n \text{ odd}}$ with
$f(1)=0$ (which yields A296062 in the \emph{On-Line Encyclopedia of
Integer Sequences} database \cite{OEIS2022} with many combinatorial
interpretations), then $f(n)$ oscillates between $0$ and $\Theta(n)$.
Thus not only monotonicity fails but also the growth order oscillates
violently, although our result yields that $f(n) = nP(\log_2n)$ for
some continuous periodic function; see \eqref{E:A296062-P}.

\paragraph{Discontinuity.}
Apart from monotonicity, there is yet another deeper reason why the
original sequence \eqref{a1} is preferred to its simplified one-sided
versions \eqref{E:fmn} and \eqref{E:fpn}: the periodic function $P$
in \eqref{a2} is always \emph{continuous}, while the corresponding
one for the solution of either \eqref{E:fmn} or \eqref{E:fpn} is
almost always \emph{discontinuous}; more precisely, in typical cases
the function $P(t)$ is discontinuous at every $t$ such that $2^t$ is
a dyadic rational, i.e., has a finite binary representation (see
\refS{Salphabeta=0} and, for details in the case $\ga+\gb=2$,
\cite[Section 8]{Hwang2017}). But why does continuity matters here?
The reason is because the periodic function $P$ when evaluated at
$\log_2n$ (see \eqref{a2}), involves indeed functions evaluated at
the dyadic rational $n/2^{\tr{\log_2n}}$ (see \eqref{b16} and
\eqref{t1y}), so that discontinuity causes the sequence (or the
original cost function) to have more violent jumps even for
neighbouring input sizes. Thus simplifying the recurrence \eqref{a1}
to either \eqref{E:fmn} or \eqref{E:fpn} has the advantage of being
easily solvable by iteration, but suffers from structural
discontinuities, or more rough oscillations.

\paragraph{Master theorems.}
Another commonly used approach to solve \eqref{E:srr} is to apply the
so-called ``master theorems" (see \cite{Bentley1980,Cormen2022}),
which are generally effective and user-friendly, but does not provide
more precise asymptotic approximations. For example, in the case of
\eqref{a1}, if $g(n) = O(n^{\rho-\ve})$, $\ve>0$, then the master
theorem \cite[\S~4.5]{Cormen2022} or \cite{Kuszmaul2021} gives $f(n)
= \Theta(n^\rho)$, where $\rho$ is defined in \eqref{rho}, while
under the same growth order of $g$, our approach (see
Corollary~\ref{C2}) again guarantees \eqref{a2} with explicitly
computable functions $P$ and $Q$. There do exist finer master
theorems that give more precise asymptotics under stronger
assumptions on $g(n)$ (such as monotonicity; see \cite{Drmota2013}),
but none of them is as precise as our identity \eqref{a2}.

\paragraph{Discrete and continuous master theorems.}
An additional feature of our interpolation-based analysis is that we
always work on the same real function $f(x)$, which coincides with
the original sequence $f(n)$ at integer parameters, unlike general 
master theorems that distinguish between ``discrete master theorems''
and ``continuous master theorems''; for example, according to 
\cite{Kuszmaul2021}:
\begin{quote}
\emph{To distinguish the two situations, we call the master theorem
without floors and ceilings the continuous master theorem and the
master theorem with floors and ceilings the discrete master theorem.}
\end{quote}
The subtleties of the two different versions (together with other 
issues) are only very recently thoroughly examined in 
\cite{Kuszmaul2021}, where they write:
\begin{quote}
\emph{Several academic works provide proofs and proof sketches of the 
discrete master theorem. To the best of our knowledge, however, all 
of these proofs are either incomplete, incorrect, or require 
sophisticated mathematics.}
\end{quote}

See also the long paper (more than 300 pages) \cite{Campbell2020} for
other delicate issues arising from divide-and-conquer recurrences.
Additionally, the chapter on ``Recurrences'' (Chapter I.4) in the
first edition of Cormen et al.'s book \cite{Cormen1990} is now
largely expanded and updated in the latest, very recent, edition
\cite[Ch.~I.4]{Cormen2022}, more than three decades after its first
edition and following the corresponding developments in clarifying
the subtleties; see \cite{Campbell2020, Cormen2022}.

For more information and references on master theorems and
divide-and-conquer recurrences, see, for example, \cite{Drmota2013,
Garet2022, Heuberger2022, Kuszmaul2021, Roura2001}. See also
\cite{Hwang2017} for more references on other approaches (including
complex-analytic, Tauberian, renewal, fractal geometry, and Ansatz or
guess-and-prove, called the substitution method in \cite{Cormen2022})
used in the literature for solving \eqref{a1}.

\paragraph{Periodic equivalence of sequences.} 
For a more canonical way to group or classify the diverse periodic
functions $P$, we will introduce in Section~\ref{sec-pe} a useful
notion called ``periodic equivalence'', roughly meaning that
different sequences share, modulo amplitude and scale, the same
oscillating part. For example, denote by $S_{2,1}(n)$ (A006046) the
total number of odd entries in the first $n$ rows of Pascal triangle;
then $f(n) = S_{2,1}(n)$ satisfies \eqref{a1} with
$(\alpha,\beta)=(2,1)$, $g(n)=0$ and $f(1)=1$. We have $S_{2,1}(n)=
n^{\log_23}P(\log_2n)$; see \refE{Ex:pascal}. Then the following
OEIS sequences, all satisfying \eqref{a1} with
$(\alpha,\beta)=(2,1)$, are periodically equivalent to $S_{2,1}(n)$
(involving, up to scale and amplitude, the same $P$):
\begin{center}
\begin{tabular}{cccc}
OEIS id. & $g(n)$ & $f(1)$ & $f(n)$ \\ \hline
A051679 & $\frac18n^2 - \begin{cases}
	\frac n4,& \text{$n$ even}\\
	\frac18,& \text{$n$ odd}
\end{cases}$ & $0$
& $\binom{n+1}2-S_{2,1}(n)$\\
A080978 & $-2$ & $3$ & $2S_{2,1}(n)+1$ \\
A159912 & $\tr{\frac n2}$ & $1$ & $2S_{2,1}(n)-n$\\
A171378 & $\cl{\frac n2}^2-\mathbf{1}_{n \text{ odd}}$ 
& $0$ & $n^2-S_{2,1}(n)$\\
A267700 & $\tr{\frac n2}$ & $0$ & $S_{2,1}(n)-n$\\ \hline
 % & $c$ & $b$ & $(b+\frac c2)S_{2,1}(n)-\frac12c$
\end{tabular}	
\end{center}
See \refE{Ex:pascal} for more periodically equivalent sequences. 

\paragraph{Smoothness of the periodic function.}
While the periodic equivalence is introduced to identify the same
fluctuating part of different sequences satisfying the \emph{same}
recurrence, we also examine the varying smoothness nature exhibited
by \emph{different} recurrences. The main motivating observation is
that most periodic functions we obtain have visible cusps (likely to
be non-differentiable points), and there are other classes of
functions (such as Lipshitz and H\"older continuous) between the
class of continuous functions and that of continuously differentiable
functions. We will thus clarify the different, characteristic, 
inherent types of H\"older continuity of the interpolated function 
$f(x)$ and the periodic function $P(t)$ when $(\alpha,\beta)$ varies.

For example, in Figure~\ref{F:Sab}, we plot the periodic functions
$P$ (as defined in \eqref{a2}) when $f$ satisfies \eqref{a1} with
$g(n)=0$ and $f(1)=1$, and with $(\alpha,\beta)=(\alpha,1)$,
$\alpha=2,3,4,5$ (the lower blue curves) and
$(\alpha,\beta)=(1,\beta)$, $\beta=2,3,4,5$ (the upper green curves).
Our results show that these periodic functions are H\"older
continuous with exponent $\log_2(1+\alpha^{-1})$ and
$\log_2(1+\beta^{-1})$, respectively, and we conjecture that these
exponents are the best possible; see Section~\ref{S:smooth} for
details. Thus for these recurrences, the larger the values of
$\alpha$ or $\beta$, the ``less smooth" the periodic functions.

\begin{figure}[!ht]
	\begin{center}
		\includegraphics[height=3.5cm]{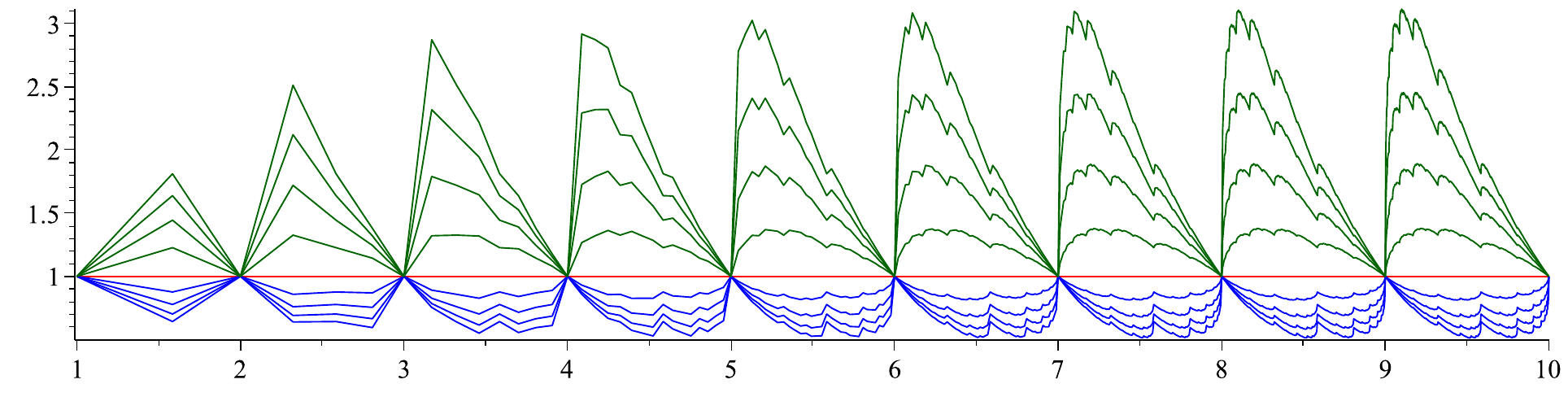}
	\end{center}
	\caption{Periodic fluctuations of the functions $P(\log_2n)
	=f(n)n^{-\log_2(\alpha+\gb)}$ for $n=2,\dots,1024$ when $f$ 
	satisfies \eqref{a1} with $g(n)=0$, $f(1)=1$ and with (the blue 
	curves) $\beta=1$ and $\alpha=2,3,4,5$ (from top to bottom),
	and (the green curves) $\alpha=1$ and $\beta=2,3,4,5$ (from 
	bottom to top).} \label{F:Sab}
\end{figure}

Moreover, we show that many of the periodic functions appearing in
our analysis are indeed not continuously differentiable at all points
in $(0,1)$. Several examples are piecewise differentiable with jumps
in the derivative at some points (see, for example, \refR{Ra=b} and
\refEs{E52} and \ref{E22-piece}). Other examples are less smooth; for
brevity of presentation, we discuss only one case (A006581) in detail
in \refE{E22-ns} and show (in Appendix~\ref{AE22-ns}) that the
periodic function is continuous but nowhere differentiable, leaving
such a deeper property for other sequences to the interested reader.

\paragraph{A generating function viewpoint.}
To see how the general cases differ from the special case
$(\alpha,\beta)=(1,1)$ or more generally $\alpha=\beta$, we consider
the generating function
\begin{equation}\label{gg1}
    A(z) 
    := \sum_{n\ge1} f(n) z^n,
\end{equation}
which, by \eqref{a1} satisfies the functional equation 
\begin{equation}\label{gg2}
    A(z) 
    = \frac{\beta+(\alpha+\beta)z+\alpha z^2}{z}\, A(z^2)
    + B(z),
\end{equation}
where
\begin{equation}\label{gg3}
    B(z) 
    := f(1)(1-\beta)z+\sum_{n\ge2}g(n)z^n.
\end{equation}
We see that if $\alpha=\beta$, then the functional equation becomes
\begin{equation}\label{gg4}
    A(z) 
    = \frac{\alpha(1+z)^2}z\,A(z^2)+B(z),
\end{equation}
so that the generating function $\bar{A}(z) := \frac{(1-z)^2}z\,A(z)$
of the second difference of $f(n)$ satisfies the simpler equation
\begin{equation}\label{gg5}
    \bar{A}(z) 
    = \alpha \bar{A}(z^2)+\bar{B}(z),
\end{equation}
where $\bar{B}(z) := \frac{(1-z)^2}z\,B(z)$. Assuming for simplicity
(without real loss of generality) that $\bar{B}(0)=0$, we then
obtain, by iteration, the exact solution
\begin{equation}\label{gg6}
    A(z) 
    = \frac{z}{(1-z)^2}\sum_{k\ge0}
    \alpha^{k}\bar{B}\lpa{z^{2^k}}. 
\end{equation}
Such a neat representation is the basis of the analytic approach
introduced in \cite{Flajolet1994}, but is not available in general
when $\alpha\ne\beta$. We will therefore use a different method.

This paper is structured as follows. We develop in
Section~\ref{S:recurrence} the required technicalities in order to
prove \eqref{a2}, and then address the smoothness properties of $P$
in Section~\ref{S:smooth}. These two sections provide a theoretical
foundation to the resolution of the divide-and-conquer recurrences of
the form \eqref{a1}. Applications of our theory to concrete examples
are discussed in Section~\ref{S:app1} for $\alpha\ne\beta$ and in
Section~\ref{S:app2} for $\alpha=\beta$. We then extend very briefly
our analysis to nonpositive $\alpha$ or $\beta$ in
Section~\ref{S:nonpositive} and to general $q$-ary recurrence
\eqref{E:qary} in Section~\ref{S:qary}. For completeness, an appendix
on the connection of our approach to Mellin transforms is given,
together with two others providing detailed proofs of some results in
the paper.

\paragraph{Notation.}

For convenience, we introduce the operator $\Lambda_{\alpha, \beta}$ 
as follows:
\begin{equation}
    \Lambda_{\alpha ,\beta}[f](n)
    := f(n)-\alpha f\left( \left\lfloor \tfrac{n}{2}
    \right\rfloor \right) -\beta f\left( \left\lceil     
    \tfrac{n}{2}\right\rceil\right) .
\end{equation}

Let, for $x>0$, 
\begin{align}\label{L,theta}
    L_{x}
    &:=\left\lfloor \log_{2}x\right\rfloor, 
    &\gth_x&:=\frax{\log_2x}=\log_2x-L_x\in[0,1),
\end{align}
and $L_0 := 0$.

For real $x$ and $y$, we let $x\vee y:=\max(x,y)$.

\section{The recurrence $\Lambda_{\alpha ,\beta}[f]=g$}
\label{S:recurrence}

\emph{Here and throughout this section, we assume that
$\alpha,\beta>0$, and define $g(1):=0$}.
The recurrence (\ref{a1}) can be rewritten as
\begin{equation}\label{b1}
	\begin{split}
	    f(2n) 
	    &=(\alpha +\beta )f(n)+g(2n),   \\
	    f(2n+1) 
	    &=\alpha f(n)+\beta f(n+1)+g(2n+1),
	\end{split}
\end{equation}
for $n\ge 1$, with given $f(1)$ and $g(n)$, $n\ge 2$. 

\paragraph{Extending the sequence $f(n)$ to a function $f(x)$ in
  $\mathbb{R}^{+}$.}  
We extend the sequence $f(n)$ to a continuous function $f(x)$
defined for real $x\ge1$ by interpolation between the integers with
scaled copies of a certain function $\varphi:\oi\to\oi$ constructed
below:
\begin{align}\label{b2}
%	\begin{split}
    f(n+t)
    &:= f(n)+\varphi(t)\left( f(n+1)-f(n)\right) \notag\\
    &= \left( 1-\varphi(t)\right) f\left( n\right) 
    + \varphi(t)f\left( n+1\right) ,
%	\end{split}
\end{align}
for $n\ge 1$ and $0\le t\le 1$; in a similar way, we construct 
$g(x)$ from $g(n)$. The function $\varphi$ is constructed so that 
we will have
\begin{equation}\label{b3}
    f(x)
    = (\alpha +\beta )f\left( \frac{x}{2}\right) 
    + g(x), \qquad x\ge2.  
\end{equation}
We further define $g(x):=0$ for $x\in[0,1)$; thus $g(x)$ is defined 
for $x\ge0$.

The function $\varphi(t)$ depends on $\alpha$ and $\beta$; 
we sometimes write it as $\varphi_{\alpha,\beta}(t)$ to emphasise the 
dependence on parameters.

\paragraph{Construction of the interpolation function $\varphi $.} 
We first give relations on $\varphi $ that imply  the functional
equation (\ref{b3}). The existence of such a $\gf$ is shown later.

\begin{lemma}\label{Lemma1}
Let $\varphi(t)=\varphi_{\alpha,\beta}(t)$ be a function on $[0,1]$
such that $\varphi(0)=0$, $\varphi(1)=1$ and
\begin{equation}\label{b4}
    \varphi(t)
    = \begin{cases}
        \tfrac{\beta}{\alpha +\beta}\,\varphi(2t), 
        & \text{if\/ }t\in [0,\frac{1}{2}], \\
        \tfrac{\alpha}{\alpha +\beta}\,\varphi(2t-1)
        +\tfrac{\beta}{\alpha +\beta}, 
        & \text{if\/ }t\in [\frac{1}{2},1].
    \end{cases} 
\end{equation}
Assume that \eqref{b1} holds. Then \eqref{b3} holds if we extend 
$f(n)$ and $g(n)$ by \eqref{b2} to $f(x)$ and $g(x)$, respectively, 
for all real $x\ge 1$.
\end{lemma}

\begin{proof}
If $n\ge 1$ and $0\le t\le \frac{1}{2}$, then, using the interpolation
(\ref{b2}), the recurrences (\ref{b1}) and the assumption (\ref{b4}), 
we have 
\begin{align}\nonumber
    \MoveEqLeft f(2n+2t) -g(2n+2t)  \\
    &= (1-\varphi(2t))(f(2n) -g(2n)) 
    +\varphi(2t)(f(2n+1) -g(2n+1))
    \nonumber \\
    &= (1-\varphi(2t)) (\alpha+\beta) f(n) 
    +\varphi(2t)\bigpar{\alpha f(n) +\beta f(n+1) }  
	\nonumber \\
    &= (\alpha+\beta) f(n) -\varphi(2t)\beta f(n) 
    +\varphi(2t)\beta f(n+1)  \nonumber \\
    &= (\alpha+\beta) f(n) -(\alpha +\beta) \varphi(t)f(n) 
    +(\alpha+\beta) \varphi(t)f(n+1)  \nonumber \\
    &= (\alpha+\beta) \bigpar{(1-\varphi(t)) f(n) 
    +\varphi(t)f(n+1)}  \nonumber \\
    &= (\alpha+\beta) f(n+t) . \label{b5} 
\end{align}
Similarly, the lower part of (\ref{b4}) can be rewritten as
\begin{align}\label{b4l}
    \xpar{\alpha +\beta}\,\varphi\left(\tfrac{1}{2}+t\right) 
    = \alpha\varphi(2t)+\beta,\qquad t\in[0,\tfrac{1}{2}],
\end{align}
and thus, still for $n\ge 1$ and $0\le t\le \frac{1}{2}$, 
\begin{align}\nonumber
    \MoveEqLeft f(2n+1+2t)-g(2n+1+2t)  \\
    &= (1-\varphi(2t))(f(2n+1) -g(2n+1))
    +\varphi(2t)(f(2n+2) -g(2n+2))  \nonumber \\
    &= (1-\varphi(2t))(\alpha f(n) +\beta f(n+1))
    +\varphi(2t)(\alpha+\beta) f(n+1)  \nonumber \\
    &= \alpha f(n) +\beta f(n+1) -\varphi(2t)\alpha f(n)
    +\varphi(2t)\alpha f(n+1)  \nonumber \\
    &= \alpha f(n) +\beta f(n+1) -\left((\alpha+\beta)
    \varphi\left(\tfrac{1}{2}+t\right) -\beta \right)
    \bigpar{f(n)-f(n+1)}\nonumber \\
    &= (\alpha+\beta) f(n) -(\alpha +\beta)
    \varphi\left(\tfrac{1}{2}+t\right) f(n)
    +(\alpha+\beta) \varphi\left(\tfrac{1}{2}+t\right) f(n+1)
    \nonumber \\
    &= (\alpha+\beta)
    \left(\left(1-\varphi\left(\tfrac{1}{2}+t\right) 
	\right) f(n)
    +\varphi\left(\tfrac{1}{2}+t\right)f(n+1) \right)  
	\nonumber \\
    &= (\alpha+\beta) f\left(n+\tfrac{1}{2}+t \right).  
	\label{b6} 
\end{align}
Combining (\ref{b5}) and (\ref{b6}), we obtain
\begin{equation}\label{b7}
    f(n+2t)
    = (\alpha +\beta )f\left(\tfrac{n}{2}+t\right) +g(n+2t)  
\end{equation}
for $n\ge 2$ and $0\le t\le \frac{1}{2}$, and thus (\ref{b3}) holds.
\end{proof} 

\begin{rem}\label{R2.2}
Conversely, it is immediate that if \eqref{b3} holds for the
extensions of $f$ and $g$ defined by \eqref{b2}, then \eqref{b1}
holds.
\end{rem}

In the case $\ga=\gb=1$ treated in \cite{Hwang2017}, the system of
equations (\ref{b4}) has the solution $\gf(t)=t$. In general, it is
not obvious that a continuous solution to (\ref{b4}) exists. We now
show that this is the case, and that this $\varphi $ is unique.

\begin{lemma}\label{Lemma2} 
If\/ $\alpha ,\beta >0$, then there exists a unique continuous
function $\varphi(t)=\gf_{\ga,\gb}(t)$ on $[0,1]$ such that
$\varphi(0)=0$, $\varphi(1)=1$ and \eqref{b4} holds. Moreover,
$\varphi $ is strictly increasing.
\end{lemma}

\begin{proof}\multlinegap=0pt
The equation (\ref{b4}) and the conditions $\varphi(0)=0$, $\varphi
(1)=1$ define recursively $\varphi(t)$ uniquely for dyadic rational
$t\in [0,1]$; hence, there is at most one continuous $\varphi $
satisfying these requirements. The existence of such a $\varphi$ can 
be proved in several different ways. 

\paragraph{First proof (probabilistic).}
Let $\gf$ be the distribution function $\gf(t):=\Pr(X\le t)$ of the
random variable $X\in\oi$ defined by the binary expansion
$X=0.B_1B_2\dots$, where the bits $B_1,B_2,\dots$ are independent and
with $\Pr(B_i=1)=\frac{\ga}{\ga+\gb}$. Since the random variable
$X':=0.B_2B_3\dots=2X-B_1$ has the same distribution as $X$, we see
that
\begin{equation}
	\begin{split}
	    \gf(t)
	    &:= \Pr(X\le t)\\
	    &=\begin{cases}
	        \Pr(B_1=0)\Pr(X'\le 2t),
	        % =\frac{\gb}{\ga+\gb}\,\gf(2t),
 	        & t\in[0,\frac12], 
	        \\
	        \Pr(B_1=0)+\Pr(B_1=1)\Pr(X'\le 2t-1),
	        % =\frac{\gb}{\ga+\gb}+\frac{\ga}{\ga+\gb}\,\gf(2t-1),
 	        & t\in[\frac12,1],  
	    \end{cases}\\
		&= \begin{cases}
			\frac{\gb}{\ga+\gb}\,\gf(2t) ,
		    & t\in[0,\frac12], \\
			\frac{\gb}{\ga+\gb}+\frac{\ga}{\ga+\gb}\,\gf(2t-1),
			& t\in[\frac12,1],
		\end{cases}
	\end{split}
\end{equation}
which verifies \eqref{b4}. Furthermore, $\gf$ is continuous and
strictly increasing on $\oi$, with $\gf(0)=0$ and $\gf(1)=1$.

\paragraph{Second proof (digital sums).}
An alternative approach begins with an explicit construction, which
will also be useful later. For $t=\sum_{j\ge 1}b_{j}2^{-j}\in 
[0,1]$, where $b_{j}\in \{0,1\}$, define 
\begin{align}\label{b8a}
    \varphi(t)
    := \sum_{j\ge 1}b_{j}
	\left(\frac{\alpha}{\beta}\right)^{b_{1}+\cdots +b_{j-1}}
    \biggpar{\frac{\beta}{\alpha +\beta}}^{j}.
\end{align}
Equivalently, define, for $t=\sum_{k\ge 1}2^{-e_{k}} \in\oi$, where
$1\le e_{1}<e_{2}<\cdots $ is a finite or infinite sequence,
\begin{equation}\label{b8}
    \varphi \biggpar{t=\sum_{k\ge 1}2^{-e_{k}}} 
    := \sum_{k\ge 1}
    \frac{\alpha^{k-1}\beta^{e_{k}-k+1}}
    {\left( \alpha +\beta \right)^{e_{k}}}.
\end{equation}
We first show that both \eqref{b8a} and \eqref{b8} are well defined 
for dyadic rational $t$ with two different representations, namely,
\begin{align}\label{mag}
    \sum_{1\le i\le k}2^{-e_{i}}
    = \sum_{1\le j<k}2^{-e_{j}}
    + \sum_{j\ge 1}2^{-( e_{k}+j) }.
\end{align}
The value in \eqref{b8} for the last term on the left-hand side of
\eqref{mag} is $\frac{\alpha^{k-1}\beta^{e_{k}-k+1}}{\left(\alpha
+\beta \right)^{e_{k}}}$; the value for the sum on the right-hand
side equals
\begin{equation}
    \sum_{j\ge 1}
    \frac{\alpha^{(k+j-1)-1}\beta^{e_{k}+j-( k+j-1)+1}}
    {(\alpha +\beta)^{e_{k}+j}}
    =\frac{\alpha^{k-1}\beta^{e_{k}-k+1}}
    {(\alpha+\beta)^{e_{k}}}
    \sum_{j\ge 1}\frac{\alpha^{j-1}\beta}
    {(\alpha+\beta)^{j}}
    =\frac{\alpha^{k-1}\beta^{e_{k}-k+1}}
    {(\alpha+\beta)^{e_{k}}},
\end{equation}
since the final geometric series sums to 1. Hence, the two
representations in \eqref{mag} yield the same sum \eqref{b8}, which
always equals the sum \eqref{b8a}, and thus $\gf$ is well-defined
on $[0,1]$.

Clearly, $\gf(0)=0$ and $\gf(1)=1$, again by summing the same
geometric sum. Moreover, \eqref{b4} follows easily from \eqref{b8a},
considering the cases $b_1=0$ and $b_1=1$ separately.

Next, if $t=\sum_{1\le j\le N} b_j 2^{-j}$ is a dyadic rational in
$[0,1)$, and $s:=\sum_{1\le j\le N} b_j$, then, by \eqref{b4} and
induction on $N\ge0$, we have
\begin{align}\label{bb10}
    \gf\bigpar{t+2^{-N}}-\gf(t)
    =\frac{\ga^s\gb^{N-s}}{(\ga+\gb)^N}.
\end{align}
It follows from \eqref{bb10} that $\gf$ is strictly increasing on the
set of dyadic rationals in $\oi$. Furthermore, suppose that
$t_1=j2^{-N}$ is a dyadic rational in $[0,1)$, and let
$t_2:=t+2^{-N}$. If $t\in[t_1,t_2]$, then both $t$ and $t_2$ have
binary expansions beginning with the expansion of $t_1$ (choosing the
infinite representation for $t_2$), and it follows from \eqref{b8}
that $\gf(t_1)\le \gf(t)\le\gf(t_2)$. Now, suppose $0\le t<u\le 1$,
and choose $N$ such that $2^{1-N}<u-t$. Then there exist $j$ and $k$
with $j2^{-N} \le t\le (j+1)2^{-N}<k2^{-N}\le u\le(k+1)2^{-N}$.
Hence, by what was just shown, $\gf(t)\le \gf
\bigpar{(j+1)2^{-N}}<\gf\bigpar{k2^{-N}}\le\gf(u)$. Thus, $\gf$ is
strictly increasing.

The monotonicity implies that any discontinuity of $\gf$ must be a
jump. Define
\begin{align}
    \gD_+\gf(t)
    := \lim_{u\searrow t}\gf(u)-\gf(t), \qquad t\in[0,1),  
\end{align}
the jump to the right at $t$, and let $R:=\sup_{t\in[0,1)}\gD_+(t)$.
It follows from \eqref{b4} that $R=\frac{{\ga\lor\gb}}{\ga+\gb} R$,
and thus $R=0$. Hence, $\gf$ is right-continuous. Similarly, $\gf$ 
is left-continuous.

\begin{figure}[!ht]
\begin{center}
	\includegraphics[height=3.3cm]{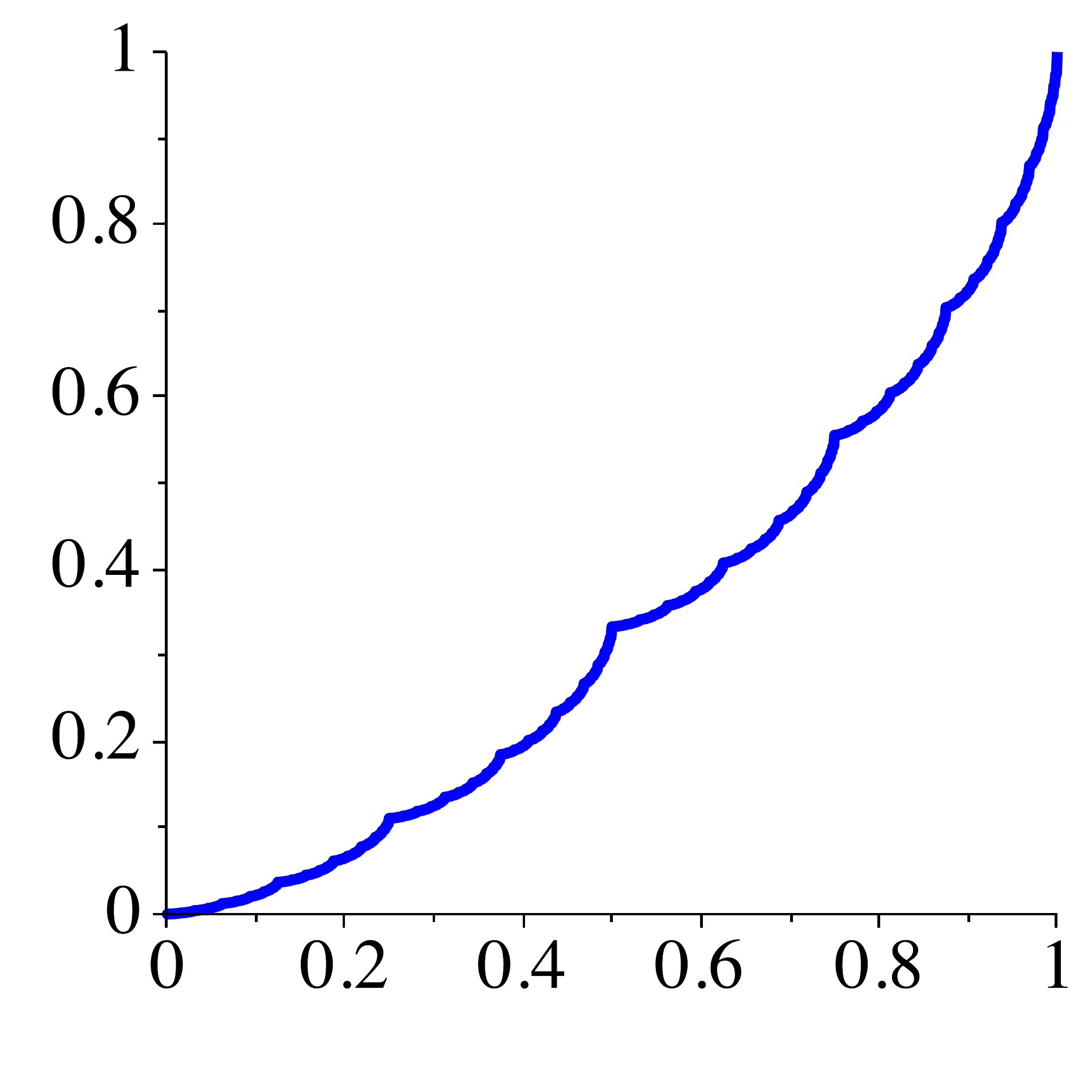}\,
	\includegraphics[height=3.3cm]{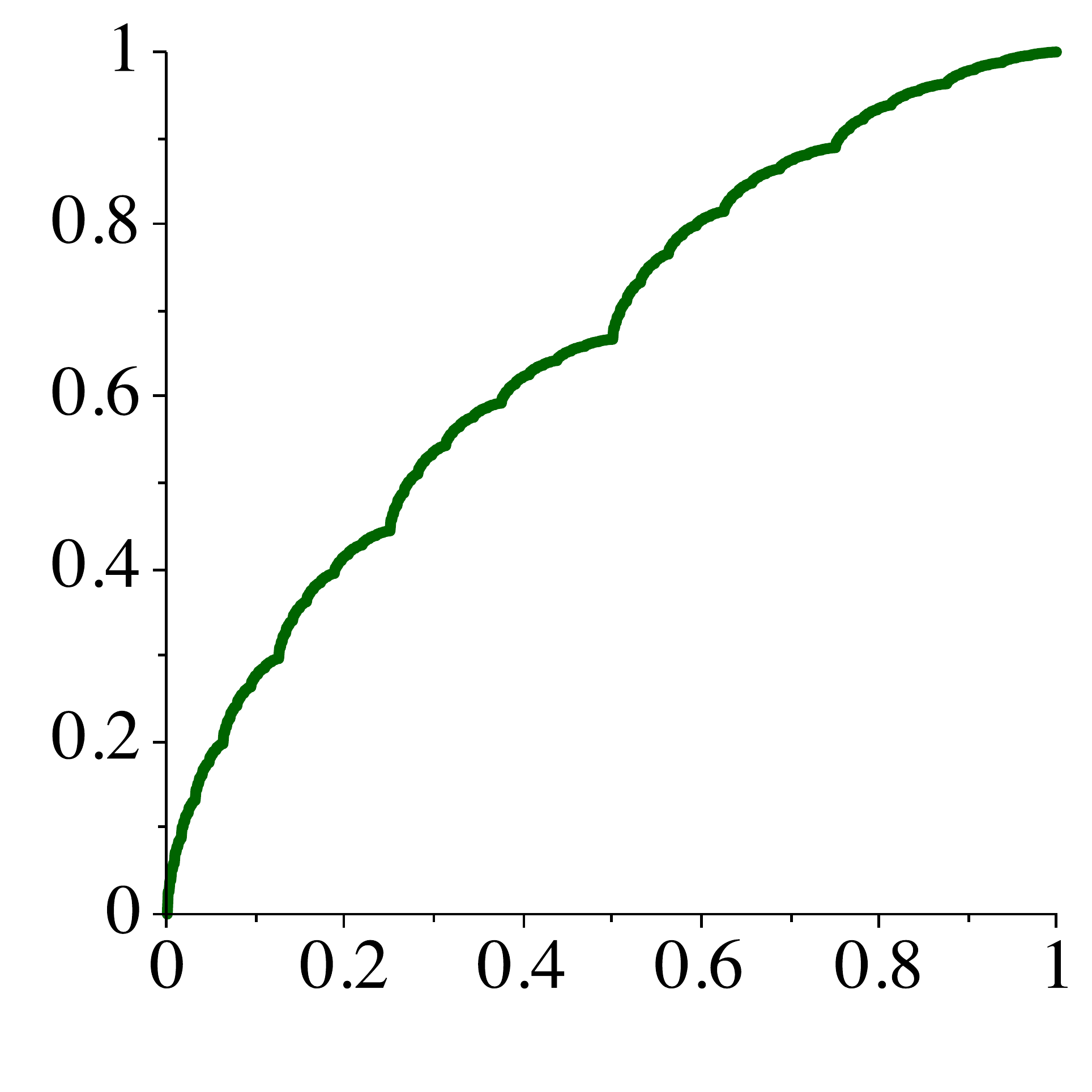}\,
	\includegraphics[height=3.3cm]{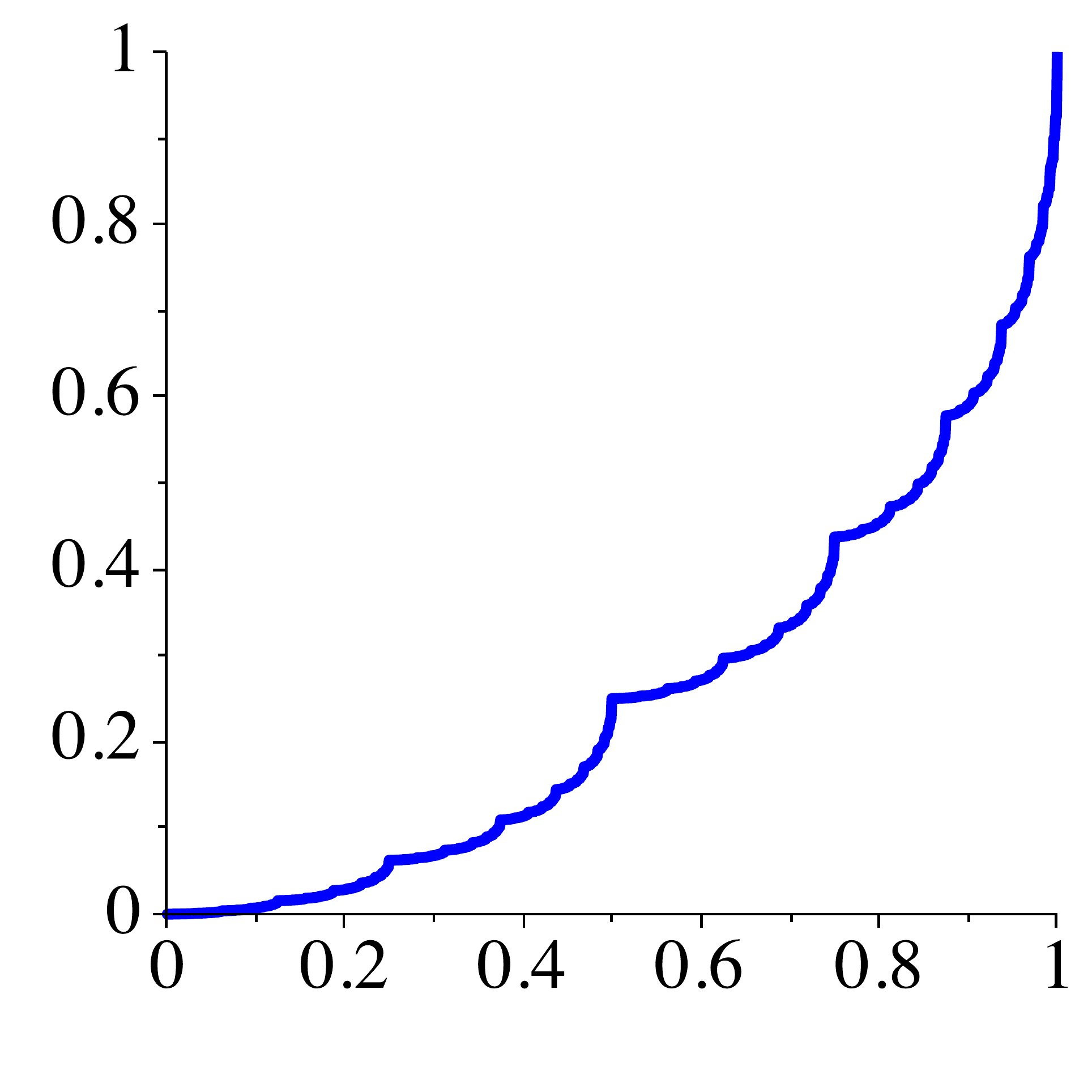}\,
	\includegraphics[height=3.3cm]{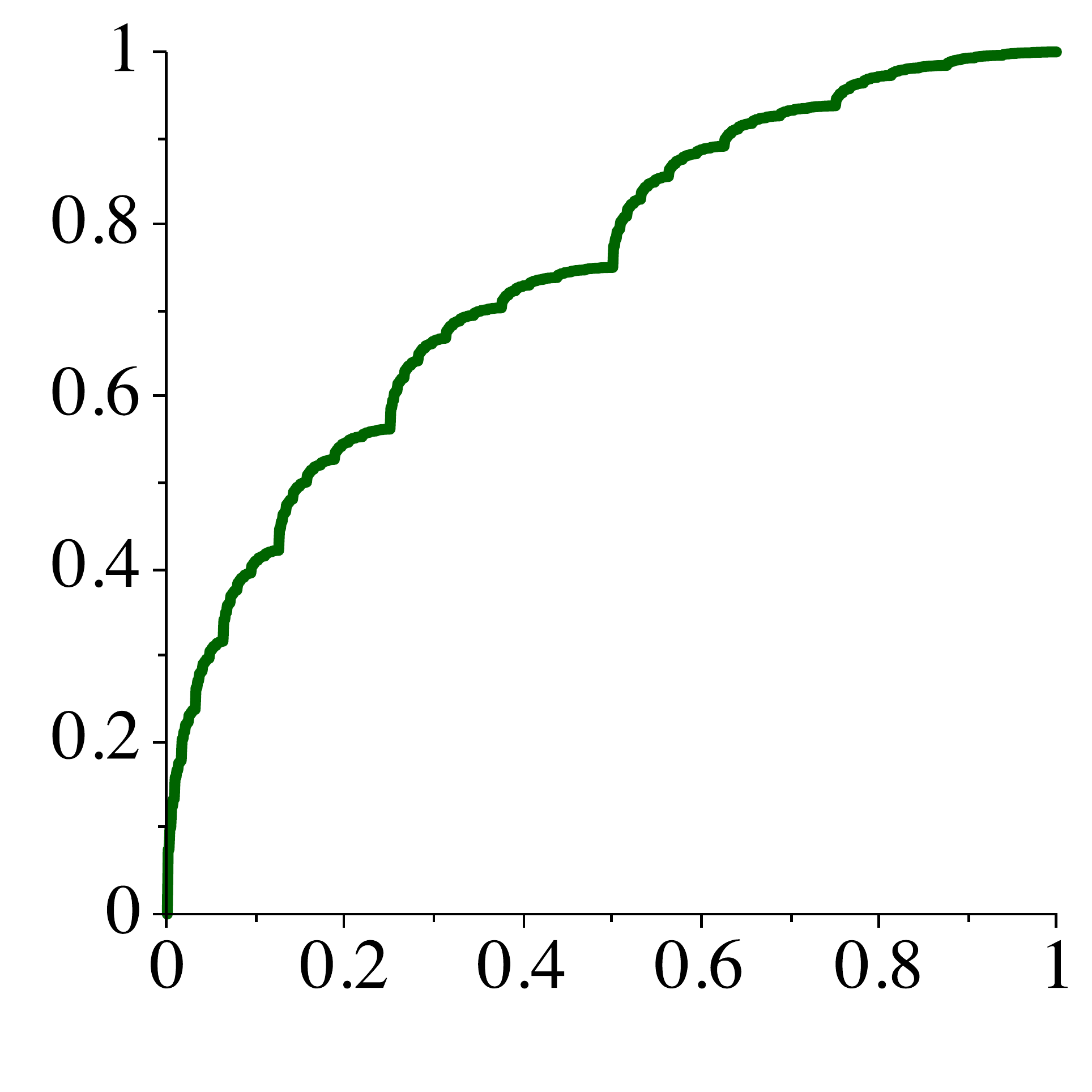}
\end{center}
\vspace*{-.3cm} 
\caption{{The function $\varphi$ for $(\alpha,\beta)=(2,1),
(1,2), (3,1), (1,3)$ (from left to right).}}
\label{Fig1}
\end{figure}

\paragraph{Third proof (recursive construction).}
Yet another alternative, similar to the construction of Koch's
snowflake curve \cite{vonKoch1904}, is to define
$\varphi^{}_{0}(t):=t, t\in [0,1]$, and then recursively let
$\varphi^{}_{k+1}$ consist of two suitably scaled copies of
$\varphi^{}_{k}$; more precisely
\begin{equation}\label{b9}
    \varphi^{}_{k+1}(t)
    := \begin{cases}
        \tfrac{\beta}{\alpha +\beta}\,\varphi^{}_{k}(2t), 
        & \text{if }t\in [0,\frac{1}{2}], \\
        \tfrac{\alpha}{\alpha +\beta}\,\varphi^{}_{k}(2t-1)
        +\tfrac{\beta}{\alpha+\beta}, 
        & \text{if }t\in [\frac{1}{2},1];
    \end{cases}
\end{equation}
see Figure \ref{Fig2} for an illustration. Then, by induction,
$|\gf_{k+1}(t)-\gf_k(t)|\le \bigpar{\frac{{\ga\lor\gb}}{\ga+\gb}}^k$
for $k\ge0$ and $t\in[0,1]$; consequently, the functions
$\varphi^{}_{k}(t)$ converge uniformly on $\oi$ to a continuous
function $\gf$ satisfying (\ref {b4}), with $\gf(0)=0$ and
$\gf(1)=1$. It follows also that $\gf$ is weakly increasing; it is
then easy to show, using \eqref{b4}, that such $\gf$ is strictly
increasing, but we omit this, since we have already shown this using
the other constructions.
\end{proof}

\begin{rem}\label{Rlinear}
If $\ga=\gb$, then \eqref{b4} is solved by $\gf(t)=t$, just as in
the case $\ga=\gb=1$ treated in \cite{Hwang2017}, so $f(x)$ is
defined in \eqref{b2} by linear interpolation. If $\ga\neq\gb$,
then this is not the case; a linear interpolation would not yield
\eqref{b3}.

Figure \ref{Fig1} shows that $\gf$ has a typical fractal shape when
$\ga\neq\gb$. We will show in Lemma \ref{LHolder} and \refR{RHolder}
below that $\gf$ is \Holder{} continuous but not Lipschitz continuous
when $\ga\neq\gb$.
\end{rem}

\vspace*{-.3cm}
\begin{figure}[!h]
\begin{center}
	\includegraphics[height=3.2cm]{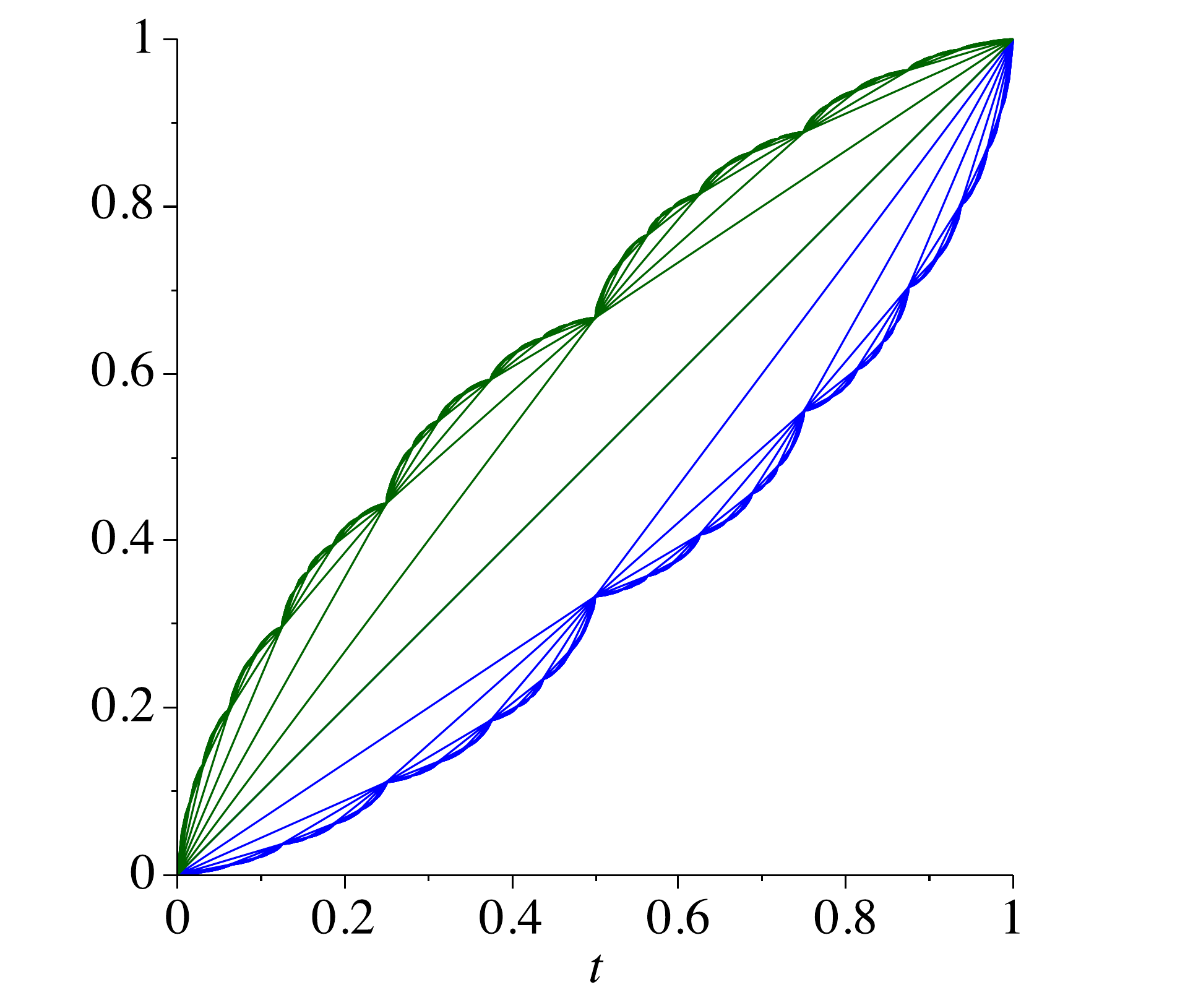}\,
	\includegraphics[height=3.2cm]{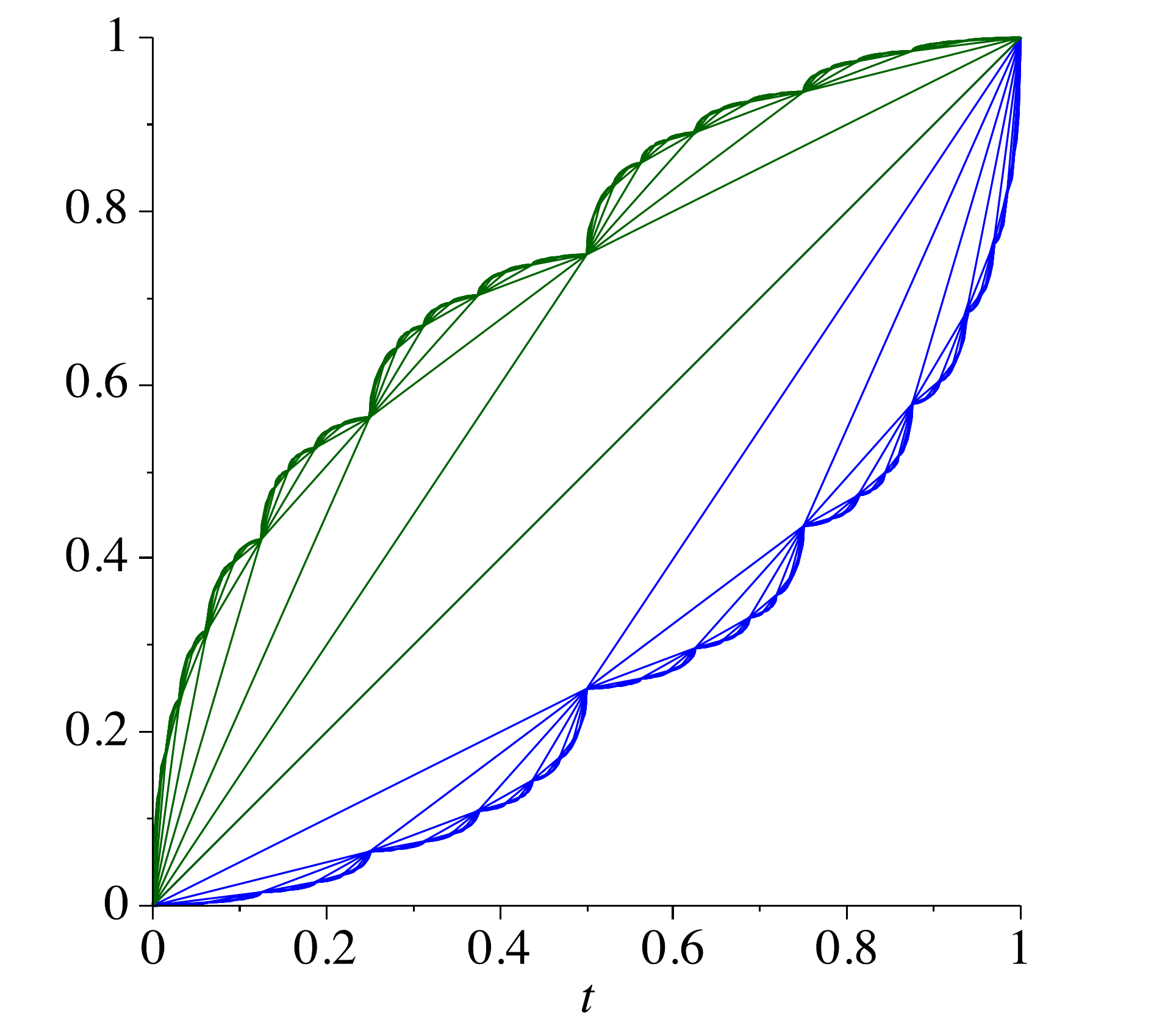}\,
	\includegraphics[height=3.2cm]{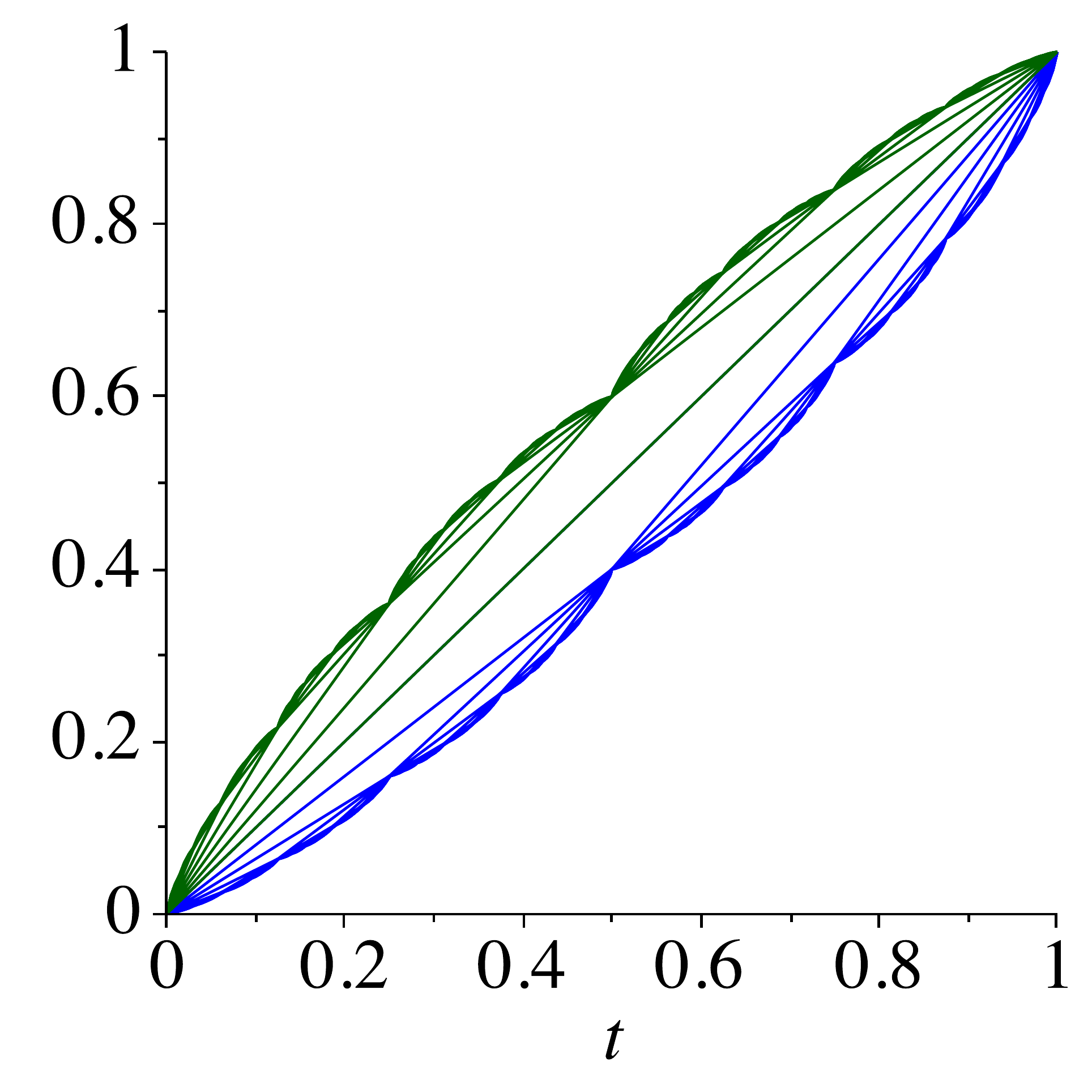}\,
	\includegraphics[height=3.2cm]{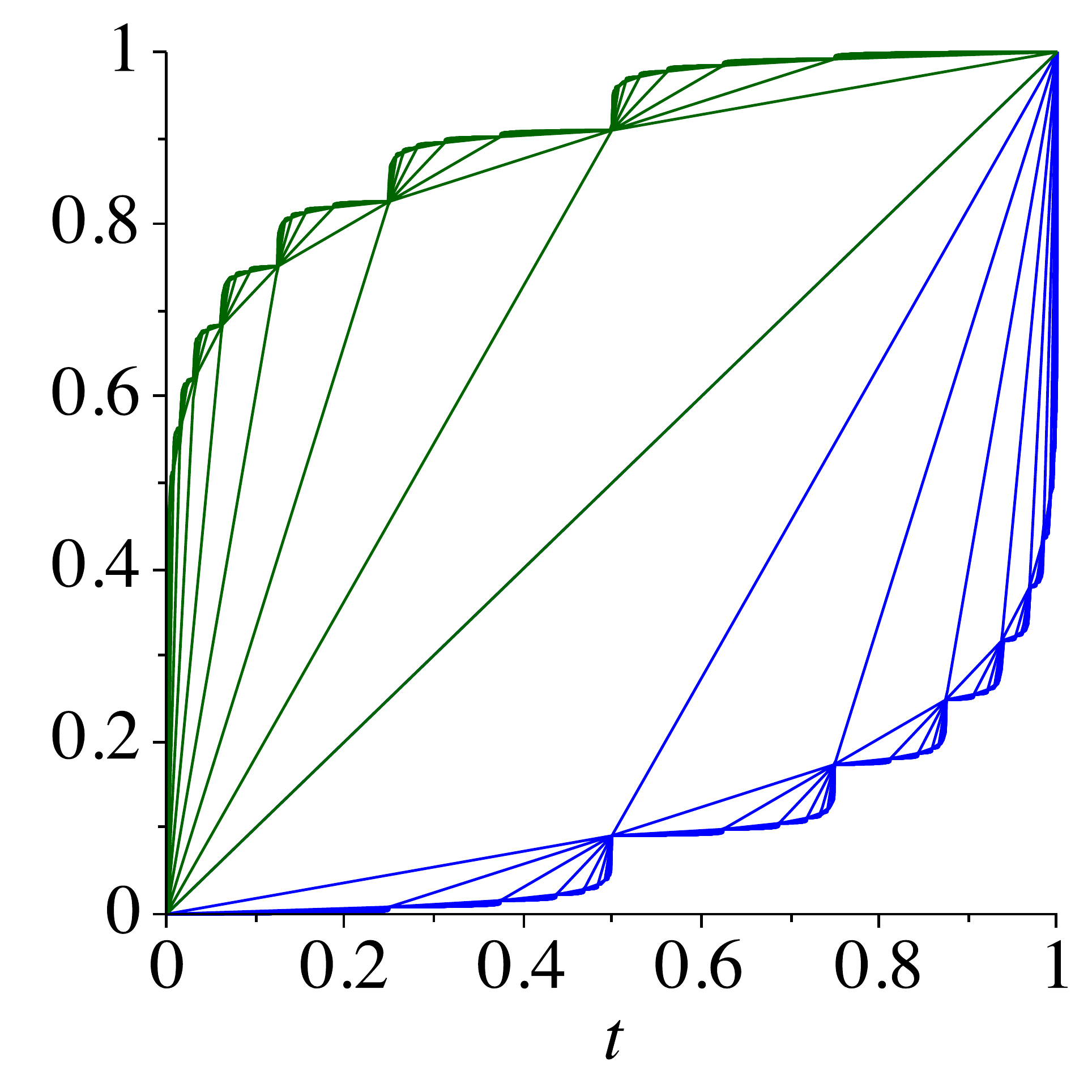}
\end{center}
\vspace*{-.5cm} \caption{{The functions $\varphi^{}_k$ for
$k=0,\dots,15$, where $(\alpha,\beta)=(2,1), (1,2)$ (left),
$(\alpha,\beta)=(3,1), (1,3)$ (middle-left), $(\alpha,\beta)=(3,2),
(2,3)$ (middle-right), and $(\alpha,\beta)=(10,1), (1,10)$
(right).}} \label{fig-koch} \vspace*{-.1cm}
\label{Fig2}
\end{figure}

\begin{rem}
It follows from \eqref{b4} that
\begin{align}
	\varphi_{\beta,\alpha}(t) 
	= 1-\varphi_{\alpha,\beta}(1-t),
\end{align}
so the graph of $\varphi_{\beta,\alpha}$ equals the graph of 
$\varphi_{\alpha,\beta}$ rotated about $(\frac12,\frac12)$. See
Figures~\ref{Fig1} and \ref{fig-koch} for a graphical illustration of 
some examples.
\end{rem}

\begin{remark}\label{Rintegral}
For later use we note that
\begin{align}\label{rintegral}
    \int_0^1\gf(x)\dd x
	=\frac{\gb}{\ga+\gb}.
\end{align}
To see this, we may denote the integral by $I$ and integrate both
sides of \eqref{b4}, which then yields
\begin{align}
	I 
	= \frac{\gb}{\ga+\gb}\cdot\frac12I+ 
	  \frac{\ga}{\ga+\gb}\cdot\frac12I
	+ \frac{\gb}{\ga+\gb}\cdot\frac12
	= \frac12I + \frac12 \frac{\gb}{\ga+\gb};
\end{align}
solving $I$ gives \eqref{rintegral}.
\end{remark}

\paragraph{Identities.} 
By iterating the functional equation (\ref{b3}), we obtain, for
$x\ge1$,
\begin{equation}\label{b10}
    f(x)
    =\sum_{0\le k<m}(\alpha +\beta )^{k}g(2^{-k}x)
    +(\alpha+\beta)^{m}f(2^{-m}x) ,
    \qquad 0\le m\le L_x. 
\end{equation}

This leads to the following identity. 
\begin{lemma}\label{Lemma3} 
Assume that $f$ satisfies \eqref{b1} and \eqref{b2}. Then,
for $x\ge 1$,
\begin{equation}\label{b15}
    x^{-\rho }f(x)
    =\sum_{k\ge 0}(2^{-k}x)^{-\rho}g(2^{-k}x)
    +f(1)P_0\left( \log_{2}x\right) ,  
\end{equation}
where $\rho:=\log_2(\ga+\gb)$ and 
\begin{equation}\label{b16}
    P_0\left( t\right)
    := \left( 1+( \alpha +\beta-1) 
    \varphi\lpa{2^{\{t\}}-1}\right) 
    (\alpha +\beta)^{-\{t\}}  
\end{equation}
is a continuous $1$-periodic function satisfying $P_0(0)=P_0(1)=1$.
\end{lemma}
Since $g(x)=0$ for $x\le1$, the sum in \eqref{b15} is indeed finite.
\begin{proof}
First, \eqref{b16} yields $P_0(0)=1$ and $\lim_{t\nearrow 1}P_0(t)
=\bigpar{1+(\ga+\gb-1)\gf(1)}(\ga+\gb)\qw =(\ga+\gb)(\ga+\gb)\qw
=1=P_0(1)$. Furthermore, $P_0$ is continuous on $[0,1)$ since $\gf$
is continuous, and $P_0$ is $1$-periodic. It follows that $P_0$ is
continuous on $\mathbb R$.

By \eqref{b1}, $f(2)=\left(\alpha +\beta \right) f(1)+g(2)$. Thus,
for $1\le x< 2$, recalling $g(1)=0$,
\begin{align}\label{b16a}
    f(x) 
    &= f(1)+\varphi(x-1)(f(2)-f(1))  \notag\\
    &= f(1)+\varphi(x-1) \bigpar{(\alpha+\beta-1)f(1)+g(2)}\notag\\
    &= f(1)\bigpar{1+(\ga+\gb-1)\gf\xpar{x-1}}+g(x)\notag\\
    &= f(1)P_0\bigpar{\log_2x}(\ga+\gb)^{\log_2x}+g(x).
\end{align}

Now consider $x\ge1$. Take $m=L_{x}$ in (\ref{b10}) and use
\eqref{b16a} with $x$ replaced by $2^{-L_x}x=2^{\gth_x}\in[1,2)$;
we then obtain, by the relation $\ga+\gb=2^\rho$,
\begin{align}\label{b17}
    f(x) 
    &=\sum_{0\le k<L_{x}}(\alpha +\beta)^{k}g(2^{-k}x)
    +(\alpha+\beta )^{L_{x}}f(2^{-L_{x}}x) \notag\\
    &=\sum_{0\le k\le L_{x}}(\alpha +\beta )^{k}g(2^{-k}x)
    +(\alpha+\beta )^{L_{x}+\gth_x}f(1)P_0\bigpar{\gth_x}\notag \\
    &=\sum_{ k\ge0}2^{\rho k}g(2^{-k}x)
    +2^{\rho\log_2x}f(1)P_0\bigpar{\log_2x}.
\end{align}
This implies (\ref{b15}).
\end{proof}

\begin{rem}\label{Ra=b}
If $\ga=\gb$, then $\gf(t)=t$ by \refR{Rlinear}, and thus \eqref{b16} 
yields
\begin{equation}\label{b16=}
    P_0\left( t\right)
    := \bigpar{ 1+( \alpha +\beta-1) 
    \lpa{2^{\{t\}}-1}}
    (\alpha +\beta)^{-\{t\}} . 
\end{equation}
In the case $\ga=\gb=1$ studied in \cite{Hwang2017}, and also in the
case $\ga=\gb=\frac12$, this yields $P_0(t)\equiv1$. In all other
cases with $\ga=\gb>0$, the periodic function $P_0(t)$ is infinitely
differentiable in $(0,1)$, but a simple calculation shows that the
derivative has a jump at the integers; hence, $P_0$ is Lipschitz but
not continuously differentiable.
\end{rem}

\begin{rem}\label{Rb=1}
If $\gb=1$, then \eqref{b16} and \eqref{b4} yield also
  \begin{align}\label{b161}
    P_0\left( t\right)
    = (\alpha +1)^{1-\{t\}}\varphi\lpa{2^{\{t\}-1}}.  
\end{align}
\end{rem}

We can now give an extension of Theorem 2 in our previous paper 
\cite{Hwang2017}.
\begin{thm}\label{Theorem1} 
Suppose that $f$ and $g$ are given by \eqref{b1} and \eqref{b2}.
The following are equivalent.
\begin{romenumerate}
\item \label{Theorem1a}
$n^{-\rho }f(n)=P(\log_{2}n)+o(1)$ as $n\to \infty $, for some
continuous $1$-periodic function $P$ on $\mathbb{R}$.
    \item \label{Theorem1ab} 
    $x^{-\rho }f(x)=P(\log_{2}x)+o(1)$ as $x\to \infty $, for some 
    continuous $1$-periodic function $P$ on $\mathbb{R}$.
    
    \item \label{Theorem1b}
    \begin{align}\label{t1b}
        x^{-\rho }f(x)
        =P(\log_{2}x)+o(1) \qquad \text{as $x\to \infty $}, 
    \end{align}
    for some $1$-periodic function $P$ on $\mathbb{R}$.
    
    \item \label{Theorem1c}
The sum
\begin{align}\label{t1Q}
    Q(x):=\sum_{m\ge 1}2^{-\rho m}g(2^{m}x)  
\end{align}
converges uniformly
    for $x\in [1,2]$.
\end{romenumerate}
Furthermore, when these conditions hold, 
\begin{equation}\label{t1z}
    f(x)
    =x^{\rho }P(\log_{2}x)-Q(x), \qquad x\ge1,
\end{equation}
where $Q(x)=o(x^\rho)$ as $x\to\infty$, $Q(x)$ is defined by
\eqref{t1Q} for all $x>0$, and the continuous periodic function
$P(t)$ is given by
\begin{equation}\label{t1y}
    P(t)
    =\sum_{m\in \mathbb{Z}}2^{-\rho (m+t)}g(2^{m+t})
    +f(1)P_0(t), \qquad t\in\mathbb{R},
\end{equation}
with $P_0(t)$ given by \eqref{b16}. 
\end{thm}
Note that \eqref{t1z} is not only an identity but also an 
asymptotic expansion.  

Before proving \refT{Theorem1}, we give two partial results.

\begin{prop}\label{Proposition1} 
Suppose that $h(x)$ is a function such that $h(x)$ lies between
$h\left( \left\lfloor x\right\rfloor \right) $ and $h\left(
\left\lceil x\right\rceil \right)$. Then, the following are 
equivalent.
\begin{romenumerate}
    \item\label{P1a}
    $n^{-\rho }h(n)=P(\log_{2}n)+o(1)$ as $n\to \infty $, for some
    continuous $1$-periodic function $P$ on $\mathbb{R}$.
    
    \item\label{P1b}
    $x^{-\rho }h(x)=P(\log_{2}x)+o(1)$ as $x\to \infty $, for some
    continuous $1$-periodic function $P$ on $\mathbb{R}$.
\end{romenumerate}
\end{prop}

\begin{proof}
\ref{P1b}$\implies$\ref{P1a} is trivial.  

The proof of \ref{P1a}$\implies$\ref{P1b} is very similar to the
proof of (i)$\implies $(ii) in \cite[Theorem 2]{Hwang2017}, so we
omit some details. First, $\log_2x-\log_2\floor{x}=O(1/x)=o(1)$ for
large $x$ and $P$ is uniformly continuous, so $P(\log_2x)
=P\bigpar{\log_2\floor{x}}+o(1)$. Hence, \ref{P1a} implies that
\begin{align}
    \floor{x}^{-\rho} h\bigpar{\floor{x}}
    =P\bigpar{\log_2\floor{x}}+o(1)
    =P\bigpar{\log_2{x}}+o(1)
\end{align}
and consequently, ${x}^{-\rho} h\bigpar{\floor{x}}
=P\bigpar{\log_2{x}}+o(1)$. Similarly, ${x}^{-\rho}
h\bigpar{\ceil{x}}=P\bigpar{\log_2{x}}+o(1)$, and \ref{P1b} follows.
\end{proof}

\begin{prop}\label{Proposition2} 
Suppose that $g(x)$ is a continuous function on $(0,\infty)$ with 
$g(x)=0$ for $x\le 1$. Define
\begin{equation}\label{b12}
    h(x)
    :=\sum_{k\ge 0}2^{k\rho }g(2^{-k}x).  
\end{equation}
Then, the following are equivalent.
\begin{romenumerate}
    \item \label{P2a}
    $x^{-\rho }h(x)=P_1(\log_{2}x)+o(1)$ as $x\to \infty $, for some
    continuous $1$-periodic function $P_1$ on $\mathbb{R}$.
    \item  \label{P2b}
    $x^{-\rho }h(x)=P_1(\log_{2}x)+o(1)$ as $x\to \infty $, for 
	some $1$-periodic function $P_1$ on $\mathbb{R}$.
    \item \label{P2c} 
    $Q(x):=\sum_{k\ge 1}2^{-\rho k}g(2^{k}x)$ converges uniformly
    for $x\in [1,2]$.
\end{romenumerate}
Furthermore, when these conditions hold, $Q(x)$ is defined for all 
$x>0$, $Q(x)=o(x^\rho)$ as $x\to\infty$,
\begin{equation}\label{b13}
    P_1(t)
    =\sum_{m\in \mathbb{Z}}2^{-\rho (m+t)}g(2^{m+t}),
    \qquad t\in\mathbb{R}, 
\end{equation}
and
\begin{equation}\label{b14}
    h(x)
    =x^{\rho }P_1(\log_{2}x)-Q(x),\qquad x>0.  
\end{equation}
\end{prop}
\begin{proof}
Again, the proof differs mainly notationally from the proofs of the
corresponding implications in \cite[Theorem 2]{Hwang2017}, and we
omit some details. Let $G_m(x):=\sum_{k=0}^m2^{-\rho k}g(2^{k}x)$,
and note that \ref{P2c} is equivalent to the property that $G_m(x)$
converges uniformly on $[1,2]$ to $G(x) :=Q(x)+g(x)$.

\ref{P2a}$\implies$\ref{P2b}. Trivial.

\ref{P2b}$\implies$\ref{P2c}.
Suppose that $y\in[1,2]$ and $m\ge0$. Then \eqref{b12} yields
\begin{align}\label{bb1}
    h\bigpar{2^my}
    =\sum_{0\le k\le m} 2^{k\rho}g\bigpar{2^{-k+m}y}
    =\sum_{0\le j\le m} 2^{(m-j)\rho}g\bigpar{2^{j}y}
    =2^{m\rho}G_m(y).
\end{align}
Hence, taking $x=2^my$, \ref{P2b} implies, as $m\to\infty$, uniformly
for $y\in[1,2]$,
\begin{align}\label{bb2}
    y^{-\rho}G_m(y)
    =(2^my)^{-\rho}  h\bigpar{2^my}
    =P_1\bigpar{m+\log_2y}+o(1)
    =P_1\bigpar{\log_2y}+o(1).
\end{align}
Hence $y^{-\rho}G_m(y)$ converges uniformly on $[1,2]$, and thus
$G_m(y)$ converges uniformly on $[1,2]$.

\ref{P2c}$\implies$\ref{P2a}.
Conversely, \eqref{bb1} now yields
\begin{align}\label{bb3}
    (2^my)^{-\rho}  h\bigpar{2^my}
    =y^{-\rho}G_m(y)
    =y^{-\rho}G(y)+o(1)
\end{align}
as $m\to\infty$, uniformly for $y\in[1,2]$. Next, we show that
\begin{align}\label{bb4}
    P_1(t)
    :=\sum_{m\in\mathbb{Z}} 2^{-\rho(m+t)}g\bigpar{2^{m+t}}.
\end{align}
is a well-defined $1$-periodic function. First, if $t\in[0,1]$, then
all terms with $m<0$ vanish and thus
\begin{align}\label{bb5}
    P_1(t)
    :=\sum_{m\ge0} 2^{-\rho(m+t)}g\bigpar{2^{m+t}}
    =2^{-\rho t} G(2^t),
\end{align}
where the sum converges uniformly for $t\in\oi$ by assumption. This
shows that the sum in \eqref{bb4} converges for $t\in[0,1]$, and 
that $P_1(t)$ is continuous there. Furthermore, the sum in 
\eqref{bb4} is $1$-periodic in $t$, and thus the sum converges for 
all real $t$ and defines a $1$-periodic continuous function 
$P_1(t)$.

Taking $m=L_x$ and $y=2^{\gth_x}$ in \eqref{bb3} yields as
$x\to\infty$, using \eqref{bb5} and the periodicity of $P_1$,
\begin{align}
    x^{-\rho}h(x)
    =2^{-\rho\gth_x}G\bigpar{2^{\gth_x}}+o(1)
    =P_1\bigpar{\gth_x}+o(1)
    =P_1\bigpar{\log_2 x}+o(1).
\end{align}
Hence \ref{P2a} holds.

Finally, \eqref{b13} is \eqref{bb4}, which also shows that the sum
defining $Q(x)$ converges for all $x>0$, and that
\begin{align}
    h(x)+Q(x)
    &=\sum_{k\ge 0}2^{k\rho }g(2^{-k}x)
    +\sum_{k\ge 1}2^{-\rho k}g(2^{k}x) \notag\\
    &=\sum_{k\in\mathbb Z}2^{-k\rho }g(2^{k}x)\notag\\
    &=x^\rho P_1\bigpar{\log_2x}.
\end{align}
Moreover, %for $x\in[1,2]$,
\begin{align}
  (2^mx)^{-\rho}Q(2^mx) =x^{-\rho} \sum_{k\ge1} 2^{-\rho(k+m)}g(2^{k+m}x) 
=x^{-\rho} \sum_{k\ge m+1} 2^{-\rho k}g(2^{k}x), 
\end{align}
which by \ref{P2c} converges to 0 as $m\to\infty$, uniformly for
$x\in[1,2]$. Hence, $x^{-\rho}Q(x)\to0$ as $x\to\infty$.

This completes the proof of the proposition. 
\end{proof}

\begin{proof}[Proof of \refT{Theorem1}]
First, \ref{Theorem1a}$\iff$\ref{Theorem1ab} follows from  
\refP{Proposition1}, with $h(x)=f(x)$.

Next, define $h(x):=f(x)-x^\rho f(1)P_0\bigpar{\log_2x}$, and note
that \eqref{b15} implies \eqref{b12}. Thus
\ref{Theorem1ab}$\iff$\ref{Theorem1b}$\iff$\ref{Theorem1c} and the
last sentence of the statement follow from \refP{Proposition2}, with
$P(t):=f(1)P_0(t)+P_1(t)$.
\end{proof}

A more practical condition than uniform convergence is the 
following. 
\begin{cor}\label{C1}
Define
\begin{equation}
    A_{m}
    :=\sup_{2^{m}\le n\le 2^{m+1}}|g(n)|.
\end{equation}
If\/ $\,\sum_{m}2^{-m\rho }A_{m}<\infty $ then \eqref{t1b} and
\eqref{t1z} hold, where $P$ is continuous, periodic and given by
\eqref{t1y}. \qed
\end{cor}

\begin{cor}\label{C2}
Suppose that $g(n)=O\bigpar{n^{\rho-\eps}}$ for some $\eps>0$. Then
\eqref{t1b} and \eqref{t1z} hold with $P$ continuous, periodic and
given by \eqref{t1y}, and $Q(x)=O(x^{\rho-\eps})$. \qed
\end{cor}

\begin{example}\label{Eg=0}
One simple but important case is when $g(n)=0$, $n\ge2$, i.e., the
recurrence $\Lambda_{\ga,\gb}[f]=0$. By suitable normalisations, we
may assume $f(1)=1$. As the sequence $f$ satisfying
$\Lambda_{\alpha,\beta}[f]=0$ with $f(1)=1$ plays a fundamental role
in most of our applications, \emph{we denote the solution by
$S_{\alpha,\beta}(n)$ throughout this paper}.

\refT{Theorem1} applies trivially, with $Q(x)=0$ and $P(t)=P_0(t)$, 
and thus   
\begin{align}\label{Sgagb}
	S_{\ga,\gb}(n)
	= f(n)
	= n^{\log_2(\ga+\gb)}P_0(\log_2 n) 
,\qquad n\ge1,
\end{align}
where $P_0(t)$ is given by \eqref{b16}; thus also, more explicitly,
\begin{align}\label{Sgagb2}
	S_{\ga,\gb}(n)&
    = \bigpar{ 1+( \alpha +\beta-1) 
    \varphi\lpa{2^{\para{\log_2n}}-1}}
    (\alpha +\beta)^{\floor{\log_2 n}}.
\end{align}

(In the case $\ga=\gb$, when $\gf(x)=x$ by \refR{Rlinear}, 
\eqref{Sgagb} was found, in an equivalent form, by \cite{Ellul2005}.)
\end{example}

\begin{example}\label{Eodd}
Another simple but important case is when $g(2m)=0$ for all $m\ge1$,
i.e., $g(n)$ is non-zero only for odd $n$. Suppose also, for example,
that $g(n)=O(n^{\rho-\eps})$ so that \refC{C2} applies. Then,
\eqref{t1Q} shows that $Q(n)=0$ for every integer $n$, and thus
\eqref{t1z} yields the identity $f(n)=n^\rho P(\log_2 n)$, $n\ge1$.
\end{example}

\begin{remark}\label{RO}
We have here studied the case with $o$-estimates and convergence in
e.g.\ \refT{Theorem1}. Analogously, we note that it follows easily
from \refL{Lemma3} and \eqref{bb1} that
\begin{equation}
	 \begin{split}
		f(n)=O(n^\rho)
		&\iff
		f(x)=O(x^\rho) \quad\text{for }x\ge1 \\
		&\iff
		% G_m(x):=
		\sum_{0\le k\le m}2^{-\rho k}g(2^{k}x)=O(1)
		\quad\text{for $x\in[1,2]$ and $m\ge0$}.
	 \end{split}
\end{equation}
\end{remark}

\begin{remark}\label{Rfast}
We concentrate in this paper on the case when $g(n)$ grows more
slowly than $n^\rho$ (and \eqref{t1Q} converges), and thus the sums
in \eqref{b10} and \eqref{b15} are dominated by terms with $k$ large
(more precisely, $k=L_x+O(1)$ in \eqref{b15}). One might also
consider the opposite case, when $g(n)$ grows more rapidly that
$n^\rho$. In this case, the sums in \eqref{b10} and \eqref{b15} are
dominated by their first terms, which shows that (typically, at
least) $f(n)$ grows at the same rate as $g(n)$, and has the same
smoothness properties. (In particular, there is no smoothening effect
as we can see in \eqref{t1z}.) We consider one example as part of
\refE{EMoser}, but we otherwise leave this case to a future study. 
(Except for several examples where we reduce to a slower growing $g$ 
by subtracting a polynomial.)
\end{remark}

\section{Smoothness properties of the periodic function $P$}\label{S:smooth}

We prove in this section that under certain conditions on $g(n)$
stronger than those in \refT{Theorem1}, the periodic function $P$
is \Holder{} continuous, and has an absolutely convergent Fourier
series expansion. We also show that the interpolating function
$\gf$ is \Holder{} continuous, and thus the interpolated function
$f(x)$ is not only continuous but always locally \Holder{} continuous.

In this section, starting from the recursion \eqref{a1}, we tacitly 
assume that $\alpha,\beta>0$ and $g(1)=0$; also $P_0$, $P_1$ and $P$ 
are functions defined as in \refS{S:recurrence}; furthermore, we 
define
\begin{align}\label{E:lambda}
    \lambda
    :=\log_{2}\frac{\alpha +\beta}
    {\alpha \vee \beta}\in (0,1].
\end{align}
Note that $\gl=1$ if and only if $\ga=\gb$. 

\paragraph{Bounded variation, Lipschitz continuity and H\"older 
continuity.}
We recall some standard definitions. A function $\phi $ is 
\emph{Lipschitz continuous} on an interval $[a,b]$ if there exists 
a positive number $C$ such that
\begin{equation}
    |\phi(x)-\phi(y)| 
    \le C|x-y| \quad \text{(}x,y\in [a,b]\text{).}
\end{equation}
This definition extends to \emph{H\"older continuity} by replacing 
the last inequality by
\begin{equation}
    | \phi(x)-\phi(y)| 
    \le C| x-y|^{\gamma}\quad \text{(}x,y\in [a,b]\text{)},
\end{equation}
for some $0<\gamma \le 1$. Let $\mathrm{H}_{\gamma}[a,b]$ be the 
space of functions $\phi $ on $[a,b]$ such that the seminorm
\begin{equation}\label{c1}
    \| \phi\|'_{\mathrm{H}_{\gamma}[a,b]}
    :=\sup_{a\le x<y\le b}\frac{|\phi(x)-\phi(y)|}
    {|x-y|^{\gamma}}  
\end{equation}
is finite; this is a Banach space with the norm
\begin{equation}\label{c2}
    \|\phi\|_{\mathrm{H}_{\gamma}[a,b]}
    :=\|\phi\|'_{\mathrm{H}_{\gamma}[a,b]}
    +\sup_{x\in [a,b]}| \phi(x)| .  
\end{equation}

A function $\phi $ is of \emph{bounded variation} on $[a,b]$ if its
total variation is bounded. Such a function is differentiable almost
everywhere. Let $\mathrm{BV}[a,b]$ be the space of functions on
$[a,b]$ of bounded variation, with the norm
\begin{equation}
    \|\phi\|_{\mathrm{BV}[a,b]}
    :=V(\phi;a,b)+\sup_{[a,b]}|\phi|,
\end{equation}
where $V$ denotes the total variation.

Note that both $\mathrm{BV}[a,b]$ and $\mathrm{H}_{\gamma}[a,b]$ are
Banach algebras. Also $\mathrm{H}_{1}\subset \mathrm{BV}$, namely,
Lipschitz continuity implies bounded variation.

\paragraph{Smoothness of $\gf$ and $P$.}
We prove first that the interpolating function $\gf$ is H\"older
continuous. Note that $\gf$ is trivially of bounded variation since
it is monotone.

\begin{lemma}\label{LHolder}
    $\varphi \in\mathrm{H}_{\lambda}[0,1]$.
\end{lemma}

\begin{proof}
First, let $x=m2^{-N}$ and $y=(m+1)2^{-N}$ for some integers 
$N\ge0$ and $m<2^N$. Then \eqref{bb10} implies
\begin{align}\label{lh1}
    \gf(y)-\gf(x) 
    \le \Bigpar{\frac{\ga\vee \gb}{\ga+\gb}}^{N}
    =\bigpar{2^{-\gl}}^{N}
    =2^{-N\gl}.
\end{align}

For general $x$ and $y$ with $0\le x< y\le 1$, let $N:=\ceil{-\log_2
\left( y-x \right)}$ and $k:=\floor{x2^N}$. Then $2^{-N}\le y-x \le
2^{1-N}$ and $k2^{-N} \le x<y\le (k+3)2^{-N}$. Hence, using
\eqref{lh1} thrice,
\begin{equation}
    \varphi(y)-\varphi(x)
    \le \gf\bigpar{(k+3)2^{-N}} - \gf\bigpar{k2^{-N}}
    \le  3\cdot 2^{-N\gl}
    \le 3|y-x|^\gl.
\end{equation}
This proves that $\varphi \in\mathrm{H}_{\lambda}[0,1]$.
\end{proof}

\begin{lemma}\label{Lemma5}
    $P_0\in \HH_\gl[0,1]\cap \BV[0,1]$.
\end{lemma}

\begin{proof}
For $t\in[0,1]$, we may replace $\{t\}$ by $t$ in \eqref{b16}. It
then follows from \refL{LHolder} that the first factor in
\eqref{b16} belongs to $\mathrm{H}_{\lambda}\oi$, and so does the
second factor since it has a bounded derivative. Hence, $P_0\in
\mathrm{H}_{\lambda}\oi$.

Similarly, both factors in \eqref{b16} then are monotone on
$[0,1)$, and therefore of bounded variation. Hence $P_0\in
\mathrm{BV}[0,1]$.
\end{proof}

To treat the function $P$, we need a smoothness assumption on the 
sequence $g(n)$. Let
\begin{align}\label{AB}
    A_{m}
    :=\max_{2^{m}\le n\le 2^{m+1}}| g(n)| 
	\eqtext{and}
	B_{m}:=\max_{2^{m}\le n<2^{m+1}}| g(n+1)-g(n)| .
\end{align}

\begin{lemma}\label{Lemma4} 
If\/ $\rho>0$ and 
\begin{equation}\label{c3}
    \sum_{m\ge0}2^{(1-\rho )m}B_{m}
    <\infty,
\end{equation}
then  
\begin{align}\label{c3A}
    \sum_{m\ge0}2^{-\rho m }A_{m}
    <\infty .
\end{align}
\end{lemma}

Note that $\rho>0$ is equivalent to $\ga+\gb>1$.

\begin{proof}
Let $m\ge1$. For every $n$ lying in the interval $2^{m}\le n\le 
2^{m+1}$,
\begin{align}\label{cj1}
    |g(n)|
	&\le |g(2)|+\sum_{2\le j<n}|g(j+1)-g(j)|\notag\\
    &\le |g(2)|+\sum_{1\le k\le m}2^{k}B_{k},
\end{align}
implying that
\begin{equation}\label{cj2}
    A_{m}
    \le |g(2)|+\sum_{1\le k\le m}2^{k}B_{k}.
\end{equation}
Thus,
\begin{align}\label{cj3}
    \sum_{m\ge1}2^{-\rho m}A_{m}
    &\le \sum_{m\ge1}2^{-\rho m}
    \biggpar{|g(2)|+\sum_{1\le k\le m}2^{k}B_{k}}\notag\\
    &=|g(2)|\frac{2^{-\rho}}{1-2^{-\rho}}+\frac{1}{
    1-2^{-\rho}}\sum_{k\ge1}2^{(1-\rho)k}B_{k}.
\end{align}
This proves the lemma.
\end{proof}

\begin{lemma}  \label{LProp1}
If\/ $\rho>0$ and \eqref{c3} holds, then $P\in \HH_\gl[0,1]\cap 
\BV[0,1]$.
\end{lemma}

\begin{proof}
Let $g_{m}(x):=g(2^{m}x)$. Then, by (\ref{t1y}) and $g(x)=0$ for 
$x\le 1$, we have
\begin{equation}\label{c5}
    P(t)
    =2^{-\rho t}\sum_{m\ge0}
    2^{-\rho m}g_{m}(2^{t})+f(1)P_0(t)
	\qquad(t\in [0,1]).
\end{equation}
Since the function $2^{-\rho t}$ belongs to $\mathrm{BV}[0,1]$, 
which is a Banach algebra, we see that
\begin{equation}\label{c6}
    \Vert P(t)\Vert_{\mathrm{BV}[0,1]}
    \le \CC\sum_{m\ge0}2^{-\rho m}
    \Vert g_{m}(2^{t})\Vert_{\mathrm{BV}[0,1]}
    +|f(1)|\,\Vert P_0(t)\Vert_{\mathrm{BV}[0,1]},  
\end{equation}
for some constant $\CCx>0$. Furthermore, by the monotonicity of
the interpolating function $\varphi$, we obtain
\begin{align}\label{c7}
    \Vert g_{m}(2^{t})\Vert_{\mathrm{BV}[0,1]}
	&=\Vert g_{m}(x)\Vert_{\mathrm{BV}[1,2]}\notag \\
    &=\Vert g(x)\Vert_{\mathrm{BV}[2^{m},2^{m+1}]}\notag\\
    &=\sup_{2^m \le x\le 2^{m+1}}|g(x)| +
	\sum_{2^{m}\le n<2^{m+1}}|\Delta g(n)|\notag\\
	&\le A_m+ 2^{m}B_{m},
\end{align}
where $\Delta g(n):=g(n+1)-g(n)$. It follows from
(\ref{c6})--(\ref{c7}), \refL{Lemma5}, (\ref{c3}) and \refL{Lemma4}
that $P(t)\in \mathrm{BV}[0,1]$.

Similarly, for the H\"older norm, we have
\begin{align}\label{c8}
    \Vert P_1(t)\Vert_{\mathrm{H}_{\lambda}[0,1]}
    &\le \CC\sum_{m\ge 0}2^{-\rho m}
    \Vert g_{m}(2^{t})\Vert_{\mathrm{H}_{\lambda}[0,1]} \notag\\
    &\le \CC\sum_{m\ge 0}2^{-\rho m}
    \Vert g_{m}(x)\Vert_{\mathrm{H}_{\lambda}[1,2]},
\end{align}
where, furthermore, by the definitions \eqref{c1}--\eqref{c2} and 
\eqref{AB}, 
\begin{align}\label{c8a}
    \Vert g_{m}(x)\Vert_{\mathrm{H}_{\lambda}[1,2]} 
    &= \Vert g_{m}(x)\Vert'_{\mathrm{H}_{\lambda}[1,2]} 
    +\sup_{x\in[1,2]}|g_m(x)| \notag\\
    &= 2^{\lambda m} 
	\Vert g(x)\Vert'_{\mathrm{H}_\lambda[2^{m},2^{m+1}]}
    +\sup_{x\in[2^{m},2^{m+1}]}|g(x)| \notag\\
    &= 2^{\lambda m} 
	\Vert g(x)\Vert'_{\mathrm{H}_\lambda[2^{m},2^{m+1}]}
    +A_m.
\end{align}
In order to bound $\Vert g\Vert'_{\mathrm{H}_\lambda[2^{m},
2^{m+1}]}$, we estimate $|g(y)-g(x)|$ for $2^{m}\le x\le y\le 2^{m+1}
$. By splitting the interval $[x,y]$ into $[x,\left\lceil
x\right\rceil ]$, $[\left\lceil x\right\rceil ,\left\lfloor
y\right\rfloor ]$ and $[\left\lfloor y\right\rfloor ,y]$, it suffices
(up to a constant factor in the norm) to consider the two cases $n\le
x\le y\le n+1$ and $x=n$, $y=n+\eta$, where $n$ and $\eta$ are
integers.

In the first case, $n\le x\le y\le n+1$ with $2^{m}\le n<2^{m+1}$, we
have $g(y)-g(x)=\Delta g(n)\left( \varphi (y)-\varphi(x)\right)$.
Since $\varphi \in \mathrm{H}_{\lambda}[0,1]$ by \refL{LHolder},
\begin{equation}\label{c9}
    |g(y)-g(x)|
    \le \CC|\Delta g(n)|\,|y-x|^{\lambda}
    \le \CCx B_{m}|y-x|^{\lambda}.  
\end{equation}
In the second case,
\begin{equation}\label{c10}
    |g(y)-g(x)|
    =|g(n+\eta)-g(n)|
    \le \sum_{0\le i<\eta}|\Delta g(n+i)|
    \le \eta B_{m}
    \le B_{m}2^{m(1-\lambda )}\eta^{\lambda}.
\end{equation}
Combining the two cases, we obtain from (\ref{c9}) and (\ref{c10})
\begin{equation}\label{c11a}
    \Vert g\Vert_{\mathrm{H}_{\lambda}[2^{m},2^{m+1}]}^{\prime }
    \le \CC B_{m}2^{(1-\lambda )m}.  
\end{equation}
Consequently, \eqref{c8a} yields
\begin{equation}
    \Vert g_m\Vert_{\mathrm{H}_{\lambda}[1,2]}
    \le \CCx2^{m} B_{m} +A_m. \label{c11b}
\end{equation}

It follows from (\ref{c8}), (\ref{c11b}), (\ref{c3}) and
\refL{Lemma4} that $P_1\in \mathrm{H}_{\lambda}[0,1]$, and then $P\in
\mathrm{H}_{\lambda}[0,1]$ by \refL{Lemma5}.
\end{proof}

\begin{rem}\label{RHolder}
\refL{LHolder} is best possible: $\gf\notin \mathrm{H}_{\gam }\oi$
for $\gam>\gl$. In particular, $\gf$ is not Lipschitz continuous
(and not differentiable) unless $\ga=\gb$. To see this, it suffices
to note that \eqref{bb10} yields $\gf(2^{-j})-\gf(0)=
\bigpar{\frac{\gb}{\ga+\gb}}^j$ and $\gf(1)-\gf(1-2^{-j})=
\bigpar{\frac{\ga}{\ga+\gb}}^j$ for all $j\ge1$, and one of these
equals $\bigpar{\frac{\ga\vee\gb}{\ga+\gb}}^j
=\bigpar{2^{-j}}^{\gl}$.

Hence, \eqref{b16} shows that if $\ga+\gb\neq1$, then also
\refL{Lemma5} is best possible: $P_0\notin \mathrm{H}_{\gam }\oi$ for
$\gam>\gl$. Furthermore, \eqref{t1y} shows that typically also
$P\notin \mathrm{H}_{\gam }\oi$ for $\gam>\gl$. (Also in the case
$f(1)=0$, since $g(x)$ has the same smoothness as $\gf$.) However,
note that $P$ may be more smooth in special cases (which means that
there is cancellation of non-smoothness in \eqref{t1y}); for example,
we may take any continuously differentiable periodic function $P(t)$
and let $f(n):=n^\rho P(\log_2n)$ and then define $g(n)$ by
\eqref{b1}.
\end{rem}

\paragraph{Fourier series.} 
The periodic function $P(t)$ may be described by its Fourier
coefficients; these are given by the following formula, where we use
the notation
\begin{align}\label{chik}
	\chi_{k}:=\frac{2k\pi }{\log 2}\ii,
	\qquad  k\in\bbZ.
\end{align}

\begin{thm}\label{Theorem4}
\begin{thmenumerate}
\item \label{T4a}
If \refT{Theorem1}\ref{Theorem1c} holds (and thus all statements in
\refT{Theorem1}), then the Fourier coefficients
$\hP(k):=\int_{0}^{1}P(t)e^{-2k\pi it}\dd t$ of $P(t)$ are given by
\begin{equation}\label{c4}
    \widehat{P}(k)
    =\frac{1}{\log 2}\int_{1}^{\infty }
    \frac{g(u)}{u^{\rho + \chi_{k}+1}}\dd u
    +\frac{f(1)}{\log 2}\int_{0}^{1}
    \frac{1+(\alpha +\beta -1)\varphi(u)}
    {(1+u)^{\rho +\chi_{k}+1}}\dd u,
%\qquad k\in\bbZ,
\end{equation}
where the first integral is (at least) conditionally convergent.

\item\label{T4b} 
If \eqref{c3A} holds, then \eqref{c4} holds, with absolutely 
convergent integrals. 

\item \label{T4c}
If\/ $\ga+\gb>1$ and \eqref{c3} holds, then $P(t)$ has an absolutely
convergent Fourier series $P(t)=\sum_{k\in
\mathbb{Z}}\widehat{P}(k)e^{2\pi ikt}$, for $t\in \mathbb{R}$, where
the coefficients are given by \eqref{c4}, with absolutely convergent
integrals.
\end{thmenumerate}
\end{thm}

\begin{proof}
\pfitemref{T4a}
The Fourier coefficients of $P(t)$ are given by, using the
definitions (\ref{t1y}), \eqref{rho} and (\ref{b16}),
\begin{align}\label{c41}
    \widehat{P}(k)
    &:=\int_{0}^{1}P(t)e^{-2k\pi it}\dd t
    \notag\\
    &=\sum_{m\in \mathbb{Z}}\int_{0}^{1}g\left( 2^{m+t}\right)
    2^{-\rho(m+t)}e^{-2k\pi it}\dd t
    \notag\\
    &\qquad+f(1)\int_{0}^{1}
    \left( 1+(\alpha +\beta -1)\varphi\left(2^{t}-1\right)
    \right) (\alpha +\beta )^{-t}e^{-2k\pi it}\dd t
    \notag\\
    &=\int_{-\infty}^{\infty}g\xpar{ 2^{t}}
    2^{-\rho t}e^{-2k\pi it}\dd t
    \notag\\
    &\qquad+f(1)\int_{0}^{1}
    \left( 1+(\alpha +\beta -1)\varphi\left( 2^{t}-1\right)
    \right) 2^{-\rho t}e^{-2k\pi it}\dd t,
\end{align}
where the sum and the integral over $(-\infty,\infty)$ are
(conditionally) convergent because the sum in \eqref{t1y} converges
uniformly, and $g(x)=0$ for $x\le1$. (The integrals over $\oi$ are,
trivially, absolutely convergent.) Finally, \eqref{c4} follows by the
changes of variables $u=2^t$ and $u=2^t-1$, respectively, in the two
integrals.

\pfitemref{T4b}
\refC{C1} shows that all statements in \refT{Theorem4} hold, and thus
part \ref{T4a} applies, which gives \eqref{c4}. The absolute
convergence of the first integral follows by \eqref{c3A}.

\pfitemref{T4c}
By \refL{Lemma4}, we have \eqref{c3A}, and thus \ref{T4b} applies.
Since $P$ is a continuous $1$-periodic function, the absolute
convergence of the Fourier series follows directly from Lemma
\ref{LProp1} by a theorem of Zygmund (see \cite[p.~241,
VI.(3.6)]{Zygmund1959} or \cite[p.~35]{Katznelson1968}).
\end{proof}

The formula \eqref{c4} shows a connection with Mellin transforms;
this is explored further in \refApp{AMellin}.

The conditions of \refT{Theorem4}\ref{T4c} are, of course, not
necessary for absolute convergence of the Fourier series of $P(t)$.
Another simple case is the following.

\begin{example}\label{Eg=0b} 
If $\ga,\gb>0$ and $g(n)=0$, $n\ge2$, we may normalise by $f(1)=1$ as
in \refE{Eg=0}; thus $f(n)=S_{\ga,\gb}(n)$ and \eqref{Sgagb} holds.
$P(t)=P_0(t)$ has an absolutely convergent Fourier series, given by
\eqref{c4} with $g(x)=0$. This follows as in the proof of
\refT{Theorem4}, now using \refL{Lemma5}.
\end{example}

Many other examples of $P(t)$ with absolutely convergent Fourier
series are given in Sections \ref{S:app1}--\ref{S:app2}. An example
where the Fourier series is not absolutely convergent is part of
\refE{EMoser} (for some $\ga$).

In the case $\ga=\gb$, the integrals in \eqref{c4} can easily be
evaluated explicitly, and we obtain the following extension of the
case $\ga=\gb=1$ in \cite[Theorem 3]{Hwang2017}. Let, as in
\cite{Hwang2017},
\begin{align}\label{D1}
	D(s)
	:=\sum_{n\ge2} g(n)\Bigpar{(n-1)^{-s}-2n^{-s}+(n+1)^{-s}},
\end{align}
for all complex $s$ such the sum is convergent; note that if
$g(n)=O(n^{\rho-\eps})$ for some $\eps\ge0$ (this holds at least with
$\eps=0$ whenever \refT{Theorem1} applies), then the sum in
\eqref{D1} is absolutely convergent in the half-space $\Re s
>\rho-1-\eps$, and thus $D(s)$ is analytic there. Note also (as in
\cite{Hwang2017}) that if $\Re s$ is large enough, \eqref{D1} can be
rearranged as a Dirichlet series (using $g(0)=g(1)=0$)
\begin{align}\label{D2}
	D(s)
	=\sum_{n\ge1} \Bigpar{g(n+1)-2g(n)+g(n-1)}n^{-s}.
\end{align}

\begin{cor}\label{ChPa=b}
If $\ga=\gb>0$ and \refT{Theorem1}\ref{Theorem1c} holds (and thus all
statements in \refT{Theorem1}), then, assuming $\rchi_k\neq0,1$,
\begin{align}\label{c4+}
	\hP(k) &= \frac{1}{(\rho+\chi_k)(\rho-1+\chi_k)\log 2}
    \Bigpar{D(\rho-1+\chi_k) 
    + \frac{(2\ga-1)(\ga-1)}{\ga}f(1)}.
\end{align}
\qed
\end{cor}

The formula \eqref{c4+} is used (often tacitly) in numerous examples 
below.

\begin{remark}\label{Rexc}
The two exceptional cases are $k=0$ and either $\ga=\gb=\frac12$ or
$\ga=\gb=1$; in these cases, \eqref{c4+} is of the form $0/0$ and is
replaced by a suitable limit form. The case $\ga=1$ is included in
\cite[Theorem 3]{Hwang2017}; the case $\ga=\frac12$ is similar, but
we omit the details.
\end{remark}

\begin{proof}
\refT{Theorem4} applies and yields \eqref{c4}; we treat the two
integrals in \eqref{c4} separately. For the first integral,
\cite[(2.23) in the proof of Theorem 3]{Hwang2017} holds whenever
$\ga=\gb$, and shows that
\begin{align}\label{maja}
    \int_{1}^{\infty } \frac{g(u)}{u^{\rho + \chi_{k}+1}}\dd u
    = \frac{D(\rho-1+\chi_k) }{(\rho+\chi_k)(\rho-1+\chi_k)},
\end{align}
with the sum \eqref{D1} converging at least conditionally.

Since $\ga=\gb>0$, we have $\gf(u)=u$ and thus the second integral in
\eqref{c4} is, by a simple calculation using
$2^{\rchi_k}=2^\rho=2\ga$,
\begin{align}\label{maj0}
	\int_{0}^{1}  \frac{1+(2\alpha -1)u}
	{(1+u)^{\rho +\chi_{k}+1}}\dd u
	&=\int_{0}^{1}\lrpar{\frac{2\alpha -1} {(1+u)^{\rho +\chi_{k}}}
	+    \frac{2-2\alpha} {(1+u)^{\rho +\chi_{k}+1}}}\dd u
	\notag\\&
	=\frac{(2\ga-1)(\ga-1)}{\ga}
	\cdot\frac{1}{(\rchi_k)(\rho-1+\chi_k)}.
\end{align}
Using \eqref{maja}--\eqref{maj0} in \eqref{c4} yields \eqref{c4+}.
\end{proof}

\begin{remark}\label{RMellin}
For later use we note that, still assuming $\ga=\gb$, if the
condition \refT{Theorem1}\ref{Theorem1c} holds, then
\cite[(2.23)]{Hwang2017} more generally yields the Mellin transform
\begin{align}\label{majs}
    \int_{1}^{\infty } \frac{g(u)}{u^{s+1}}\dd u
    = \frac{D(s-1) }{s(s-1)}
\end{align}
with the integral converging absolutely at least for every complex 
$s$ with $\Re s>\rho$.   
\end{remark}

\begin{example}\label{EhPa=b}
Let $\ga=\gb>0$, and consider $P_0(t)$ given by \eqref{b16=} in
\refR{Ra=b}. By \refC{ChPa=b}, with $g(n)=0$ and $f(1)=1$ (and thus
$D(s)=0$), we have the Fourier coefficients, assuming
$\rchi_k\neq0,1$:
\begin{align}\label{maj1}
	\hP_0(k)
	&=\frac{(2\ga-1)(\ga-1)}{\ga\log2}\cdot
	\frac{1}{(\rchi_k)(\rho-1+\chi_k)},
	\qquad k\in\bbZ.
\end{align}
In particular, \eqref{maj1} verifies that $P_0(t)$ has an absolutely
convergent Fourier series, as shown more generally in \refE{Eg=0b}.
If $\ga\in\para{\frac12,1}$ (so $\rho=0$ or 1), \eqref{maj1} shows
that $\hP_0(k)$ vanishes for every $k\neq0$, which is obvious since
then $P_0(t)\equiv1$ by \refR{Ra=b}. In all other cases, \eqref{maj1}
yields $\bigabs{\hP_0(k)}=\Theta\bigpar{k^{-2}}$ as $k\to\pm\infty$,
which agrees well with the fact that $P_0(t)$ is Lipschitz but not
$C^1$; see again \refR{Ra=b}.
\end{example}

\begin{remark}\label{R:int-phi}
Again, by the recursive relations in \eqref{b4}, we can express the 
second integral in \eqref{c4} in the series form 
\begin{align}\label{E:int-phi}
	&\frac{(\rho+\chi_k)(\alpha+\beta)}{\alpha+\beta-1}
	\int_0^1 \frac{1+(\alpha+\beta-1)\varphi(t)}
	{(1+t)^{\rho+\chi_k+1}}\,\dd t \nonumber\\
	&\quad =  1+ \sum_{m\ge0}\sum_{0\le j<2^m}
	\alpha^{\nu(j)}\beta^{m+1-\nu(j)}
	\Lpa{\frac1{(2^m+j+\tfrac12)^{\rho+\chi_k}}
	-\frac1{(2^m+j+1)^{\rho+\chi_k}}},
\end{align}
which is more useful for numerical purposes, where $\nu(j)$ denotes 
the number of $1$s in the binary expansion of $j$. See 
Appendix~\ref{AHat-Pk} for a proof. 
\end{remark}

\section{Applications, I. $\alpha\ne\beta$}
\label{S:app1}

We discuss applications of our results in this section, grouping them
according to the growth order of $g$. Most examples are taken from
OEIS, sometimes with a shift of the index (which for simplicity of
presentation is not explicitly specified in this paper). For example,
if
\begin{align}
    f(n) = \alpha f\lpa{\ltr{\tfrac {n+d}2}}
    +\beta f\lpa{\lcl{\tfrac {n+d}2}}+ g(n)
    \qquad(n\ge n_0\ge1),
\end{align}
then $\bar{f}(n) := f(n+d)$ satisfies 
\begin{align}
    \bar{f}(n) 
    =\alpha \bar{f}\lpa{\ltr{\tfrac{n}2}}
    +\beta \bar{f}\lpa{\lcl{\tfrac{n}2}}+ g(n+d)
    \qquad(n\ge n_0-d).
\end{align}
Note also that if \eqref{a1} holds only for $n\ge n_0$, we can make
it hold for all $n\ge2$ by redefining $g(n)$ for $2\le n<n_0$. See
A294456 (contained in \refE{E11}) for an example.

Some of the examples are defined in OEIS by a recursion of the form
\eqref{a1}; in other examples, such a recursion is stated as a
property; in yet other examples below, no such recursion is given
explicitly in OEIS, but can be concluded from other properties given
there. Of course, every sequence $f(n)$ satisfies
$\gL_{\ga,\gb}[f]=g$ for some sequence $g(n)$; we are only interested
in cases when $g(n)$ has a simple explicit form, and in particular
does not grow too fast. We regard polynomial terms in $f(n)$ as
essentially trivial, so we also include examples when they dominate
$f$ and the periodic fluctuations constitute a lower-order term. (In
such cases our theorems apply only after subtracting a suitable
polynomial.)

To avoid trivialities, we discard in our discussions sequences from
OEIS whose generating functions are rational with all singularities
on the unit circle. (For example, polynomials.) Such sequences are in
the thousands in OEIS.

For notational convenience, we insert the subscript to $\varphi$ by
writing $\varphi^{}_{\alpha,\beta}(t)$ whenever necessary. The
symbols $f(n), g(n)$ and $P(t)$ are all generic and may differ from
one instance to the other; we also specify explicitly them as
$f_{\text{A006046}}(n)$, $g_{\text{A006046}}(n)$, and
$P_{\text{A006046}}(t)$ if needed. Note that our indexing of the
sequence $f$ may differ from that on OEIS by a shift; for example,
$f_{\text{A006581}}(n)=\text{A006581}(n+1)$ for $n\ge1$. Also the
format $g(n) = \begin{cases}\cdots\\ \cdots \end{cases}$ in the
tables without explicit mention always means the values of $g(n)$ in
the even and odd cases, respectively.

We continue to assume $\ga,\gb>0$, and consider in this section cases
with $\ga\neq\gb$. On the other hand, examples in the special cases
when $\alpha=\beta$ exhibit more structural properties and explicit
expressions, and will be discussed in Section~\ref{S:app2}. The
properties are very similar to the case when $\alpha=\beta=1$ that we
already examined in detail in \cite{Hwang2017}, although there are
also subtle differences on the smoothness of the periodic functions.

\subsection{Periodic equivalence}\label{sec-pe}

We introduce a simple notion here, very useful in identifying the
relation between sequences. The main case is when two sequences $f_1$
and $f_2$ both satisfy \refT{Theorem1} (with the same $\rho$), and
the corresponding periodic functions $P_1$ and $P_2$ are the same, or
more generally proportional. It will be convenient to be a bit more
general, and regard polynomial terms as trivial. We thus define:
\begin{definition}
Two sequences $f_1$ and $f_2$ are said to be \emph{periodically
equivalent} if, for some $\rho$, $f_j(n)=n^\rho P_j(\log_2n)+p_j(n)$,
$j=1,2$, where $p_j(n)$ are polynomials and $P_j$ are periodic
functions such that ${P}_1(k)=c{P}_2(k)$ for some constant $c\neq0$.
For simplicity, we write $f_1\Bumpeq f_2$.
\end{definition}
It is possible to extend the definition to \emph{asymptotically
periodic equivalence} by allowing $f_j(n)=n^\rho
P_j(\log_2n)+p_j(n)+o(n^\rho)$ for $j=1,2$, but the above definition
without $o(1)$-term is sufficient for our use in this paper.

We begin with the simplest cases when $\gL_{\ga,\gb}[f](n)=g(n)=0$.
Such cases cover also the situation when $g(n)=c$ for $n\ge2$ because
normalising $f$ by $\bar{f}(n) := f(n) +\frac{c}{\alpha+\beta-1}$
yields the recurrence $\Lambda_{\alpha,\beta}[\bar{f}]=0$ with
$\bar{f}(1) = f(1) + \frac{c}{\alpha+\beta-1}$. For convenience of
reference, we state this observation, in a somewhat generalised form, 
as a lemma.

\begin{lemma} \label{lmm-g-const}
If\/ $\alpha+\beta>1$, then two sequences defined by
$\Lambda_{\alpha,\beta}[f_1]=0$ with $f_1(1)\ne0$ and
$\Lambda_{\alpha,\beta}[f_2]=c$ are connected by
\begin{align}\label{l4.1}
    f_2(n) 
	= \left(f_2(1)+\frac{c}{\alpha+\beta-1}\right)
    \frac{f_1(n)}{f_1(1)} 
    -\frac{c}{\alpha+\beta-1}.
\end{align}
Hence, if $f_2(1) +\frac{c}{\alpha+\beta-1}\ne0$, the sequences
$f_1(n)$ and $f_2(n)$ are periodically equivalent with the underlying
periodic function $P$ satisfying $P\in \HH_\gl[0,1]\cap \BV[0,1]$, 
where $\lambda$ is defined in \eqref{E:lambda}.
\end{lemma}
\begin{proof}
It is easy to verify that \eqref{l4.1} satisfies the recursion
$\Lambda_{\alpha,\beta}[f_2]=c$. and thus \eqref{l4.1} holds by
induction. Thus $f_1\Bumpeq f_2$ if $f_2(1)
+\frac{c}{\alpha+\beta-1}\ne0$. The H\"older continuity and bounded
variation of $P$ follow from Lemma~\ref{LProp1}.
\end{proof}

\begin{lemma} \label{lmm-g-linear}
Write $f\sim\XLambda_{\alpha,\beta}[c;d,e]$ if\/  
$\Lambda_{\alpha,\beta}[f]=cn+\begin{cases} 
	d, & \text{$n$ even}\\ 
	e, & \text{$n$ odd} \end{cases}$ for $n\ge2$.
Suppose that $\alpha+\beta>2$, and that
\begin{align}\label{L-31a}
    f_1\sim \XLambda_{\alpha,\beta}[c;d,e]
    \quad\text{and}\quad 
    f_2\sim\XLambda_{\alpha,\beta}
    \left[0;0,
    \frac{(\alpha-\beta)c}{\alpha+\beta-2}-d+e\right],
\end{align}
where $f_2(1)=f_1(1)+\frac{2c}{\alpha+\beta-2}
+\frac{d}{\alpha+\beta-1} \neq0$. Then $f_1\Bumpeq f_2$. Furthermore,
\begin{align}\label{L-31b}
    f_1(n) = n^\rho P(\log_2n)-
    \frac{2cn}{\alpha+\beta-2}
    -\frac{d}{\alpha+\beta-1}\qquad(n\ge1)
\end{align}
for a periodic function $P\in \HH_\gl[0,1]\cap \BV[0,1]$.
\end{lemma}
\begin{proof}
The normalised sequence 
\begin{align}\label{L-31c}
    \bar{f}(n) := f_1(n) + \frac{2cn}{\alpha+\beta-2}
    +\frac{d}{\alpha+\beta-1}
\end{align}
satisfies 
\begin{align}\label{L-31d}
    \Lambda_{\alpha,\beta}[\bar{f}] 
    = \left(\frac{(\alpha-\beta)c}{\alpha+\beta-2}-d+e\right)
    \cdot \mathbf{1}_{n\text{ is odd}},
\end{align}
with the initial condition $\bar{f}(1) =
f_1(1)+\frac{2c}{\alpha+\beta-2} +\frac{d}{\alpha+\beta-1}$. Thus,
$\bar{f}=f_2$. Furthermore, $f_2(n)=n^\rho P(\log_2n)$ for a periodic
function $P$ by \refE{Eodd}. This yields \eqref{L-31b}. Again, the
H\"older continuity and bounded variation of $P$ follow from
Lemma~\ref{LProp1}.
\end{proof}

\subsection{$\Lambda_{\alpha, \beta}[f]=0$}
\begin{example}[Generating polynomial of the sum-of-digits 
	function]\label{Enu}
As an immediate application of \refT{Theorem4}, we consider the
following partial sum
\begin{equation}\label{Sn-theta}
    f(n)
    :=\sum_{0\le k<n}\alpha^{\nu(k)}\qquad (n\ge 0),  
\end{equation}
where $\alpha>0$ and $\nu (n)$ denotes the number of $1$s in the
binary expansion of $n$. Such sums with various $\alpha$ have been
encountered and studied in a large number of different contexts;
see the recent survey \cite{Chen2014} and the references therein
for more information. Then by \eqref{Sn-theta} and the recurrence
relation 
\begin{align}\label{nu-rec}
    \nu(n)
    =\nu\left(\left\lfloor\tfrac{n}{2}\right\rfloor\right) 
    +\mathbf{1}_{n\text{ is odd}}\qquad (n\ge 1),  
\end{align}
we see that $f$ satisfies $\Lambda_{\alpha ,1}[f]=0$ with $f(1)=1$,
or (see \refE{Eg=0})
\begin{equation}\label{eq:Sa1}
	f(n) = S_{\alpha,1}(n).
\end{equation}
Thus $g(n)\equiv 0$ for all $n$, so \eqref{c3} is trivial and
\refT{Theorem4}\ref{T4c} applies. Furthermore, \eqref{t1y} yields
$P(t)=P_0(t)$ given by \eqref{b16}, and \refT{Theorem1} or
\refE{Eg=0} shows that
\begin{align}
    f(n)
	=n^{\log_2(\alpha+1)}P_0(\log_2 n) 
	\qquad(n\ge1), 
\end{align}    
where, by \refT{Theorem4}\ref{T4c} and \eqref{b161},  
\begin{align}
    P_0(t) 
	:= (\alpha+1)^{-\{t\}}
	\lpa{1+\alpha\varphi^{}_{\alpha,1}
    \lpa{2^{\{t\}}-1}}
    = (\alpha+1)^{1-\{t\}}
	\varphi^{}_{\alpha,1}\lpa{2^{\{t\}-1}}
\end{align}
has an absolutely convergent Fourier series for $\alpha>0$. This
extends and improves the result established in \cite{Grabner2005} for
$\alpha$ in the range $(\sqrt{2}-1,\sqrt{2}+1)$, where a completely
different approach was employed: instead of Zygmund's theorem applied
above (see the proof of \refT{Theorem4}), the proof in
\cite{Grabner2005} used a theorem of Bernstein saying that a periodic
function $P\in \mathrm{H}_{\lambda}$ with $\lambda >\frac{1}{2}$ has
an absolutely convergent Fourier series. Note that $\lambda
>\frac{1}{2}$ when $\alpha \in (\sqrt{2}-1,\sqrt{2}+1)$.

The Fourier coefficients $\widehat{P_0}(k)$ are by \refT{Theorem4}
given by, with $\chi_{k}:=\frac{2k\pi i}{\log 2}$,
\begin{equation}\label{c20}
    \widehat{P_0}(k)
    =\frac{1}{\log 2}\int_{0}^{1}
    \frac{1+\alpha \varphi^{}_{\alpha,1}(u)}
    {(1+u)^{\log_{2}(\alpha+1 )+\chi_{k}+1}}\dd u
    \qquad(k\in\mathbb{Z}).  
\end{equation}

The same type of results also hold for the recurrence 
$\Lambda_{1,\alpha}[f]=0$, the only difference being replacing the 
underlying interpolation function $\varphi^{}_{\alpha,1}$ by 
$\varphi^{}_{1,\alpha}$.
\end{example}

\begin{example}[Recurrence with minimisation or maximisation]
\label{Ex-min}
A class of sequences satisfying recurrences of the form
\begin{align}\label{E4.4}
    \mu(n)
    =\min_{1\le k\le \tr{\frac n2}}\{\alpha
    \mu(k)+\beta \mu(n-k)\}\qquad (n\ge 2)
\end{align}
with $\mu(1)=1$ was studied in \cite{Chang2000} to solve the AND-OR
Problem. It is proved there that if $\beta \ge \alpha$ are positive
integers, then the minimum in \eqref{E4.4} is reached at $k=\lfloor
\frac{n}{2}\rfloor $, so that $\mu(n)=f(n)=S_{\alpha,\beta}(n)$. In
this case, \refT{Theorem1} or \refE{Eg=0} and \refE{Eg=0b} imply that
$\mu(n) = n^{\log_2(\alpha+\beta)}P_0(\log_2n)$, where $P_0$ is a
periodic function with an absolutely convergent Fourier series, given
by \eqref{b16}. We can extend this to the cases (\emph{i})
$\beta\ge\alpha$, $\beta\ge1$ and (\emph{ii}) $\alpha\ge\beta$,
$\alpha+\beta\le 1$; see Appendix~\ref{Aminmix}.

On the other hand, it is not difficult to see that if $\beta\le 1$,
$\ga+\gb\ge1$, and $\ga\ge\gb^2$, then the minimum in \eqref{E4.4} is
attained at $k=1$, and we get a simple geometric expression for
$f(n)$. (Apart from the trivial case $\ga+\gb=1$, when $\mu(n)=1$,
these are the only $(\ga,\gb)$ for which the minimum always is
attained at $k=1$.) The behaviour of the recursion \eqref{E4.4} for
the remaining $(\ga,\gb)$ seems to be more complicated. (For example,
the case $1<\gb<\ga$.)

In \refApp{Aminmix}, we show in a similar manner that the solution to
the corresponding recurrence with maximisation
\begin{align}\label{rr-max}
    \mu(n)
    =\max_{1\le k\le \tr{\frac n2}}\{\alpha
    \mu(k)+\beta \mu(n-k)\}\qquad (n\ge 2)
\end{align}
with $\mu(1)=1$ is given by $\mu(n)=f(n)=S_{\alpha,\beta}(n)$
whenever $\alpha\ge \beta$, $\beta\le 1$ and $\alpha+\beta\ge1$. We
again obtain, by \refEs{Eg=0} and \ref{Eg=0b}, that $f(n) =
n^{\log_2(\alpha+\beta)} P_0(\log_2n)$, where $P_0$ is the periodic
function in \eqref{b16}, with an absolutely convergent Fourier
series. The case $\ga>\gb=1$ of \eqref{rr-max} was solved (in an
equivalent form) in \cite{AMM-2009}.

For more recurrences with minimisation or maximisation, see
\cite{Fredman1974,Hwang2003} and the references therein.
\end{example}

\begin{example}[OEIS sequences satisfying 
\protect{$\Lambda_{\alpha,\beta}[f]=0$ with $f(1)=1$} and thus
$f(n)=S_{\alpha,\beta}(n)$]\label{ex4-3}
We collect OEIS sequences of this category in \refTab{tb-ex4-3},
where $f$ and $P$ are both generic symbols, not necessarily the same
in each occurrence; some of the ``popular" sequences will be discussed
in detail below.

\begin{center}
\begin{longtable}{llll}
\multicolumn{4}{c}{{}} \\
\multicolumn{1}{c}{OEIS id.} &
\multicolumn{1}{c}{$(\alpha,\beta)$} &
\multicolumn{1}{c}{Description} &
\multicolumn{1}{c}{$f(n)$} \\ \hline    
A000027 & $(1,1)$ & Natural numbers & $n$\\
A064194 & $(1,2)$ & Gates in AND/OR problem \cite{Chang2000}
& $n^{\log_2(3)}P(\log_2n)$ \\
A006046 & $(2,1)$ & Odd entries in Pascal's triangle 
& $n^{\log_2(3)}P(\log_2n)$ \\
A268524 & $(1,3)$ & $\Lambda_{1,3}[f]=0$ with $f(1)=1$
& $n^2P(\log_2n)$ \\
A130665 & $(3,1)$ & $S_{3,1}(n)$ (see \eqref{Sn-theta});
also \eqref{rr-max} %"Formula", line 1
& $n^2P(\log_2n)$ \\
A073121 & $(2,2)$ & appeared in \cite{Ellul2005}
& $n^2P(\log_2n)$ \\
A268527 & $(1,4)$ & $\Lambda_{1,4}[f]=0$ with $f(1)=1$
& $n^{\log_2(5)}P(\log_2n)$ \\
A116520 & $(4,1)$ & \eqref{rr-max} with $(\alpha,\beta)=(4,1)$
& $n^{\log_2(5)}P(\log_2n)$ \\
A268526 & $(2,3)$ & $\Lambda_{2,3}[f]=0$ with $f(1)=1$
& $n^{\log_2(5)}P(\log_2n)$ \\
A268525 & $(3,2)$ & $\Lambda_{3,2}[f]=0$ with $f(1)=1$
& $n^{\log_2(5)}P(\log_2n)$ \\
A130667 & $(5,1)$ & \eqref{rr-max} with $(\alpha,\beta)=(5,1)$
& $n^{\log_2(6)}P(\log_2n)$ \\
A116522 & $(6,1)$ & Limit of the power of a matrix
& $n^{\log_2(7)}P(\log_2n)$ \\
A161342 & $(7,1)$ & 3-D cellular automaton
& $n^3P(\log_2n)$ \\
A116526 & $(8,1)$ & Limit of the power of a matrix
& $n^{\log_2(9)}P(\log_2n)$ \\
A116525 & $(10,1)$ & Limit of the power of a matrix
& $n^{\log_2(11)}P(\log_2n)$ \\
A116524 & $(12,1)$ & $\Lambda_{12,1}[f]=0$ with $f(1)=1$
& $n^{\log_2(13)}P(\log_2n)$ \\
A116523 & $(16,1)$ & $\Lambda_{16,1}[f]=0$ with $f(1)=1$
& $n^{\log_2(17)}P(\log_2n)$ \\ \hline
\caption{OEIS sequences of the form $S_{\ga,\gb}(n)$ (\refE{ex4-3}).}
\label{tb-ex4-3}
\end{longtable}    
\end{center}
We will see that these sequences play to some extent a prototypical 
role for more general recurrences with nonzero $g$. 

On the other hand, the only periodic function in this table that 
admits a closed-form expression in terms of elementary functions is 
when $(\alpha,\beta)=(2,2)$ for which we have $P(t) = 
2^{-\{t\}}\lpa{3-2^{1-\{t\}}}$, by \eqref{b16=}. 
\end{example}

We note also the following example, where $\gL_{\ga,\gb}[f]$ is not 
zero, but a constant.

\begin{example}\label{Enu0}
We can generalise \refE{Enu} by considering the partial sum, for 
$\ga,\gb>0$,
\begin{equation}\label{Sum-ab}
    f(n)
    :=\sum_{0\le k<n}\alpha^{\nu(k)}\gb^{\nu_0(k)}\qquad (n\ge 0),  
\end{equation}
where $\nu(n)$ is as above, and similarly $\nu_0 (n)$ denotes the
number of $0$s in the binary expansion of $n$ (with $\nu_0(0):=0$,
A080791).

In analogy with \eqref{nu-rec}, we have the recurrence relation
\begin{align}\label{nu0-rec}
    \nu_0(n)
    =\nu_0\left(\left\lfloor\tfrac{n}{2}\right\rfloor\right) 
    +\mathbf{1}_{n\text{ is even}}\qquad (n\ge 1),  
\end{align}
and it follows easily that $f$ satisfies  $f(1)=1$ and
\begin{align}\label{gab}
	\Lambda_{\alpha ,\beta}[f](n)
	=g(n)=1-\gb
	\qquad (n\ge2).
\end{align}
Assume  $\ga+\gb>1$. Then, by \eqref{gab} and \refL{lmm-g-const}, 
\begin{equation}\label{fab}
	f(n) 
	= \frac{\ga}{\ga+\gb-1}S_{\alpha,\gb}(n)
	+\frac{\gb-1}{\ga+\gb-1}\qquad (n\ge 1).
\end{equation}
Equivalently, 
\begin{equation}\label{Sn-ab}
	S_{\ga,\gb}(n)
	=\frac{\ga+\gb-1}{\ga}\sum_{0\le k<n}
	\alpha^{\nu(k)}\gb^{\nu_0(k)}-\frac{\gb-1}{\ga}
	\qquad (n\ge 1).
\end{equation}
By \eqref{fab} and \refE{Eg=0}, we have
\begin{align}
	f(n)
	=n^{\log_2(\alpha+\beta)}P(\log_2 n) 
	+\frac{\gb-1}{\ga+\gb-1}
	\qquad(n\ge1), 
\end{align}
with $P(t)=P_0(t)$ given by \eqref{b16}. Furthermore, by
\refT{Theorem4} or \refE{Eg=0b}, $P(t)$ has an absolutely convergent
Fourier series.
\end{example}

\subsection{$(\alpha,\beta)=(1,2)$ and $(\alpha,\beta)=(2,1)$}

The large number of concrete sequences discussed here and in the 
following sections show the usefulness and power of the notion of 
``periodic equivalence''.

\begin{example}[\protect{$\Lambda_{1,2}[f]=g$}] \label{E12} We begin
with A064194 in \refE{ex4-3}, which satisfies $\Lambda_{1,2}[f]=0$
and thus equals $S_{1,2}(n)$. This sequence enumerates the number of
gates in the AND/OR problem (\refE{E4.4}) \cite{Chang2000}. It also
counts the number of multiplications needed to multiply two degree
$n$ polynomials using Karatsuba's algorithm 
\cite[Exercise 4.3.3-17]{Knuth1997}, 
as well as the total number of odd entries in the
$(n+1)\times (n+1)$ matrix $\bigl[\binom{i+j}{i}\bigr]_{i,j=0}^n$.
Another interpretation in terms of Sierpi\'nski-like arrays can be
found on Peter Karpov's webpage at
%\SJm{FIXA hyperref}%\href{http://inversed.ru/index.htm}{inversed.ru}. %%%SJ
%\href{http://inversed.ru/index.htm}{inversed.ru}.

In this case, $f(n) = n^{\log_23}P(\log_2n)$ for $n\ge1$, where, by 
Lemma~\ref{Lemma3} or Theorem~\ref{Theorem1},
\begin{align}
    P(t) = 3^{-\{t\}}
    \lpa{1+2\varphi^{}_{1,2}\lpa{2^{\{t\}}-1}}.
\end{align}
Here the interpolation function $\varphi^{}_{1,2}$ has by \eqref{b8} 
the form (see Figure~\ref{Fig1})
\begin{align}
    \varphi^{}_{1,2}\biggl(\sum_{k\ge 1}2^{-e_{k}}\biggr) 
    =\sum_{k\ge1}2^{e_{k}-k+1}3^{-e_{k}}
    \qquad (1\le e_{1}<e_{2}<\cdots ).
\end{align}

Some other sequences periodically equivalent to $S_{1,2}(n)
=f_{\text{A064194}}(n)$ are given in \refTab{tb-E12}, where
\textsf{AND} denotes the bitwise logic AND operator; we also denote
by $b_j(n)$ the $(j+1)$st bit (from right to left) in the binary
expansion of $n$: $b_j(n) =
\tr{\frac{n}{2^j}}-2\tr{\frac{n}{2^{j+1}}}$ for $j=0,1,\dots,L_n$,
and $\nu_0(n)$ is as in \refE{Enu0} the number of $0$s in the binary
expansion of $n$.

\begin{table}[!ht]
\begin{center}
\begin{small}
\begin{tabular}{cccc}
%\begin{tabular}
\multicolumn{4}{c}{{}} \\
\multicolumn{1}{c}{OEIS id.} &
\multicolumn{1}{c}{Description} & 
\multicolumn{1}{c}{$g(n)$} &
\multicolumn{1}{c}{$f(n)$} \\ \hline   
A080572 
& $\sum\limits_{0\le k,j<n}\mathbf{1}_{k\textsf{ AND } j\ne0}$
& $\tr{\frac n2}^2-\mathbf{1}_{n \text{ odd}}$
& $n^2-S_{1,2}(n)$ \\
A268514 
& $\sum\limits_{1\le j<n}2^{\nu_0(j)}$
& $1$ & $\frac12(S_{1,2}(n)-1)$ \\
A325102 
& \makecell{\#(pairs $(k,m)$, $1\le k,m\le n$ such that\\
$(b_j(k),b_j(m))\ne(1,1)$, $0\le j\le \min\{L_k,L_m\}$)}
& $2\cl{\frac n2}-2$ & $S_{1,2}(n)-2n+1$ \\
A325103 & \makecell{\#(pairs $(k,m)$, $1\le k<m\le n$ such that\\
$(b_j(k),b_j(m))\ne(1,1)$, $0\le j\le L_k$)} 
& $\cl{\frac n2}-1$ & $\frac12	S_{1,2}(n)-n+\frac12$ \\
A325104 & \makecell{\#(pairs $(k,m)$, $1\le k<m\le n$ such that\\
$(b_j(k),b_j(m))=(1,1)$ for some $0\le j\le L_k$)}
& $\binom{\tr{\frac n2}+1}{2}-1$
& $\frac{n^2-n+1}2-\frac12S_{1,2}(n)$ \\ \hline
%\end{tabular}
\end{tabular}
\caption{Sequences periodically equivalent to $S_{1,2}(n)$
(\refE{E12}).}
\label{tb-E12}
\end{small}
\end{center}
\end{table}

Observe that 
\begin{align}
	f_{\text{A325103}}(n)- f_{\text{A325103}}(n-1)
	= \text{A115378}(n) 
	= 2^{\nu_0(n)}-1,
\end{align}
and this provides a proof for the recurrence $\Lambda_{1,2}[f]
=\cl{\frac n2}-1$ satisfied by $f_{\text{A325103}}(n)$. The
consideration of the other two sequences A325102 and A325104 is
similar. On the other hand, the sequence $2^{\nu_0(n)}$ corresponds
to A080100 whose partial sum satisfies $\Lambda_{1,2}[f]=-1$ with
$f(1) = 1$. Then $f(n) = \frac12(S_{1,2}(n)+1)$.

Another example with $(\ga,\gb)=(1,2)$ is the sequence A086845, which
counts the number of comparators used in Bose and Nelson's sorting
networks \cite{Bose1962}, and satisfies $\Lambda_{1,2}[f]=\tr{\frac
n2}$ with $f(1)=0$. \refL{lmm-g-linear} applies with
$(c,d,e)=(\frac12,0,-\frac12)$, giving $f(n) =
n^{\log_23}P(\log_2n)-n$; see \cite{Grabner2005}. In fact,
\refL{lmm-g-linear} and its proof show that the periodically
equivalent sequence $h(n):=f(n)+n$ satisfies
$\Lambda_{1,2}[h]=-\mathbf{1}_{n\text{ odd}}$ with $h(1)=1$, and thus
$h(n)=n^{\log_23}P(\log_2n)$ by \refE{Eodd}.
\end{example}

\begin{figure}[!ht]
\begin{center}
\includegraphics[height=3.3cm]{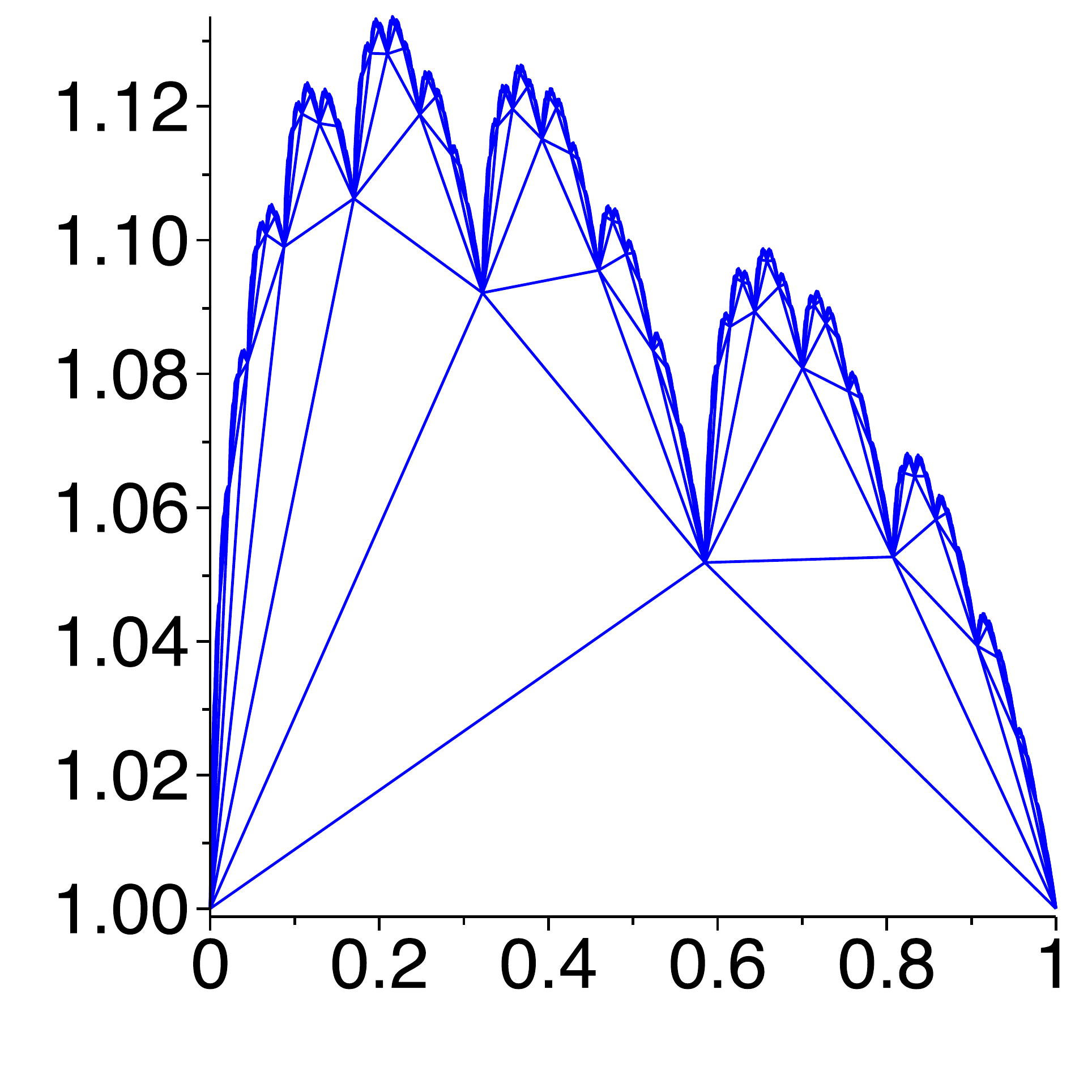}\,
\includegraphics[height=3.3cm]{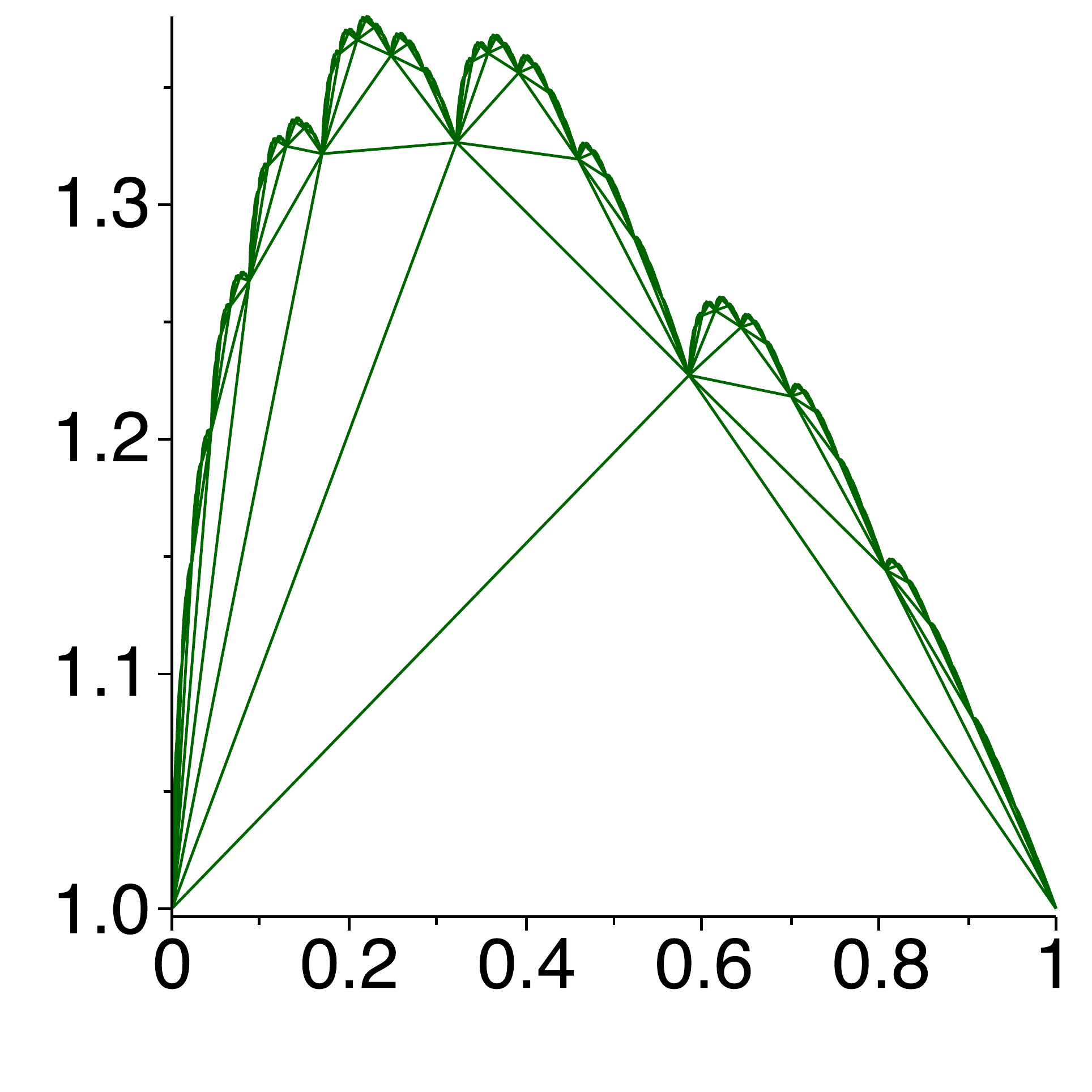}\,
\includegraphics[height=3.3cm]{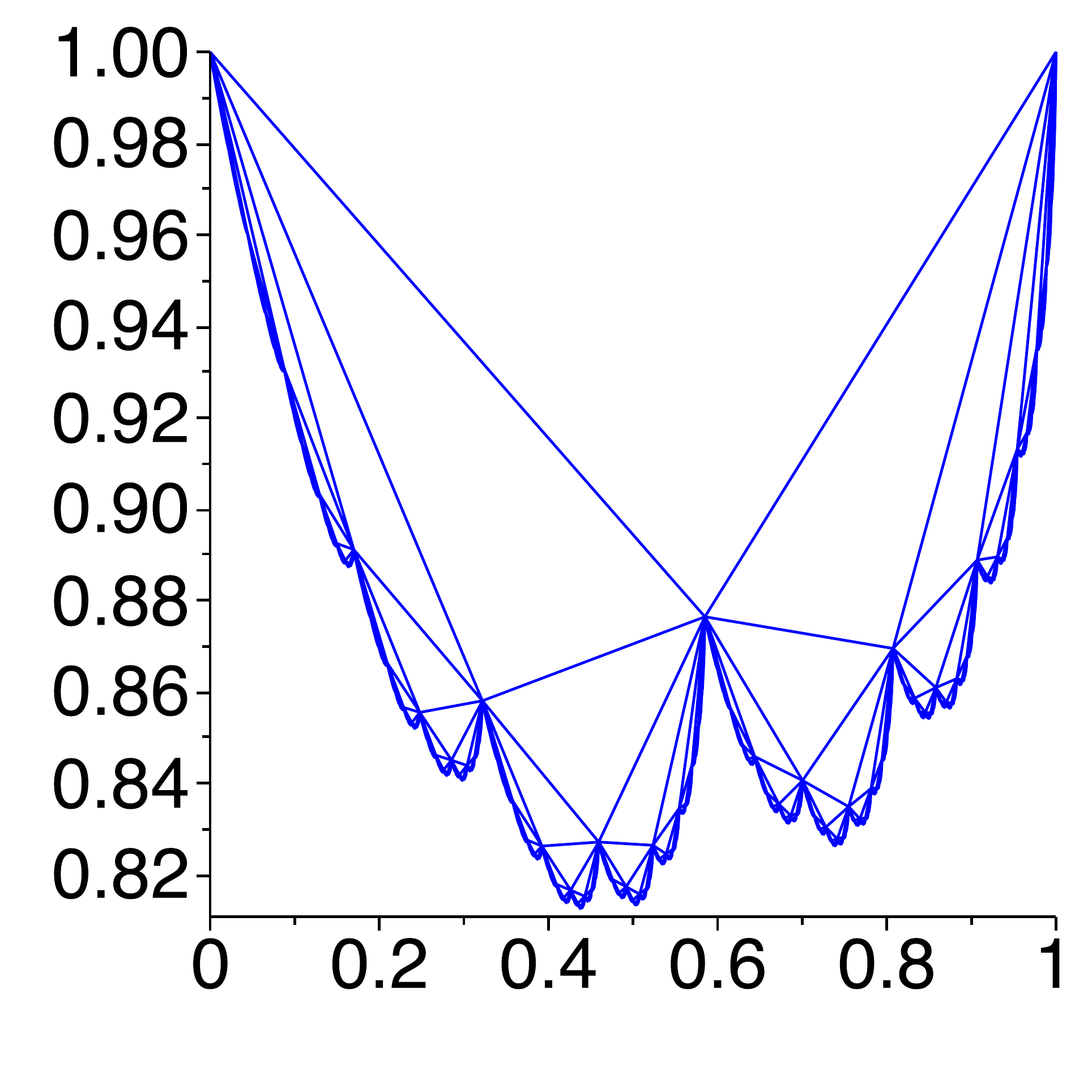}\,
\includegraphics[height=3.3cm]{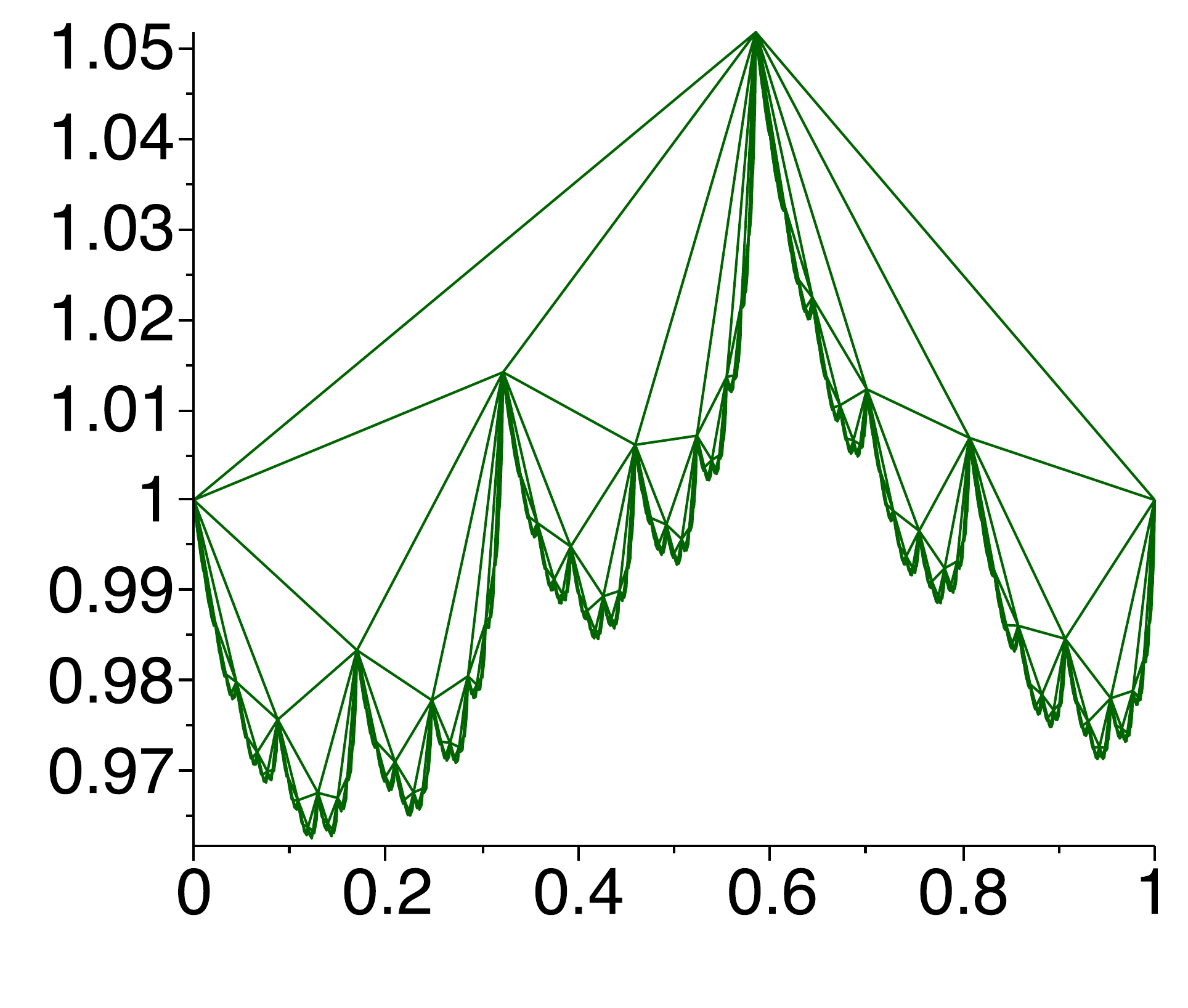}
\end{center}
\vspace*{-.3cm} 
\caption{{The periodic functions arising from the four sequences 
(in left to right order) $\Lambda_{1,2}[f]=\tr{\frac n2}$, 
$\Lambda_{1,2}[f]=\cl{\frac n2}$, $\Lambda_{2,1}[f]=\tr{\frac n2}$, 
$\Lambda_{2,1}[f]=\cl{\frac n2}$ all with $f(1)=0$, as approximated 
by the fractional part 
$\bigl\{\frac{f(2^k+j)+2^k+j}{(2^k+j)^{\log_23}}\bigr\}$ for 
$0\le j<2^k$ and $k=1,2,\dots,10$.}} \label{Fig-BN}
\end{figure}

\begin{example}[\protect{$\Lambda_{2,1}[f]=g$}]\label{Ex:pascal}
A006046$(n)$ is the total number of odd entries in first $n$ rows of
Pascal's triangle. This sequence $f(n)$ equals $S_{2,1}(n)$ (defined
in \refE{Eg=0}) and has a rich literature with different extensions
and connections; for example, it equals \eqref{Sn-theta} with
$\ga=2$; see below, the OEIS page, the survey papers
\cite{Stolarsky1977, Chen2014} and Finch's book \cite[\S
2.16]{Finch2003} for more information. By \refE{Eg=0}, we have
$f(n)=n^{\log_23}P_0(\log_2n)$, where
$P_0\in\mathrm{H}_{\log_23-1}[0,1]$ by \refL{Lemma5}; see
\cite{Flajolet1994,Grabner2005}. Some periodically equivalent
sequences (possibly with a shift) are given as follows.
\begin{itemize}

\item A051679: Total number of even entries in the first $n$ rows
of Pascal's triangle, namely, $f(n) = \binom{n+1}2-S_{2,1}(n)$.
Then $f(1) = 0$ and
\begin{equation}\label{E:A051679}
	\Lambda_{2,1}[f]= 
	\tfrac18\,{n}^{2}-
	\begin{cases}	
    \frac n4, & \text{if $n$ is even};\\ 
    \frac18, & \text{if $n$ is odd}.\\	
    \end{cases}
\end{equation}

\item A064406: The accumulation of the number of even entries
(A048967) over the number of odd entries (A001316) in row $n$ of
Pascal's triangle (A007318); in other words,
$\A051679(n)-\A006046(n)$. Thus $f(n) = \binom{n+1}{2} -2S_{2,1}(n)$.
Then $f$ satisfies the same recurrence \eqref{E:A051679} but with the
different initial condition $f(1) = -1$. The sequence is positive
except for the first 18 terms.

\item A074330: $S_{2,1}(n+1)-1$.

\item A080978: $f(n) = 2 S_{2,1}(n) + 1$. Then $\Lambda_{2,1}[f]= -2$ 
with $f(1) = 3$.

\item A151788: $f(n) := \frac12\,(3S_{2,1}(n)-1)= \frac32
S_{2,1}(n)-\frac12$ and satisfies $\Lambda_{2,1}[f]=1$ with $f(1) =
1$.

\item A159912: $f(n)=\sum_{j<n}\lpa{2^{\nu(2j+1)}-1}$ satisfies
$\Lambda_{2,1}[f]=\tr{\frac n2}$ and the relation $f(n) =
2S_{2,1}(n)-n$.

\item A160720: Number of ``ON" cells in a certain $2$-dimensional 
cellular automaton: $f(1)=1$ and 
\begin{equation}
    \Lambda_{2,1}[f]= 2n-2-2\times \mathbf{1}_{n\text{ odd}}.
\end{equation}
One has $f(n) = 4(S_{2,1}(n)-n)+1$. 

\item A160722: Number of ``ON" cells in a certain $2$-dimensional
cellular automaton based on Sierpi\'nski triangles. Then
$\Lambda_{2,1}[f] =2\tr{\frac n2}$ with $f(1)=1$ and $f(n) =
3S_{2,1}(n)-2n$.

\item A171378: $f(n)=n^2-S_{2,1}(n)$. Then
$\Lambda_{2,1}[f]=\cl{\frac n2}^2-\mathbf{1}_{n\text{ odd}}$ with
$f(1) = 0$.

\item A193494: Worst case of an unbalanced recursive algorithm over
all $n$-node binary trees; $2A(n-1)+1$ satisfies the max-recursion
\eqref{rr-max} with $(\ga,\gb)=(2,1)$. Thus the sequence $f(n) := 
A(n-1)= \frac12(S_{2,1}(n)-1)$ satisfies $\Lambda_{2,1}[f]=1$ with 
$f(1)=0$; see \refE{Ex-min} and \refP{PAmax}.

\item A256256: Number of ``ON" cells in a cellular automaton on
triangular grid, which is $6S_{2,1}(n)$ and satisfies the recurrence
$\Lambda_{2,1}[f]=0$ with $f(1)=6$.

\item A262867: Total number of ``ON" cells in a cellular automaton.
$f(n)=n^2-S_{2,1}(n)+1 =f_{\A171378}(n)+1$, which satisfies
$\Lambda_{2,1}[f] = \cl{\frac n2}^2-2- \mathbf{1}_{n\text{ odd}}$.

\item A266532: Number of $Y$-toothpicks in a cellular automaton. We
then get the recurrence $\Lambda_{2,1}[f]=3\tr{\frac n2}-2$ and
$f(n)=3(S_{2,1}(n)-n)+1$.

\item A267610: Accumulated number of ``OFF'' cells in a cellular
automaton. This is $f(n) = S_{2,1}(n)-2n+1$, and
$\Lambda_{2,1}[f]=2\tr{\frac n2}$ with $f(1)=0$.

\item A267700: ``Tree" sequence in a $90$ degree sector of some
cellular automaton. (Also the partial sum of A038573.) This sequence
satisfies $\Lambda_{2,1}[f]=\tr{\frac n2}$ with $f(1)=0$, so that
$f(n) = S_{2,1}(n)-n$.

\end{itemize}

A different example with $(\ga,\gb)=(2,1)$ is the sequence A137294,
which arises in a polynomial-time algorithm for a sowing game; see
\cite[p.\ 289]{Erickson1996} for more information. It satisfies
$\Lambda_{2,1}[f]=1 + \mathbf{1}_{n\text{ odd}}$ with $f(1)=0$. By
\refE{Eodd} applied to $f(n)+\frac12$, we have $f(n) =
n^{\log_23}P(\log_2n)-\frac12$. See Figure \ref{Fig-AA}.

\begin{figure}[!ht]
\begin{center}
\includegraphics[height=3cm]{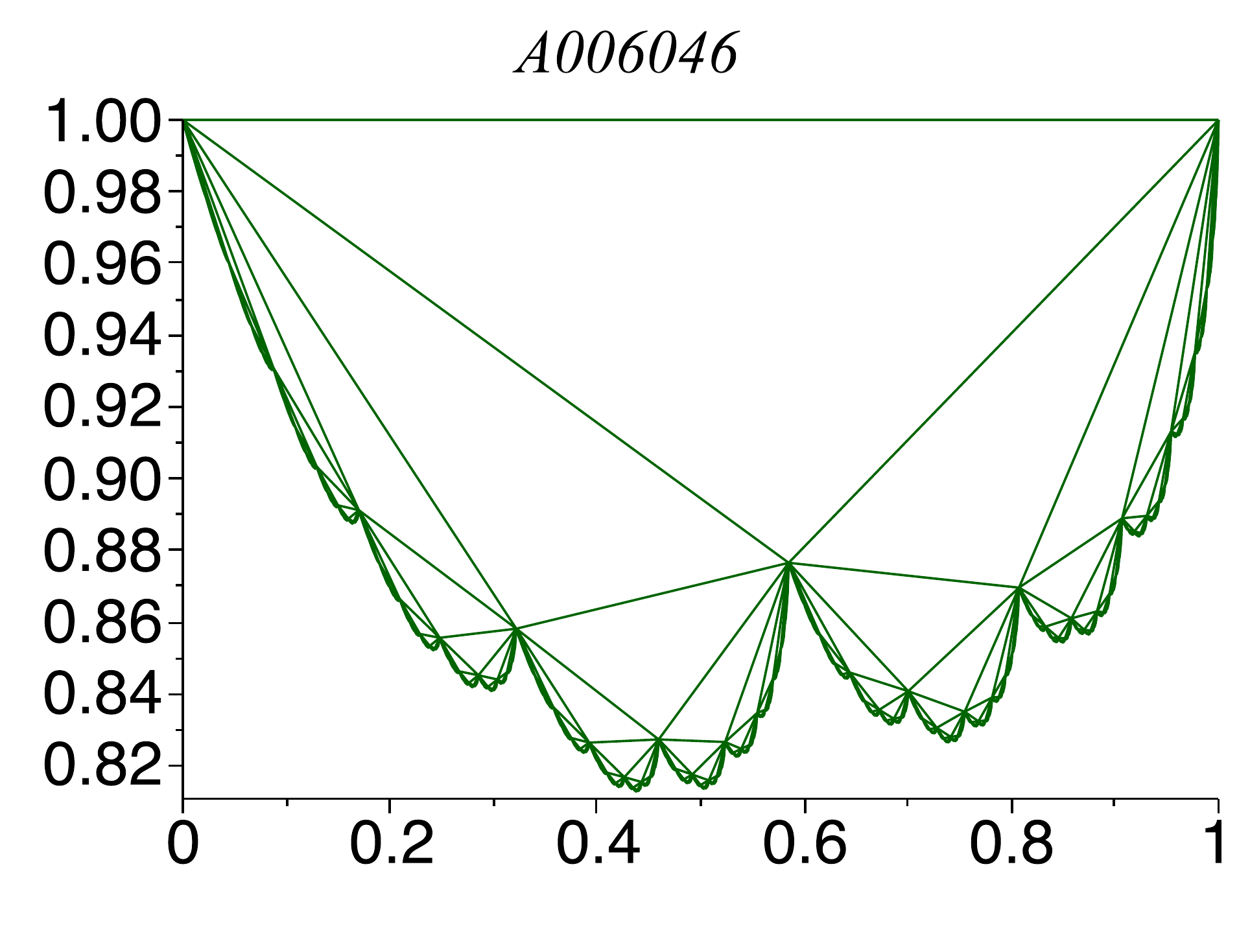}\;
\includegraphics[height=3cm]{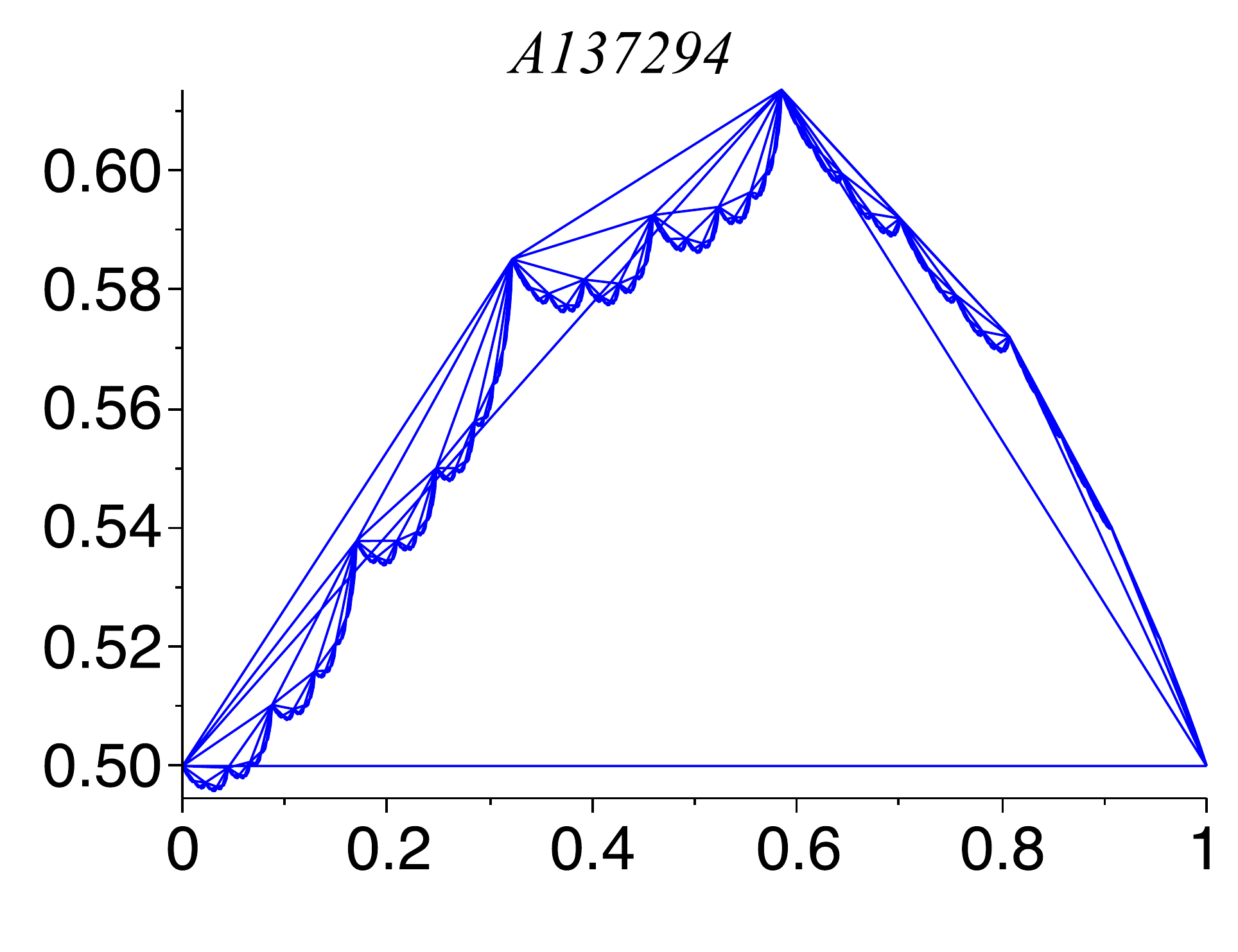}
\end{center}
%\vspace*{-.3cm} 
\caption{The periodic functions arising from A006046 and A137294 (see 
\refE{Ex:pascal}), respectively.}
\label{Fig-AA}
\end{figure}

\end{example}

\subsection{Sequences satisfying $\Lambda_{\alpha,\beta}[f]=g$ 
with $\alpha+\beta\ge4$}

\begin{example}[\protect{$\alpha+\beta=4$}]\label{E31}

The sequence A268524 satisfies $\Lambda_{1,3}[f]=0$ with $f(1)=1$ and
thus equals our $S_{1,3}(n)$, as listed in Example~\ref{ex4-3}.

When $(\alpha,\beta)=(3,1)$, the prototype sequence $S_{3,1}(n)$
corresponds to A130665, which satisfies $\gL_{3,1}[f]=0$ and
$f(1)=1$; see \refE{ex4-3}. Some variants of this sequence from OEIS,
all having $\Lambda_{1,3}[f]=g$ constant, are given in the following
table; they arise mostly from the combinatorics of the
Ulam--Warburton cellular automaton.

\begin{table}[!ht]
\begin{center}
\begin{tabular}{lcl|lcl}
\multicolumn{6}{c}{{}} \\
\multicolumn{1}{c}{OEIS id.} &
\multicolumn{1}{c}{$(f(1),g(n))$} &
\multicolumn{1}{c}{$f(n)$} &
\multicolumn{1}{c}{OEIS id.} &
\multicolumn{1}{c}{$(f(1),g(n))$} &
\multicolumn{1}{c}{$f(n)$} \\ \hline    
A147562 & $(1,1)$ &  $\frac13(4S_{3,1}(n)-1)$ &
A151914 & $(4,-4)$ & $\frac43(2S_{3,1}(n)+1)$ \\
A151917 & $(1,-1)$ & $\frac13(2S_{3,1}(n)+1)$ & 
A151920 & $(0,1)$ & $\frac13(S_{3,1}(n)-1)$ \\
A160410 & $(4,0)$ & $4S_{3,1}(n)$ &
A160412 & $(3,0)$ & $3S_{3,1}(n)$\\ \hline
\end{tabular}
\caption{Sequences periodically equivalent to $S_{3,1}(n)$
(\refE{E31}).}
\label{tb-E31}
\end{center}
\end{table}

Three other sequences satisfying $\Lambda_{3,1}[f]=g$ are given 
below. 
\begin{table}[!ht]
\begin{center}
\begin{tabular}{lllcl}
\multicolumn{5}{c}{{}} \\
\multicolumn{1}{c}{OEIS id.} &
\multicolumn{1}{c}{Context} &
\multicolumn{1}{c}{$g(n)$} &
\multicolumn{1}{c}{$f(1)$} &
\multicolumn{1}{c}{$f(n)$}\\ \hline        
A183060 & Cellular automaton & $-n+\begin{cases}
    2\\
    3
\end{cases}$ & $1$ & $n^2P(\log_2n)+n-\frac23$\\
A183126 & Toothpick sequence & $-4n+\begin{cases}
    3\\
    7
\end{cases}$ & $7$ & $n^2P(\log_2n)+4n-1$\\
A183148 & Toothpick sequence & $-3n+\begin{cases}
    3\\
    6
\end{cases}$ & $4$ & $n^2P(\log_2n)+3n-1$\\ \hline
\end{tabular}    
\caption{Further sequences with $\Lambda_{3,1}[f]=g$
(\refE{E31}).}
\label{tb-E31b}
\end{center}
\end{table}

Although very different in appearance, these sequences are all
periodically equivalent to $S_{3,1}(n)$ because, by
Lemma~\ref{lmm-g-linear}, the right-hand side of \eqref{L-31a} are
all of the form $\Lambda_{3,1}[f;0,0,0]$ when $(\alpha,\beta)=(3,1)$
and $(c,d,e) = (-1,2,3), (-4,3,7)$ and $(-3,3,6)$, respectively. More
precisely, we have the relations ($f_{\text{A130665}}(n)=S_{3,1}(n)$)
\begin{align}
    f_{\text{A183060}}(n) 
    &= \tfrac23 f_{\text{A130665}}(n)+n-\tfrac23,\\
    f_{\text{A183126}}(n) 
    &= 4 f_{\text{A130665}}(n)+4n-1,\\
    f_{\text{A183148}}(n) 
    &= 2 f_{\text{A130665}}(n)+3n-1.
\end{align}
\end{example}

\begin{example}[\protect{$\alpha+\beta=5$}]
The sequence A268527 in \refE{ex4-3} satisfies $\Lambda_{1,4}[f]=0$
with $f(1)=1$, and thus equals our $S_{1,4}(n)$.

Three other sequences were found with $(\alpha,\beta)=(4,1)$. The
first is A116520 which satisfies $\Lambda_{4,1}[f]=0$ and equals
$S_{4,1}(n)$; see \refE{ex4-3}. Another sequence A151790 equals
$\frac14(5S_{4,1}(n)-1)$. It satisfies $\Lambda_{4,1}[f]=1$ with
$f(1)=1$; as a check, Lemma~\ref{lmm-g-const} and \eqref{l4.1} yield
\begin{align}
    f_{\text{A151790}}(n) 
    = \tfrac54f_{\text{A116520}}(n)-\tfrac14
    = \tfrac54 S_{4,1}(n)-\tfrac14.
\end{align}
The last sequence with the pattern $(\alpha,\beta)=(4,1)$ we found is
A273578, which is the total number of ``ON" cells in a 2-D cellular
automaton. It satisfies $f(1) = 1$ and
\begin{align}
	\Lambda_{4,1}[f]=\begin{cases}	
        \frac{1}{2}\,{n}^{3}+\frac{n}{2}-4,
        & \text{if $n$ is even};\\
        \frac{1}{2}\,{n}^{3}+\frac{3}{2}\,{n}^{2}-2\,n-4,
        & \text{if $n$ is odd}.\\	
    \end{cases} 
\end{align}
This sequence is also periodically equivalent to
$f_{\text{A116520}}(n)$. To see this, we consider the difference
$\Delta(n) = \frac43n^3-\frac13n+1-f(n)$, which satisfies
$\Lambda_{4,1}[\Delta] = 0$ with $\Delta(1) = 1$; thus
$\gD(n)=S_{4,1}(n)= f_{\text{A116520}}(n)$. From this we deduce the
identity
\begin{align}
    f(n) = \tfrac43\,n^3 +n^{\log_25}P(\log_2n)
    -\tfrac13n+1\qquad(n\ge1), 
\end{align}
and the relation
\begin{align}
    f_{\text{A273578}}(n) 
    = \tfrac43\,n^3-\tfrac13n+1-
    f_{\text{A116520}}(n). 
\end{align}

Two other sequences with $\alpha+\beta=5$ are given in \refE{ex4-3}:
A268526, which satisfies $\Lambda_{2,3}[f]=0$, and A268525, which
satisfies $\Lambda_{3,2}[f]=0$, both with $f(1)=1$.
\end{example}

\begin{example}[\protect{$\alpha+\beta=6$}] \label{Ea+b=6}
Sequences in OEIS of this type have to do either with digital sums or
cellular automata. They include A130667 $(=S_{5,1}(n))$ from
\refE{ex4-3}.
\begin{table}[!ht]
\begin{center}
\begin{tabular}{lllll}
\multicolumn{5}{c}{{}} \\
\multicolumn{1}{c}{OEIS id.} &
\multicolumn{1}{c}{$(\alpha,\beta)$} &
\multicolumn{1}{c}{$g(n)$} &
\multicolumn{1}{c}{$f(1)$} &
\multicolumn{1}{c}{$f(n)$}\\ \hline        
A130667 & $(5,1)$ & $0$ & $1$ & $n^{\log_26}P(\log_2n)$\\
A151781 & $(5,1)$ & $1$ & $1$ & $n^{\log_26}P(\log_2n)-\frac15$\\
A186410 & $(5,1)$ & $-4\tr{\frac n2}+4$ & $1$ & 
$n^{\log_26}P(\log_2n)+n-\frac45$ \\
A270106 & $(4,2)$ & $-1$ & $1$ & $n^{\log_26}P(\log_2n)+\frac15$\\
A273500 & $(4,2)$ & $\begin{cases}	
        \frac13\,{n}^{3}+\frac23\,n-4,\\
        \frac13\,{n}^{3}+{n}^{2}-\frac73\,n-4,\\
    \end{cases}$ & $1$ & 
    $\begin{matrix} 
        \frac43n^3+n^{\log_26}P(\log_2n)\\
        \hfill\qquad-\frac13n+\frac45
    \end{matrix}$\\
A273562 & $(4,2)$ & $\begin{cases}
        \frac16n^3-\frac23n+1,\\
        \frac16n^3+\frac12n^2-\frac{13}6n+\frac32,\\
    \end{cases}$
& $0$ & $\begin{matrix}
    \frac23n^3+n^{\log_26}P(\log_2n)\\
\hfill\qquad +\frac13n-\frac15
\end{matrix} $\\ \hline
\end{tabular}    
\caption{Sequences with $\gL_{\ga,\gb}[f]=g$ with $\ga+\gb=6$
(\refE{Ea+b=6}).}
\label{tb-Ea+b=6}
\end{center}
\end{table}

The sequence $S_{4,2}(n)$ is not in OEIS, but the sequence A270106
(which comes from a cellular automaton) equals the sum \eqref{Sum-ab}
with $(\ga,\gb)=(4,2)$, as follows from the discussion in OEIS of its
increments A189007. Hence, \refE{Enu0} shows that
\begin{align}\label{S42}
	f_{A270106}(n) = \tfrac45S_{4,2}(n)+\tfrac15.
\end{align}

The six sequences above lead only to two periodically distinct ones
(say, A130667 and A270106) because, by Lemma~\ref{lmm-g-const} and
Lemma~\ref{lmm-g-linear}, we have (with
$f_{\text{A130667}}(n)=S_{5,1}(n)$)
\begin{align}
    f_{\text{A151781}}(n) 
    &= \tfrac65f_{\text{A130667}}(n)-\tfrac15\\
    f_{\text{A186410}}(n)
    &= \tfrac45 f_{\text{A130667}}(n)+n-\tfrac45,
\end{align}
and, recall also \eqref{S42}, % ($S_{4,2}(n)$ not in OEIS)
\begin{align}
    f_{\text{A273500}}(n) 
    &= \tfrac43n^3-\tfrac13n+1-f_{\text{A270106}}(n)\\
    f_{\text{A273562}}(n)
    &= \tfrac23n^3+\tfrac13n-f_{\text{A270106}}(n),
\end{align}
so that 
\begin{align}
    f_{\text{A273500}}(n) 
    = \tfrac23n^3-\tfrac23n+1+f_{\text{A27356}}(n).
\end{align}
\end{example}

\begin{example}[\protect{$\alpha+\beta\ge7$}]\label{Ea+b>=7}
For $\ga+\gb\ge7$, we found the following examples with $\gb=1$. (The
ones with $g(n)=0$ and $f(1)=1$ appear also in \refE{ex4-3}.)
\begin{table}[!ht]
\begin{center}
\begin{tabular}{llllll}
\multicolumn{6}{c}{{}} \\
\multicolumn{1}{c}{OEIS id.} &
\multicolumn{1}{c}{$(\alpha,\beta)$} &
\multicolumn{1}{c}{$g(n)$} &
\multicolumn{1}{c}{$f(1)$} &
\multicolumn{1}{c}{$f(n)$} &
\multicolumn{1}{c}{$xS_{\alpha,1}(n)+y$}\\ \hline  
A116522 & $(6,1)$ & $0$ & $1$ 
& $n^{\log_27}P(\log_2n)$ & $S_{6,1}(n)$\\
A151792 & $(6,1)$ & $1$ & $1$ 
& $n^{\log_27}P(\log_2n)-\frac16$
& $\frac{7}{6}S_{6,1}(n)-\frac16$\\
A151793 & $(7,1)$ & $1$ & $1$ 
& $n^{3}P(\log_2n)-\frac17$
& $\frac{8}{7}S_{7,1}(n)-\frac1{7}$\\
A160428 & $(7,1)$ & $0$ & $8$ 
& $n^{3}P(\log_2n)$ & $8S_{7,1}(n)$\\
A161342 & $(7,1)$ & $0$ & $1$ 
& $n^{3}P(\log_2n)$ & $S_{7,1}(n)$\\
A116526 & $(8,1)$ & $0$ & $1$ 
& $n^{\log_29}P(\log_2n)$ & $S_{8,1}(n)$\\
A255764 & $(8,1)$ & $1$ & $1$ 
& $n^{\log_29}P(\log_2n)-\frac18$
& $\frac{9}{8}S_{8,1}(n)-\frac1{8}$\\
A255765 & $(9,1)$ & $1$ & $1$ 
& $n^{\log_210}P(\log_2n)-\frac19$
& $\frac{10}{9}S_{9,1}(n)-\frac1{9}$\\
A116525 & $(10,1)$ & $0$ & $1$ 
& $n^{\log_211}P(\log_2n)$ & $S_{10,1}(n)$\\
A255766 & $(10,1)$ & $1$ & $1$ 
& $n^{\log_211}P(\log_2n)-\frac1{10}$ 
& $\frac{11}{10}S_{10,1}(n)-\frac1{10}$\\
A116524 & $(12,1)$ & $0$ & $1$ 
& $n^{\log_2 13}P(\log_2n)$ & $S_{12,1}(n)$\\
A116523 & $(16,1)$ & $0$ & $1$ 
& $n^{\log_217}P(\log_2n)$ & $S_{16,1}(n)$\\ \hline
\end{tabular}
\caption{Sequences with $\gL_{\ga,\gb}[f]=g$ for $\ga+\gb\ge7$
(\refE{Ea+b>=7}).}
\label{tb-Ea+b>=7}
\end{center}
\end{table}

The only  example we found in OEIS with $\beta\ne1$ is A269589
with $(\ga,\gb)=(3,4)$:
\begin{align}
  	\Lambda_{3,4}[f] 
	= -6\cdot 2^{\nu_0(n)}\mathbf{1}_{n\text{ odd}},
\end{align}
with $f(1) =1$; this sequence enumerates the number of triples
$(i,j,k)\in[0,n-1]^3$ such that their bitwise \textsf{AND} is zero.
By \refE{Eodd}, we have $f(n) = n^{\log_27}P(\log_2n)$ for some
periodic function $P$.
\end{example}

\section{Applications II. $\alpha=\beta$}
\label{S:app2}

We group in this section examples satisfying the recurrence
$\Lambda_{\alpha,\alpha}[f]=g$. Since the interpolating function
$\varphi^{}_{\alpha,\alpha}(t) = t$ for every $\ga>0$, this is
similar to the case $\ga=\gb=1$ treated in \cite{Hwang2017}. In
particular, $\varphi_{\ga,\ga}$ is linear on $\oi$, and therefore we
can derive in many cases a closed-form solution in terms of
elementary functions. In cases when the periodic functions do not
have simple explicit forms, we can often derive explicit Fourier
expansions in terms of known functions such as Riemann's or Hurwitz's
zeta functions. As the situations and analysis are very similar to
the case when $(\alpha,\beta)=(1,1)$, we omit most of the details,
which can be found in \cite{Hwang2017}. Note, however, one difference
between the cases $\ga>1$ and $\ga=1$: as discussed in \refR{Ra=b},
the periodic function $P_0(t)$ in \refL{Lemma3} is not continuously
differentiable when $\ga>1$; hence, typically, the periodic function
$P(t)$ in \refT{Theorem1} also is not continuously differentiable.

We begin with two examples for a general $\ga$.
\begin{example}\label{EIIa} 
Consider the sum $f(n):=\sum_{1\le k<n} \ga^{L_k}$. The case $\ga=2$
is A063915. Since $L_k=\nu(k)+\nu_0(k)-1$ for $k\ge1$, we have by
\refE{Enu0} and \eqref{fab}
\begin{align}\label{eiia}
	f(n)
	=\sum_{1\le k<n}\ga^{L_k}
	=\frac{1}{\ga}\llpa{\sum_{0\le k<n}
	\ga^{\nu(k)+\nu_0(k)}-1}
	=\frac{1}{2\ga-1}\bigpar{S_{\ga,\ga}(n)-1},
\end{align}
at least provided $\ga>\frac12$; the result holds indeed for any
$\ga\neq\frac12$ since both sides are polynomials in $\ga$ for fixed
$n$. It follows easily (cf.\ \eqref{gab}), that $\gL_{\ga,\ga}[f](n)
=g(n)=1$.
\end{example}

\begin{example}\label{EIIb}
A more general pattern that we found in several OEIS sequences (all
with $\ga=\gb=2$; see Table \ref{tb-22-1} below) is of the form in
Lemma~\ref{lmm-g-linear}, i.e.,
\begin{align}\label{maj2}
    \Lambda_{\ga,\ga}[f]= g(n) = c n+ \begin{cases}
        d,&\text{ if $n$ is even}; \\
        e,&\text{ if $n$ is odd}.
    \end{cases}
\end{align}
Assume for simplicity $\ga>1$, so \refL{lmm-g-linear} applies. Then 
\eqref{L-31b} yields the solution 
\begin{equation}
    f(n) = n^\rho P(\log_2n)-\frac{cn}{\alpha-1}
	-\frac{d}{2\alpha-1}\qquad(n\ge1).
\end{equation} 
Furthermore, it follows, from the proof of \refL{lmm-g-linear}, 
that the Fourier coefficients of the periodic function $P$ are given 
by (assuming first $\ga>2$ so that $\rho>2$),
\begin{align}\label{maj3}
	\hP(k)
	&=\frac{1}{(\rchi_k)(\rchi_k-1)\log2}
	\biggl(\frac{(2\ga-1)(\ga-1)}{\ga}\bar{f}(1)\notag\\
	&\hskip4em +
	(e-d)\sum_{j\ge1} \Bigpar{(2j)^{-(\rchi_k-1)}-2(2j+1)^{
	-(\rchi_k-1)}+(2j+2)^{-(\rchi_k-1)}}\biggr)\notag\\
	&=\frac{(2\ga-1)(\ga-1)f(1) + (2\ga-1)(c+e) -\ga d
	+2(e-d)(2-\ga)\zeta(\richi_k)}{\ga(\rchi_k)(\richi_k)\log 2}.
\end{align}
By analytic continuation (temporarily allowing complex $\ga$ and
$\rho$), \eqref{maj3} holds for all $\ga$ with $\ga>1$ (so $\rho>1$),
but we have to be careful when $\ga=2$ and thus $\rho=2$, since then
$\zeta$ has a pole at $\richi_0=1$, and we have to interpret (by
continuity) $(2-\ga)\zeta(\richi_0)=-2\log 2$. The formula simplifies
when $\ga=2$, since then all other terms with $\zeta$ disappears. See
further \refE{E52}, where an explicit formula for $P(t)$ is given for
$\ga=2$.

From a generating function viewpoint, the fact that $\alpha=2$ is
special may be due to the identity
\begin{equation}
    \sum_{k\ge0}\frac{\alpha^k z^{2^k}}
	{1+z^{2^k}} = \frac{z}{1-z}
	\eqtext{iff}\alpha=2.
\end{equation}
\end{example}

\subsection{$(\alpha,\beta)=(1,1)$}
\begin{example}[Sequences not in \cite{Hwang2017}]\label{E11}
As this case has already been discussed in detail in
\cite{Hwang2017}, we only list sequences (together with their
closed-form expressions) that are not included in \cite{Hwang2017}.
We use the pattern $f(n) = nP(\log_2n)-Q(n)$. (Some recursions start
at some $n>2$.)
\begin{small}
\begin{center}
\begin{longtable}{ccccc}
\hline
\textbf{OEIS} & \textbf{$g(n)$} & \textbf{Initials} 
& $P(t)$ & $-Q(n)$ \\
\hline
\endfirsthead
\multicolumn{5}{c}%
{\tablename\ \thetable\ -- \textit{Continued from previous page}} \\
\hline
\textbf{OEIS} & \textbf{$g(n)$} & \textbf{Initials}
& $P(t)$ & $-Q(n)$  \\
\hline
\endhead
\hline \multicolumn{5}{r}{\textit{Continued on next page}} \\
\endfoot
\hline
\endlastfoot

\makecell[l]{A277267\\(binary trees)}
& $ 1 $ & $ \begin{array}{l}
\{f(j)\}_{2\le j\le 3} \\
=\{0,0\}\end{array}$
& $\max\left\{\begin{array}{l}
    2^{-1-\{t\}}\\1-2^{-\{t\}}\end{array}\right\}$
& $-1$ \\
\makecell[l]{A279521\\(binary trees)} & $ 1 $ 
& $ \begin{array}{l}
\{f(j)\}_{2\le j\le 3} \\
=\{0,1\}\end{array}$ 
& $\min\left\{\begin{array}{l}
    1-2^{-1-\{t\}}\\ 2^{-\{t\}}\end{array}\right\}$
    & $-1$\\ %\hline 
\makecell[l]{A294456\\(recursion)} & $ 2 $ & $\begin{array}{l}
\{f(j)\}_{1\le j\le 2} \\
=\{0,1\}\end{array}$ 
& $\min\left\{\begin{array}{l}
  2-2^{-1-\{t\}}\\ 1+2^{-\{t\}}\end{array}\right\}$ 
  & $-2$ \\ %\hline 
\makecell[l]{A295513\\(binary length)} & $ n $ & $ f(1) = -1 $ 
& $1-\{t\}-2^{1-\{t\}}$ & $n\log_2n$ \\ %\hline 
\makecell[l]{A296062\\(binary trees)} & $\mathbf{1}_{n\text{ odd}}$ 
& $ f(1) = 0 $ & see below & $0$ \\ %\hline 
\makecell[l]{A296349\\(digital sum)} & $ n-2 $ & $ f(1) = 1 $ &
$1-\{t\}-2^{1-\{t\}}$ & $n\log_2n+2$\\ %\hline 
\makecell[l]{A297531\\(subword\\complexity)}
& $ 0 $ & $\begin{array}{l}
\{f(j)\}_{4\le j\le 7} \\   %a(n+1)
=\{13,17,21,24\}
\end{array}$ &
$\min\left\{\begin{array}{l}
	4-3\cdot 2^{-2-\{t\}}\\
	3+3\cdot 2^{-2-\{t\}}\\
	2+5\cdot 2^{-1-\{t\}}
\end{array}\right\}$ & $0$ \\ %\hline 
\makecell[l]{A301336\\(digital sum)} 
& $2-\mathbf{1}_{n\text{ odd}}$ 
& $ f(1) = -1 $ & see below & $-2$ \\ %\hline 
\makecell[l]{A303548\\(Hamming\\weight)} 
& $\mathbf{1}_{n\equiv 3\bmod 4}$ & $ f(1) = 0 $ 
& see below & $0$ \\ %\hline 
\makecell[l]{A316936\\(word\\complexity)} & 
$\begin{array}{l}
    \frac{1}{2}{n}^{2}-2n\\
    +\begin{cases} 
	    1\\
    	\frac{1}{2} 
    \end{cases}
\end{array}$ & $ f(1) = 3 $ &
$-1+2\{t\}+2^{2-\{t\}}$ & $\begin{array}{l}
    n^2-2n\log_2n\\
    -1 \end{array}$\\ %\hline
\caption{Some sequences with $\gL_{1,1}[f]=g$
(\refE{E11}).}
\label{tb-E11}
\end{longtable}
\end{center}
\end{small}
In particular, the sequence A297531 gives the maximum possible
subword complexity over all binary overlap-free words of a given
length, A301336 counts the difference between the total number of
$1$s and the total number of $0$s in the binary expansions of
$\{0,\dots,n-1\}$, and A303548 equals the sum of the distances from
$n$ to the nearest number with a given Hamming weight.

Note that some of these sequences are discussed in \cite{Hwang2017} 
and were subsequently added to OEIS:
\begin{center}
\begin{tabular}{c|cccccc}
OEIS & A294456 & %A295513 & 
A296062 & %A296349 & %A301336 & 
A303548\\ \hline
Example\ in \cite{Hwang2017} & 3.2 & %5.1 & 
3.6 & %5.1 & %3.6 & 
3.7
\end{tabular}
\end{center}

Here $\A295513(n)= \A001855(n)-1$, $\A296349(n)=\A083652(n-1)$, and 
\begin{align}
    f_{\text{A296349}}(n) = f_{\text{A295513}}(n)+2,
    \quad
    f_{\text{A301336}}(n) =n-2-f_{\text{A296062}}(n).
\end{align}
On the other hand, with $\chi_{k}:=\frac{2k\pi i}{\log 2}$ as above, 
using \cite[Theorem 3 and Example 3.6]{Hwang2017}, 
\begin{align}\label{E:A296062-P}
    P_{\text{A296062}}(t) 
    &= 2-\log_2\pi+\frac1{\log2}
    \sum_{k\ne0}\frac{1+2\zeta(\chi_k)}{\chi_k(1+\chi_k)}
    \, e^{2k\pi i t}\\
    P_{\text{A301336}}(t) 
    &= 1-P_{\text{A296062}}(t) \\
    P_{\text{A303548}}(t) \label{A303548}
    &= \tfrac12\log_2\pi-2\log_2\Gamma\lpa{\tfrac34}
    +\frac1{\log2}
    \sum_{k\ne0}\frac{\zeta(\chi_k)-2\zeta(\chi_k,\frac34)}
    {\chi_k(1+\chi_k)}\, e^{2k\pi i t}.
\end{align}
Finally, the sequence A330038, satisfying $\Lambda_{1,1}[f]=\tr{\frac 
n2}$ with $f(1)=1$, is nothing but A000788 (best case of mergesort or 
partial sum of $\nu(n)$) plus $n$; see \cite[Example 5.2]{Hwang2017}. 
\end{example}

\subsection{$(\alpha,\beta)=(2,2)$}

While most analysis here parallels that in \cite{Hwang2017}, we will
see that there are subtle differences in the genesis of periodic
oscillations. In particular, as remarked above, the periodic function
$P(t)$ is generally not continuously differentiable.

\begin{example}[Periodic functions differentiable except at integers] 
	\label{E52}
Consider the simple case when $\Lambda_{2,2}[f]=c$ for $n\ge2$.
Then, by \refL{lmm-g-const}, \refE{Eg=0} and \eqref{b16} (or directly
by \eqref{b15} and \eqref{b16}), we deduce that
\begin{align}
	f(n)
	&=\lpa{f(1)+\tfrac13\,c}S_{2,2}(n)-\tfrac13\,c
	= \lpa{f(1)+\tfrac13\,c}2^{L_n}\lpa{3n-2^{L_n+1}}
	-\tfrac13\,c\qquad(n\ge1),
\end{align}
so that $f(n) = n^2P(\log_2n)-\frac13c$, where
\begin{equation}
	P(t) =  \lpa{f(1)+\tfrac13\,c}P_0(t)
	= \lpa{f(1)+\tfrac13\,c} 2^{-\{t\}}(3-2^{1-\{t\}}),
\end{equation}
is continuously differentiable on $\oi$, but not at integer $t$, 
where the derivative $P'(t)$ has a jump.

A more general pattern that we identified from the OEIS sequences 
is of the form (see Lemma~\ref{lmm-g-linear} and \refE{EIIb})
\begin{align}\label{L22-linear}
    \Lambda_{2,2}[f]= g(n) = c n+ \begin{cases}
        d,&\text{ if $n$ is even}; \\
        e,&\text{ if $n$ is odd}.
    \end{cases}
\end{align}
Then the solution is given by $f(n) = n^2P(\log_2n)-cn-\frac13\,d$ 
for $n\ge1$, where
\begin{equation}\label{maj4}
    P(t) = d-e+\lpa{f(1)+c-\tfrac23\,d+e}
    2^{-\{t\}}\lpa{3-2^{1-\{t\}}},
\end{equation}
which easily is verified using \refL{lmm-g-linear},
\eqref{L-31c}--\eqref{L-31d}, and the special case
$\gL_{2,2}[n^2]=-\Iodd$.

This can be compared with the recursion $\Lambda_{\alpha,\alpha}[f] =
g(n)$ for a general $\ga$ (with the same $g(n)$) studied in
\refE{EIIb}, where it is readily checked that the Fourier
coefficients \eqref{maj3} (with $\ga=2$) agree with \eqref{maj4} and
\eqref{maj1}.

\medskip
\noindent
\begin{minipage}{0.7\textwidth}
The prototype sequence in this category is A073121:
\begin{equation}
    f_{\text{A073121}}(n)
    =S_{2,2}(n)
    =n^2P_{\text{A073121}}(\log_2n),
\end{equation}
where $P_{\text{A073121}}(t)=P_0(t):=2^{-\{t\}}\lpa{3-2^{1-\{t\}}}$,
which arises as an upper bound for the cardinality of some regular
expressions; see \cite{Ellul2005} for the more general form
$\Lambda_{\alpha,\alpha}[f]=0$. The Fourier coefficients of $P_0(t)$
are, by \eqref{maj1}, of the form
\end{minipage}
\begin{minipage}{0.25\textwidth}
    \hfill
    \includegraphics[width=0.8\textwidth]{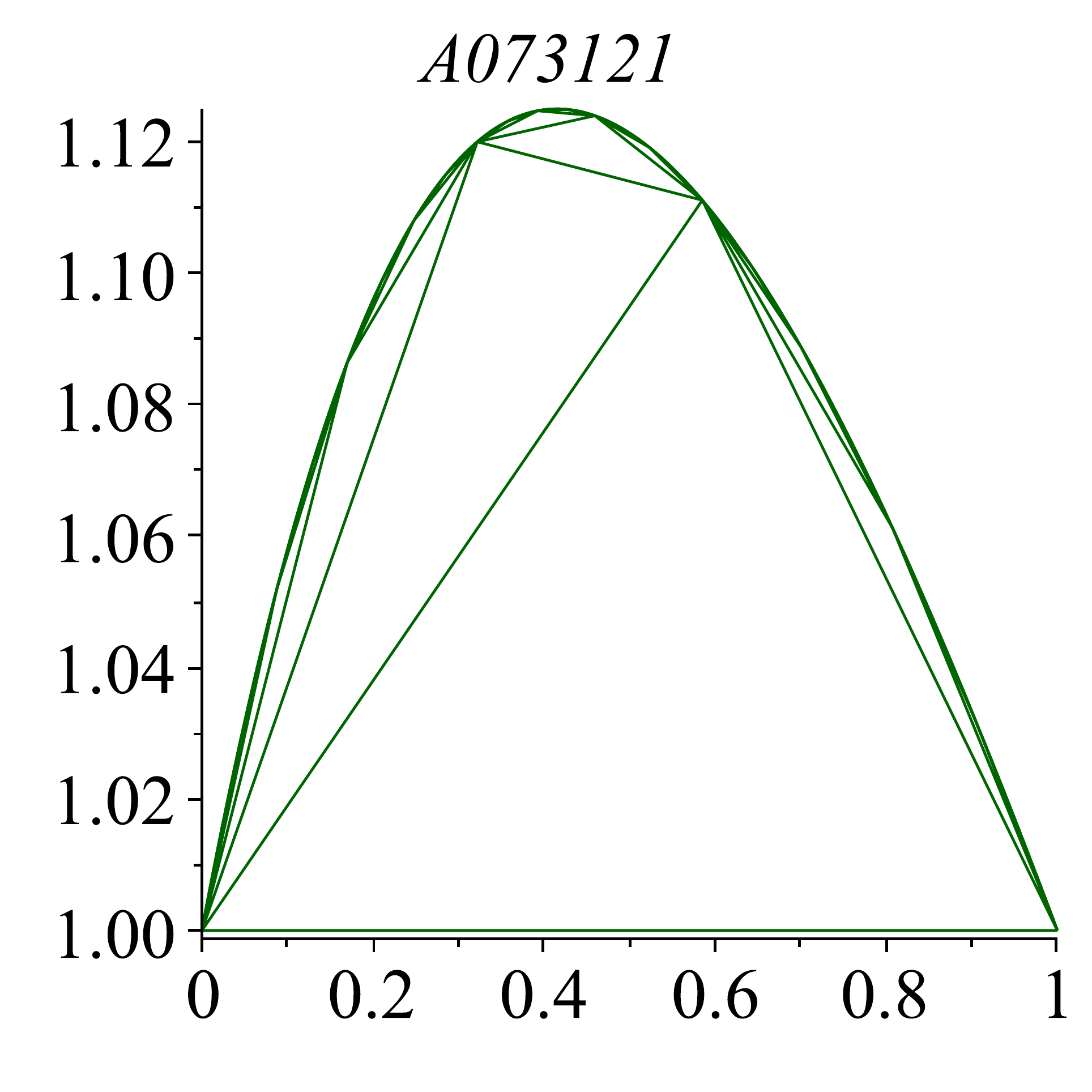}
\end{minipage}

\begin{align}\label{hp022}
	\hP_0(k)=\frac{3}{2\log 2\cdot(1+\chi_k)(2+\chi_k)}.
\end{align}

With $S_{2,2}(n)$, the solution to \eqref{L22-linear} can 
alternatively be written as
\begin{equation}
    f(n) 
    = (d-e)n^2+\lpa{f(1)+c-\tfrac23\,d+e}S_{2,2}(n)
    -cn-\tfrac13d.
\end{equation}
Table \ref{tb-22-1} gives more examples from OEIS that are 
periodically equivalent to $S_{2,2}(n)$.
\end{example}
\begin{longtable}{llllll}
\multicolumn{5}{c}{{}} \\
\multicolumn{1}{c}{OEIS id.} &
\multicolumn{1}{c}{Context} &
\multicolumn{1}{c}{$g(n)$} &
\multicolumn{1}{c}{$f(1)$} &
\multicolumn{1}{c}{$f(n)$} \\ \hline    
A063915 & \makecell[l]{$\sum_{1\le k<n}2^{L_k}$}
& $1$ & $0$ & $\frac13(S_{2,2}(n)-1)$ \\
A073121 & \makecell[l]{Regular\\expressions}
& $0$ & $1$ & $S_{2,2}(n)$ \\
A181497 & \makecell[l]{Combinatorial\\sequence}
& $-1$ & $1$ & $\frac13(2S_{2,2}(n)+1)$ \\
A236305 & \makecell[l]{Nim\\game} & $\begin{cases}
    0 \\ -3
\end{cases}$ & $1$ & 
$3n^2-2S_{2,2}(n)$ \\
A255748 & \makecell[l]{Cellular\\automaton}
& $\begin{cases}
    -\frac12n+2 \\ -\frac12n+\frac52
\end{cases}$ & $0$
& $-\frac12n(n-1)+\frac23(S_{2,2}(n)-1)$ \\
A256249 & \makecell[l]{Josephus\\problem}
& $\begin{cases}
    1 \\ 0
\end{cases}$ & $0$ 
& $n^2-\frac13(2S_{2,2}(n)+1)$ \\
A256250 & \makecell[l]{Cellular\\automaton}
& $\begin{cases}
    1 \\ -3
\end{cases}$ & $1$
& $4n^2-\frac13(8S_{2,2}(n)+1)$ \\
A256266 & \makecell[l]{Cellular\\automaton}
& $\begin{cases}
    -3n+12 \\ -3n+15
\end{cases}$ & $0$
& $-3n(n-1)+4(S_{2,2}(n)-1)$ \\ 
A256534 & \makecell[l]{Cellular\\automaton} & 
$\begin{cases}
    0 \\ -12
\end{cases}$ & $4$
& $4(3n^2-2S_{2,2}(n))$ \\
A261692 & \makecell[l]{Cellular\\automaton} 
& $\begin{cases}
    1 \\ 2
\end{cases}$ & $0$
& $-n^2+\frac13(4S_{2,2}(n)-1)$ \\
A262620 & \makecell[l]{Cellular\\automaton} 
& $\begin{cases}
    1 \\ 5
\end{cases}$ & $1$
& $-4n^2+\frac13(16S_{2,2}(n)-1)$ \\
A266538 & \makecell[l]{Josephus\\problem}
& $\begin{cases}
    2 \\ 0
\end{cases}$ & $0$
& $2n^2-\frac23(2S_{2,2}(n)+1)$ \\ \hline
\caption{Sequences satisfying $\Lambda_{2,2}[f]=g$ and their 
relations to $S_{2,2}(n)$ (\refE{E52}).}
\label{tb-22-1}
\end{longtable}    
%\end{center}
\vspace*{-.5cm}

\begin{example}[Piecewise differentiable periodic functions]
\label{E22-piece} 
When the recurrence $\Lambda_{2,2}[f]=g$ is satisfied only for $n\ge
n_0>1$, the resulting periodic function is specified according to the
initial conditions $f(n)$ with $n\le n_0$. For simplicity, we
consider the sequence A080075 (Proth numbers), which denotes the
numbers of the form $(2r+1)2^k+1$ for $k\ge1$ and $2r+1<2^k$ and has
many variants. The sequence $f(n)=\text{A080075}(n-1)$ then satisfies
$\Lambda_{2,2}[f]=-3$ with the initial conditions $f(2)=3$ and
$f(3)=5$. The solution is then given by $f(n) = n^2P(\log_2n)+1$ for
$n\ge2$ with
\begin{equation}
    P(t) = \begin{cases}
        2^{-\{t\}}\lpa{1-2^{-1-\{t\}}},
        &\text{if }t\in[0,\log_2\frac32];\\
        2^{1-\{t\}}\lpa{1-2^{-\{t\}}}, 
        &\text{if }t\in[\log_2\frac32,1];
    \end{cases}
\end{equation}
see \cite[Examples 3.1 and 3.2]{Hwang2017} for similar behaviour.
\begin{center}
\begin{minipage}{0.65\textwidth}
\begin{tabular}{lrcl}
\multicolumn{4}{c}{{}} \\
\multicolumn{1}{c}{OEIS id.} &
\multicolumn{1}{c}{$g(n)$} &
\multicolumn{1}{c}{$(f(2),f(3))$} &
\multicolumn{1}{c}{$f(n)$} \\ \hline    
A080075 & $-3$ & $(3,5)$ &  \\
A082662 & $0$ & $(1,2)$ & $\frac12(f_{\text{A080075}}(n)-1)$ \\
A112714 & $3$ & $(1,3)$ & $f_{\text{A080075}}(n)-2$ \\
A116882 & $0$ & $(2,4)$ & $f_{\text{A080075}}(n)-1$ \\
A260711 & $0$ & $(8,16)$ & $4(f_{\text{A080075}}(n)-1) $ \\ \hline
\end{tabular}
\end{minipage}
\begin{minipage}{0.25\textwidth}
    \vspace*{.4cm}
    \includegraphics[width=0.8\textwidth]{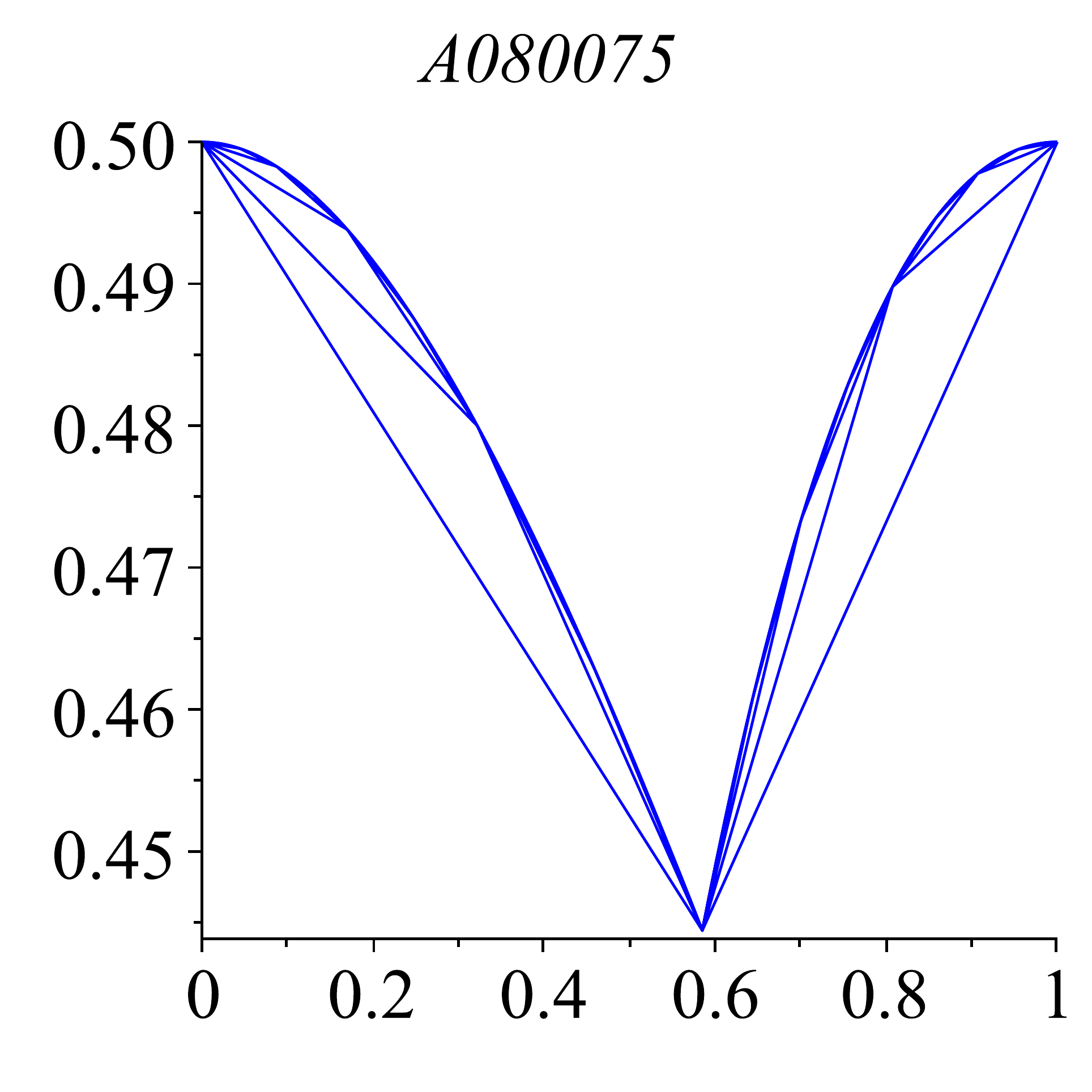}
\end{minipage}
\end{center}
\end{example}

\begin{example}[Non-differentiable periodic functions]\label{E22-ns}
A few sequences defined as the partial sum of the bitwise operator
between $j$ and $n-j$ or their complements satisfy
$\Lambda_{2,2}[f]=g$ with different $g$; see \refTab{tb-22-ns} for a
summary; in all cases $f(1)=0$. Note that the \textsf{NOT} operator
$\bar{j}$ uses the full number of bits $L_n+1$ for each $0\le j<n$
and equals
\begin{equation}
    (\underbrace{1\dots 1}_{L_n-s}
    0\bar{b}_{s-1}\cdots\bar{b_0})_2 
    \text{ if }j=(1b_{s-1}\cdots b_0)_2. 
\end{equation}

\begin{center}
\begin{longtable}{clll}
\multicolumn{4}{c}{{}} \\
\multicolumn{1}{c}{OEIS id.} &
\multicolumn{1}{c}{Description} &
\multicolumn{1}{c}{$g(n)$} & 
\multicolumn{1}{c}{$f(n)$} \\ \hline
A006581 & $\sum\limits_{1\le j\le n-2} 
    \bigl(j \textsf{ AND } (n-1-j) \bigr)$
& $\begin{cases}
    0\\
    \frac {n-1}2 
\end{cases}$ & $n^2P(\log_2n)$ \\ 
A006582 & $\sum\limits_{1\le j\le n-2} 
    \bigl(j \textsf{ XOR } (n-1-j) \bigr)$
& $\begin{cases} 
    3n-6\\
    2n-6 
\end{cases}$ & $n^2P(\log_2n)-3n+2$ \\ 
A006583 & $\sum\limits_{1\le j\le n-2} 
    \bigl(j \textsf{ OR } (n-1-j) \bigr)$
& $\begin{cases} 
    3n-6 \\
    \frac 52n-\frac{13}2 
\end{cases}$ & $n^2P(\log_2n)-3n+2$ \\ 
A090889 & $\sum\limits_{1\le j<n} jv_2(j)(n-j)$ 
& $\begin{cases}
        \frac1{12}\,{n}^{3}-\frac 13\,n\\
        \frac1{12}\,{n}^{3}-\frac1{12}\,n
    \end{cases}$ & $\frac16n^3+n^2P(\log_2n)+\frac n3$\\
A099027 & $\sum\limits_{0\le j<n} 
    \bigl(\bar{j}\textsf{ AND } (n-1-j) \bigr)$
& $\begin{cases} 
    \frac{n}{2} \\
    0  
\end{cases}$ & $n^2P(\log_2n)-\frac n2$ \\ \hline
\caption{Sequences satisfying $\Lambda_{2,2}[f]=g$ with
non-smooth periodic functions (\refE{E22-ns}).}\label{tb-22-ns}
\end{longtable}
\end{center}
Here $v_2(n)$ denotes the dyadic valuation of $n$ (exponent of the
highest power of $2$ dividing $n$). Note that the recurrence
provided on OEIS for A090889 is incorrect (and the generating
function misses a factor of $2$). 

These apparently different sequences are all periodically
equivalent. Indeed, by the recurrences and induction, we can prove
the relations
\begin{equation}
	\begin{split}
    f_{\text{A006582}}(n) 
    &= (n-1)(n-2)-2f_{\text{A006581}}(n) \\
    f_{\text{A006583}}(n) 
    &= (n-1)(n-2)-f_{\text{A006581}}(n) \\
    f_{\text{A090889}}(n) 
    &= \tfrac16n(n-1)(n-2)+f_{\text{A006581}}(n)\\
    f_{\text{A099027}}(n) 
    &= \tfrac12n(n-1)-f_{\text{A006581}}(n).
	\end{split}
\end{equation}
The Fourier expansion of $P_{\text{A006581}}(t)$ is, by \eqref{c4+} 
and standard calculations, given by
\begin{equation}\label{E:P-A006581}
    P_{\text{A006581}}(t)
    = \frac12-\frac1{4\log 2}
    +\frac1{\log 2}\sum_{k\ne0}
    \frac{\zeta(\chi_k)}{(1+\chi_k)(2+\chi_k)}
    \, e^{2k\pi i t}.
\end{equation}
We prove in Appendix~\ref{AE22-ns} that the continuous function 
$P_{\textrm{A006581}}$ is nowhere differentiable, and that it is not 
Lipschitz.

Another sequence $\A048641$ is defined as the sum
$\sum_{k<n}\gamma(k)$, where $\gamma(k)= k\textsf{ XOR
}\tr{\frac12k}$ denotes the numerical value of the binary reflected
Gray code of $k$ (A003188), and satisfies the recurrence
$\Lambda_{2,2}[f]= \frac12\lpa{n-\sin\frac12n\pi}$ with $f(1)=0$. We
then obtain $f(n) = n^2P(\log_2n)-\frac12n$ with
\begin{align}\label{E:A048641}
    P(t) 
    = \frac{\pi}{8\log 2}+ \frac1{4\log 2}
    \sum_{k\ne0}\frac{\zeta\lpa{1+\chi_k,\frac14}
    -\zeta\lpa{1+\chi_k,\frac34}}
    {(1+\chi_k)(2+\chi_k)}\, e^{2k\pi i t}.
\end{align}
We leave the question whether this function is nowhere differentiable 
as an open problem.

Yet another related example is A022560; this is the sum $f(n) :=
\sum\limits_{1\le k<n}2^{v_2(k)}(n-k)$, which satisfies the
recurrence $\Lambda_{2,2}[f]= \tr{\frac14n^2}$. (Note that one of the
recursions given in OEIS is incorrect.) In this case $g(n)$ grows too
rapidly for \refT{Theorem1} to be directly applicable, since the
series \eqref{t1Q} diverges. A modification of our arguments, see
\refE{EB} for details, shows that we have
\begin{align}\label{f022560}
	f_{\A022560}(n)=\tfrac14n^2\log_2n+n^2P_{\A022560}(\log_2n)
\end{align}
with a logarithmic factor in the leading term, and periodic
fluctuations given by
\begin{equation}\label{P022560}
    P_{\A022560}(t) 
    = \frac38+\frac{2\gamma-3}{8\log 2} 
    + \frac1{2\log 2}
    \sum_{k\ne0}\frac{\zeta\lpa{1+\chi_k}}
    {(1+\chi_k)(2+\chi_k)}\, e^{2k\pi i t}.
\end{equation}
See Figure \ref{fig-robust}.
\end{example}
\begin{figure}[!ht]
\centerline{\footnotesize
\begin{tabular}{c c c}
\includegraphics[height=2.6cm]{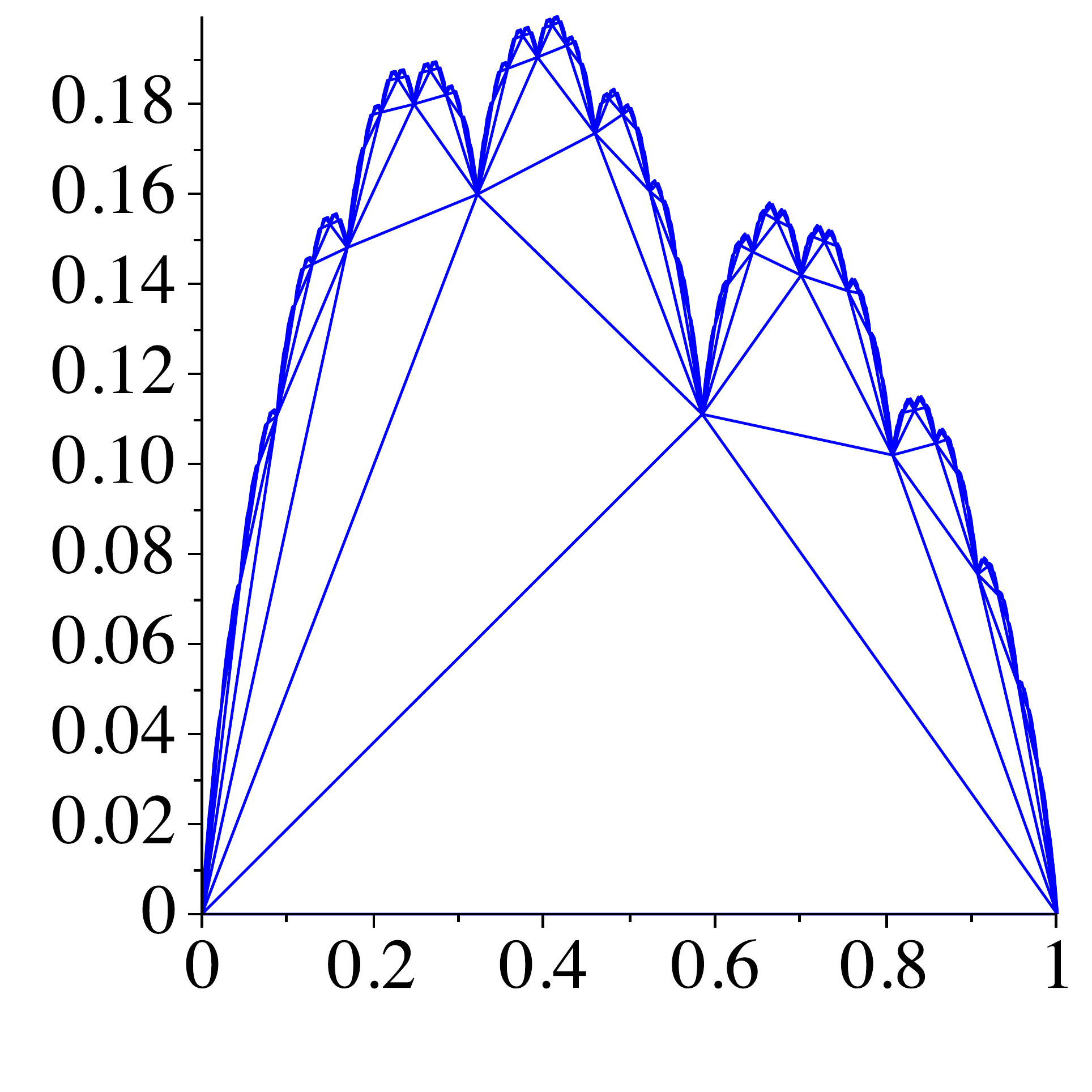} &
\includegraphics[height=2.6cm]{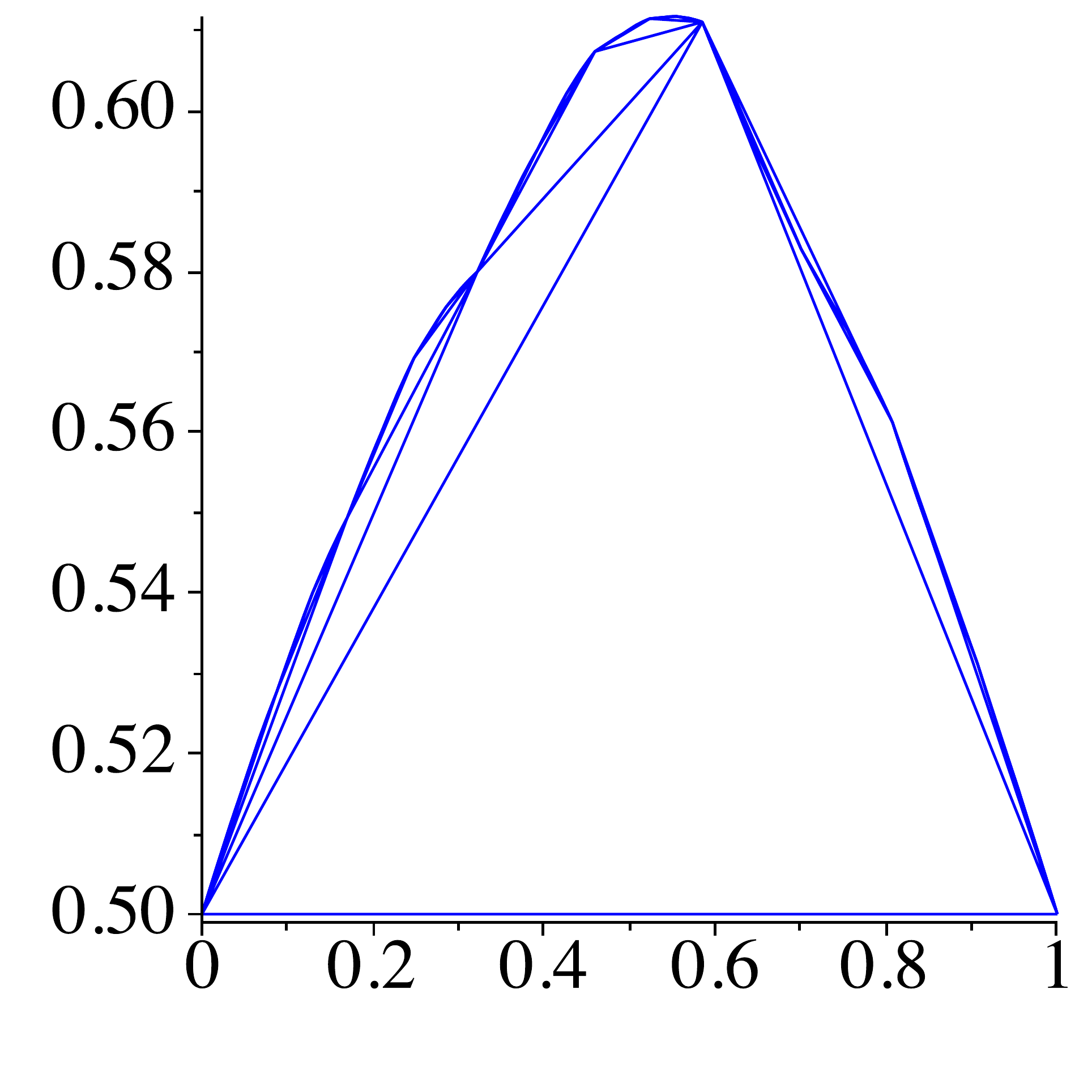} & 
\includegraphics[height=2.6cm]{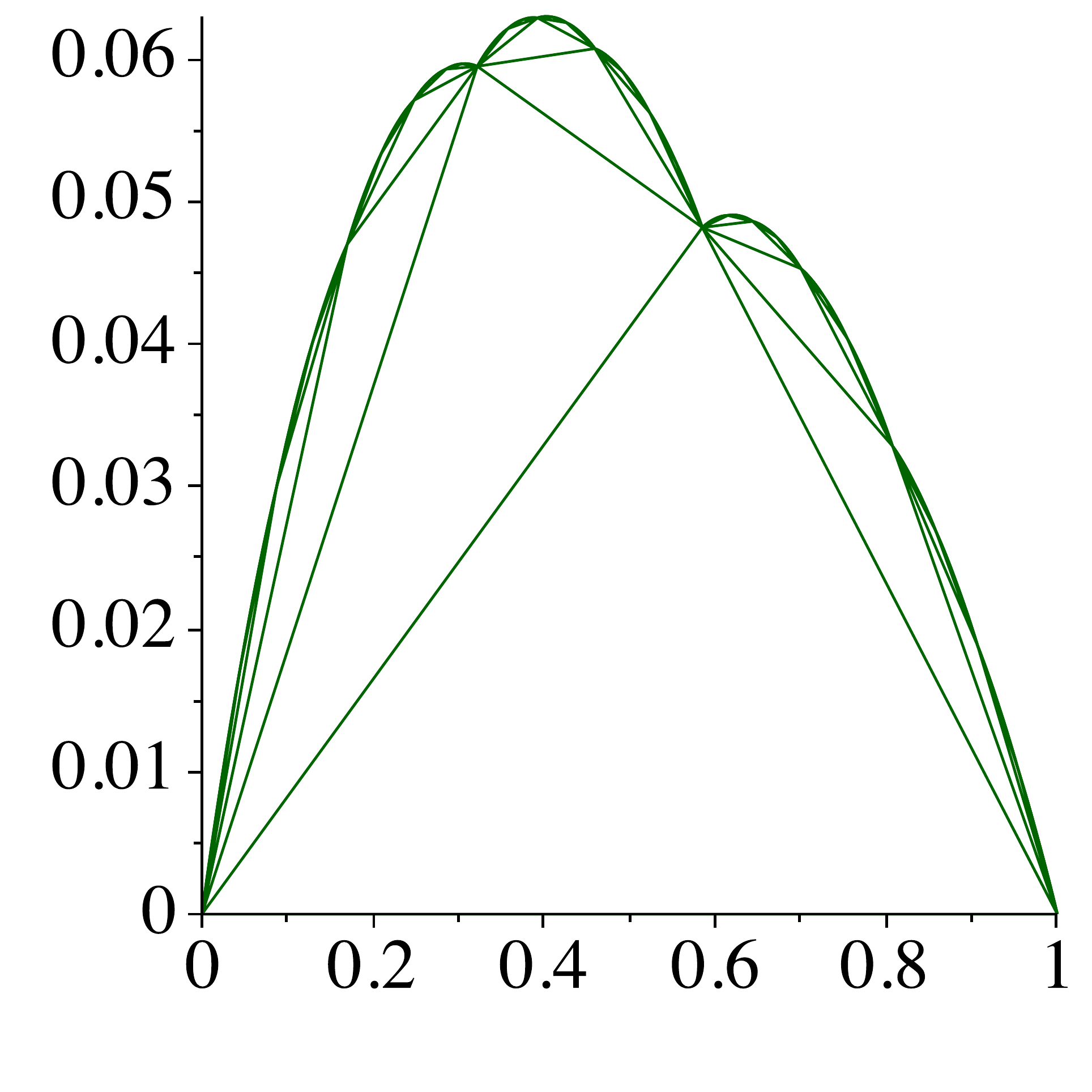} \\
A006581 & A048641 & A022560 
\end{tabular}}
\caption{Fluctuations (properly normalised) in the three periodically
distinct cases in \refE{E22-ns}. At least the first is nowhere
differentiable.} \label{fig-robust}
\end{figure}

\begin{example}[Sensitivity test: small variations inducing big
differences] \label{E5.7} While the previous examples show that many
different toll functions lead to periodically equivalent
oscillations, the same recurrence also exhibits the opposite
sensitivity property, as we now examine. We begin with the difference
$f(n) = f_{\text{A048641}}(n)-\binom{n}2$, which gives A048644 and
satisfies the recurrence $\Lambda_{2,2}[f]= \mathbf{1}_{n \bmod
4\equiv 3}$ with the solution $f(n) = n^2P(\log_2n)$, where
$P(t)=P_{\A048641}(t)-\frac12$ is given by \eqref{A048644} below.

It is interesting to see that changing $3$ to other remainders
results in drastically different periodic functions; compare also
\cite[Example 3.7]{Hwang2017}. Let $\Lambda_{2,2}[f_j] =\mathbf{1}_{n
\bmod 4\equiv j}$ for $j=0,1,2,3$ with $f_j(1)=0$. Then
\begin{align}
    f_j(n) = \begin{cases}
        n^2P_0(\log_2n) -\frac13
        +\frac14\cdot 
        \mathbf{1}_{n \text{ odd}}, & 
        \text{ if }j=0,\\
        n^2P_j(\log_2n) , & 
        \text{ if }j=1,3,\\
        n^2P_2(\log_2n) 
        -\frac14\cdot
        \mathbf{1}_{n \text{ odd}}, & 
        \text{ if }j=2,\\
    \end{cases}
\end{align}
where $P_0(t)$ is continuously differentiable on $\oi$ (not the same
as $P_0(t)$ from Section \ref{S:recurrence}), $P_2(t)=\frac14$ is
\emph{a constant}, and $P_1(t)$ has many visible cusps in the unit
interval (see Figure~\ref{fig-sensitivity}). More explicit
expressions are given by
\begin{align}
    P_0(t) &= \tfrac34-2^{1-\{t\}}+\tfrac134^{1-\{t\}}, \\
    P_1(t) &= \frac{3}{4\log 2}
    -\frac{\pi}{8\log 2}-\frac12
    -\frac1{4\log 2}\sum_{k\ne0}
    \frac{\zeta(1+\chi_k,\frac14)-\zeta(1+\chi_k,\frac34)-6}
    {(1+\chi_k)(2+\chi_k)}\, e^{2k\pi i t},\\
    P_2(t) &=\frac14,\\
    P_3(t) &
    = \frac{\pi}{8\log 2}-\frac12+ \frac1{4\log 2}
    \sum_{k\ne0}\frac{\zeta\lpa{1+\chi_k,\frac14}
    -\zeta\lpa{1+\chi_k,\frac34}}
    {(1+\chi_k)(2+\chi_k)}\, e^{2k\pi i t}; \label{A048644}
\end{align} 
see Figure~\ref{fig-sensitivity} for an illustration. The difference
$\zeta(s,\frac14)-\zeta(s,\frac34)$ is also known as the Dirichlet
beta function; it appears also in several other formulas above and
below. Note that $f_2(n) = \ltr{\frac14n^2}$, which equals the 
quarter-squares A002620, and the sum $f_1+f_2+f_3+f_4$ equals A063915
in \refE{EIIa} and \refTab{tb-22-1}. Also the mean values of these
periodic functions are given by 
\begin{align}
	\lpa{-\tfrac1{2\log 2}+\tfrac34, 
	\tfrac{3}{4\log 2}-\tfrac{\pi}{8\log 2}-\tfrac12,
	\tfrac14,
	\tfrac{\pi}{8\log 2}-\tfrac12}
	\approx\lpa{0.02865, 0.01548, 0.25, 0.06655},  
\end{align}
respectively, and we see that the mean value of $P_2$ is much larger 
than those of other three. 

Similarly, by comparing with \refTab{tb-22-1}, the sequence A256249
for a sum for the Josephus problem has $f_{\A256249}=f_0+f_2$. The
sequence A266540 for another sum in the Josephus problem has $f(1)=0$
and $\Lambda_{2,2}[f]=1+\cos\lpa{\frac12n\pi}$ (which has the
periodic pattern $(2,1,0,1)$ for $n\ge0$); hence
$f_{\A266540}=2f_0+f_1+f_3$. It is written on the OEIS page that ``It
appears that this sequence has a fractal (or like-fractal)
behaviour.'' This is untrue because from the generating function given
there
\begin{equation}
    \sum_{n\ge1}f(n) z^n
    = \frac{z^3(1+z^2)}{(1-z^2)(1-z)^2}
    -\frac{z}{(1-z)^2}\sum_{k\ge2}2^{k-1}z^{2^k},
\end{equation}
or from our approach, we can derive the identity 
$f(n) = n^2P(\log_2n)-\frac23+\frac12\mathbf{1}_{n\text{ is odd}}$
with $P(t) = \frac12-2^{-\para{t}}+\frac234^{-\para{t}}$. 

Changing the recurrence to
$\Lambda_{2,2}[f] =1-\cos\lpa{\frac12n\pi}$ does not alter the
smooth nature of the periodic function because $f(n) =
n^2P(\log_2n)-\frac12\mathbf{1}_{n\text{ is odd}}$, where $P(t) =
-\frac12+3\cdot 2^{-t}-2\cdot 4^{-t}$. However, switching cosine
function to sine function does change the nature of the periodic
oscillation because we then have the solution $f(n) =
n^2P(\log_2n)-\tfrac13$ when (\emph{i})
$g(n)=1+\sin\lpa{\frac12n\pi}$, where
\begin{equation}
    P(t) = \frac1{\log 2}-\frac{\pi}{4\log 2}
    -\frac1{2\log 2}\sum_{k\ne0}
    \frac{\zeta\lpa{1+\chi_k,\frac14}
    -\zeta\lpa{1+\chi_k,\frac34}-4}{(1+\chi_k)(2+\chi_k)}
    \, e^{2k\pi i t},
\end{equation}
and (\emph{ii}) $g(n)=1-\sin\lpa{\frac12n\pi}$, where
\begin{equation}
    P(t) = \frac{\pi}{4\log 2}-\frac1{2\log 2}
    +\frac1{2\log 2}\sum_{k\ne0}
    \frac{\zeta\lpa{1+\chi_k,\frac14}
    -\zeta\lpa{1+\chi_k,\frac34}-2}{(1+\chi_k)(2+\chi_k)}
    \, e^{2k\pi i t}.
\end{equation}
\end{example}

\begin{figure}[!ht]
\centerline{\footnotesize
\begin{tabular}{c c c c}
\includegraphics[height=2.6cm]{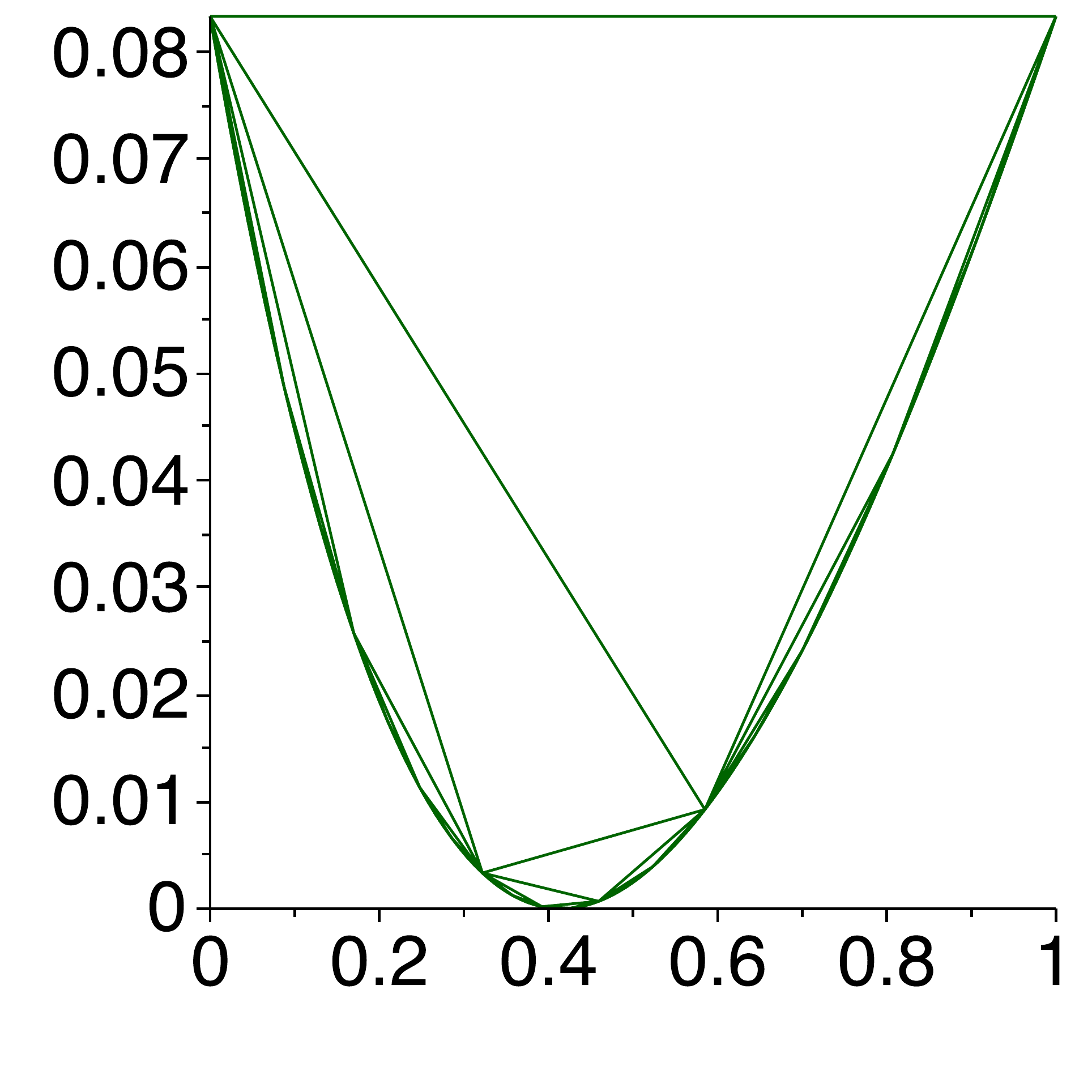} &
\includegraphics[height=2.6cm]{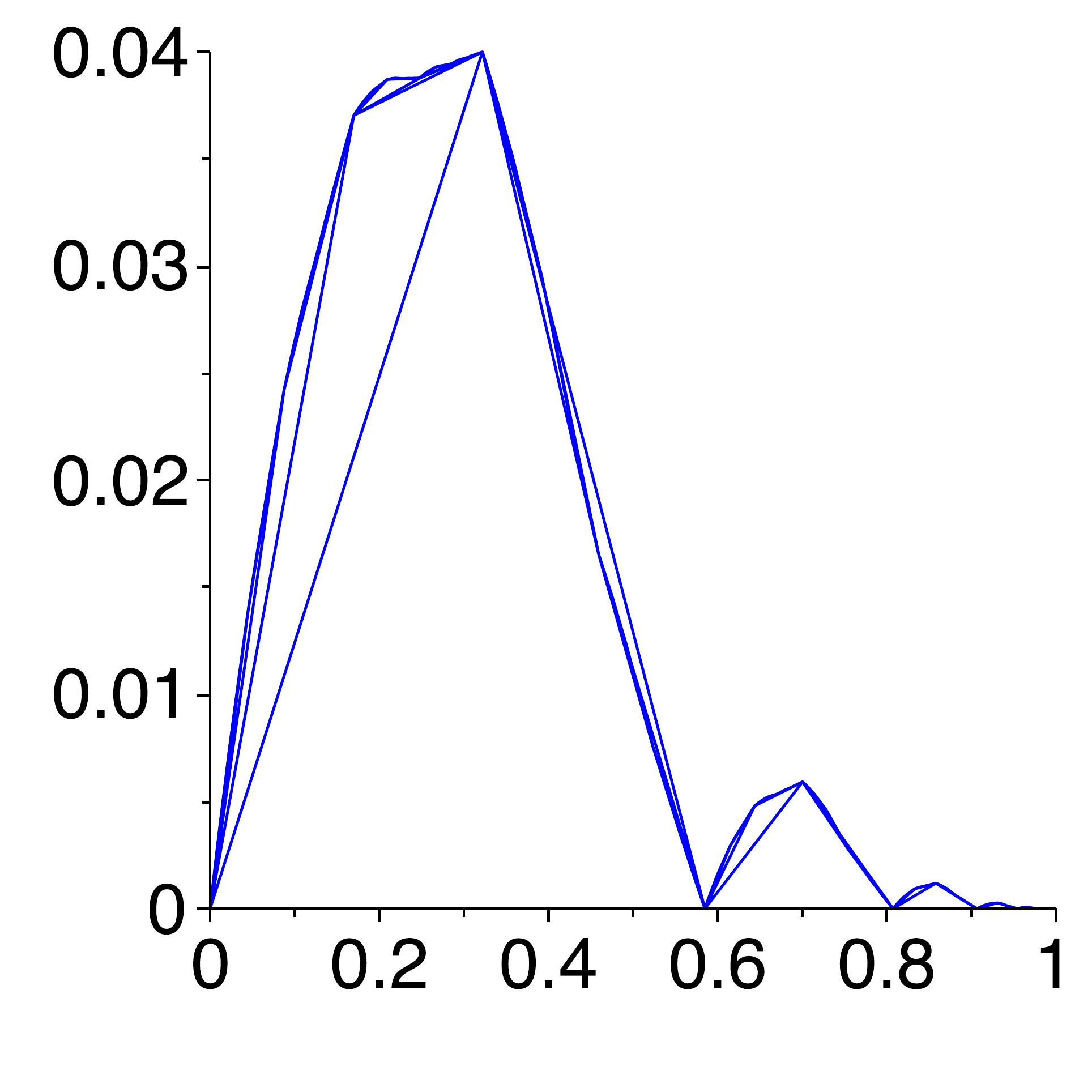} &
\includegraphics[height=2.6cm]{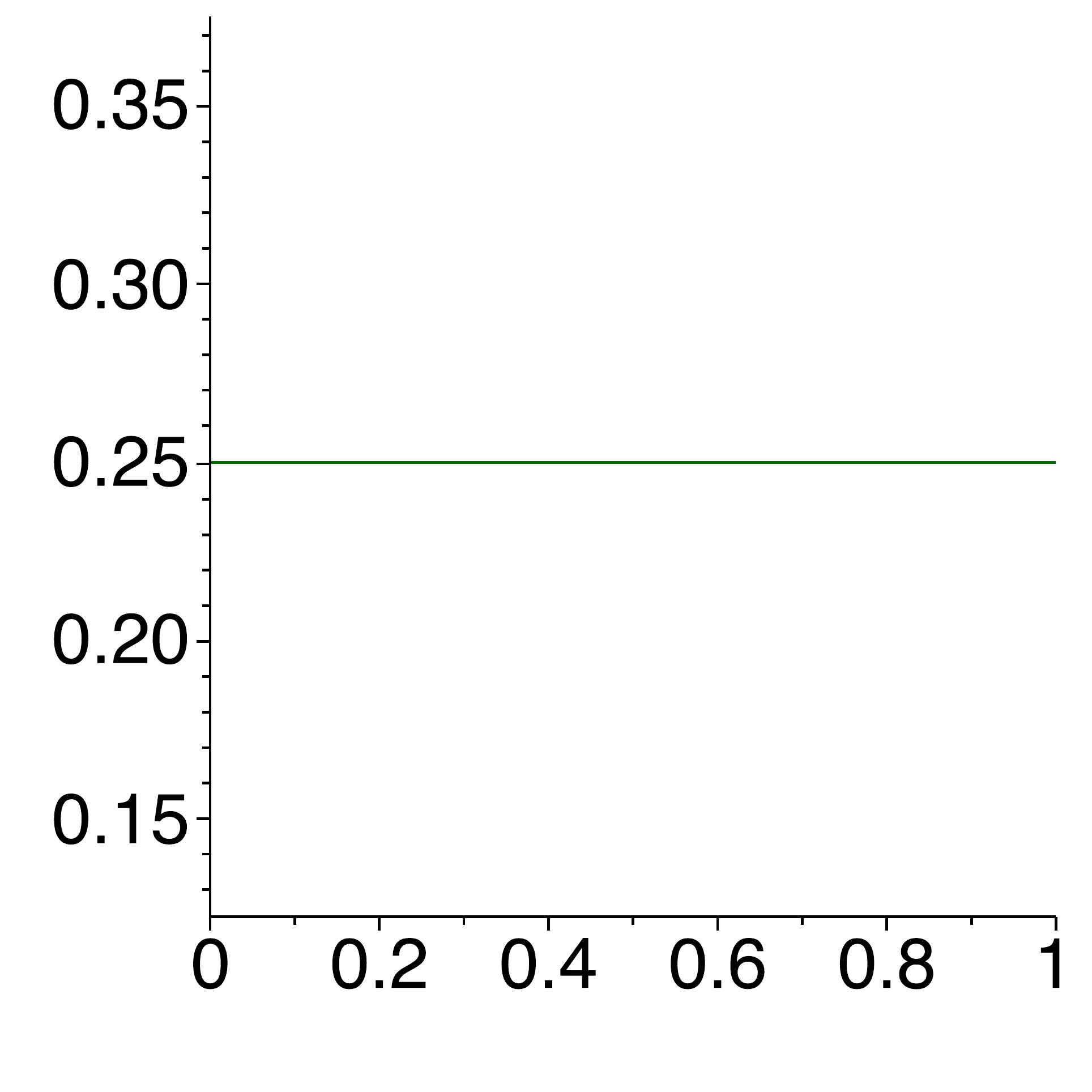} &
\includegraphics[height=2.6cm]{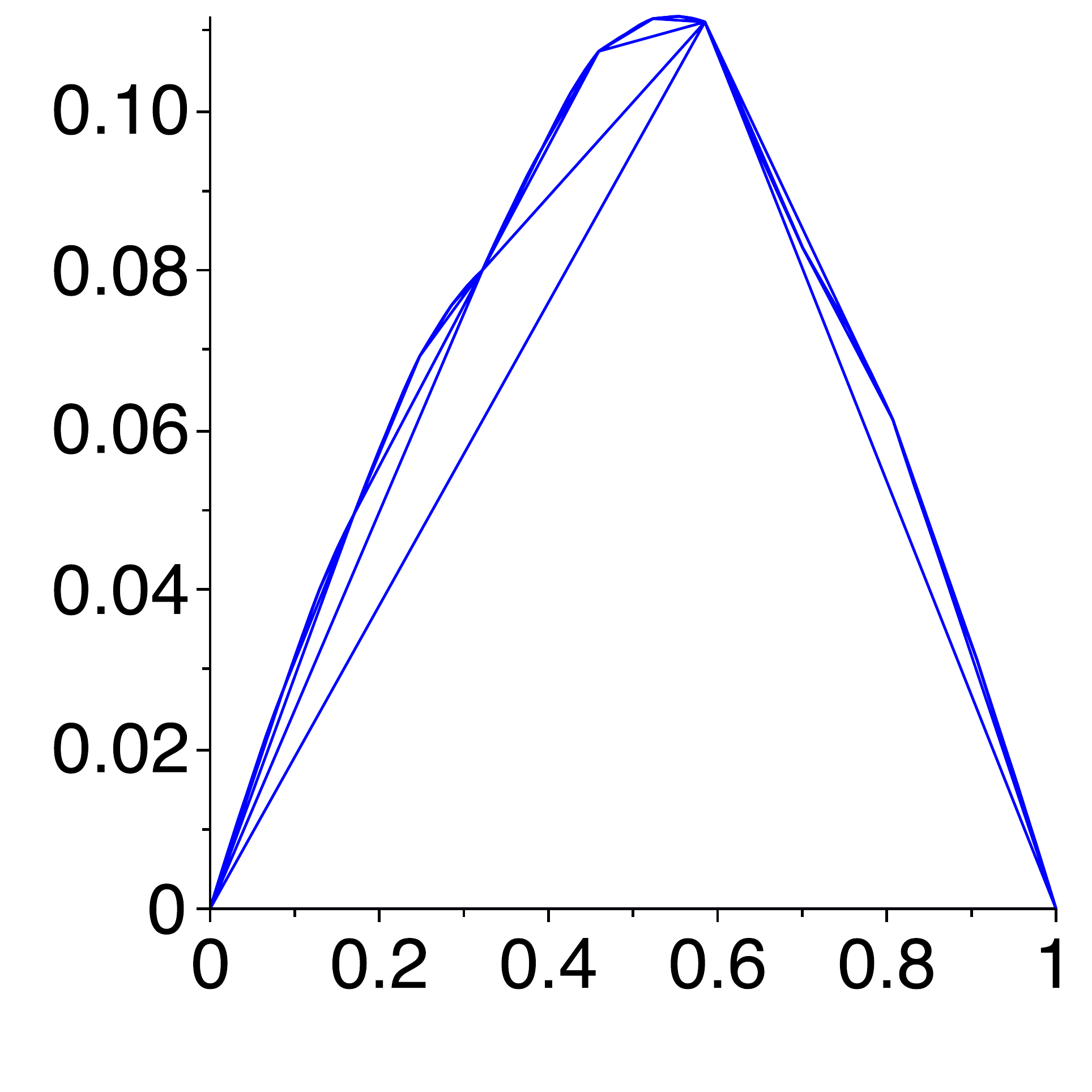} \\
$g(n)=\mathbf{1}_{n\bmod 4\equiv 0}$ &
$g(n)=\mathbf{1}_{n\bmod 4\equiv 1}$ &
$g(n)=\mathbf{1}_{n\bmod 4\equiv 2}$ &
$g(n)=\mathbf{1}_{n\bmod 4\equiv 3}$ \\
\includegraphics[height=2.6cm]{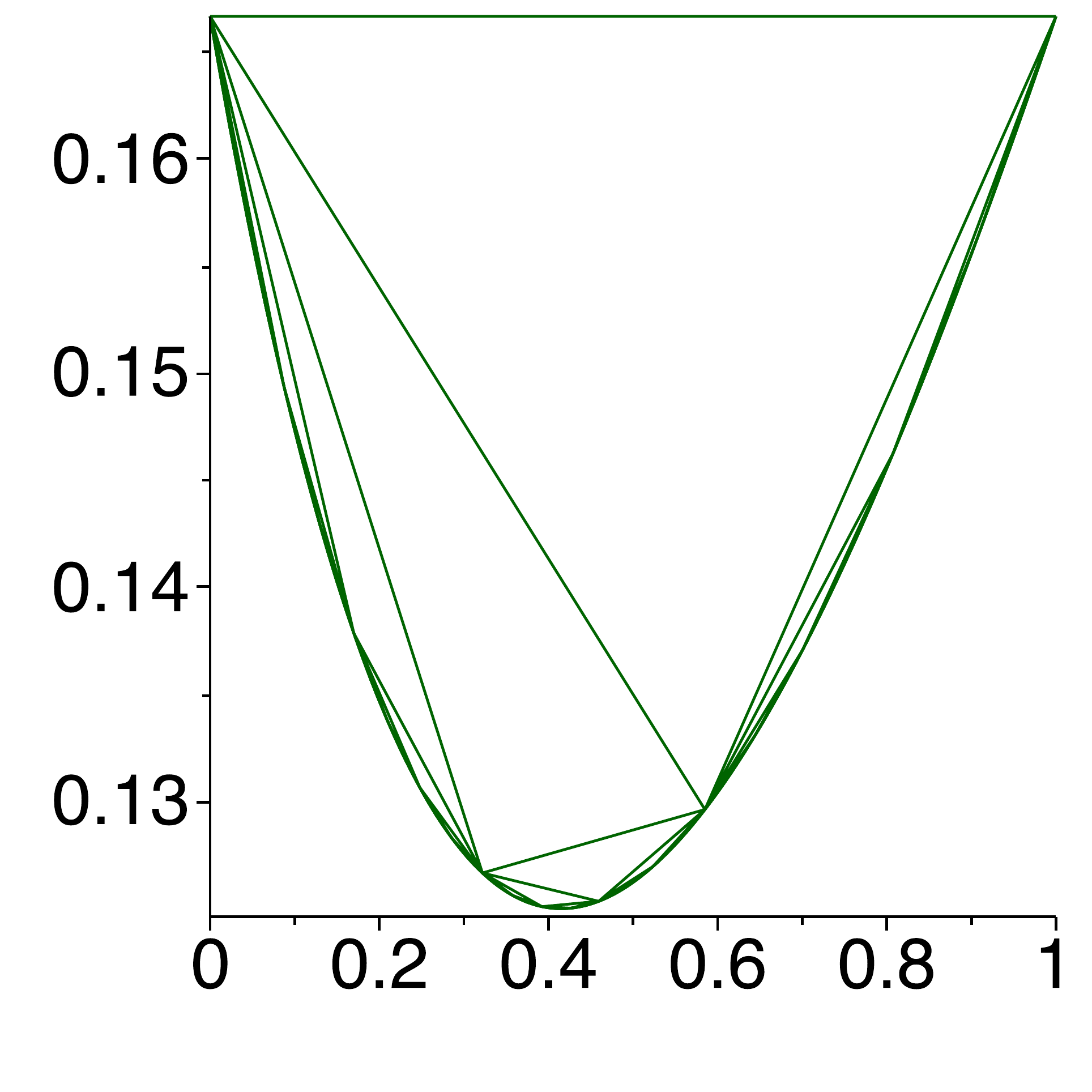} &
\includegraphics[height=2.6cm]{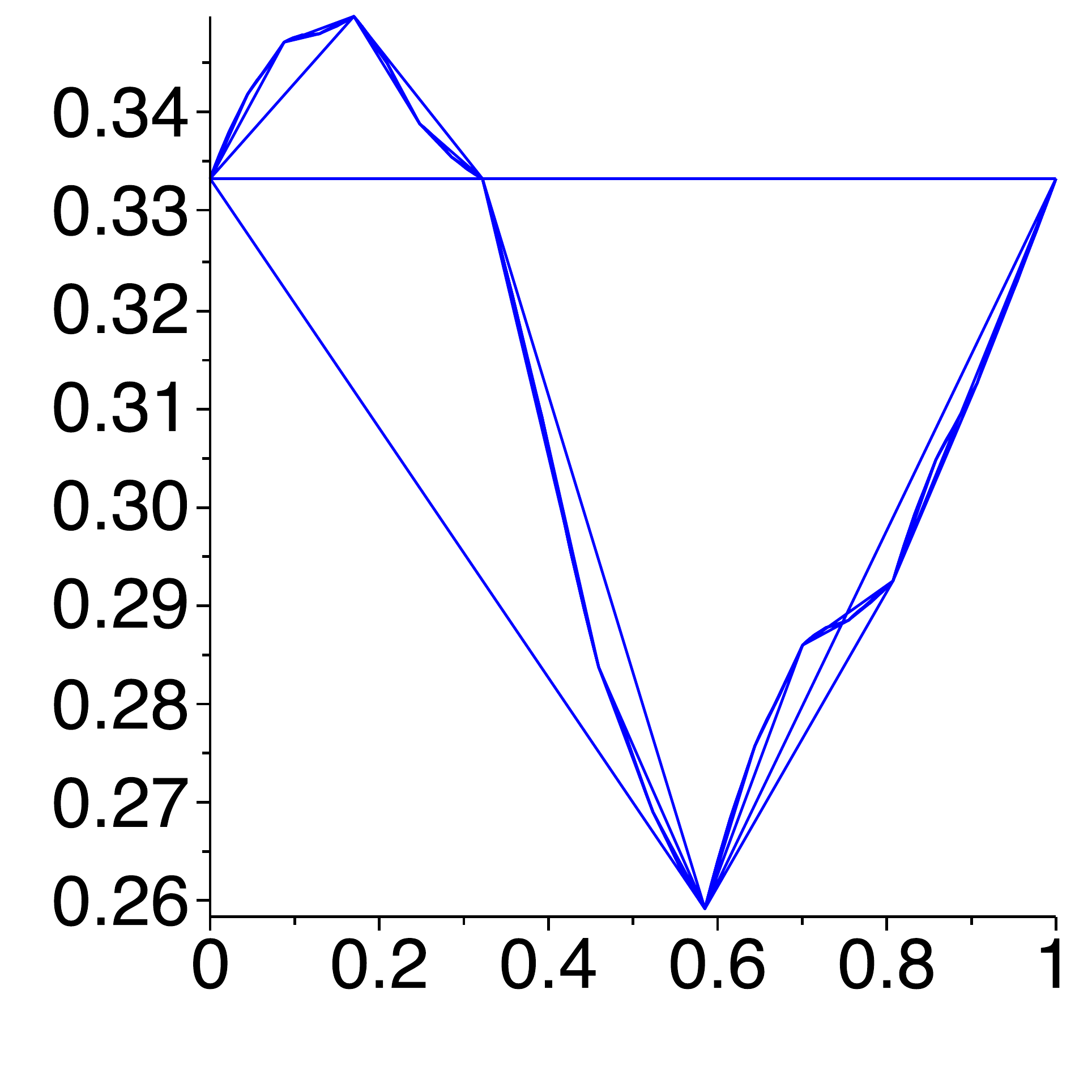} &
\includegraphics[height=2.6cm]{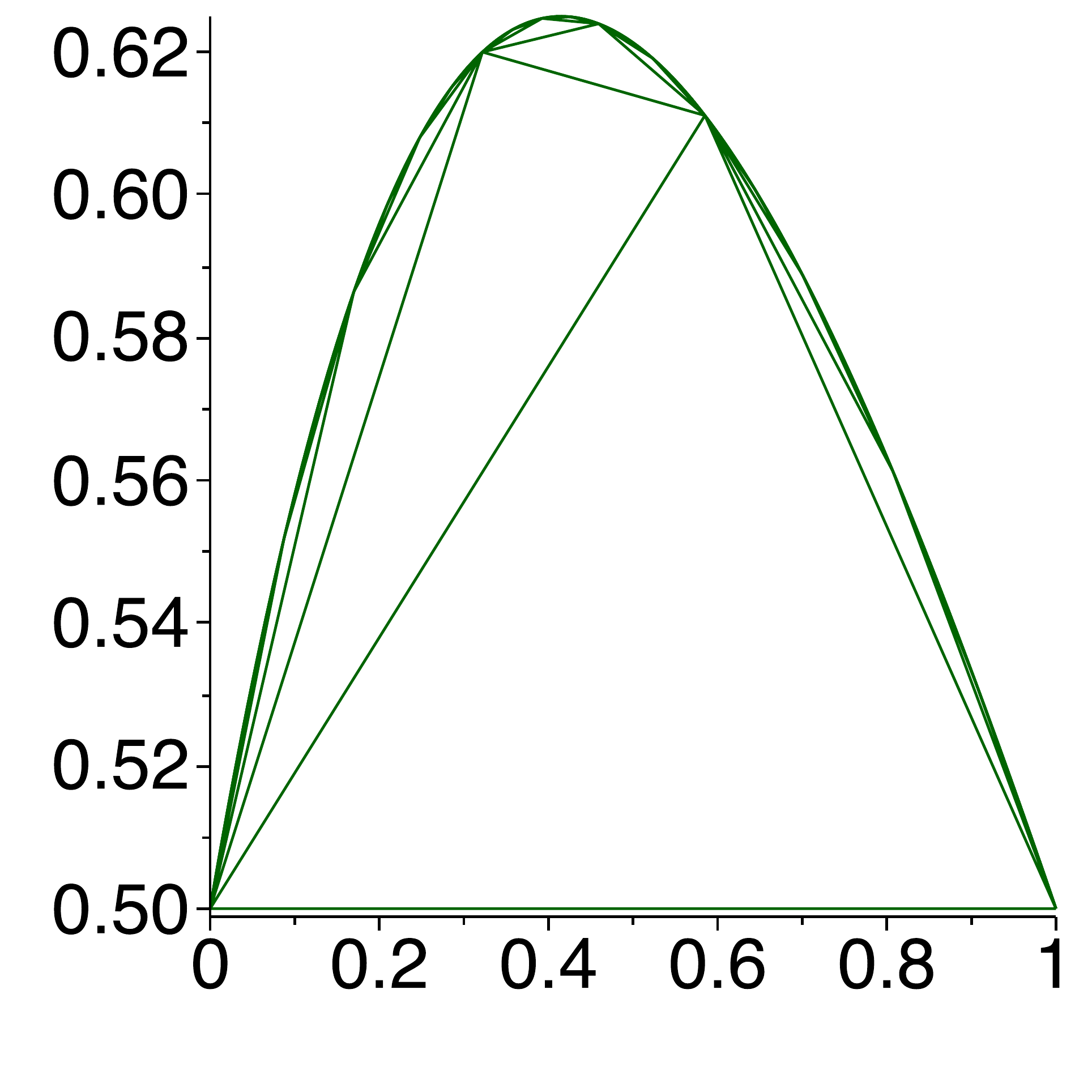} &
\includegraphics[height=2.6cm]{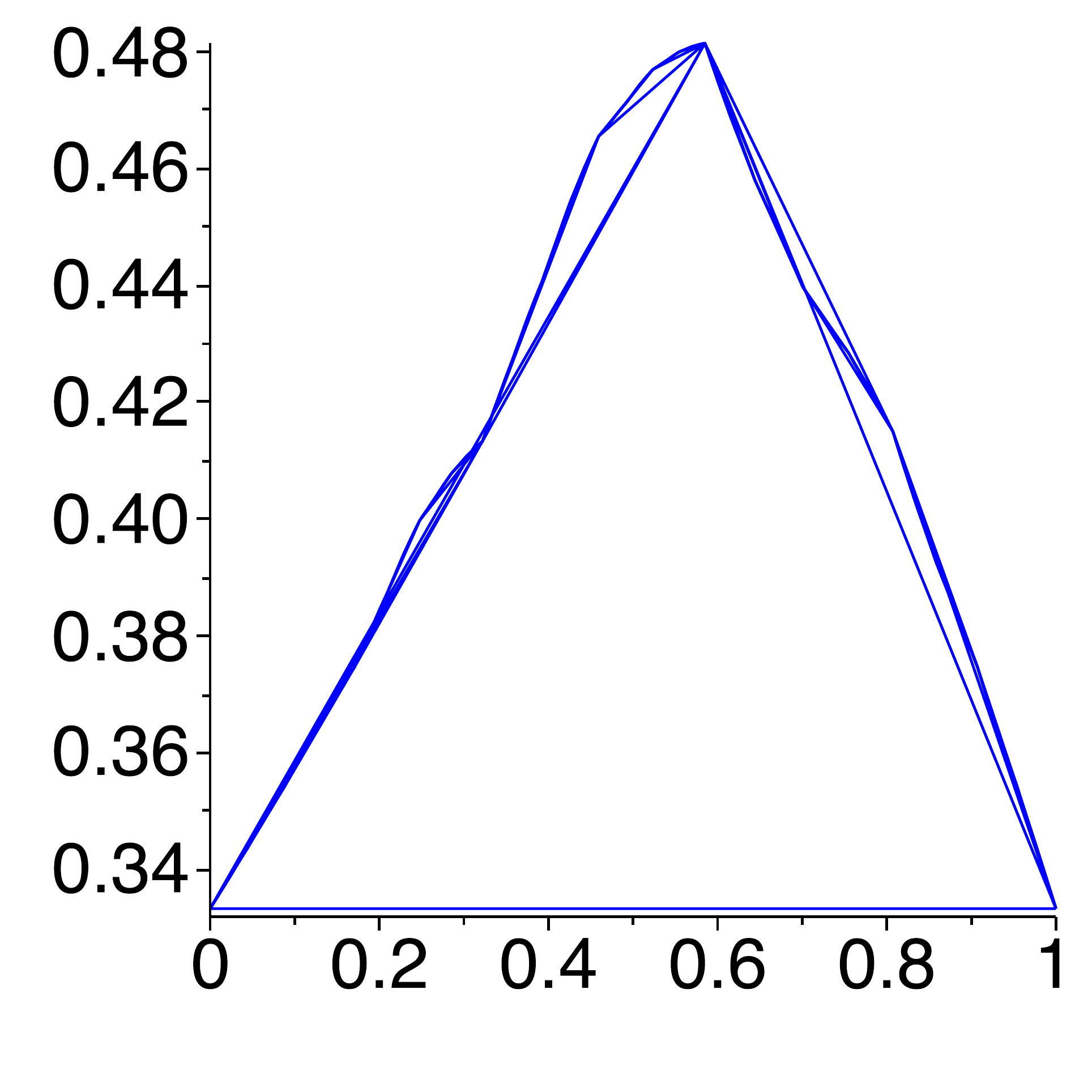} \\
$g(n)=1+\cos\lpa{\frac12n\pi}$ &
$g(n)=1+\sin\lpa{\frac12n\pi}$ &
$g(n)=1-\cos\lpa{\frac12n\pi}$ &
$g(n)=1-\sin\lpa{\frac12n\pi}$ \\
\end{tabular}}
\caption{{Sensitive dependence of the periodic functions on $g(n)$.}}
\label{fig-sensitivity}
\end{figure}

Consider finally the partial sum of A229762:
\begin{equation}
	f(n) 
	:= \sum_{0\le k<n}
	\lpa{\lpa{k \;\textsf{XOR}\; \tr{\tfrac k2}}\;
	\textsf{AND}\;\tr{\tfrac k2}}.
\end{equation}
Then $\Lambda_{2,2}[f]= \tr{\frac{n+1}4}$. Let $\bar{f}(n) := 
f(n)+\frac n4-\frac18\mathbf{1}_{n\text{ odd}}$. Then $\bar{f}(1) = 
\frac18$ and for $n\ge2$
\begin{equation}
    \Lambda_{2,2}[\bar{f}]
	=\frac{((n\bmod 4)-1)^2-1}8
	=\begin{cases}
    -\frac18, &\text{if }n \bmod 4 \equiv 1,\\
	\frac38, &\text{if }n \bmod 4 \equiv 3,\\
	0, &\text{otherwise}.
    \end{cases}
\end{equation}
We then deduce that $f(n) = n^2P(\log_2n) -\frac n4
+\frac18{\mathbf{1}_{n\text{ odd}}}$, where we have, using
\eqref{c4+} again, $P(t) = \frac12P_3(t)+\frac18$ 
with $P_3(t)$ given in \eqref{A048644}. 
%%\begin{equation}\label{asb}
%%	f(n) = n^2P(\log_2n)
%%	-\frac n4 +\frac{\mathbf{1}_{n\text{ odd}}}8,
%%\end{equation}
%%where we have,  using \eqref{c4+} again,
%%\begin{equation}\label{asa}
%%	P(t) = \frac{\pi}{16\log 2}-\frac18
%%	+\frac1{8\log 2}\sum_{k\ne0}
%%	\frac{\zeta(1+\chi_k,\frac14)-
%%	\zeta(1+\chi_k,\frac34)}{(1+\chi_k)(2+\chi_k)}
%%	\,e^{2k\pi i t}.
%%\end{equation}

Similarly, the partial sum of the sequence A229763 satisfies 
$\Lambda_{2,2}[f]= \frac12(1-(-1)^{\tr{\frac12n}})$ and can be dealt 
with by the same procedure.

\subsection{$(\alpha,\beta)=(4,4)$}
Similarly to the $(2,2)$ case, sequences in OEIS of this category
include digital sums connected to 2-D arrays and double sums
involving bitwise logic operators. Note that a degenerate case occurs
when $f(1)=0$ and $g(n) = 3n$ when $n$ is even and $g(n)=0$
otherwise; in this case, $f(n) = n^3-n$.

\begin{example}[Different sequences, same periodic 
	oscillations]\label{E44}
We consider the following eight OEIS sequences in which the first 
four have $g(n)$ involving $\sin\frac12 n\pi$ (of period $4$) while 
the last four have $g(n)$ depending simply on the parity of $n$. 

\begin{table}[!ht]
\begin{center}
\begin{tabular}{lllll}
\multicolumn{5}{c}{{}} \\
\multicolumn{1}{c}{OEIS id.} &
\multicolumn{1}{c}{Context} &
\multicolumn{1}{c}{$g(n)$} &
\multicolumn{1}{c}{$f(1)$} &
\multicolumn{1}{c}{$f(n)$}\\ 
\hline \Tstrut
A163242 & 2-D Gray code
& $\frac32\lpa{n-\sin\frac12n\pi}$
& $0$ & $n^3P(\log_2n)-\frac12n$\\
A163365 & Hilbert curve
& \makecell[l]{$\frac12\lpa{3+\cos n\pi}n$\\
  \;\;$-\sin\frac12n\pi$} 
& $0$ & $n^3P(\log_2n)-\frac23n$\\  
A163477 & Hilbert curve
& \makecell[l]{$\frac18\lpa{3+\cos n\pi}n$\\
  \;\;$-\frac14\sin\frac12n\pi$} 
& $0$ & $n^3P(\log_2n)-\frac16n$\\ 
A163478 & 2-D Gray code
& $\frac12\lpa{n-\sin\frac12n\pi}$ 
& $0$ & $n^3P(\log_2n)-\frac16n$
\Bstrut
\\
\hline\Tstrut          
A224923 & $\sum\limits_{1\le i,j<n}(i\textsf{ XOR }j)$
& $\tr{\frac12n^2}$ 
& $0$ & $n^{3}P(\log_2n)-\frac12n^2$ \\
A224924 & $\sum\limits_{1\le i,j<n}(i\textsf{ AND }j)$
& $\begin{cases}	
        \frac14\,{n}^{2}\\
        \frac14(n-1)(n-5) 
    \end{cases}$ 
& $0$ & $n^3P(\log_2n)-\frac14n^2$ \\
A241522 & game of Nim
& $\begin{cases}	
        0\\
        -3(2n-1)
    \end{cases}$ & $1$ & $n^3P(\log_2n)$ \\
A258438 & $\sum\limits_{1\le i,j<n}(i\textsf{ OR }j)$
& $\begin{cases}	
        \frac14n(7n-12)\\ 
        \frac14(n-1)(7n-11)
    \end{cases}$ 
& $0$ & $n^3P(\log_2n)-\frac74n^2+n$
\\ \hline
\end{tabular}
\caption{Sequences satisfying $\Lambda_{4,4}[f]=g$ 
(\refE{E44}).}
\label{tb-44}
\end{center}
\end{table}

\medskip

While it is visible that
$f_{\text{A163242}}(n)=3f_{\text{A163478}}(n)$ and
$f_{\text{A163365}}(n)=4f_{\text{A163477}}(n)$, it is less
transparent but can be proved by induction that
\begin{equation}
	\begin{split}
    f_{\text{A163477}}(n)
    &= \tfrac1{12}n(n^2-1)+\tfrac16f_{\text{A163242}}(n)\\
    f_{\text{A224924}}(n)
    &= \tfrac12n^2(n-1)-\tfrac12f_{\text{A224923}}(n)\\
    f_{\text{A241522}}(n)
    &= n^2(2n-1)-2f_{\text{A224923}}(n)\\
    f_{\text{A258348}}(n)
    &=\tfrac12n(n-1)(n-2)+\tfrac12f_{\text{A224923}}(n),
	\end{split}
\end{equation}
implying that these eight sequences lead indeed to only two 
periodically different ones: the first four and the last four. 

Also the Fourier expansions are given by 
\begin{align}\label{A163242}
    P_{\text{A163242}}(t)
    &= \frac{\zeta\lpa{2,\frac14}-
    \zeta\lpa{2,\frac34}}{32\log 2}
    +\frac{3}{16\log 2}\sum_{k\ne0}
    \frac{\zeta\lpa{2+\chi_k,\frac14}-
    \zeta\lpa{2+\chi_k,\frac34)}}
    {(2+\chi_k)(3+\chi_k)}
    \,e^{2k\pi i t}\\
    P_{\text{A224923}}(t)
    &= \frac{\pi^2}{24\log 2}
    +\frac{3}{2\log 2}\sum_{k\ne0}
    \frac{\zeta(2+\chi_k)}{(2+\chi_k)(3+\chi_k)}
    \,e^{2k\pi i t},
\end{align}
respectively. The two series are absolutely convergent. 
\end{example}
\begin{figure}[!ht]
\centerline{\footnotesize
\begin{tabular}{c c c}
\includegraphics[height = 3cm]{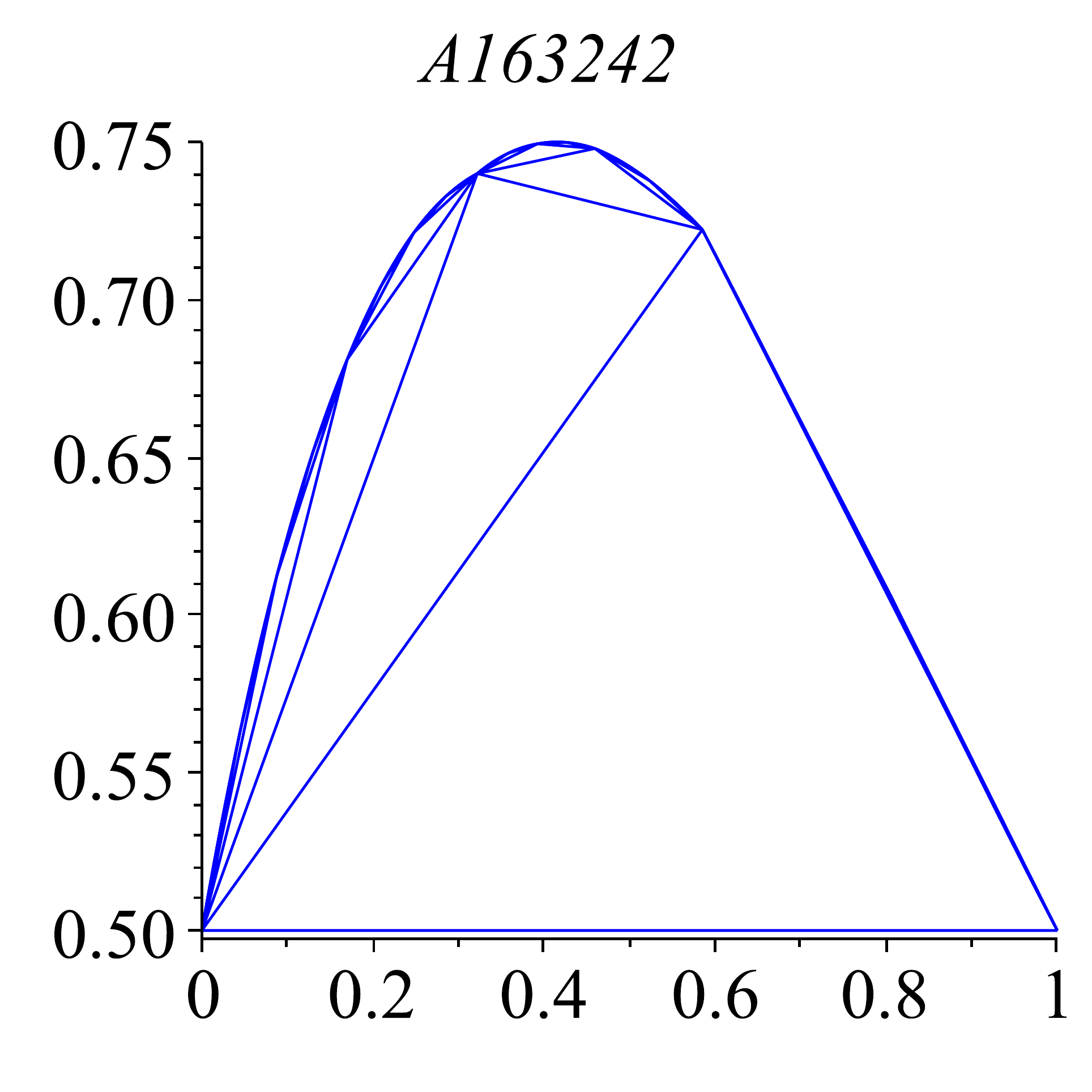} &
\includegraphics[height = 3cm]{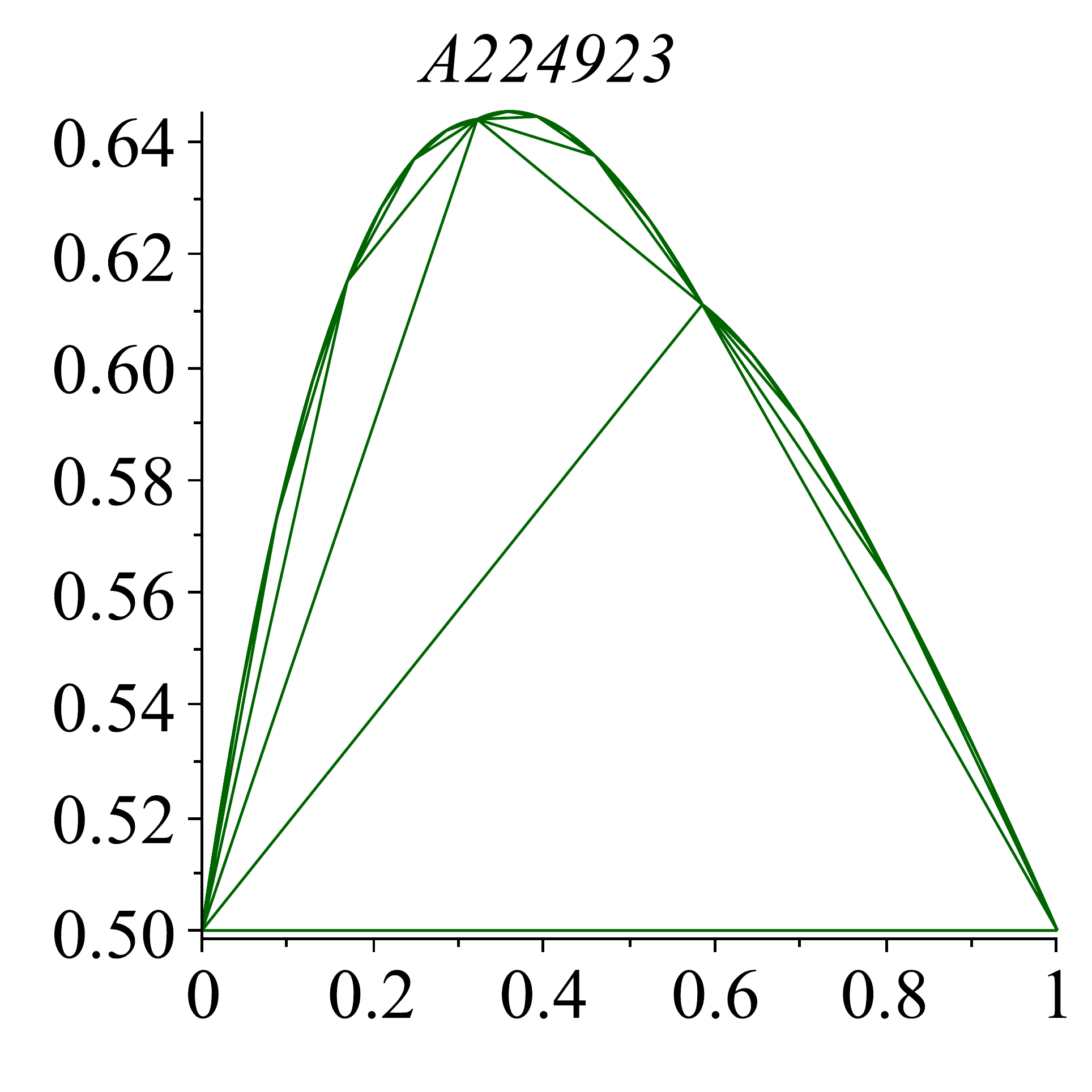} &
\includegraphics[height = 3cm]{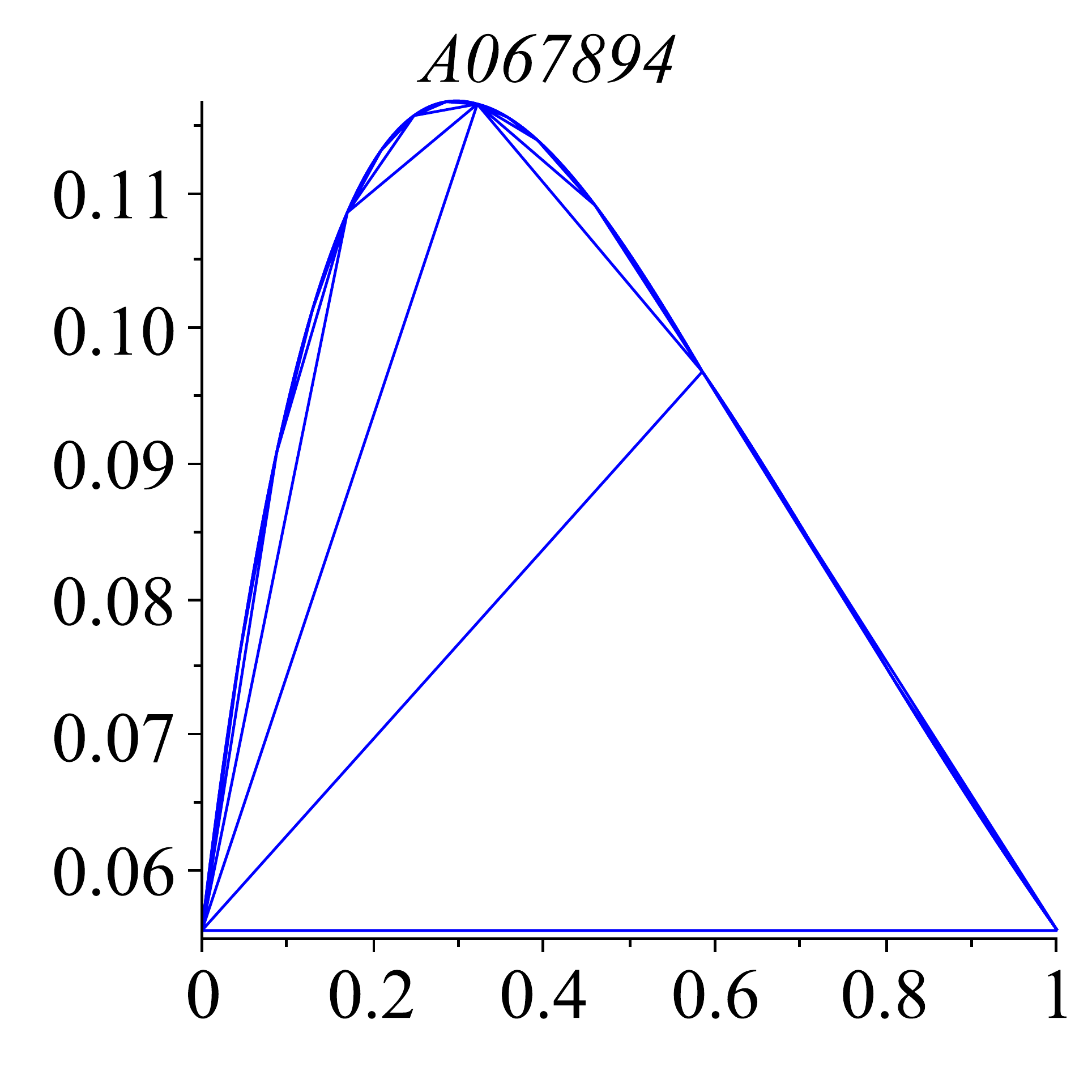} \\
A163242 & A224923 & A067894 \\
\end{tabular}}
\caption{{The periodic functions in the two representative $(4,4)$ 
cases (\refE{E44}) and A067894 (\refE{A067894}) as approximated by 
$n^{-\rho}(f(n)+Q(n))$ for $n=2^k+j$ with $k=0,1,\dots$ and $0\le 
j\le 2^k$.}} \label{fig-dab-4-4}
\end{figure}

\subsection{$(\alpha,\beta)=(10,10)$}

\begin{example}[A067894] \label{A067894}
Write $0, \dots, n-1$ in binary and add as if they were decimal
numbers. Then $f(1) = 0$ and
\begin{equation}\label{no0}
    \Lambda_{10,10}[f] = \tr{\tfrac n2}.
\end{equation}
The solution is
\begin{equation}
    f(n) = n^{\log_2(20)}P(\log_2n)-\tfrac1{18}n,
\end{equation}
In particular, we can derive the Fourier series expansion (as in
our previous paper) for $P_{\text{A067894}}(t)$:
\begin{equation}
    P_{\text{A067894}}(t) 
    = \frac4{5\log 2}\sum_{k\in\mathbb{Z}}
    \frac{\zeta(\rho-1+\chi_k)}{(\rho-1+\chi_k)(\rho+\chi_k)}
    \, e^{2k\pi i t},
\end{equation}
where $\rho=\log_220$. 

An extension by replacing $10$ by other values is discussed in 
\refE{EMoser}. 
\end{example}

\subsection{Partial sums of $\Lambda_{\alpha,0}[f]=g$}

One simple way to generate sequences satisfying
$\Lambda_{\alpha,\alpha}[f]=g$ with $\alpha=\beta$ is to consider the
partial sum $f(n) := \sum_{1\le k<n}h(k)$, where
$\Lambda_{\alpha,0}[h]=\xi$; then
\begin{equation}\label{erik}
    \Lambda_{\alpha,\alpha}[f] 
    = h(1)+\sum_{2\le k<n}\xi(k),
\end{equation}
for $n\ge2$. Such a sum (after normalising by $n$) gives the average
order and more smooth asymptotics than the original sequence (which
leads almost always to functions with discontinuities; see
\refS{Salphabeta=0} and \cite{Hwang2017} for more details). For
example, the partial sum of the sequence A006520 (satisfying
$\Lambda_{2,0}[h]=\lcl{\frac12n}$) gives A022560 discussed in
\refE{E22-ns} with $g(n)=\ltr{\frac14n^2}$.

In a similar manner, if we define $f(n) := \sum_{1\le k\le n}
h(k)$, where $\Lambda_{0,\beta}[h]=\xi$; then
\begin{equation}\label{er2}
    \Lambda_{\beta,\beta}[f] 
    = (1-\beta)h(1)+\sum_{2\le k\le n}\xi(k), 
\end{equation}
for $n\ge2$. Note that the equations $\Lambda_{\alpha,0}[h]=\xi$ and
$\Lambda_{0,\alpha}[h]=\xi$ are equivalent up to a shift of both $h$
and $\xi$; thus it suffices to consider only one of them.

The number of such sequences on OEIS exceeds several hundred after
removing sequences whose generating functions are rational (with all
singularities on the unit circle). So far we discussed only OEIS
sequences leading to $(1,1)$, $(2,2)$, $(4,4)$ and $(10,10)$, but in
such partial sum constructions, sequences with different values of
$\alpha$ are frequently found; see \refTab{tb-EMoser}.

\begin{example}[Partial sum of Moser-de Bruijn sequences:  $\alpha\ge0$] \label{EMoser}

The Moser-de Bruijn sequence A000695 consists of the integers whose
digits in base 4 are in $\para{0,1}$; equivalently, $\A000695(n)$ is
obtained by reading the binary representation of $n$ in base 4. This
sequence $h(n)$ satisfies $\Lambda_{4,0}[h]=\mathbf{1}_{n\text{
odd}}$. Hence, \eqref{erik} shows that the partial sum
$f(n):=\sum_{k<n} h(n)$ satisfies
\begin{equation}\label{er3}
    \Lambda_{4,4}[f](n) = \tr{\tfrac n2},
\qquad n\ge2.
\end{equation}

We can here replace 4 by any base $\ga>1$ (possibly non-integer);
more generally,  we can take any real $\ga$ and define, for any 
integer $n=\sum_{j\ge0}b_j 2^j$ with $b_j\in\para{0,1}$,
\begin{align}\label{moser}
	h\Bigpar{\sum_{j\ge0}b_j 2^j}:=\sum_{j\ge0}b_j\ga^j.
\end{align}
(One might also take complex $\ga$, but we leave that case to the
adventurous reader.) We then have the generating function
\begin{equation}
    \sum_{n\ge1}h(n)z^{n}
    = \frac1{1-z}\sum_{k\ge0}\frac{\ga^k z^{2^k}}
    {1+z^{2^k}}.
\end{equation}
\refTab{tb-EMoser} lists many of such ``$\ga$-Moser-de Bruijn
sequences'' that we found on OEIS. Note that $\ga=1$ gives
$\A000120(n)= \nu(n)$, the number of $1$s in the binary expansion of
$n$ (see \refE{Enu} and \eqref{nu-rec}), and that $\ga=2$ gives the
trivial case $\A000027(n)=n$.

Taking partial sums as above gives us a sequence $f(n):=\sum_{k<n} 
h(n)$ with
\begin{equation}\label{er4}
    \Lambda_{\ga,\ga}[f](n) 
	= \tr{\tfrac n2}, \qquad n\ge2.
\end{equation}
The resulting sequence $f$ is found in OEIS in the cases $\alpha=1$,
$2$, $10$, and $-1$, which give A000788, A000217 ($\binom{n}2$),
A067894 (Example~\ref{A067894}), and A005536 \cite[Example
7.1]{Hwang2017}, respectively.

\begin{table}[!ht]
\begin{center}
\begin{tabular}{ccccccc}\\   
$\alpha$ & $1$ & $2$ & $3$ & $4$ & $5$ & $6$ \\ \hline
OEIS & A000120 & A000027 & A005836 & A000695 & A033042 & A033043\\
$\alpha$ & $7$ & $8$ & $9$ & $10$ & $11$ & $12$ \\ \hline
OEIS & A033044 & A033045 & A033046 & A007088 & A033047& A033048	\\
$\alpha$  & $13$ & $14$ & $15$ & $16$ & $17$ & $18$\\ \hline
OEIS  & A033049 & A033050 & A033051 & A033052 & A197351 & A197352\\
$\alpha$  & $19$ & $20$ & $64$ & $100$ & $-1$ & $-2$\\ \hline
OEIS  & A197353 & A063012 & A135124 & A063010 & A065359 & A053985\\
\end{tabular}
\end{center}
\caption{$\ga$-Moser-de Bruijn sequences in OEIS (\refE{EMoser}).}
\label{tb-EMoser}
\end{table}

\paragraph{$\mathbf{\boldsymbol\alpha\boldsymbol>\frac12}$.}
For any $\ga>1$, we have, by \eqref{er4}, \eqref{t1z} and 
\eqref{t1Q}, 
\begin{align}\label{er5}
	f(n) 
	= n^{1+\log_2 \ga}P(\log_2n)-\frac{n}{2(\ga-1)}.
\end{align}
The same result holds for $\frac12<\ga<1$ by considering 
$f_1(n):=f(n)+\frac{1}{2(\ga-1)}n$, which satisfies
$\gL_{\ga,\ga}[f_1](n)=g_1(n):=-\{\frac n2\}$. 

\paragraph{$\mathbf{\boldsymbol\alpha\boldsymbol\le\frac12}$.}
If $0<\ga\le\frac12$, then $\rho\le0$, and even if we consider
$f_1(n)$ as above, the sum \eqref{t1Q} does not converge uniformly
since the individual terms do not converge uniformly to 0. Hence,
\refT{Theorem1} shows that \eqref{er5} cannot hold with a continuous
$P(t)$. In fact, if $s(x)$ is the sawtooth function that is defined
by $s(n):=\Iodd$ for $n\in\bbZ$, with linear interpolation between
the integers, then $g(x)=\frac12x-\frac12s(x)$ for all $x\ge0$, and
thus \eqref{b15} implies that, in view of $f(1)=0$,
\begin{align}\label{nn1}
	f(x)
	= \frac{x}{2(1-\ga)}
	-\frac12\sum_{k\ge0} (2\ga)^k s(2^{-k}x),
	\qquad x\ge0.
\end{align}
(This actually holds as soon as $|\ga|<1$.) If $0<\ga<\frac12$, then
the sum in \eqref{nn1} converges uniformly for all $x\ge0$ to a
bounded continuous function on $[0,\infty)$. Note that this function
is dominated by the first few terms in the sum, which are periodic
with small periods $(2,4,\dots)$. Hence $f(n)-\frac{n}{2(1-\ga)}$ can
be approximated arbitrarily well by a periodic function, without any
scaling; such behaviour is very different from the case $\ga>\frac12$
when we instead have a periodic function of $\log_2n$, scaled by
$n^\rho$; see also \refR{Rfast}.

If $\ga=\frac12$, the sum in \eqref{nn1} is $O(\log x)$ (for
$x\ge2$), but again, there is no smooth asymptotic behaviour. For
example, with $f_1(n):=f(n)-n$, we have
$f_1(\frac13(2^{2\ell}-2^{2k})) =-\frac23(\ell-k)+O(1)$ for $0\le
k<\ell$; taking $0\le k\le \frac12 \ell$, say, shows a rather large
variation on a relatively small interval of length $O(n\qq)$. Note
also $f_1(2^k)=-1$ for all $k\ge0$.

The case $\ga=0$ is trivial, with $f(n)=g(n)=\floor{\frac{n}2}$. The
case $\ga<0$ will be discussed in \refE{EMoser-}.

\paragraph{Fourier expansion.}
We have, by \eqref{D2} with $g(n) = \floor{\frac n2}$,
when $\Re s$ is large enough,
\begin{align}\label{per1}
	D(s)
	=\sum_{n\ge1}(\Iodd-\Ieven)n^{-s}
	=(1-2^{1-s})\zeta(s).
\end{align}
It follows by \refC{ChPa=b} that if $\ga>\frac12$ with
$\alpha\ne1,2$, then the periodic function $P(t)$ in \eqref{er5} has
the Fourier expansion 
\begin{equation}\label{per2}
    P(t) 
    = \frac{\alpha-2}{\alpha\log 2}
	\sum_{k\in\mathbb{Z}}
	\frac{\zeta(\rho-1+\chi_k)}
	{(\rho-1+\chi_k)(\rho+\chi_k)}\,e^{2k\pi i t}. 
\end{equation}
Alternatively, by \eqref{mel4}, $f(1)=0$, \eqref{majs} and 
\eqref{per1}, we derive the Mellin transform of $f(x)$:
\begin{align}\label{per3}
	f^*(s)
	=\frac{(1-2^{s+2})\zeta(-s-1)}{(1-\ga2^{s+1})s(s+1)},
	\qquad \Re s<\min(-1, -\rho).
\end{align}
This extends to all real $\ga$, with $\rho:=\log(2|\ga|)$ as in
\refE{EMoser-} when $\ga\le0$. Hence, Mellin inversion yields (after
a change of variable) the integral formula
\begin{equation}\label{per4}
    f(n) = \frac1{2\pi i}
    \int_{\gs-i\infty}^{\gs+i\infty}
    \frac{(1-2^{1-s})\zeta(s)n^{s+1}}
    {s(s+1)(1-\alpha 2^{-s})}\dd s,
\end{equation}
for any $\ga$ and any $\gs>\max(\rho-1,0)$. This provides by standard
methods an alternative proof of several of the results above. 

In particular, when $\ga=2$, we have $f(n)=\binom n2$ and 
$P(t)=\frac12$. When $\ga=1$, we can use \refT{TB1} and conclude that 
\begin{align}
	f(n)=\tfrac12n\log_2n+nP(\log_2n), 
\end{align}
where $P(t)$ is the Trollope--Delange fractal function (see 
\cite{Chen2014,Delange1975}):
\begin{equation}\label{per21}
    P(t) = \frac12\log_2\pi-\frac14-\frac1{2\log 2}
	-\frac1{\log 2}\sum_{l\ne0}\frac{\zeta(\chi_k)}
	{\chi_k(\chi_k+1)}\,e^{2k\pi i t}.
\end{equation}

\paragraph{Convergence of Fourier series.}
For a fixed real $\gs$, and all real $t$ with $|t|\ge2$, it is known
that
\begin{equation}\label{zeta}
|\zeta(\gs+it)|=
  \begin{cases}
\Theta(|t|^{\frac12-\gs}), & \gs<0,\\
O(|t|^{\frac12}\log|t|), & \gs=0,\\
O\bigpar{|t|^{\frac12-\frac\gs2}},&0<\gs<1,\\
O(\log|t|), & \gs=1,\\
\Theta(1),&\gs>1,    
  \end{cases}
\end{equation}
see, e.g., 
\cite[Chapter 5]{Titchmarsh1951}
or \cite[Section 1.5]{Ivic1985}. 
(Stronger results are known for $0\le\gs\le1$; the best
exponents are still not known for $0<\gs<1$.)
Hence, the Fourier series in \eqref{per2} and
\eqref{per21} are absolutely convergent if and only if
$\rho>\frac12$, i.e., if $\ga>\sqrthalf$. Thus, if
$\frac12<\ga<\sqrthalf$, then the Fourier series is \emph{not}
absolutely convergent. (As a check, we note that \eqref{zeta}
verifies that the Fourier coefficients are in $\ell^2$ for every
$\ga>\frac12$ ($\rho>0$), which is obvious from Parseval's formula.
On the other hand, for $0<\ga\le\frac12$, \eqref{per2} is not the
Fourier series of any $L^2$ function.)
\end{example}

\section{Extension to nonpositive $\ga$ or $\gb$}
\label{S:nonpositive}

We have so far considered the recurrence \eqref{a1}, or, equivalently
\eqref{b1}, with $\ga,\gb>0$. Here we discuss rather briefly $\ga$
and $\gb$ with other sign combinations. The case $\ga=\gb=0$ is
trivial, with \eqref{a1} reduced to $f(n)=g(n)$, and is ignored in
the sequel.

\subsection{Recurrences with $\alpha=0$ or $\beta=0$}
\label{Salphabeta=0}

The situation when $\alpha$ or $\beta$ equals zero is very similar to
the special cases $(\ga,\gb)=(2,0)$ and $(0,2)$ discussed in our
previous paper \cite{Hwang2017}, and we give only some brief
comments. Such sequences abound in OEIS; see \refTab{tb-EMoser} for
a few examples of one kind.

For any $\ga\neq0$, \eqref{b4} is solved by $\gf_{\ga,0}(t)=0$,
$t\in[0,1)$, and thus \eqref{b2} yields $f(x)=f(\floor x)$.
Similarly, for $\gb\neq0$, $\gf_{0,\gb}(t)=1$, $t\in(0,1]$, and thus
$f(x)=f(\ceil x)$. The main difference from the case $\ga,\gb>0$ is
that now $\gf$ is discontinuous (at an endpoint), and thus $f(x)$ is
discontinuous except in trivial cases, In the cases $(\ga,0)$ with
$\ga>0$ and $(0,\gb)$ with $\gb>0$, it is easily verified that
\refT{Theorem1} holds with modifications similar to the special cases
in \cite[Theorems 4 and 5]{Hwang2017}; note that the periodic
function $P(t)$ now is discontinuous except in trivial cases. If
$\ga<0$ or $\gb<0$, this holds with further modifications as in
\refS{SSga<0gb<0} below. We omit the details.

\subsection{Recurrences with both $\alpha$ and $\beta$ 
negative}\label{SSga<0gb<0}

In this section, we consider the case when $\ga$ and $\gb$ both are
negative; see \refTab{tb-E<0<0} for a few examples from OEIS
(discussed below). We thus assume the recurrence \eqref{a1}, or,
equivalently, \eqref{b1} with $\ga,\gb<0$. In this case we define
\begin{align}\label{rho-}
	\rho:=\log_2(|\ga|+|\gb|).  
\end{align}
Thus $\ga+\gb=-2^\rho$. As above, we define $g(1):=0$. We also extend
the definition of $\gf_{\ga,\gb}(t)$ to negative values of $\ga, \gb$
by $\gf_{\ga,\gb}(t):=\gf_{|\ga|,|\gb|}(t)$, and note that then
\eqref{b4} holds. The proof of \refL{Lemma1} applies, \emph{mutatis
mutandis}, and shows that \eqref{b3} holds in this case too. It
follows that the theory developed in Sections
\ref{S:recurrence}--\ref{S:smooth} extends to this case, but with an
important modification.

We say that a function $P(t)$ on $\bbR$ is \emph{$1$-antiperiodic} if
$P(t+1)=-P(t)$ for all $t\in\bbR$. In other words, $P(t)$ is
$1$-antiperiodic if and only if $e^{\pi it}P(t)$ is $1$-periodic.

Note that every $1$-antiperiodic function is $2$-periodic. Moreover, a
$1$-antiperiodic function that is integrable on $\oi$ (and thus on
any compact interval) has a Fourier series that can be written
\begin{align}\label{fou-}
	P(t)\sim  \sum_{k\in\bbZ} \hP(k+\tfrac12)e^{(2k+1)\pi\ii t}.
\end{align}

We collect the most important results in the following theorem,
leaving the others in Sections \ref{S:recurrence}--\ref{S:smooth} to
the reader.

\begin{thm}\label{T-}
Suppose that $\ga,\gb<0$, and let $\rho:=\log_2(|\ga|+\gb|)$. Then
\refT{Theorem1} holds with the modifications that $1$-periodic is
replaced by $1$-antiperiodic, and that \eqref{t1Q} and \eqref{t1y}
are replaced by
\begin{align}\label{t1Q-}
	Q(x):=%\sum_{m\ge 1}(-1)^m2^{-\rho m}g(2^{m}x)  
	\sum_{m\ge 1}(-2^{-\rho})^{m}g(2^{m}x)  
	=\sum_{m\ge 1}(\ga+\gb)^{-m}g(2^{m}x),
\end{align}
and
\begin{equation}\label{t1y-}
    P(t)
    =\sum_{m\in \mathbb{Z}}(-1)^m2^{-\rho (m+t)}g(2^{m+t})
    +f(1)P_0(t), \qquad t\in\mathbb{R},
\end{equation}
where now
\begin{equation}\label{b16-}
	P_0\left( t\right)
	:= (-1)^{\floor{t}}
	\bigpar{ 1+( \alpha +\beta-1) 
	\varphi\lpa{2^{\{t\}}-1}}
	2^{-\rho\{t\}}.
\end{equation}
Moreover, \refT{Theorem4} also holds, with the condition
$|\ga|+|\gb|>1$ in part \ref{T4c}, and with \eqref{c4} replaced by
\begin{equation}\label{c4-}
    \widehat{P}(k+\tfrac12)
    =\frac{1}{\log 2}\int_{1}^{\infty }
    \frac{g(u)}{u^{\rho + \chi'_{k}+1}}\dd u
    +\frac{f(1)}{\log 2}\int_{0}^{1}
    \frac{1+(\alpha +\beta -1)\varphi(u)}
    {(1+u)^{\rho +\chi'_{k}+1}}\dd u,
\end{equation}
for $k\in\bbZ$, where
\begin{align}\label{chi'k}
	\chi'_{k}
	:=\frac{(2k+1)\pi }{\log 2}\ii,
	\qquad  k\in\bbZ.
\end{align}
If $\ga=\gb$, then,  in analogy with \refC{ChPa=b}, \eqref{c4-} 
simplifies to
\begin{align}\label{c4+-}
	\hP(k+\tfrac12) 
	&= \frac{1}{(\rho+\chi'_k)(\rho-1+\chi'_k)\log 2}
	\Bigpar{D(\rho-1+\chi'_k) 
	+ \frac{(2\ga-1)(\ga-1)}{\ga}f(1)},
\end{align}
where $D(s)$ is defined in \eqref{D1}.
\end{thm}

\begin{proof}
The proof follows by the same arguments used in Sections
\ref{S:recurrence}--\ref{S:smooth}, with minor modifications. In
particular, \refL{Lemma3} and Propositions
\ref{Proposition1}--\ref{Proposition2} hold with $P_0$ and $Q$ as
above, $1$-periodic replaced by $1$-antiperiodic, and extra factors
$(-1)^k$ or $(-1)^m$ in the sums in \eqref{b15}, \eqref{b12},
\eqref{b13} and in the definition of $G_m(x)$. Note that \eqref{bb1}
now becomes
\begin{align}\label{bb1-}
    h\bigpar{2^my}
    =(-1)^m2^{m\rho}G_m(y)
	=(\ga+\gb)^mG_m(y).
\end{align}

For the proof of \eqref{c4+-}, note that \eqref{maj0} holds with
$\chi_k$ replaced by $\chi'_k$, since now $2^{\rho+\chi'_k} 
=-2^\rho=2\ga$.
\end{proof}

\begin{remark}\label{Rlog4}
Thus $P(\log_2n)$ in our formulas now is an $1$-antiperiodic function
of $\log_2n$, which implies that it is a $2$-periodic function of
$\log_2n$, or, equivalently, a $1$-periodic function of
$\frac12\log_2n=\log_4n$.
\end{remark}

\begin{remark}
Most of the formulas in Sections \ref{S:recurrence}--\ref{S:smooth}
hold verbatim for $\ga,\gb<0$ if we instead replace $\rho$ by the
complex logarithm $\rho_{\mathsf c}:=\log_2(\ga+\gb)=\rho+\pi\ii/\log
2$. However, this seems less convenient for applications.
\end{remark}

\begin{example}\label{Eg=0-}
Consider the basic case $g(n)=0$, $f(1)=1$ as in \refE{Eg=0}; we
again denote the solution $f(n)$ by $S_{\ga,\gb}(n)$. \refT{T-} shows
that
\begin{align}\label{Sgagb-}
	S_{\ga,\gb}(n) 
	= n^{\rho} P_0(\log_2 n),\qquad n\ge1,
\end{align}
where now $P_0(t)$ is given by \eqref{b16-}. It follows that
\eqref{Sgagb2} still holds, which also follows because for fixed $n$,
$S_{\ga,\gb}(n)$ is a polynomial in $\ga,\gb\in\bbR$, as is
$\varphi\lpa{2^{\para{\log_2n}-1}}(\alpha +\beta)^{\floor{\log_2 n}}$
by \eqref{b8}.

In the special case $\ga=\gb<0$, \eqref{c4+-} shows that, in analogy
to \eqref{maj1},
\begin{align}\label{maj1-}
	\hP_0(k+\tfrac12)&
	=\frac{(2\ga-1)(\ga-1)}{\ga\log2}
	\cdot\frac{1}{(\rho+\chi'_k)(\rho-1+\chi'_k)},
	\qquad k\in\bbZ,
\end{align}
and, thus (see \eqref{fou-}), $P_0(t)$ has the absolutely convergent
Fourier series
\begin{align}\label{per-}
	P_0(t)=
	\frac{(2\ga-1)(\ga-1)}{\ga\log2}
	\sum_{k\in\bbZ} \frac{e^{(2k+1)\pi\ii t}}
	{(\rho+\chi'_k)(\rho-1+\chi'_k)}.
\end{align}
\end{example}

\begin{center}
\begin{longtable}{cccc}\hline
\textbf{OEIS} & \textbf{$\left(\alpha,\beta\right)$} & 
\textbf{$g(n)$} & \textbf{Initials} \\ \hline
\endfirsthead
\multicolumn{4}{c}%
{\tablename\ \thetable\ -- \textit{Continued from previous page}} \\
\hline
\textbf{OEIS} & \textbf{$\left(\alpha,\beta\right)$} & 
\textbf{$g(n)$} & \textbf{Initials} \\ \hline
\endhead
\hline \multicolumn{4}{r}{\textit{Continued on next page}} \\
\endfoot %\hline
\endlastfoot
A005536 & $(-1,-1)$ & $\tr{\frac12n}$ & $f(1)=0$ \\ %\hline
A079947 & $(-1,-1)$ & $ n-1 $ & $ f(1) = 0 $ \\ %\hline 
A079954 & $(-1,-1)$ & $ n-2 $ & $ f(1) = 0 $ \\ %\hline 
A094120 & $(-2,-2)$ & $\tr{\frac14n^2}$ & $ f(1) = 0 $ \\ \hline 
\caption{Sequences in OEIS satisfying $\Lambda_{\ga,\gb}[f]=g$ 
with $\ga,\gb<0$.}
%(\refE{E<0}).}
\label{tb-E<0<0}
\end{longtable}
\end{center}

\begin{example}[Partial sum of Moser-de Bruijn sequences: $\alpha<0$]
	\label{EMoser-}
Let $\ga<0$. Consider the digital sum $f(n)$ defined in
\refE{EMoser} as the partial sum of the $\ga$-Moser-de Bruijn
sequence. Then $f$ satisfies \eqref{er4} with $f(1)=0$. Recall that
$\ga=-1$ gives $\A005536$ treated in \cite[Example 7.1]{Hwang2017},
and that $\ga=-2$ gives the partial sums of A053985.

For $\ga<-\frac12$, 
we obtain by \refT{T-}, 
directly if $\ga<-1$ and otherwise
by considering 
$f_1(n):=f(n)+\frac{1}{2(\ga-1)}n$,
\begin{align}
  \label{er5-}
f(n) = n^{1+\log_2 |\ga|}P(\log_2n)-\frac{n}{2(\ga-1)},
\end{align}
for a continuous $1$-antiperiodic function $P(t)$ with the Fourier
expansion (using \eqref{per1}; cf.\ \eqref{per2})
\begin{equation}\label{per2-}
	P(t) 
	= \frac{\alpha-2}{\alpha\log 2}
	\sum_{k\in\mathbb{Z}}
	\frac{\zeta(\rho-1+\chi_k')}
	{(\rho-1+\chi_k')(\rho+\chi_k')}\,e^{(2k+1)\pi i t}.
\end{equation}
As in \refE{EMoser}, the Fourier series is absolutely convergent if
and only if $\rho>\frac12$, i.e., if and only if $|\ga|>\sqrthalf$.

The case $-\frac12\le\ga<0$ is similar to the case $0<\ga\le\frac12$
discussed in \refE{EMoser}. (For $\ga=-\frac12$, consider the example
$n=\frac{1}{15}(16^\ell-16^k)$.) Also note that \eqref{per4} holds 
for all $\ga<0$.
\end{example}

\begin{example}[$\ga=\gb=-1$ and thus $\rho=1$]\label{E-1-1}
The function $h(n):=(-1)^{L_n}$ obviously satisfies
$\gL_{-1,0}[h](n)=0$, with $h(1)=1$; thus $h(n)=S_{-1,0}(n)$. By
\eqref{erik}, the partial sums $f(n):=\sum_{1\le k<n}(-1)^{L_k}$
satisfy
\begin{align}\label{e-1-1a}
	\gL_{-1,-1}[f](n)=1, \qquad n\ge2, 
\end{align}
with $f(1)=0$. It follows that 
\begin{align}\label{e-1-1b}
	f(n)=\tfrac{1}{3}-\tfrac{1}{3}S_{-1,-1}(n),
\end{align}
which is \refE{EIIa} and \eqref{eiia} with $\ga=-1$; 
cf.\ also \refE{Eg=0-}. By induction, we have the explicit formula 
$f(n) = (-1)^{L_n} n -\frac43(-2)^{L_n}+\frac13$, which implies that 
$f(n)=nP(\log_2n)+\frac13$ with the $1$-antiperiodic function
\begin{align}\label{e-1-1p}
	P(t)=(-1)^{\floor{t}}\Bigpar{1-\frac{2^{2-\frax{t}}}3}
	= \frac{2}{\log2}
	\sum_{k\in\bbZ} \frac{e^{(2k+1)\pi\ii t}}
	{\chi'_k(1+\chi'_k)},
\end{align}
where the Fourier coefficients follow from \eqref{e-1-1b} and 
\eqref{per-}. This sequence is not in OEIS, but the following two 
periodically equivalent variants are.

The sequence A079947 is $f(n) := \frac12\sum_{1\le k<n}(1+(-1)^{L_k})$
(partial sums of A030300), which, by \eqref{erik} or \eqref{e-1-1a}, 
satisfies $\Lambda_{-1,-1}[f]= n-1$ with $f(1)=0$. We then have, by
\eqref{e-1-1b},
\begin{align}\label{a079947}
	f_{\A079947}(n)
	=\frac{n-1}2+\frac{1}{6}-\frac{1}{6}S_{-1,-1}(n)
	=\frac{n}2-\frac{1}{6}S_{-1,-1}(n)-\frac{1}{3}.
\end{align}
Explicitly, $f(n) = \frac12\lpa{1+(-1)^{L_n}}n
-\frac23(-2)^{L_n}-\frac13$, so that $f(n) = nP(\log_2n)-\frac13$,
where $P(t)$ is $2$-periodic with
\begin{align}
	P(t) 
	=\frac12+(-1)^{\floor{t}}
	\Bigpar{\frac12-\frac{2^{1-\frax{t}}}3}
	= \frac12+\frac1{\log 2}\sum_{k\in\mathbb{Z}}
	\frac{e^{(2k+1)\pi i t}}{\chi'_k(1+\chi'_k)}.
\end{align} 

Similarly, A079954 is $f(n) := \frac12\sum_{1\le k<n}(1-(-1)^{L_k})$
(partial sums of A030301), which satisfies $\Lambda_{-1,-1}[f]= n-2$
with $f(1)=0$. We have $f_{\text{A079954}}(n)
=n-1-f_{\text{A079947}}(n)$. Thus, by\eqref{e-1-1b},
\begin{align}\label{a079954}
	f_{\A079954}(n)
	=\frac{n-1}2-\frac{1}{6}+\frac{1}{6}S_{-1,-1}(n)
	=\frac{n}2+\frac{1}{6}S_{-1,-1}(n)-\frac{2}{3}.
\end{align}
Explicitly, $f(n) = \frac12\lpa{1-(-1)^{L_n}}n
+\frac23(-2)^{L_n}-\frac23$, so that $f(n) = nP(\log_2n)-\frac23$,
where $P(t)$ is $2$-periodic with
\begin{align}
	P(t) 
	=\frac12-(-1)^{\floor{t}}\Bigpar{\frac12-\frac{2^{1-\frax{t}}}3}
	= \frac12-\frac1{\log 2}\sum_{k\in\mathbb{Z}}
	\frac{e^{(2k+1)\pi i t}}{\chi'_k(1+\chi'_k)}.
\end{align} 
\end{example}

\begin{example}[$\ga=\gb=-2$ and thus $\rho=2$]\label{E094120}
Consider the sequence {A094120} given by $f(n):=\sum_{1\le
k<n}\sum_{1\le j\le k}(-2)^{v_2(j)} =\sum_{1\le j<
n}(-2)^{v_2(j)}(n-j)$, where $v_2(j)$ is the dyadic valuation of $j$; 
see \refE{E22-ns} for definition. Note that this is the analogue of
A022560 in \refE{E22-ns} with $2^{v_2(j)}$ replaced by
$(-2)^{v_2(j)}$.

By definition, $f(n)$ is the partial sum of $h(k):=\sum_{1\le j\le
k}(-2)^{v_2(j)}$; it is easily seen that $\gL_{-2,0}[h](n)
=\cl{\frac{n}2}$, and thus \eqref{erik} yields
$g(n):=\Lambda_{-2,-2}[f](n) =\tr{\tfrac14n^2}$. \refT{T-} does not
directly apply because $g(n)$ grows too rapidly, but we can use a
standard trick and subtract a multiple of $n^2$. We have
$\gL_{-2,-2}[n^2]=2n^2+\Iodd$, and thus $f_1(n):=f(n)-\frac18n^2$
yields $\gL_{-2,-2}[f_1](n)=-\frac38\Iodd$ with $f_1(1)=-\frac18$.
\refT{T-} applies to $f_1(n)$ and implies that $f_1(n) =
n^2P_1(\log_2n)$ with $P_1(t)$ $1$-antiperiodic, and consequently
$f(n) = n^2P(\log_2n)$, where $P(t)=P_1(t)+\frac18$ is $2$-periodic
with Fourier expansion given by, from \eqref{c4+-} applied to $f_1$,
\begin{align}
    P(t) = \frac18+\frac3{2\log 2}\sum_{k\in\mathbb{Z}}
    \frac{\zeta(1+\chi'_k)}{(1+\chi'_k)(2+\chi'_k)}\,
    e^{(2k+1)\pi i t}.
\end{align}
See Figure \ref{fig-A094120}.
Note the difference from A022560 in \refE{E22-ns} where a logarithmic 
term appears in \eqref{f022560}.
\begin{figure}[!ht]
\begin{center}
\includegraphics[height=3cm]{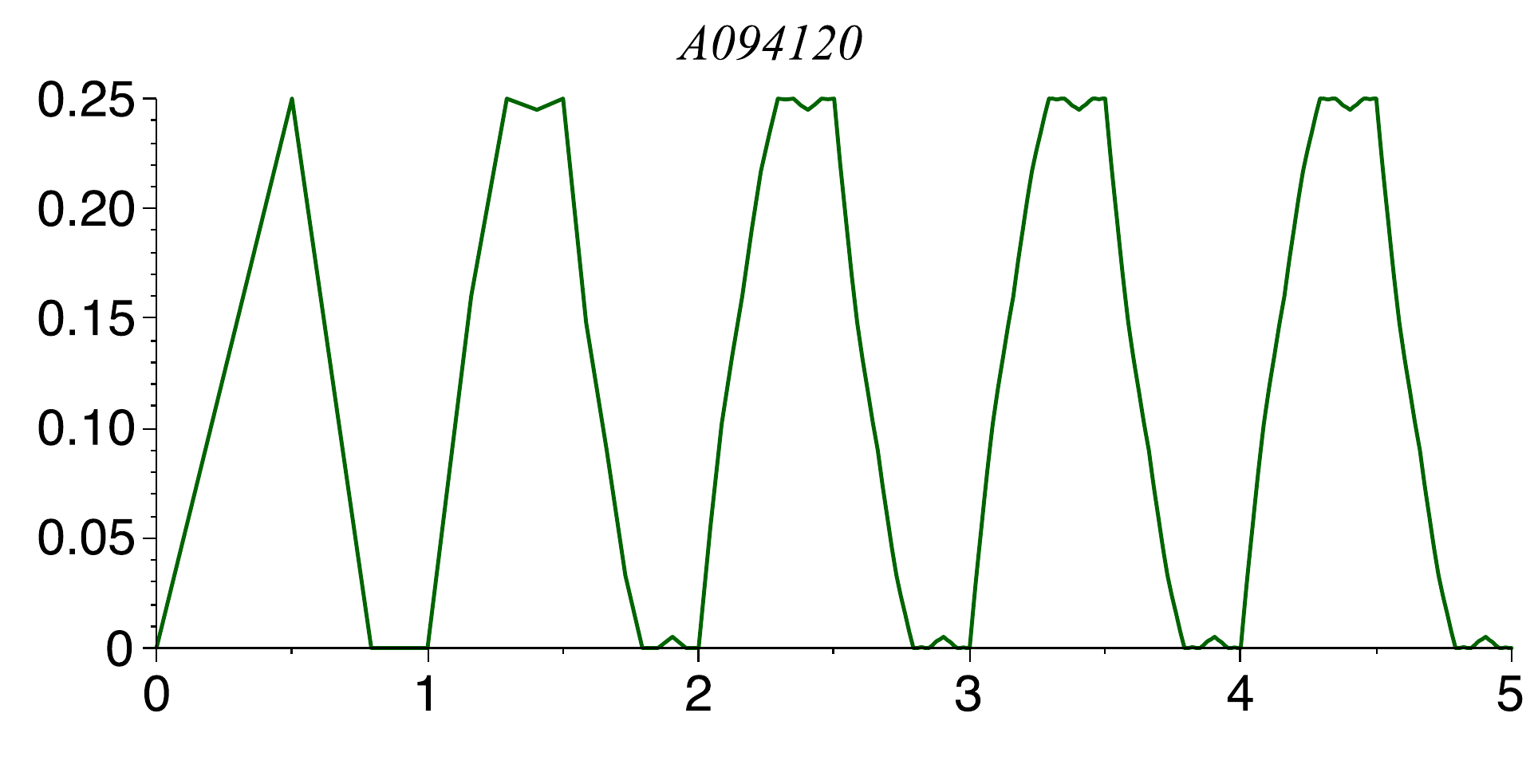}	\;
\includegraphics[height=3cm]{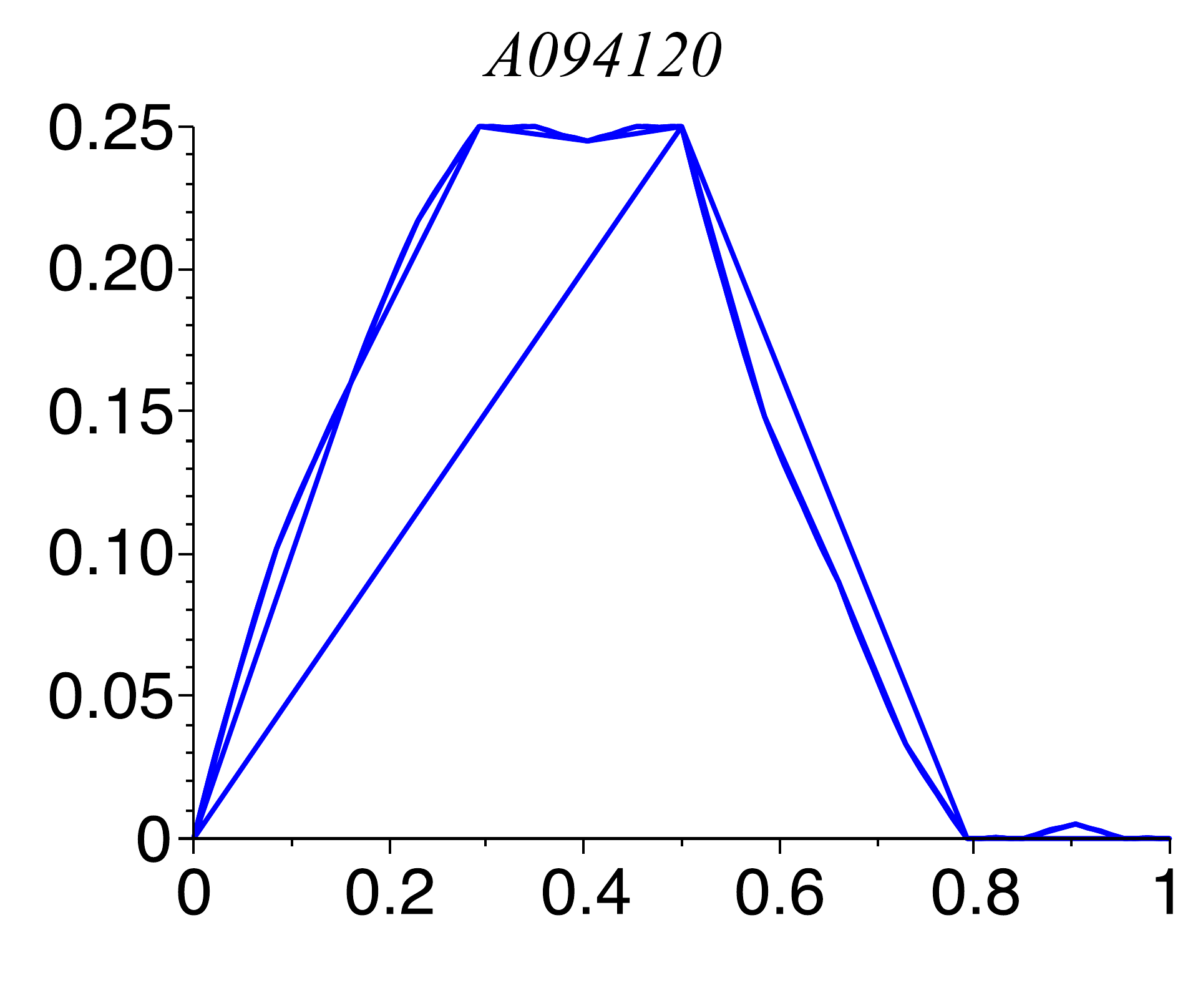}
\end{center}
\caption{Periodic fluctuations of A094120$(n)/n^2$ (\refE{E094120}).}
\label{fig-A094120}
\end{figure}
\end{example}

\subsection{Recurrences with either $\alpha$ or $\beta$ 
negative}\label{SSgagb<0} 

There are many examples in OEIS of sequences satisfying
$\gL_{\ga,\gb}[f]=g$ with $\ga\gb<0$, i.e., one of $\ga$ and $\gb$ is
positive and the other is negative. In this case, we can still define
$\gf(t)$ for dyadic rationals $t\in\oi$ by \eqref{b8} (and it has to
have this value), but since now $|\ga+\gb|<|\ga|,|\gb|$, this
function is unbounded in every interval, and thus cannot be extended
to a continuous function (or any other reasonable function) on $\oi$.
Hence our methods break down, and we have no general theorem in this
case. Moreover, we do not expect any simple asymptotics in general,
which is illustrated by the following example.

\begin{example}
Consider again $f(n)=S_{\ga,\gb}(n)$ defined by
$\gL_{\ga,\gb}[f](n)=0$ and $f(1)=1$. In particular, we have
\begin{align}\label{pm0}
    f(2n)=(\ga+\gb)f(n),
\qquad n\ge1.
\end{align}
Moreover, it is easily seen by induction, or by \eqref{fab} and
\eqref{Sum-ab}, that, for any $\ga$ and $\gb$,
\begin{align}\label{pm1}
	f(n+1)-f(n)=(\ga+\gb-1)\ga^{\nu(n)-1}\gb^{\nu_0(n)}. 
\end{align} 
By \eqref{pm0} and \eqref{pm1}, we have, for example,
\begin{align}
	f(2^k)&=(\ga+\gb)^k,  
	&& k\ge0, \label{pm3}\\
	f(2^k-1)&=(\ga+\gb)^k-(\ga+\gb-1)\ga^{k-1}, 
	&& k\ge1, \label{pm3-}\\
	f(2^k+1)&=(\ga+\gb)^k+(\ga+\gb-1)\gb^{k},  
	&& k\ge0 \label{pm3+}.
\end{align}

Consider, for definiteness, the case $\ga>0>\gb$ with
$|\ga|\ge|\gb|$. Then $0\le\ga+\gb<\ga$. Assume also $\ga+\gb\neq1$.
(Otherwise, $S_{\ga,\gb}(n)=1$ for all $n$.) We see from \eqref{pm3-}
that $|f(n)|$ may be of the order $\ga^{\log_2n}=n^{\log_2\ga}$,
although \eqref{pm3} shows that $|f(n)|$ also may be much smaller. In
fact, it is easily shown by induction using \eqref{pm0}--\eqref{pm1}
that $|f(n)|\le C\ga^{L_n}$ for some constant $C$ (depending on $\ga$
and $\gb$) and all $n$. Hence, $\xf(n):=n^{-\log_2\ga}f(n)$ is
bounded. However, $\xf(n)$ does not seem to have any simple 
asymptotic approximations, as the following arguments show.

First, \eqref{pm3} and \eqref{pm3-} show that $\xf(n)$ have
infinitely many jumps with size of order 1. More precisely, it can be
seen from \eqref{pm1} that for every $\eps>0$, there exists $\gd>0$
such that every interval $[N,(1+\eps)N]$ with $N\ge1$ contains some
$n$ such that the jump $|\xf(n+1)-\xf(n)|>\gd$.

Secondly, we cannot have $|f(n)|=n^{\log_2\ga}
\bigpar{P(\log_2n)+o(1)}$ for any $1$-periodic function $P(t)$,
because then substituting $2n$ would give $|f(2n)|=\ga
n^{\log_2\ga}\bigpar{P(\log_2n)+o(1)}$, while \eqref{pm0} would imply
$|f(2n)|=(\ga+\gb)n^{\log_2\ga}\bigpar{P(\log_2n)+o(1)}$, and
together these imply $P(\log_2 n)=o(1)$, so we would have $f(n) =
o(n^{\log_2 \ga})$, which contradicts \eqref{pm3-}.
\end{example}

There are many OEIS sequences in this category too, and most of them
have $\ga+\gb=0$. We do not discuss these examples further since we
have nothing new to add by our methods, but just mention a prototype
sequence A115384, the partial sum of Thue-Morse sequence (A010060,
the parity of the dyadic valuation), which satisfies
$\Lambda_{-1,1}[f]=\tr{\frac n2}$ with $f(1) = 0$. The exact solution
is given by $f(n) = \tr{\frac n2} +\frac14(1-(-1)^n)
(1+(-1)^{\nu(n-1)})$, where the last term indicates why there is no
simple smooth function providing good asymptotic approximation to
$f(n)-\frac n2$. See also other related sequences A076826, A159481, 
A173209 and A245710, which have a very similar behaviour.

\section{Extension from binary to $q$-ary}
\label{S:qary}

We consider  briefly the more general recurrence
\begin{align}\label{q1}
	f(n)
	= \sum_{0\le j<q}\alpha_{j}f\lpa{\ltr{\tfrac{n+j}q}}
	+g(n) \qquad (n\ge q),
\end{align}
for some integer $q\ge2$ and $q$ given constants $\ga_0,
\dots,\ga_{q-1}$; note that the case $q=2$ corresponds to \eqref{a1}. 
We now require $g(n)$ for $n\ge q$ and the initial values
$\{f(1),\dots,f(q-1)\}$.

Just as the recurrence \eqref{a1}, for example, occurs naturally in
many combinatorial and algorithmic contexts where a problem is split
into two halves, the generalisation \eqref{q1} occurs typically in
divide-and-conquer context or recursive structures where we instead 
divide the source problem into $q$ subproblems of sizes as evenly 
as possible.

The case $\ga_0=\dots=\ga_{q-1}=1$ was discussed in \cite{Hwang2017}
with several examples from the literature and OEIS. We can similarly
extend the method of \refS{S:recurrence} to treat the general
recursion \eqref{q1} under suitable conditions. We assume that
$\ga_j$ are real with
\begin{align}\label{q2}
	\max_{0\le j<q}|\ga_j| < A:=\sum_{0\le j<q}\ga_j.
\end{align}
Note that \eqref{q2} holds in the standard case when all $\ga_j>0$;
it also holds, more generally, if $\ga_j\ge0$ for all $j$ and at
least two $\ga_j$ are non-zero.

\begin{lemma}
Assume \eqref{q2}. Then there exists a unique continuous function
$\gf(t)$ on $\oi$ such that $\gf(0)=0$, $\gf(1)=1$, and for
$j=0,\dots,q-1$,
\begin{align}\label{q4}
	\gf(t) 
	= \frac{\ga_{q-1-j}}{A}\gf(qt-j)
	+ \frac{\sum_{q-j\le i<q}\ga_i}{A},   
	\qquad \text{if}\quad 
	\frac{j}{q}\le t\le \frac{j+1}{q}.
\end{align}
Moreover, if  $\ga_j\ge0$ for all $j$, then $\gf$ is strictly 
increasing.
\end{lemma}

\begin{proof}
This follows with only notational changes as in our third proof of
\refL{Lemma2} based on the recursive construction \eqref{b9}; we
obtain by an analogue of \eqref{b9} a sequence of continuous
functions $\gf_k$ that, using \eqref{q2}, converge uniformly to a
function $\gf(t)$ satisfying \eqref{q4}.
\end{proof}

We then extend $f(n)$ and $g(n)$ to functions of a real variable
$x\ge1$ by \eqref{b2} as before, and it is easily verified that
\eqref{b3} generalises to
\begin{align}\label{q5}
	f(x)=Af\lrpar{\frac{x}{q}}+g(x),
	\qquad x\ge q.  
\end{align}
We may now argue as in \refS{S:recurrence} and prove extensions of
\refT{Theorem1} and its corollaries for the recursion \eqref{q1}. We
now define
\begin{align}\label{q3}
	\rho:=\log_q A
	=\log\Bigpar{\sum_{0\le j<q}\ga_j}.
\end{align}
The simplest situation is when $g(n)=O(n^{\rho-\eps})$ for some 
$\eps>0$; then 
\begin{align}\label{q6}
    f(x)=x^\rho P(\log_q x)-Q(x),
    \qquad x\ge1,
\end{align}
where $P(t)$ is a continuous $1$-periodic function, and
\begin{align}\label{q7}
	Q(x):=\sum_{m\ge1} q^{-\rho m} g(q^mx) = o(x^\rho).   
\end{align}
We leave further details to the reader and content ourselves with the
discussion of two classes of examples. 

\subsection{Binomial coefficients not divisible by a prime $q$}
\label{SS:binomial-q}

Let $f(n)$ denote the number of binomial coefficients $\binom{m}{k}$,
$0\le k\le m<n$, that are not divisible by a given prime $q$. This
sequence has a long history, at least dating back to Fine's
\cite{Fine1947} observation that almost all binomial coefficients are
even; see, e.g., \cite{Chen2002} and the references therein. It
equals A006046 (see Example~\ref{Ex:pascal}) when $q=2$. The case
$q=3$ corresponds to A006048, while $q=5$ gives A194458. We then 
deduce the recurrence 
\begin{equation}\label{qa1}
	f(n) 
	= \sum_{0\le j<q}(q-j)f\lpa{\tr{\tfrac{n+j}q}},
\end{equation}
with $f(0)=0$ and $f(1)=1$. This is \eqref{q1} with $\ga_j=q-j$ and
$g(n)=0$; furthermore, in this example \eqref{qa1} holds for all
$n\ge0$. We have $A=\binom{q+1}2$ and $\rho := \log_q A$. Stein
\cite{Stein1989} proved that
\begin{equation}\label{qa2}
	\frac1A\le \frac{f(n)}{n^\rho}\le 1,
\end{equation}
and extended $f(n)$ to a continuous function $f(x)$; see
\cite{Franco1998,Volodin1999,Wilson1998} for finer lower bounds. Our
general approach yields the same continuous extension $f(x)$ as in 
\cite{Stein1989}; we obtain (by \eqref{q6})
\begin{equation}\label{qa3}
	\frac{f(n)}{n^{\rho}}=P(\log_qn)\qquad(n\ge1),
\end{equation}
where $P(t):=A^{-\frax t}f(q^{\frax t})$ is a continuous $1$-periodic
function. Moreover, since \eqref{qa1} holds for all $n\ge1$, it is
easily verified (using \eqref{q4}) that $f(x)=A\gf(x/q)$ for
$x\in[1,q]$, and thus
\begin{equation}
   P(t) = A^{1-\{t\}} \varphi(q^{\{t\}-1}).
\end{equation}
(Cf.\ \refR{Rb=1} for a similar simplification, for related reasons.)
Here $\varphi(t)$ satisfies $\varphi(0)=0$, $\varphi(1)=1$, and for 
$0\le j<q$,
\begin{equation}
	\varphi(t) 
	= \frac{j+1}{A}\varphi(\{qt\})
	+ \frac{\binom{j+1}{2}}{A},
	\qquad \text{if}\quad 
	\frac{j}{q}\le t\le \frac{j+1}{q}.
\end{equation}
By \eqref{qa2} or \eqref{qa3}, 
\emph{almost all binomial coefficients are divisible by any
given prime $q$} because 
\begin{equation}
	\rho = \log_q(q+1)-\log_q2+1 <2.
\end{equation}

\begin{figure}[!ht]
\begin{center}
\begin{tabular}{cccc}
	\includegraphics[height=3cm]{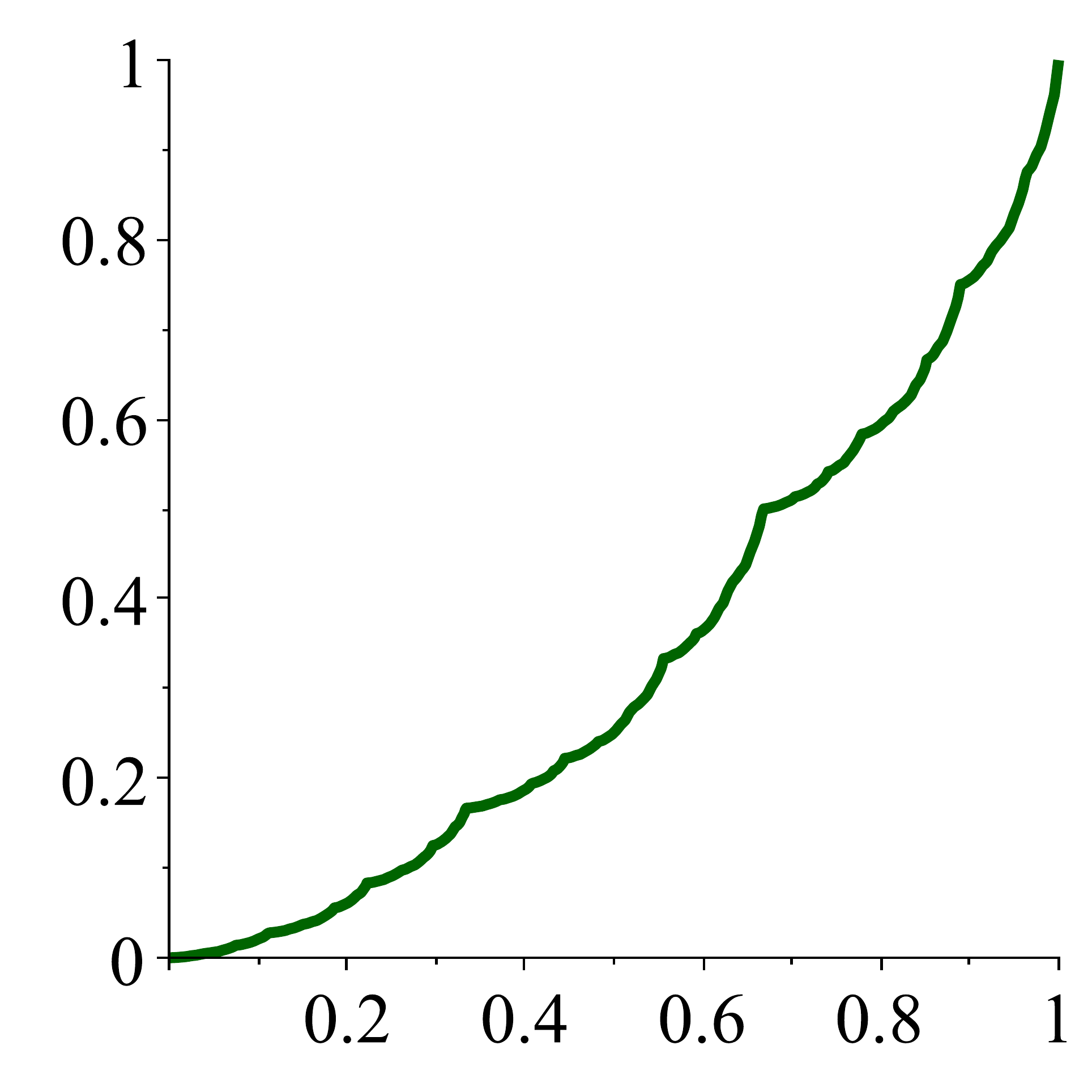}
	& \includegraphics[height=3cm]{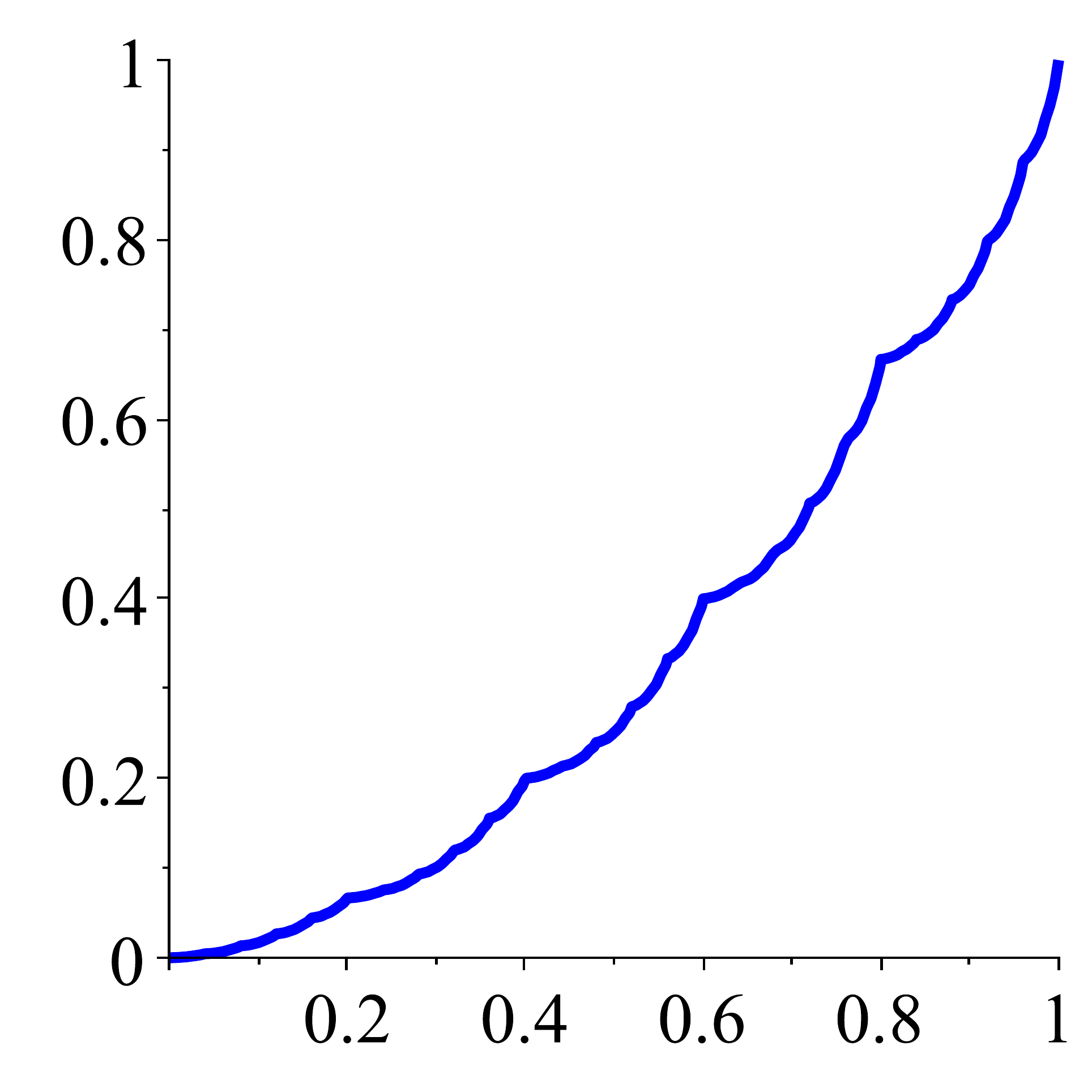}
	& \includegraphics[height=3cm]{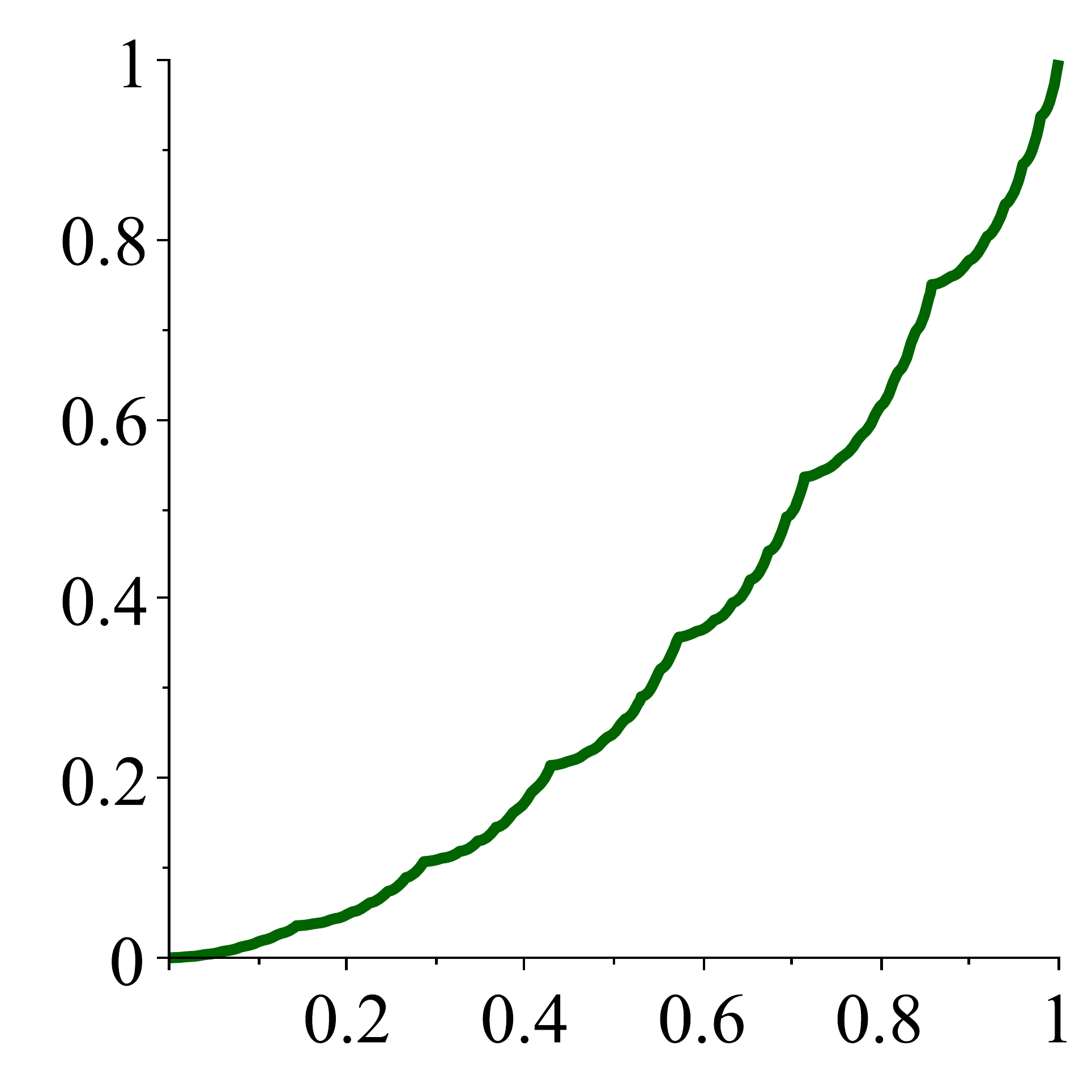}
	& \includegraphics[height=3cm]{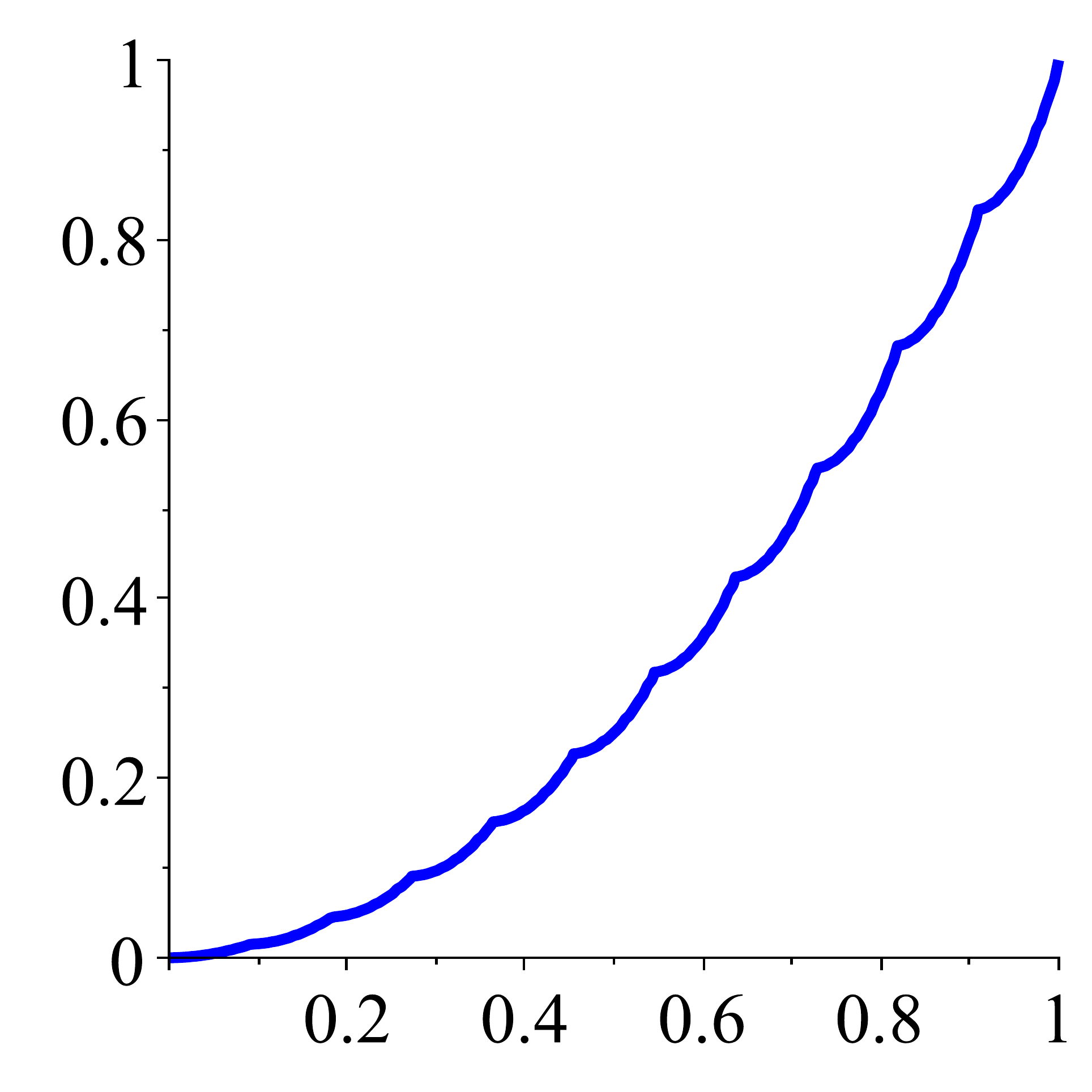}\\
	$\varphi_3$ & $\varphi_5$ & $\varphi_7$ & $\varphi_{11}$ \\
	\includegraphics[height=3cm]{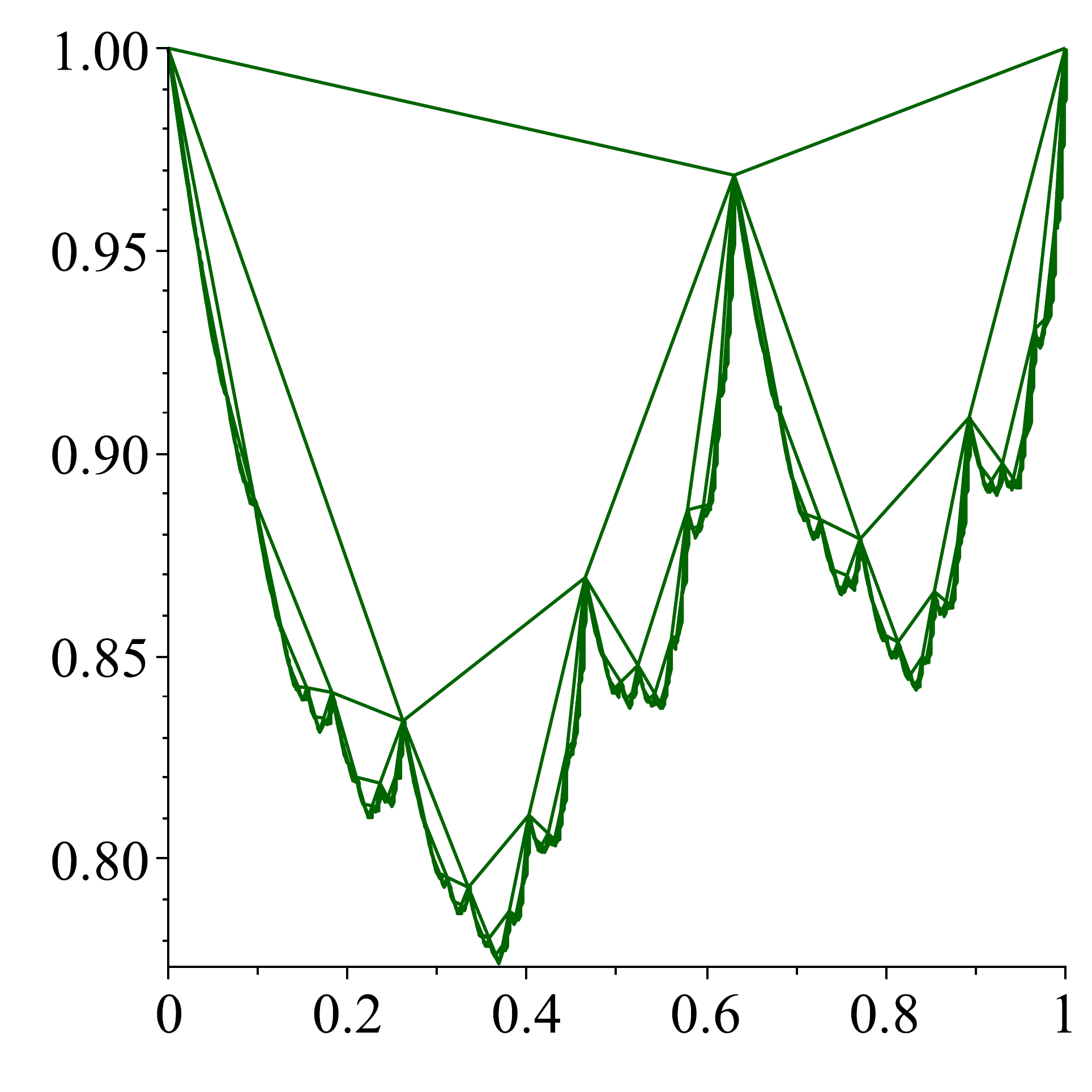}
	& \includegraphics[height=3cm]{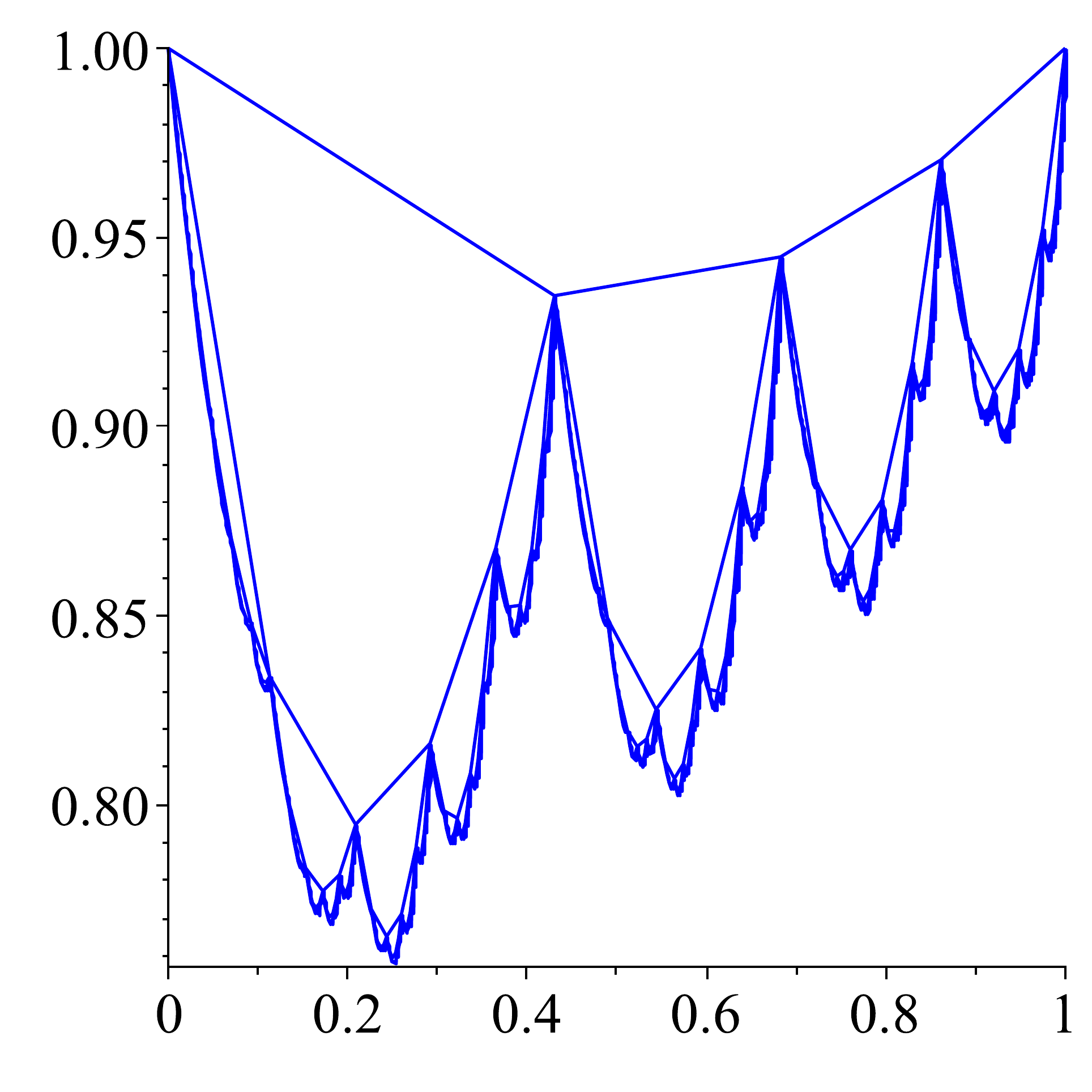}
	& \includegraphics[height=3cm]{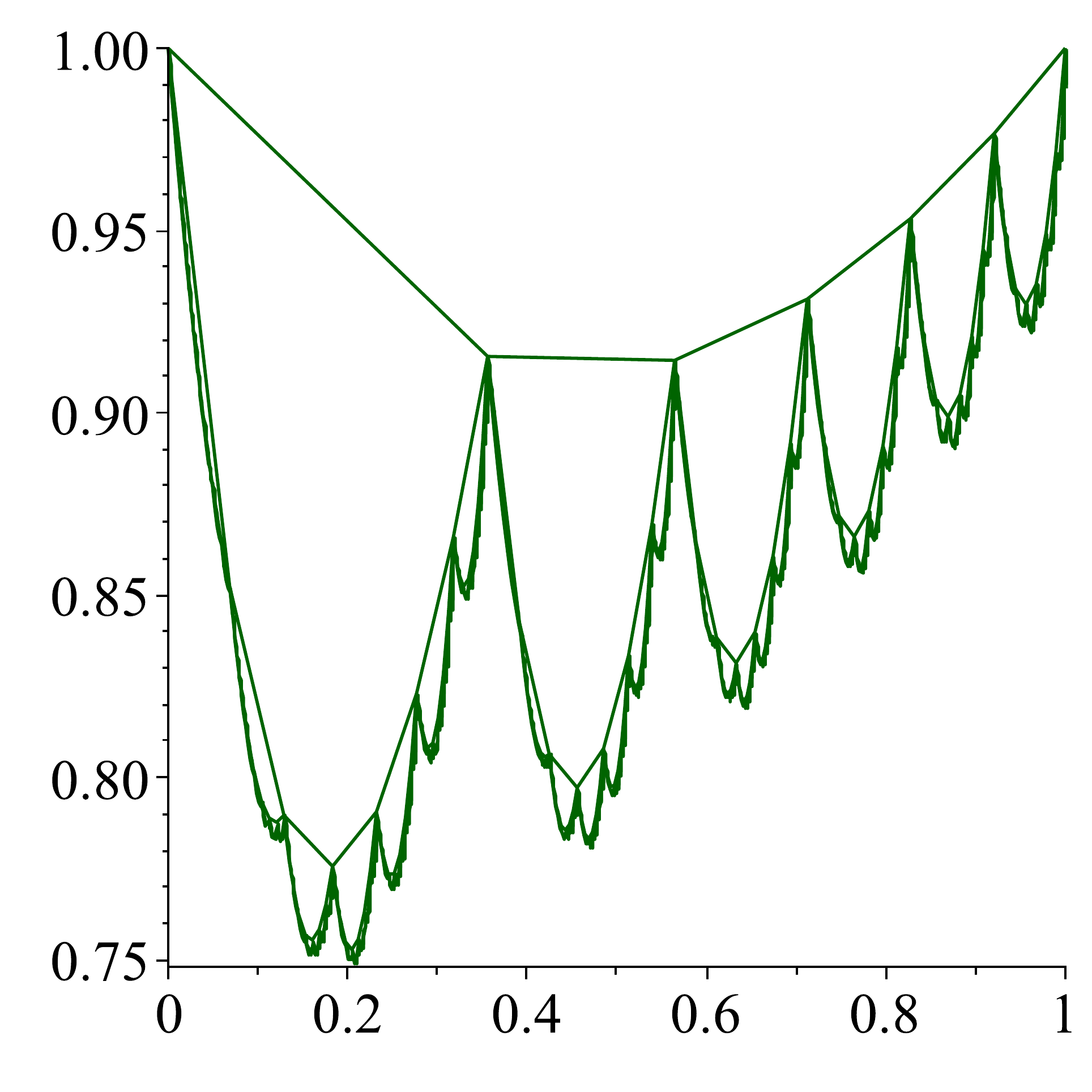}
	& \includegraphics[height=3cm]{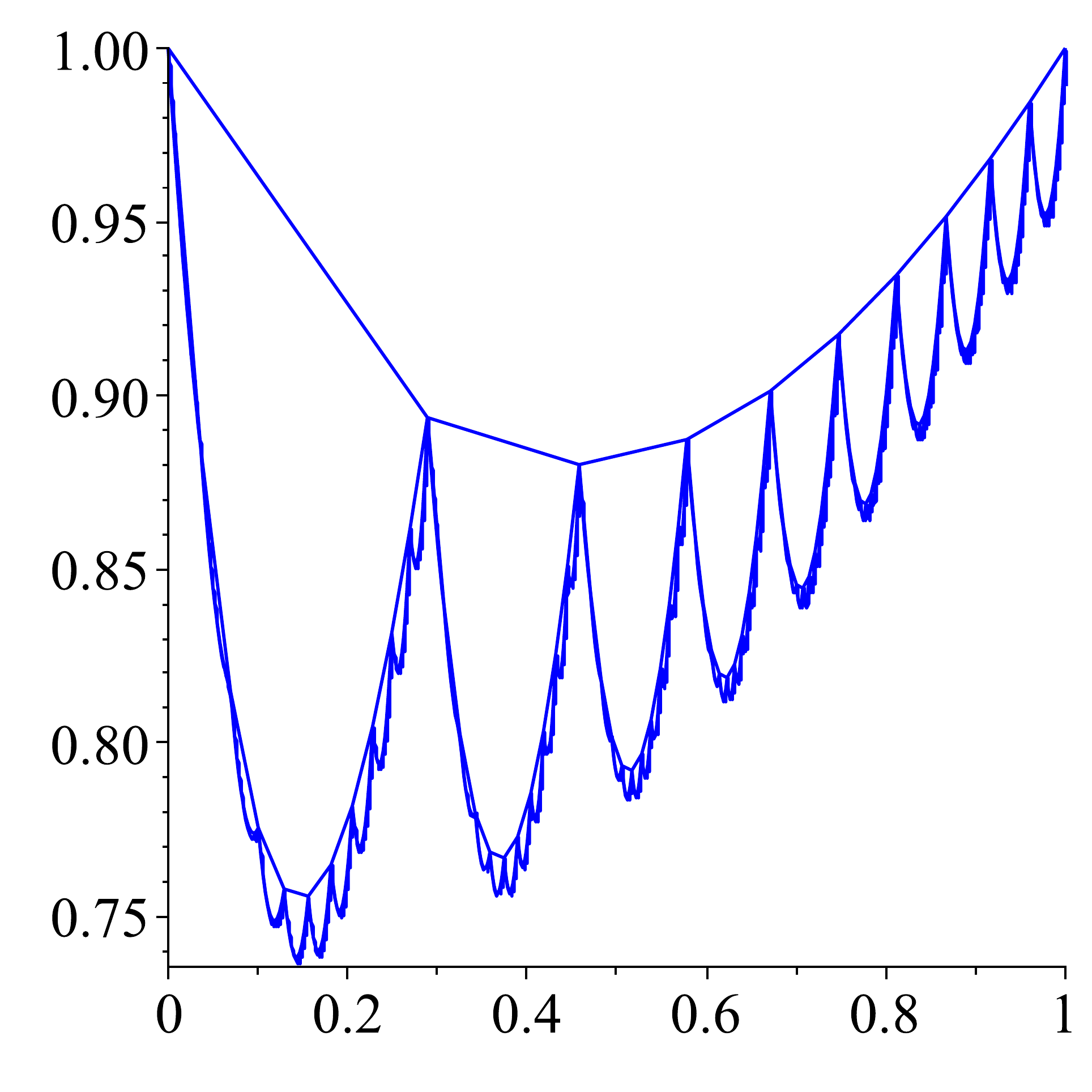}\\
	$P_3$ & $P_5$ & $P_7$ & $P_{11}$
\end{tabular}	
\end{center}
\caption{The function $\varphi$ (\refS{SS:binomial-q}) when $q=3, 5,
7, 11$ (upper half) and the periodic function $P$ for the same set of
values of $q$ (lower half).}\label{fig-binomial-q}
\end{figure}

Alternatively, it is known \cite{Chen2002} that 
\begin{equation}
	f(n) = \frac12\sum_{0\le j\le s}
	{\binom{p+1}2}^j b_j\prod_{j\le i\le s}(b_i+1),
\end{equation}
when $n = b_0 + b_1q+\cdots b_sq^s$ with $0\le b_j<q$ and 
$s=\tr{\log_q n}$, from which we can obtain an alternative 
representation for the periodic function $P$. Also the generating 
function \cite{Stein1989}
\begin{equation}
	\sum_{n\ge0}f(n)z^n
	= \frac{z}{1-z}\prod_{k\ge0}\sum_{0\le j<q}(j+1)z^{j\cdot q^k},
\end{equation}
is helpful in applying the Mellin transform approach; see 
\cite{Flajolet1994,Grabner1993,Grabner2005}. 

The more general problem of the number of multinomial coefficients
$\binom{m}{j_1,\dots,j_d}$ not divisible by a prime $q$, for $0\le
m<n$ and $j_1+\cdots+j_d=m$, $j_1,\dots,j_d\ge0$, where $d\ge1$ is
given (see \cite{Chen1998,Chen2002,Volodin1999}) can be similarly
dealt with. This number satisfies the recurrence
\begin{equation}
	f(n) 
	= \sum_{0\le j<q}\binom{q-j+d-2}{d-1}
	f\Lpa{\Bigl\lfloor\frac{n+j}q\Bigr\rfloor};
\end{equation}
we then deduce, by \eqref{q6}, the identity $f(n)= n^{\rho}
P(\log_qn)$, $n\ge1$, for some continuous periodic function $P$ with
$\rho = \log_q\binom{q+d-1}{d}$.

\subsection{Generating polynomial of Gray codes}
\label{SS:Gray}
Gray codes of integers are strings of binary words in which
neighboring code words differ by one bit only; we already discussed
some properties of the binary reflected Gray codes in
Examples~\ref{E22-ns} and \ref{E44}. Here, we consider a simple
extended version of binary Gray codes to $q$-ary ones
(non-reflected); see \cite{Cohn1963,Sharma1978}. The construction is
as follows. If
\begin{equation}
	n = \sum_{0\le j\le s}\kappa_j q^j,
	\qquad(0\le \kappa_j<q),
\end{equation}
then the Gray code of $n$ is given by $(\kappa_s',\dots,\kappa_0')$,
where $\kappa_s' = \kappa_s$ and 
\begin{equation}\label{E:kappa-p}
	\kappa_j' := (\kappa_{j}-\kappa_{j+1})\bmod q
	\qquad (0\le j<s).
\end{equation}
For simplicity, we consider the number of nonzero digits $\gamma(n):=
\sum_{0\le j\le s}\mathbf{1}_{\kappa_j'>0}$ in this Gray code
representation of $n$; other quantities such as the sum of digits
$\sum_{0\le j\le s}\kappa_j'$ can be considered similarly (a sketch
given below). Then by the recurrence
\begin{equation}
	\gamma(qk+j)
	= \gamma(k)+\begin{cases}
		0, & \text{if } k \bmod q \equiv j;\\
		1, & \text{otherwise},
	\end{cases}
\end{equation}
for $0\le j<q$, we deduce that the generating polynomial $f(n) := 
\sum_{0\le k<n} \alpha^{\gamma(k)}$ of $\gamma(n)$ satisfies
\begin{equation}
	f(n) 
	= f\lpa{\ltr{\tfrac{n}{q}}}+
	\alpha \sum_{1\le j<q} f\lpa{\ltr{\tfrac{n+j}{q}}}
	+g(n),
\end{equation}
where $g$ can be expressed as 
\begin{align}
	g(q^2k+qr+j) 
	= (1-\alpha)\alpha^{\gamma(qk+r)+1}
	\qquad(k\ge0),
\end{align}
for $0\le r\le q-2$ and $r+1\le j\le q-1$, and $g(n)=0$ for all other 
values of $n$. Alternatively, in terms of $q$-ary expansion, the 
nonzero $g(n)$, $0\le n<q^2$, occurs when $n=(\kappa_1,\kappa_0)_q$ 
is of the form:
\begin{center}
\begin{tabular}{c}
	$(0,1)_q,(0,2)_q,(0,3)_q,(0,4)_q,\dots, (0,q-1)_q$\\
	$(1,2)_q,(1,3)_q,(1,3)_q,\dots, (1,q-1)_q$\\
	$\vdots$\\
	$(q-3,q-2)_q,(q-3,q-1)_q$\\
	$(q-2,q-1)_q$
\end{tabular}
\end{center}

We then deduce from \eqref{q6} the exact and asymptotic expansion
(since $Q(n)=0$ by \eqref{q7})
\begin{equation}
	f(n) = P(\log_qn)n^{\log_q(1+(q-1)\alpha)}
	\qquad(n\ge1),
\end{equation}
whenever $\alpha>0$, for some continuous periodic function $P=P_\ga$;
note that then $g(n)=O\bigpar{(\ga\vee1)^{\log_q n}} =
O\bigpar{n^{0\vee\log_q\ga}}$ with $0\vee\log_q\ga<\rho
=\log_q(1+(q-1)\alpha)$.

\begin{figure}[!ht]
\begin{center}
	\begin{tabular}{cccc}
		\includegraphics[width=3cm]{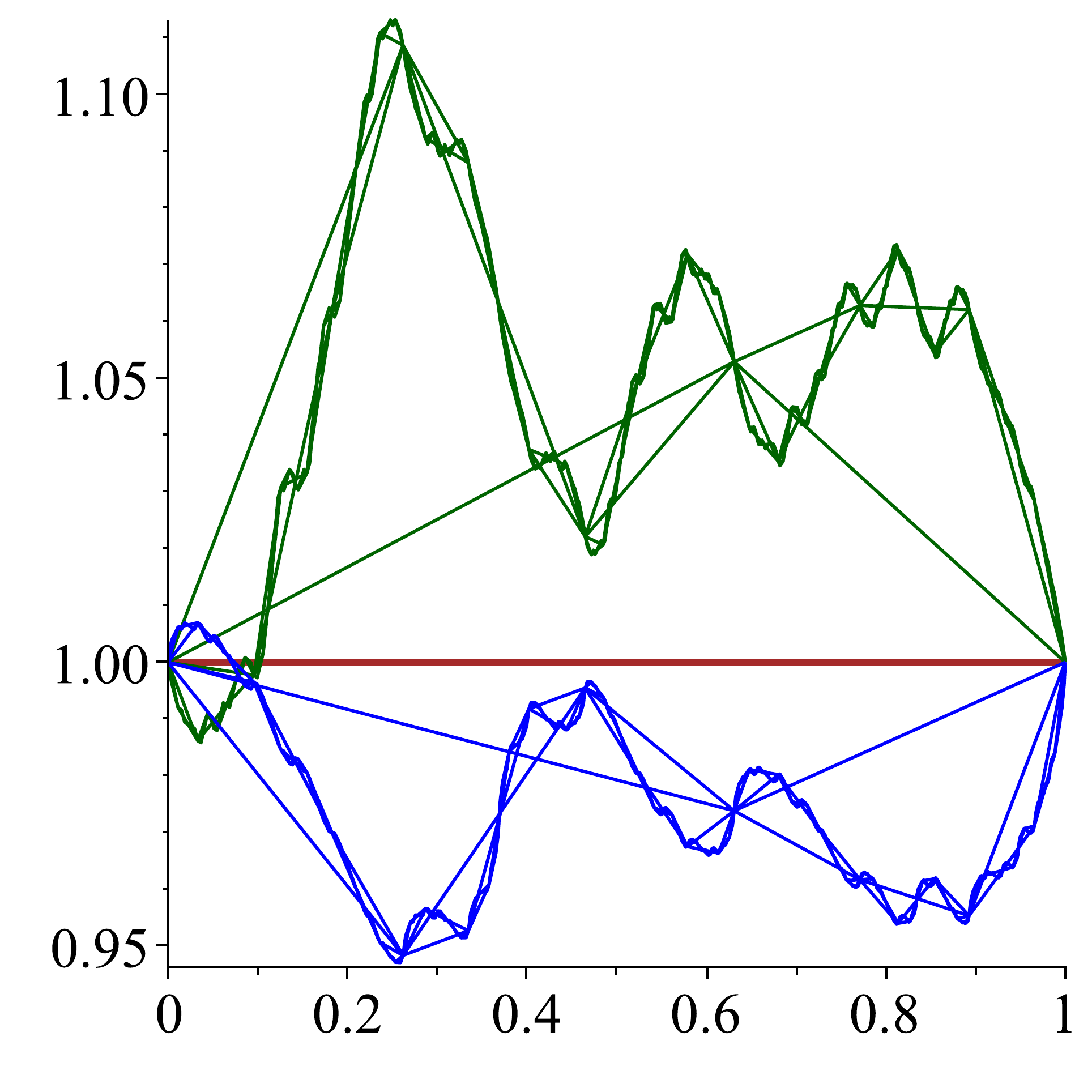} &
		\includegraphics[width=3cm]{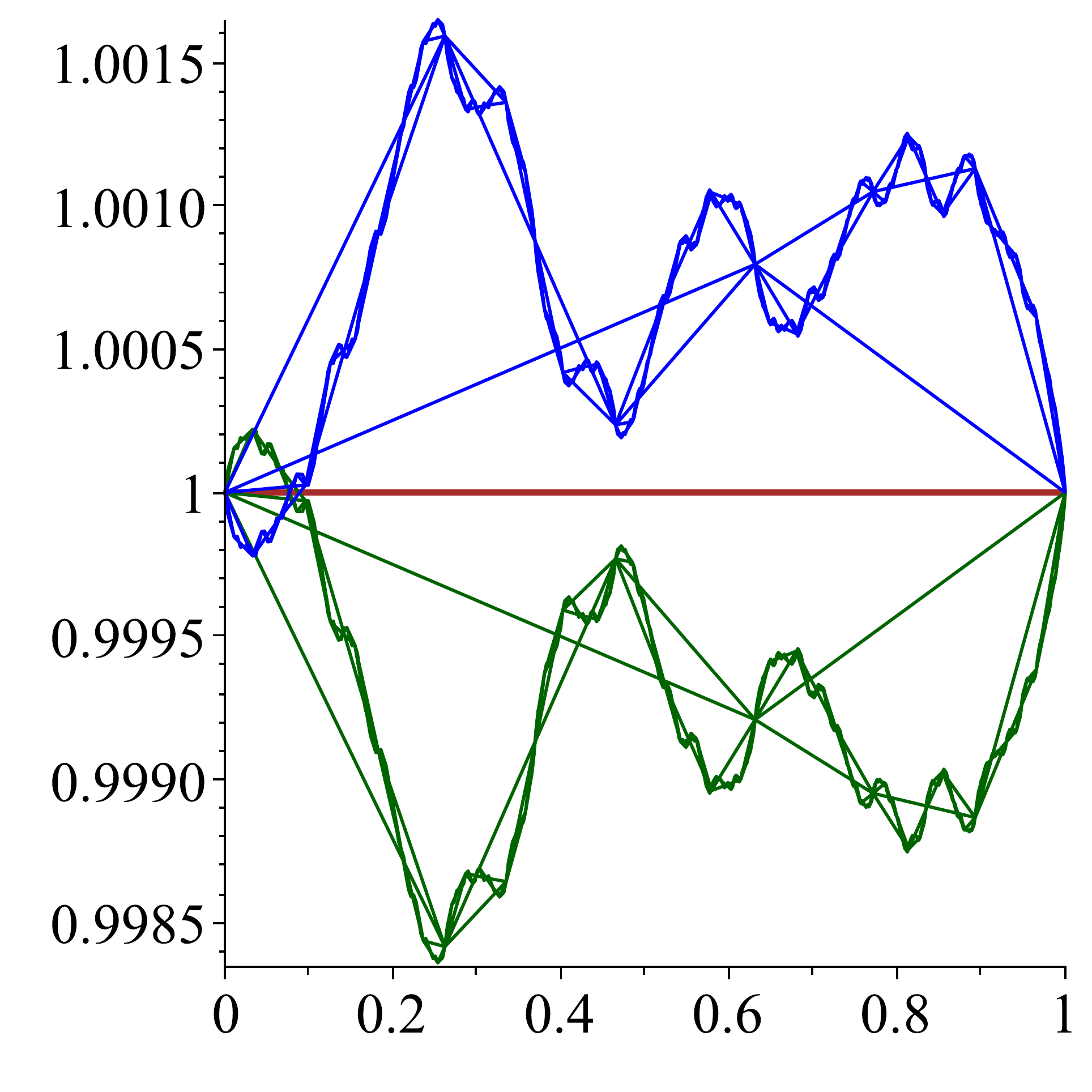} &
		\includegraphics[width=3cm]{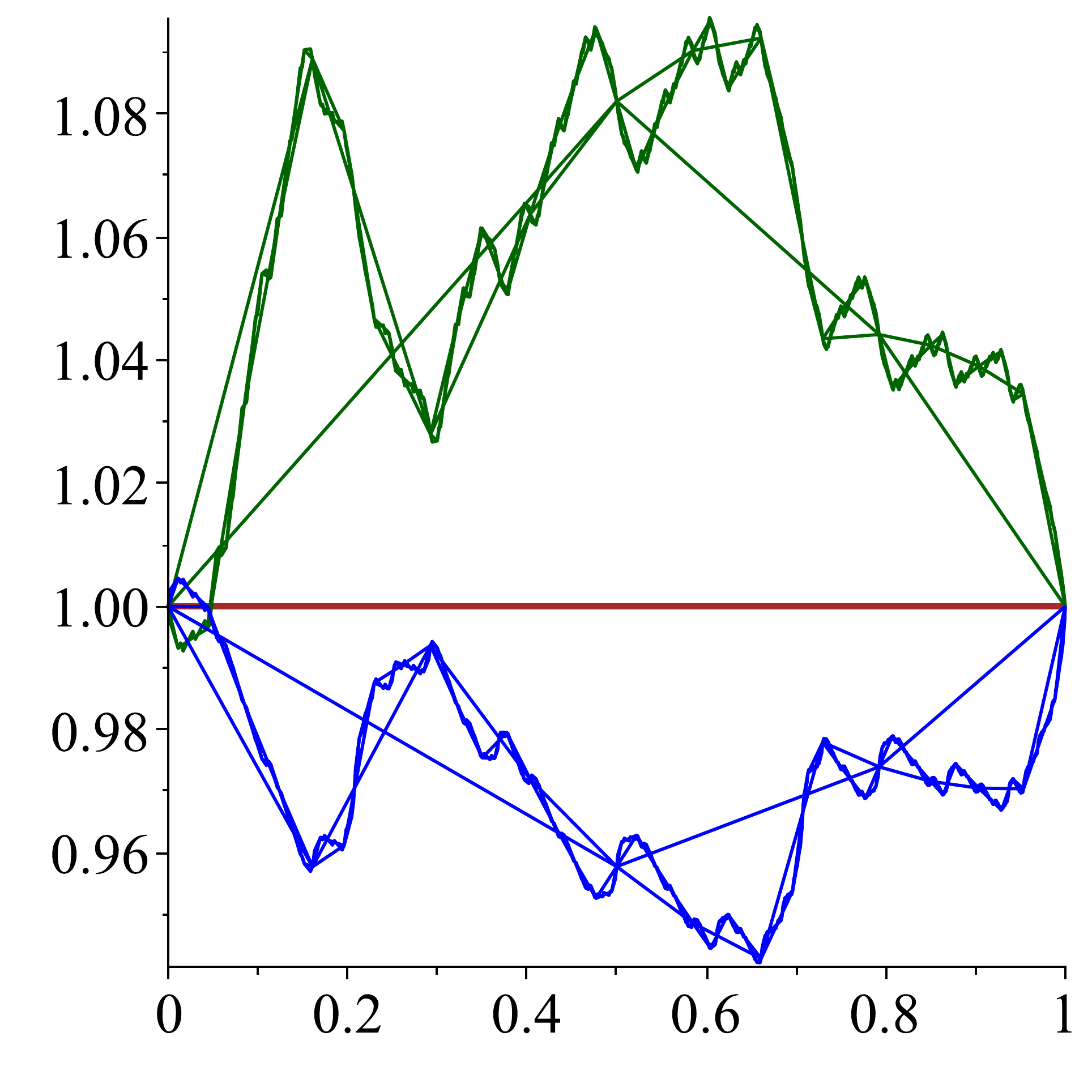} &
		\includegraphics[width=3cm]{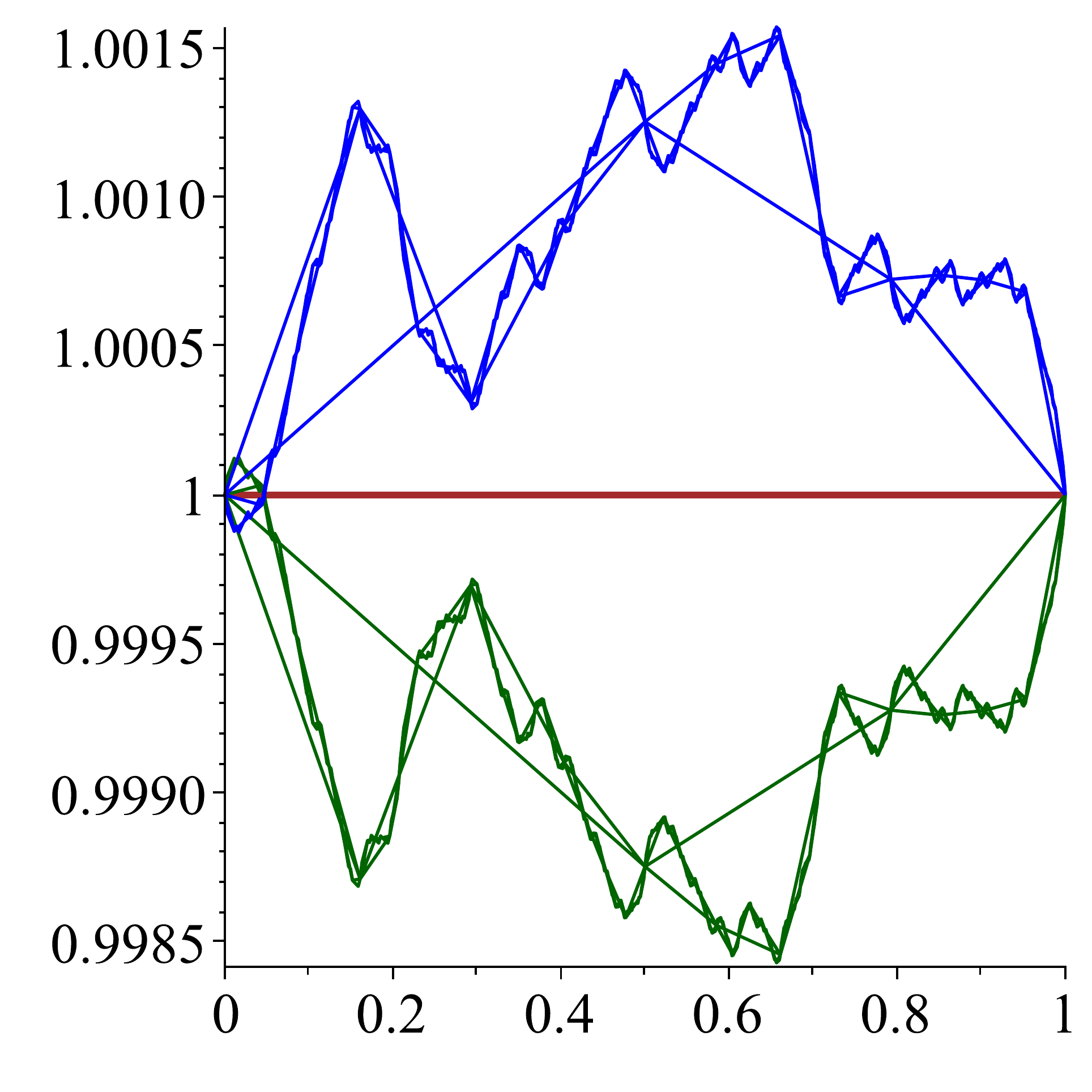}\\
		$q=3,\alpha=\frac{\sqrt{5}\pm1}2$ &
		$q=3,\alpha=1\pm0.01$ &
		$q=4,\alpha=\frac{\sqrt{5}\pm1}2$ &
		$q=4,\alpha=1\pm0.01$ 
	\end{tabular}
\end{center}
\caption{Periodic fluctuations of $f(n)n^{-\log_q(1+(q-1)\alpha)}$ in
the cases of nonzero digits of $q$-ary Gray codes for $q=3,4$ and 
different $\alpha$ (\refS{SS:Gray}).} \label{fig-Gray}
\end{figure}

Similarly, the generating polynomial of the sum-of-digits function in
such $q$-ary Gray codes $f(n) = \sum_{0\le k<n}\alpha^{\sigma(k)}$, 
$\sigma(k) := \sum_{0\le j\le s}\kappa_j'$ when $k=\sum_{0\le j\le
s}\kappa_jq^j$, satisfies the recurrence
\begin{equation}
	f(n) = \sum_{0\le j<q}\alpha^j f\lpa{\ltr{\tfrac{n+j}{q}}}
	+ g(n),
\end{equation}
where $g$ can be expressed as
\begin{align}
	g(q^2k+qr+j) 
	=\begin{cases}
	\ga^{\gs(qk+r)}\frac{(1-\ga^j)(\ga^{q-r}-\ga^{q-j})}{1-\ga} 
	& (0\le j<r),\\
	\ga^{\gs(qk+r)}\frac{(1-\ga^{j-r})(1-\ga^{q-j})}{1-\ga} 
	& (r\le j<q),
	\end{cases}
\end{align}
for $k\ge0$, $0\le r\le q-1$, and $0\le j\le q-1$. This is derived by 
the relation
\begin{align}
	\sigma(qk+j) = \sigma(k) + (j-k) \bmod q,
\end{align}
which in turn follows from \eqref{E:kappa-p}. We then deduce that 
\begin{equation}
	f(n) = P(\log_qn)n^{\log_q(1+\alpha+\cdots + \alpha^{q-1})},
\end{equation}
for $n\ge1$.

\appendix
\addcontentsline{toc}{section}{Appendices}

\section{Mellin transforms}\label{AMellin}

Mellin transforms are another useful techniques in analysing
divide-and-conquer recurrences; see \cite{Flajolet1994a,
Flajolet1995, Flajolet1994, Grabner2005, Hwang2003} and the
references therein for more information. Up to now most of the tools
we adopt to solve \eqref{a1} are direct and elementary in nature; it
is however possible to apply Mellin transforms for a more effective
characterisation of the underlying periodic oscillations, notably
calculations of the Fourier coefficients, as already observed before
in the literature (although an analytic approach often requires
stronger conditions).

Let again $\ga,\gb>0$, and $\rho:=\log_2(\ga+\gb)$, and assume that
$f(n)$ and $g(n)$ satisfy the recursion \eqref{a1}. Extend again $f$
and $g$ to $[1,\infty)$ by \eqref{b2}, with $g(1):=0$, and define
$f(x):=g(x):=0$ for $x\in[0,1)$. Denote the Mellin transform of
$f(x)$ by $f^*(s)$:
\begin{align}\label{mellin} 
	f^*(s):=\int_0^\infty f(x) x^{s-1}\dd x 
	= \int_1^\infty f(x) x^{s-1}\dd x, 
\end{align} 
for all complex $s$ such that the integral is absolutely convergent,
and similarly for $g^*(s)$. If $f(n)=O(n^c)$ for large $n$ and some
real $c$, then $f^*(s)$ exists at least in the half-plane $\Re s <
-c$, and is analytic there. If $f^*(s)$ or $g^*(s)$ extends
meromorphically to a larger domain, we use the same notation there.

Assume that $f^*(s)$ and $g^*(s)$ exist. Then \eqref{b3} implies that
\begin{align}\label{mel2}
	f^*(s)
	&=\int_1^2 f(x)x^{s-1}\dd x 
	+ \int_2^\infty \Bigpar{2^\rho  
	f\Bigpar{\frac{x}{2}}+g(x)}x^{s-1}\dd x\notag\\
	&=\int_1^2 \bigpar{f(x)-g(x)}x^{s-1}\dd x 
	+ 2^{\rho+s} f^*(s)+g^*(s).
\end{align}
Furthermore, \eqref{b15} or \eqref{b16a} shows that if $1\le x<2$,
then $f(x)-g(x)=f(1)P_0(\log_2x)x^\rho$. Hence, recalling the
definition \eqref{b16},
\begin{align}\label{mel3}
	\int_1^2 \bigpar{f(x)-g(x)}x^{s-1}\dd x &
	=f(1)\int_1^2 P_0(\log_2 x) x^{\rho+s-1}\dd x  
	\notag\\&
	=f(1)\int_1^2 \lrpar{1+(2^\rho-1)\gf\lrpar{x-1}} x^{s-1}\dd x.
\end{align}
Substituting this into \eqref{mel2} yields 
\begin{align}\label{mel4}
	\lrpar{1-2^{\rho+s}}f^*(s)
	=g^*(s) + f(1)\int_0^1\lrpar{1+(2^\rho-1)\gf(u)} 
	(1+u)^{s-1}\dd u.
\end{align}
Note that the right-hand side of \eqref{c4} equals, apart from a
factor $\frac1{\log 2}$, the right-hand side of \eqref{mel4} at
$s=-\rho-\chi_k$. On the other hand, for such $s$, the left-hand side
factor $1-2^{\rho+s}$ equals zero. Hence, combining \eqref{mel4} and
\refT{Theorem4}\ref{T4b} yields the following (see
\cite{Flajolet1995} for more information).

\begin{lemma}\label{LB1}
If\/ $g(n)=O(n^{\rho-\eps})$ for some $\eps>0$, then $g^*(s)$ is
analytic in (at least) the half-plane $\Re s < -\rho+\eps$, and
$f^*(s)$ is meromorphic in the same half-plane with only simple
poles, which are at\/ $-\rho-\chi_k=-\rho-\frac{2k\pi i}{\log 2}$ for
some $k\in\bbZ$. Furthermore, \refT{Theorem4}\ref{T4b} applies, and
\begin{align}\label{mel5}
	\hP(k)=-\Res\lrsqpar{f^*(s);s=-\rho-\chi_k}.
\end{align}
\end{lemma}
We do not claim that every $-\rho-\chi_k$ actually is a pole. In
fact, by \eqref{mel5}, $-\rho-\chi_k$ is a pole of $f^*(s)$ if
and only if $\hP(k)\neq0$.

For a better demonstration of the approach, we study the case when
$g(n)\sim n^\rho$, and thus \refT{Theorem1} does not directly apply
(although readily amenable) due to an extra logarithmic leading term
in the asymptotic approximation of $f(n)$. For a meromorphic function
$F(z)$, let $F(s)\FV$ denote \emph{finite value} (or the constant
term) in the Laurent series expansion at $z=s$; thus $F(s)\FV=F(s)$
when the latter is finite.

\begin{thm}\label{TB1}
Assume that $g(n)=n^\rho+g_0(n)$, where $g_0(n)=O(n^{\rho-\eps})$ for
some $\eps>0$. Then
\begin{equation}\label{tb1}
	f(n)=
	n^\rho\log_2 n +n^{\rho }P(\log_{2}n)-Q(n), \qquad n\ge1,
\end{equation}
where (as in \refT{Theorem1}) $P(t)$ is a continuous $1$-periodic
function and $Q(n)=o\lrpar{n^\rho}$ as \ntoo. The Fourier
coefficients of $P(t)$ are given by
\begin{equation}\label{c4B}
	\widehat{P}(k)
	=\frac{f(1)}{\log 2}\int_{0}^{1}
	\frac{1+(\alpha +\beta -1)\varphi(u)}
	{(1+u)^{\rho +\chi_{k}+1}}\dd u
	+\begin{cases}
		\frac{1}{\log 2}g^*(-\rho-\chi_k), & k\neq0,
		\\
		\frac{1}{\log 2}g^*(-\rho)\FV+\frac12, & k=0,
	\end{cases}
\end{equation}
where $g^*(s)$ is meromorphic in $\Re s<-\rho+(\eps\wedge1)$ with a
sole simple pole at $s=-\rho$. In particular, \eqref{c4} holds for
$k\neq0$.

If, moreover, $g(n)=n^\rho$ for even $n\ge2$ (i.e., $g_0(2m)=0$),
then $Q(n)=0$ for $n\ge1$.
\end{thm}

\begin{proof}
Assume, without loss of generality, $\eps<1$. The assumption and 
\eqref{b2} then imply that we have, for $n\ge1$ and $t\in[0,1]$,
\begin{align}\label{tb2g}
	|g(n+t)-g(n)|&
	\le |g(n+1)-g(n)|
	=|(n+1)^\rho-n^\rho+g_0(n+1)-g_0(n)|
	\notag\\&
	=O\bigpar{n^{\rho-\eps}},
\end{align}
which together with a Taylor expansion of $(n+t)^{s-1}$ yields, for
$\Re s<-\rho$,
\begin{align}\label{mel12g}
	g^*(s)&
	=\sum_{n\ge1} \int_{0}^{1} g(n+t)(n+t)^{s-1}\dd t
	=\sum_{n\ge1}\Bigpar{ g(n)n^{s-1} 
	+O\bigpar{n^{\rho+\Re s-1-\eps}}} \notag\\
	&=\sum_{n\ge1} n^{\rho+s-1} +
	\sum_{\ge 1} O\bigpar{n^{\rho+\Re s-1-\eps}}\notag\\
	&=\zeta(1-\rho-s)+
	\sum_{n\ge1} O\bigpar{n^{\rho+\Re s-1-\eps}}.
\end{align}
Moreover, each term in the final sum is an entire function in $s$,
and the $O$ is uniform for $s$ in any compact set;
thus the sum converges to an analytic function in $H_\eps$. 
Hence, $g^*(s)$ extends to a meromorphic function
in $H_\eps$, with a single simple pole at $s=-\rho$,
as asserted.

Let $f_1(n):=n^\rho\log_2 n$, and $f_2(n):=f(n)-f_1(n)$, and let
$g_j:=\gL_{\ga,\gb}[f_j]$, $j=1,2$. We have
\begin{align}\label{mel6}
	\gL_{\ga,\gb}[f_1](2n)
	&=(2n)^\rho\log_2(2n)-(\ga+\gb)n^\rho\log_2(n)=(2n)^\rho,
\end{align}
and, with $\psi(x):=x^\rho\log_2(x)$, using the mean-value theorem.
\begin{align}\label{mel7}
	\gL_{\ga,\gb}[f_1](2n+1)
	&=(2n+1)^\rho+2^\rho\psi\bigpar{n+\tfrac12}
	-\ga\psi(n)-\gb\psi(n+1)\notag\\
	&=(2n+1)^\rho+O\bigpar{n^{\rho-1}\log (n+1)}.
\end{align}
In other words, for all $n\ge2$,
\begin{align}\label{mel8}
	g_1(n):=  \gL_{\ga,\gb}[f_1](n)
	=n^\rho+O\bigpar{n^{\rho-1}\log n}
	=n^\rho+O\bigpar{n^{\rho-\eps}},
\end{align}
and thus 
\begin{align}\label{mel9}
	g_2(n):=g(n)-g_1(n)
	=g_0(n)+O\bigpar{n^{\rho-\eps}}
	=O\bigpar{n^{\rho-\eps}}.
\end{align}
Consequently, $g_2^*(s)$ is analytic in the half-plane
$H_\eps:=\para{s:\Re s < -\rho+\eps}$, and, by \refC{C2},
\begin{equation}\label{mel10}
	f_2(x)
	=x^{\rho }P_2(\log_{2}x)-Q_2(x), \qquad x\ge1,
\end{equation}
where $Q_2(x)=\sum_{m\ge1} 2^{-\rho m} g_2\bigpar{2^mx}$, and
$P_2(t)$ is a periodic continuous function. Since
$f(n)=f_1(n)+f_2(n)$, this shows \eqref{tb1} with $P(t):=P_2(t)$ and
$Q(t):=Q_2(t)$.

Furthermore, by \refL{LB1}, 
\begin{align}\label{mel11}
	\hP(k)=
	\hP_2(k)=-\Res\lrsqpar{f_2^*(s);s=-\rho-\chi_k}.
\end{align}
By \eqref{mel4}, $f_2^*(s)$ is meromorphic in $H_\eps$, with poles
only at $-\rho-\chi_k$. Moreover, similarly to \eqref{tb2g}, it
follows from \eqref{b2} that
$|f_1(n+t)-f_1(n)|=O\lrpar{n^{\rho-\eps}}$ for $n\ge1$ and
$t\in[0,1]$, which in turn, similarly to \eqref{mel12g}, yields
\begin{align}\label{mel12f}
	f_1^*(s)
	&=\sum_{n\ge1} \int_{0}^{1} f_1(n+t)(n+t)^{s-1}\dd t
	=\sum_{n\ge1} n^{\rho+s-1}\log_2n +
	\sum_{n\ge1} O\bigpar{n^{\rho+\Re s-1-\eps}}
	\notag\\
	&=-\frac{1}{\log 2}\zeta'(1-\rho-s)+
	\sum_{n\ge1} O\bigpar{n^{\rho+\Re s-1-\eps}}.
\end{align}
Again, each term in the final sum is an entire function in $s$, and
the $O$ is uniform for $s$ in any compact set; thus the sum converges
to an analytic function in $H_\eps$. Hence, $f_1^*(s)$ extends to a
meromorphic function in $H_\eps$, with a single (double) pole at
$s=-\rho$. Furthermore, the residue $\Res\lrsqpar{f_1^*(s);s=-\rho} =
\frac{-1}{\log 2} \Res\lrsqpar{\zeta'(1-\rho-s);s=-\rho}=0$, since
$\zeta'$ is a derivative of a meromorphic function. Consequently,
$f^*(s)=f_1^*(s)+f_2^*(s)$ is meromorphic in $H_\eps$, with poles
only at $-\rho-\chi_k$, and
\begin{align}\label{mel14}
	\Res\lrsqpar{f^*(s);s=-\rho-\chi_k}
	=\Res\lrsqpar{f_2^*(s);s=-\rho-\chi_k},
	\qquad k\in\bbZ.
\end{align}
Thus, by \eqref{mel11} and \eqref{mel14},
\begin{align}\label{mel16a}
	\hP(k)&
	=-\Res\lrsqpar{f^*(s);s=-\rho-\chi_k}.
\end{align}
Consequently,
for $k\neq0$,
\eqref{mel4} yields, 
with $s=-\rho-\chi_k$,
\begin{align}\label{mel16}
	\hP(k)\log 2
	&=g^*(-\rho-\chi_k)
	+ f(1)\int_0^1 \lrpar{1+(2^\rho-1)\gf(u)} 
	(1+u)^{-\rho-\chi_k-1}\dd u,  
\end{align}
which is \eqref{c4B} in this case. For $k=0$, by the expansions
$\zeta(z)=(z-1)\qw+O(1)$ and $1-2^z=-(\log
2)z-\frac{(\log2)^2}{2}z^2+O(z^3)$ for small $|z|$, and the relations
\eqref{mel12f} and \eqref{mel16a}, we have
\begin{align}
	(1-2^{z})f^*(-\rho+z)
	&=(1-2^{z})\Bigpar{\frac{1}{\log2}z^{-2}-\hP(0)z\qw+O(1)}
	\notag\\
	&=-z\qw-\frac{\log 2}{2}+\hP(0)\log 2+O(z).\label{mel20}
\end{align}
Hence,
\begin{align}\label{mel21}
	\hP(0) = \frac12+\frac{1}{\log 2} 
	\bigpar{(1-2^{\rho+s})f^*(s)\big|_{s=-\rho}}\FV,
\end{align}
and \eqref{c4B} for $k=0$ follows from \eqref{mel4}.

Finally, \eqref{mel6} shows that $g_1(n)=n^\rho$ for even $n$. Hence,
if $g(n)=n^\rho$ for even $n$, then $g_2(2m)=0$ for $m\ge1$, and then
$Q(n)=Q_2(n)=0$ for $n\ge1$ by \refE{Eodd}, which proves the final
claim of the theorem.
\end{proof}

\begin{example}\label{EB}
Let $\ga=\gb=2$ (so $\rho=2$), and let $g(n)=n^2$, $n\ge2$. Then
\refT{TB1} applies, with $g_0(n)=0$. Hence, $Q(n)=0$, and \eqref{tb1}
yields
\begin{equation}\label{ve0}
	f(n)=
	n^\rho\log_2 n +n^{\rho }P(\log_{2}n), \qquad n\ge1.
\end{equation}
Furthermore, \eqref{majs} shows that, at least for $\Re s<-2$,
\begin{align}
	g^*(s) = \frac{D(-s-1)}{s(s+1)},
\end{align}
where \eqref{D2} yields, recalling $g(0)=g(1)=0$,
\begin{align}
	D(s)=4+2^{-s}+\sum_{n\ge3}2n^{-s}=2\zeta(s)+2-2^{-s},
    \qquad \Re s >1.
\end{align}
Accordingly,
\begin{align}\label{ve1}
	g^*(s)
	&=\frac{1}{s(s+1)}\Bigpar{2\zeta(-s-1)+2-2^{s+1}}.
\end{align}
This shows that $g^*(s)$ is meromorphic in $\bbC$, with poles at
$0,-1$ and $-2$. The finite value $g^*(-2)\FV=\gamma-\frac{3}4$,
where $\gamma$ is Euler's constant. Hence, if for simplicity $f(1)=0$,
the Fourier series of $P(t)$ is, by \eqref{c4B} and \eqref{ve1},
\begin{align}\label{ve2}
	P(t) = \frac{4\gamma-3}{4\log 2}+\frac12
	+\frac{1}{2\log 2}\sum_{k\neq 0}\frac{4\zeta(1+\chi_k)+3}
	{(1+\chi_k)(2+\chi_k)}\,e^{2k\pi i t}.
\end{align}

This sequence $f(n)$ with $g(n)=n^2$ and $f(1)=0$ is not in OEIS, but
the sequence $f_{\A022560}(n)=\A022560(n-1)$ discussed in
\refE{E22-ns} satisfies the recurrence $\Lambda_{2,2}[f]=
\tr{\frac14n^2}$ with $f_{\A022560}(1)=0$. Hence,
$4f_{\A022560}(n)=f(n)+n^2-S_{2,2}(n)$ with $f(n)$ as above. The
formulas \eqref{f022560}--\eqref{P022560} now follow from \eqref{ve0}
and \eqref{ve2} together with \eqref{hp022}.
\end{example}

\section{A series representation for $\hat{P}(k)$}
\label{AHat-Pk} 

We prove \eqref{E:int-phi} in Remark~\ref{R:int-phi}. For notational
simplicity we consider the case $k=0$; the general case is the same,
with $\rho$ replaced by $\rho+\chi_k$ below. (Note that
$2^{\rho+\chi_k}=2^\rho=\ga+\gb$ for all $k\in\bbZ$.) Consider the
integral
\begin{align}
	J := \int_0^1 \frac{1+(\alpha+\beta-1)\varphi(t)}
	{(1+t)^{\rho+1}}\,\dd t
	= \frac{1-2^{-\rho}}{\rho} + (\alpha+\beta-1)J(1)
\end{align}
where
\begin{align}
	J(m) := \int_0^1 \frac{\varphi(t)}
	{(m+t)^{\rho+1}}\,\dd t.
\end{align}
We now express $J(1)$ in a series form as follows. First, by applying 
the recursive definition \eqref{b4} of $\varphi$:
\begin{align}
	\varphi(t) 
	= \begin{cases}
		\frac{\beta}{\alpha+\beta}\,\varphi(2t),
		&\text{if }0\le t\le\frac12;\\
		\frac{\alpha}{\alpha+\beta}\,\varphi(2t-1)
		+\frac{\beta}{\alpha+\beta},
		&\text{if }\frac12\le t\le1,
	\end{cases}
\end{align}
we obtain, for any $m\ge1$, 
\begin{align}
	J(m) &:= \int_0^1 \frac{\varphi(t)}
	{(m+t)^{\rho+1}}\,\dd t\notag\\
	&= \beta\int_0^1 
	\frac{\varphi(t)}{(2m+t)^{\rho+1}}\,\dd t
	+ \alpha\int_0^1 
	\frac{\varphi(t)}{(2m+1+t)^{\rho+1}}\,\dd t
	+ \frac{\gb}{\ga+\gb}
	\int_{\frac12}^1\frac1{(m+t)^{\rho+1}}\,\dd t\notag\\
	&=: \beta J(2m)+\alpha J(2m+1) +K(m),
\end{align}
say. Iterating this gives, for any $N\ge0$,
\begin{align}\label{cc4}
	J(1) = \sum_{0\le m<N}
	\sum_{0\le j<2^m}\alpha^{\nu(j)}\beta^{m-\nu(j)}K(2^m+j)
    +\sum_{0\le j<2^N}\alpha^{\nu(j)}\beta^{N-\nu(j)}J(2^N+j).
\end{align}
where $\nu(j)$ denotes the number of 1's in $j$'s binary expansion.
Since $|J(m)|,|K(m)| = O(m^{-\rho-1})$, the last sum in \eqref{bb4}
is $O\bigpar{(\ga+\gb)^N2^{-(\rho+1)N}}=O\bigpar{2^{-N}}$ and
similarly the inner sum in the double sum is $O(2^{-m})$; hence we
can let $N\to\infty$, which yields
\begin{align}
	J(1) = \sum_{m\ge0}\sum_{0\le j<2^m}
	\alpha^{\nu(j)}\beta^{m-\nu(j)}K(2^m+j).
\end{align}
Now
\begin{align}
	K(m) 
	=\frac{\beta}{\alpha+\beta}
	\int_{\frac12}^1\frac1{(m+t)^{\rho+1}}\,\dd t
	= \frac{\beta}{\rho(\alpha+\beta)}
	\Lpa{\frac1{(m+\tfrac12)^{\rho}}
	-\frac1{(m+1)^{\rho}}}.
\end{align}
Thus
\begin{align}
	&\int_0^1 \frac{1+(\alpha+\beta-1)\varphi(t)}
	{(1+t)^{\rho+1}}\,\dd t \notag\\
	&\quad = \frac{1-2^{-\rho}}{\rho} 
	+ \frac{\beta(\alpha+\beta-1)}{\rho(\alpha+\beta)}
	\sum_{m\ge0}\sum_{0\le j<2^m}
	\alpha^{\nu(j)}\beta^{m-\nu(j)}
	\Lpa{\frac1{(2^m+j+\tfrac12)^{\rho}}
	-\frac1{(2^m+j+1)^{\rho}}}\notag\\
	&\quad= \frac{\alpha+\beta-1}{\rho(\alpha+\beta)}
	\llpa{1+\beta \sum_{m\ge0}\sum_{0\le j<2^m}
	\alpha^{\nu(j)}\beta^{m-\nu(j)}
	\Lpa{\frac1{(2^m+j+\tfrac12)^{\rho}}
	-\frac1{(2^m+j+1)^{\rho}}}}.
\end{align}
The double sum can be converted into a single one as follows. 
\begin{align}
	&\sum_{m\ge0}\sum_{0\le j<2^m}
	\alpha^{\nu(j)}\beta^{m+1-\nu(j)}
	\Lpa{\frac1{(2^m+j+\tfrac12)^{\rho}}
	-\frac1{(2^m+j+1)^{\rho}}}\notag\\
	&\qquad= 2^\rho \sum_{m\ge0}\sum_{0\le j<2^m}
	\alpha^{\nu(j)}\beta^{m+1-\nu(j)}
	\Lpa{\frac1{(2^{m+1}+2j+1)^{\rho}}
	-\frac1{(2^{m+1}+2j+2)^{\rho}}}\notag\\
	&\qquad= 2^\rho \sum_{m\ge1}\sum_{0\le j<2^{m-1}}
	\alpha^{\nu(j)}\beta^{m-\nu(j)}
	\Lpa{\frac1{(2^{m}+2j+1)^{\rho}}
	-\frac1{(2^{m}+2j+2)^{\rho}}} \notag\\
	&\qquad= 2^\rho\sum_{k\ge2}\frac{(-1)^k}{(k+1)^{\rho}}
	\,\alpha^{\nu(\tr{2^{L_k-1}\{k/2^{L_k}\}})}
	\beta^{{L_k}-\nu(\tr{2^{L_k-1}\{k/2^{L_k}\}})}.
\end{align}

\section{Recurrences with minimisation or maximisation}
\label{Aminmix}

Consider the class of sequences satisfying recurrences of the form
\begin{align}\label{ol1}
	u(n)
	=\min_{1\le k\le
	\left\lfloor \frac{n}{2}\right\rfloor}
	\left\{ \alpha u(k)+\beta u(n-k)\right\} 
	\qquad (n\ge 2),
\end{align}
with $u(1)=1$. 

This was studied in \cite{Chang2000}, where it was shown that if
$\alpha$ and $\beta$ are positive integers with $\beta \ge
\alpha$, then the minimum in \eqref{ol1} is reached at
$k=\left\lfloor \frac{n}{2}\right\rfloor $ and the solution is given 
by 
\begin{equation}\label{e1}
    u(n)
	=1+(\alpha +\beta -1)\sum_{1\le j<n}w(j),  
\end{equation}
where
\begin{equation}\label{e10}
	w(j):=\alpha ^{L_{j}-\nu _{0}(j)}\beta^{\nu _{0}(j)}, 
\end{equation}
and $\nu_{0}(j)$ denotes the number of zeros in the binary expansion 
of $j$. We will extend this result and the arguments in 
\cite{Chang2000}, and prove the following.

\begin{prop}\label{PAA}
Let $\ga,\gb>0$ be real numbers such that either
\begin{romenumerate}

\item \label{PAA1}
$\beta \ge \alpha $ and $\beta \ge 1$, or 

\item  \label{PAA2}
$\alpha \ge\beta $ and $\alpha +\beta \le1$.
\end{romenumerate}
Then the minimum in \eqref{ol1} is reached at $k=\left\lfloor
\frac{n}{2}\right\rfloor $. Hence, \eqref{ol1} reduces to
$\Lambda_{\ga,\gb}[u]=0$, and thus \eqref{ol1} is solved by
$u(n)=S_{\ga,\gb}(n)$.
\end{prop}

\begin{proof}
We note first that $S_\gab(n)$ is given by the formula in \eqref{e1}
for any $\ga,\gb$. This follows by \eqref{Sgagb}, \eqref{b16}, and
\eqref{bb10} in \refS{S:recurrence}, or by the proof in
\cite{Chang2000}.

It thus remains to show that if $u(n)$ is defined by \eqref{e1},
then
\begin{align}\label{ol2}
	\alpha u(k)+\beta u(n-k)
	\ge \alpha u\lpa{\left\lfloor \tfrac{n}{2}
	\right\rfloor}
	+\beta u\lpa{\left\lceil \tfrac{n}{2}\right\rceil}
\end{align}
for $1\le k<\left\lfloor \frac{n}{2}\right\rfloor $. By (\ref{e1}),
the difference between the two sides of \eqref{ol2} is
\begin{align}\label{ol3}
	(\alpha +\beta -1)\left(\beta 
	\sum_{\left\lceil \frac{n}{2}\right\rceil\le j<n-k}
	w(j)-\alpha \sum_{k\le j<\left\lfloor\frac{n}{2}\right\rfloor}
	w(j)\right).
\end{align}
To prove that this is non-negative, we will show that
\begin{equation}\label{e2}
	\beta \sum_{\left\lceil \frac{n}{2}\right\rceil\le j<n-k}
	w(j)\ge \alpha \sum_{k\le j<\left\lfloor
	\frac{n}{2}\right\rfloor}w(j)
\end{equation}
if \ref{PAA1} holds, and that \eqref{e2} holds with the inequality
reversed if \ref{PAA2} holds. The key is the following claim:\\
\begin{alphenumerate}\em
\item \label{PAAa}
If\/ $\beta \ge \alpha $ and $\beta \ge 1$, then
\begin{align}\label{paaa}
	\beta w(n+2^{j})
	\ge \alpha w(n),\quad \text{for all }
	n\ge 1\text{ and }j\ge 0.  
\end{align}
\item\label{PAAb}
If\/ $\alpha \ge \beta $ and $\beta \le 1$, then
\begin{align}\label{paab}
	\beta w(n+2^{j})\le \alpha w(n),\quad 
	\text{for all }n\ge 1\text{ and }j\ge 0.
\end{align}
\end{alphenumerate}
\emph{Proof of the claim.} From (\ref{e10}), we have
\begin{align}\label{e5}
	\frac{\beta w(n+2^{j})}{\alpha w(n)}&
	=\left( \frac{\beta }{\alpha }\right)^{
	1+\nu _{0}(n+2^{j})-\nu _{0}(n)}
	\alpha^{L_{n+2^{j}}-L_{n}} \notag\\
	&=\left( \frac{\beta }{\alpha }\right)^{
	1+\nu _{0}(n+2^{j})-\nu _{0}(n)-(L_{n+2^{j}}-L_{n})}
	\beta^{L_{n+2^{j}}-L_{n}}. 
\end{align}
It is easily seen that
\begin{equation}\label{e6}
	1+\nu _{0}(n+2^{j})-\nu _{0}(n)
	\ge L_{n+2^{j}}-L_{n}\ge 0.
\end{equation}
Both parts of the claim thus follow from \eqref{e5}. \qed

To show (\ref{e2}), or its converse in case \ref{PAA2}, which will
then complete the proof of the proposition, we combine the claim
above with a pairing between the sets, with
$m:=\lrfloor{\frac{n}{2}}-k$,
\begin{equation}
	\left\{\left\lfloor \frac{n}{2}\right\rfloor -m,
	\cdots ,\left\lfloor \frac{n}{2}\right\rfloor -1\right\} 
	\quad \text{and}\quad \left\{ \left\lceil 
	\frac{n}{2}\right\rceil,
	\cdots ,\left\lceil \frac{n}{2}\right\rceil +m-1\right\},
\end{equation}
such that the difference between the elements of each pair is a power 
of $2$. In other words, the proof is completed by applying the 
following lemma (with a translation).
\end{proof}

\begin{lemma}
Let $n\ge1$ and
\begin{align}
	A &:=\{1,2,\cdots ,n\}, \\
	C_{1} &:=\{n+1,n+2,\cdots ,2n\}, \\
	C_{2} &:=\{n+2,n+3,\cdots ,2n+1\}.
\end{align}
There exist one-to-one mappings
\begin{equation}
	h_{1}:A\rightarrow C_{1}\quad 
	\text{and}\quad h_{2}:A\rightarrow C_{2},
\end{equation}
such that for each $k\in A$ there exist $j_{1},j_{2}$ with
\begin{align}
	h_{1}(k)
	&=k+2^{j_{1}},\quad 
	\text{for some } j_1\ge0, \label{Ah1}\\
	h_{2}(k)
	&=k+2^{j_{1}},\quad \text{for some } j_2\ge0.\label{Ah2}
\end{align}
\end{lemma}

\begin{proof}
We prove the existence of $h_1$ by induction. Write $n=2^{\ell}+j$
for some $\ell\ge 0$ and $0\le j< 2^{\ell}$. (Thus,
$\ell=L_n$.) We want to show that there exist a one-to-one mapping
\begin{equation}
	h_{1}:\{1,\cdots ,2^{\ell}+j\}
	\rightarrow \{2^{\ell}+j+1,\cdots ,2^{\ell+1}+2j\},
\end{equation}
such that for each $1\le k \le 2^{\ell}+j$, \eqref{Ah1} 
holds.

If $j=0$, we simply define $h_1(k):=k+2^\ell$.

If $j\ge1$, we first define $h_1(k)$ for $k\le 2j$ by
$h_1(k):=k+2^{\ell+1}$. This gives a mapping from $\{1,2,\cdots
,2j\}$ to $\{2^{\ell+1}+1,\cdots ,2^{\ell+1}+2j\}$. We remove these
two blocks, and it remains to define a one-to-one mapping satisfying
\eqref{Ah1} between $\{2j+1,\cdots ,2^{\ell}+j\}$ and
$\{2^{\ell}+j+1,\cdots ,2^{\ell+1}\}$. By subtracting $2j$ from each
term, it is equivalent to showing that there exist such a mapping
from $\{1,\cdots ,m\} $ to $\{m+1,\cdots ,2m\}$, where
$m=2^{\ell}-j$; this is true by the induction hypothesis.

The proof of the existence of $h_2$ is similar (but with 
$n+1=2^\ell+j$); we omit the details.
\end{proof}

If we replace $\min$ by $\max$ in \eqref{ol1}, we obtain a similar
result; see also \cite{AMM-2009} for $\ga>\gb=1$.

\begin{prop}\label{PAmax}
Let $\ga,\gb>0$ be real numbers such that $\alpha \ge \beta$,
$\beta \le 1$, and $\alpha+\beta \ge 1$. Then the maximum
in the recursion
\begin{align}\label{olmax}
	u(n)=\max_{1\le k\le 
	\left\lfloor \frac{n}{2}\right\rfloor}
	\left\{ \alpha u(k)+\beta u(n-k)\right\} 
	\qquad (n\ge 2),
\end{align}
with $u(1)=1$, is attained at $k=\left\lfloor
\frac{n}{2}\right\rfloor $. Hence, \eqref{olmax} reduces to
$\Lambda_{\ga,\gb}[u]=0$, and thus \eqref{olmax} is solved by
$u(n)=S_{\ga,\gb}(n)$.
\end{prop}

\begin{proof}
The proof above shows that under these conditions, \eqref{paab}
holds, and hence \eqref{e2} holds in the opposite direction. Thus the
difference in \eqref{ol3} is $\le0$. (We do not obtain a second case
with $\ga+\gb<1$; we then would need \eqref{paaa}, but the condition
$\gb\ge1$ in \ref{PAAa} above is incompatible with $\ga+\gb<1$.)
\end{proof}

The results above can be extended to the recurrences of the form
\begin{align}
	u(n)=\min_{1\le k\le
	\left\lfloor \frac{n}{2}\right\rfloor
	}\left\{ \alpha u(k)+\beta u(n-k)\right\} 
	+c\qquad (n\ge 2),
\end{align}
with the same conditions as above on $\ga$ and $\gb$; see 
\cite{Fredman1974,Hwang2003} for more general versions. The solution 
for the recurrence
\begin{align}
	u(n)
	=\ga u\left( \left\lfloor \frac{n}{2}\right\rfloor \right) 
	+\gb u\left( \left\lceil \frac{n}{2}\right\rceil \right) +c
\end{align}
is, with $w(j)$ as in \eqref{e1}--\eqref{e10} above (see also
\eqref{l4.1})
\begin{equation}
	u(n)
	=u(1)+\bigpar{(\alpha +\beta -1)u(1)+c} 
	\sum_{1\le j<n}w(j).
\end{equation}

\section{Nowhere differentiability of $P_{\text{A006581}}(t)$}
\label{AE22-ns}

We prove in this appendix the fractal nature of the periodic function
$P(t)$ arising from A006581 (discussed in \refE{E22-ns} with Fourier
expansion given in \eqref{E:P-A006581}), namely,
$\Lambda_{2,2}[f]=g$ where $g(n) := \{\frac{n}2\}(n-1)$:
\begin{align}
    f(n) = n^2P(\log_2n),	
\end{align}
where, with $\bar{g}(x) := g(x)/x^2$,
\begin{align}
	P(t) = \sum_{m\in\mathbb{Z}} 4^{-m-t}g(2^{m+t})
	=\sum_{m\in\mathbb{Z}} \bar{g}(2^{m+t}).
\end{align}
Here $g(x)$ is extended from $g(n)$ as in \eqref{b2} with $\varphi(t)
=t$: 
\begin{align}
	\bar{g}(x) = \frac1{2x^2}\times\begin{cases}
		\{x\}\tr{x}, 
		&\text{if }\tr{x} \text{ is even};\\
		(1-\{x\})(\tr{x}-1),
		&\text{if }\tr{x} \text{ is odd}.
	\end{cases}
\end{align}
Since $\bar{g}(2^{m+t})=0$ for $t\in[0,1]$ and $m\le 0$, we have 
\begin{align}
	P(t) = \sum_{m\ge1} \bar{g}(2^{m+t})
\qquad (0\le t\le 1).
\end{align}

The method of proof used here to prove the nowhere differentiability 
of $P$ is standard and similar to that for the Takagi function given 
in the survey paper \cite{Allaart2011}.

Let $t\in[0,1)$, and define 
\begin{align}
    \tau_n := \log_2\frac{\tr{2^{n+t}}}{2^n}
	\eqtext{and}
	\tau_n' := \log_2\frac{\tr{2^{n+t}}+1}{2^n}.
\end{align}
To prove that $P$ does not have a finite derivative at $t$, it 
suffices to show that the sequence
\begin{align}\label{E:P-ratio}
    \frac{P(\tau_n)-P(\tau_n')}
	{\tau_n-\tau_n'}
	= \sum_{1\le m\le n}\frac{\bar{g}(2^{m+\tau_n})-
	\bar{g}(2^{m+\tau_n'})}{\tau_n-\tau_n'}
\end{align}
does not converge to a finite limit. Here we used the relation 
$\bar{g}(2^{m+\tau_n})=\bar{g}(2^{m+\tau_n'})=0$ for $m>n$. Now 
for $\theta\in[0,1)$
\begin{align}
	\left\{
	\begin{aligned}
		\bar{g}(2k+\theta)
		&= \frac{k\theta}{(2k+\theta)^2}\\
		\bar{g}(2k+1+\theta)
		&= \frac{k(1-\theta)}{(2k+1+\theta)^2}
	\end{aligned}
	\right.
\end{align}
so that, taking the right derivative at integer points here and below,
\begin{align}\label{ccg'}
	\left\{
	\begin{aligned}
		\bar{g}'(2k+\theta)
		&= \frac{k(2k-\theta)}{2(2k+\theta)^3}\\
		\bar{g}'(2k+1+\theta)
		&= -\frac{k(2k+3-\theta)}{(2k+1+\theta)^3}
.	\end{aligned}
	\right.
\end{align}
If $1\le m\le n$, then
\begin{align}\label{cc5}
	\floor{2^{m+\tau_n}} \le 2^{m+\tau_n}
	\le 2^{m+t}
	< 2^{m+\tau'_n}
	\le \floor{2^{m+\tau_n}}+1.
\end{align}
It follows that $h(x):=\bar g(2^x)$ is infinitely differentiable on
$[m+\tau_n,m+\tau_n']$, and it is easily seen from \eqref{ccg'} that
$h''(x)=O(1)$ (uniformly in $m$ and $n$). We have, for some
$\tau''_n\in(\tau_n,\tau'_n)$,
\begin{align}
	\frac{\bar{g}(2^{m+\tau_n})-
	\bar{g}(2^{m+\tau_n'})}{\tau_n-\tau_n'}
	&=\frac{h(m+\tau_n)-h(m+\tau_n')}{\tau_n-\tau_n'}
	= h'(m+\tau''_n)\notag\\
	&=h'(m+t)+O(|\tau''_n-t|)
	=h'(m+t)+O(|\tau'_n-\tau_n|)\notag\\
	&=h'(m+t)+O(2^{-n}).
\end{align}
Hence, \eqref{E:P-ratio} implies that
\begin{align}
    \frac{P(\tau_n)-P(\tau_n')}{\tau_n-\tau_n'}
	= \sum_{1\le m\le n}h'(m+t)+O(n2^{-n}),
\end{align}
and thus, if $P$ is differentiable at $t$, then the sum
\begin{align}\label{cc6}
    \sum_{ m\ge 1}h'(m+t)
\end{align}
converges (and equals $P'(t)$).

On the other hand, it follows easily from \eqref{ccg'} that if
$x\ge2$, then $x\bar g'(x)\ge \frac19$ for even $\tr{x}$ and $x\bar
g'(x)\le -\frac14$ for odd $\tr{x}$. Hence, $|h'(m+t)|=|2^{m+t}\bar
g'(2^{m+t})\log 2|\ge \frac{\log2}9$ for all $m\ge1$. Consequently,
the sum \eqref{cc6} diverges for any $t$, and thus $P$ is nowhere
differentiable.

Moreover, we note that if $2^t$ is a dyadic rational, then for all
large $m$, $2^{m+t}$ is an even integer, and thus
$h'(m+t)\ge\frac{\log 2}{9}$. Hence, in this case the sum \eqref{cc6}
diverges to $+\infty$, and thus so does
$\bigpar{P(\tau_n)-P(\tau'_n)}/(\tau_n-\tau'_n)$ in
\eqref{E:P-ratio}; consequently, $P$ is not Lipschitz.

We do not know whether $P$ is H{\"o}lder continuous, and leave that
as an open problem. Note that \refL{LProp1} does not apply since
\eqref{c3} does not hold.

\bibliographystyle{abbrv}
\bibliography{dac-ab-2022-oct}

\end{document}